\documentclass[
reprint,
showkeys,
aps,
prx, twocolumn,
floatfix,
superscriptaddress
]{revtex4-2}

\usepackage{amsmath,amssymb}
\usepackage{graphicx}
\usepackage{bm}
\usepackage[dvipsnames]{xcolor}
\usepackage{subcaption}
\usepackage{multirow}

\usepackage[utf8]{inputenc}
\usepackage[cal=boondoxo]{mathalfa}
\usepackage[ruled, linesnumbered]{algorithm2e}
\usepackage[english]{babel}
\usepackage{mathtools}

\usepackage{xpatch}
\makeatletter
\xpatchcmd{\@ssect@ltx}{\@xsect}{\protected@edef\@currentlabelname{#8}\@xsect}{}{}
\xpatchcmd{\@sect@ltx}{\@xsect}{\protected@edef\@currentlabelname{#8}\@xsect}{}{}
\makeatother

\usepackage{hyphenat}
\usepackage{verbatim}

\usepackage{amsthm}
\usepackage[colorlinks,linkcolor=darkblue,citecolor=darkblue,linktocpage,hypertexnames=true,pdfpagelabels]{hyperref}
\usepackage[capitalise]{cleveref}
\crefname{observation}{Observation}{Observation}
\newtheorem{theorem}{Theorem}
\newtheorem{corollary}[theorem]{Corollary}
\newtheorem{lemma}[theorem]{Lemma}
\newtheorem{definition}[theorem]{Definition}
\newtheorem{proposition}[theorem]{Proposition}
\newtheorem{example}[theorem]{Example}
\newtheorem{remark}[theorem]{Remark}

\usepackage{tikz}
\usetikzlibrary{shapes.geometric, calc, math, arrows, decorations, positioning, circuits.logic.US}
\usepackage{tkz-euclide}
\usetikzlibrary{backgrounds}
\usepackage{tikz-cd}
\tikzset{shifted path/.style args={from #1 to #2 by #3}{insert path={
let \p1=($(#1.east)-(#1.center)$),
\p2=($(#2.east)-(#2.center)$),\p3=($(#1.center)-(#2.center)$),
\n1={veclen(\x1,\y1)},\n2={veclen(\x2,\y2)},\n3={atan2(\y3,\x3)} in
(#1.{\n3+180+asin(#3/\n1)}) to (#2.{\n3-asin(#3/\n2)})
}}}

\usepackage{enumitem}

\newcommand{\mc}{\mathcal}
\newcommand{\mr}{\mathrm}

\definecolor{darkblue}{rgb}{0.0, 0.0, 0.55}
\definecolor{amethyst}{rgb}{0.6, 0.4, 0.8}
\definecolor{ao}{rgb}{0.0, 0.0, 1.0}
\definecolor{amber}{rgb}{1.0, 0.75, 0.0}
\definecolor{amaranth}{rgb}{0.9, 0.17, 0.31}

\newcommand{\Del}{\Delta}
\newcommand{\Shift}{\Gamma}
\newcommand{\enc}{\mr{enc}}
\newcommand{\dec}{\mr{dec}}
\newcommand{\phys}{\mr{P}}
\newcommand{\physSet}{\mc{Ph}}
\newcommand{\degeneracy}{d}
\newcommand{\simul}{\mr{sim}}
\newcommand{\simSet}{\mc{Sim}}
\newcommand{\Cone}{\mr{Cone}}

\newcommand{\Scale}{\mr{Scale}}

\makeatletter
\def\l@subsubsection#1#2{}
\makeatother

\begin{document}

\title{The Structure of Emulations in Classical Spin Models: Modularity and Universality}

\author{Tobias Reinhart}
\email{tobias.reinhart@uibk.ac.at}
\affiliation{Institute for Theoretical Physics, University of Innsbruck,\\ Technikerstr.\ 21a,\ A-6020 Innsbruck, Austria}

\author{Benjamin Engel}
\affiliation{Institute for Theoretical Physics, University of Innsbruck,\\ Technikerstr.\ 21a,\ A-6020 Innsbruck, Austria}

\author{Gemma De les Coves}
\affiliation{Departament d'Enginyeria, Universitat Pompeu Fabra, Carrer T\`anger 122, 08018 Barcelona, Catalonia}
\affiliation{ICREA, Instituci\'o Catalana d'Estudis i Recerca Avan\c{c}ats, Passeig Llu\'is Companys 23, 08010 Barcelona, Catalonia }

\begin{abstract}
The theory of spin models intersects with condensed matter physics, complex systems, graph theory, combinatorial optimization, computational complexity and neural networks. Many ensuing applications rely on transforming complicated spin models to simpler ones. What is the structure of such transformations? Here, we provide a framework to study and construct emulations between spin models. A spin model is a set of spin systems, and emulations are efficiently computable families of simulations with arbitrary energy cut-off, where a source spin system simulates a target system if, below the cut-off, the target Hamiltonian is encoded in the source Hamiltonian. We prove that emulations preserve important properties, as they induce reductions between computational problems such as computing ground states, approximating partition functions and approximate sampling from Boltzmann distributions. Emulations are modular (they can be added, scaled and composed), and allow for universality, i.e.\ under emulations, certain spin models have maximal reach. We prove that a spin model is universal if and only if it is scalable, closed and functional complete. Because the characterization is constructive, it provides a step-by-step guide to construct emulations. We prove that the 2d Ising model with fields is universal, for which we also provide two new crossing gadgets. Finally, we show that simulations can be computed by linear programs. While some ideas of this work are contained in [G.\ De las Cuevas and T.\ S.\ Cubitt, Science \textbf{351}, 1180 (2016)], we provide new definitions and theorems. This framework provides a toolbox for applications involving emulations of spin models.  
\end{abstract}

\maketitle
\tableofcontents

\section{Introduction}

\begin{figure*}[th]
    \centering
    \includegraphics[width=0.95\textwidth]{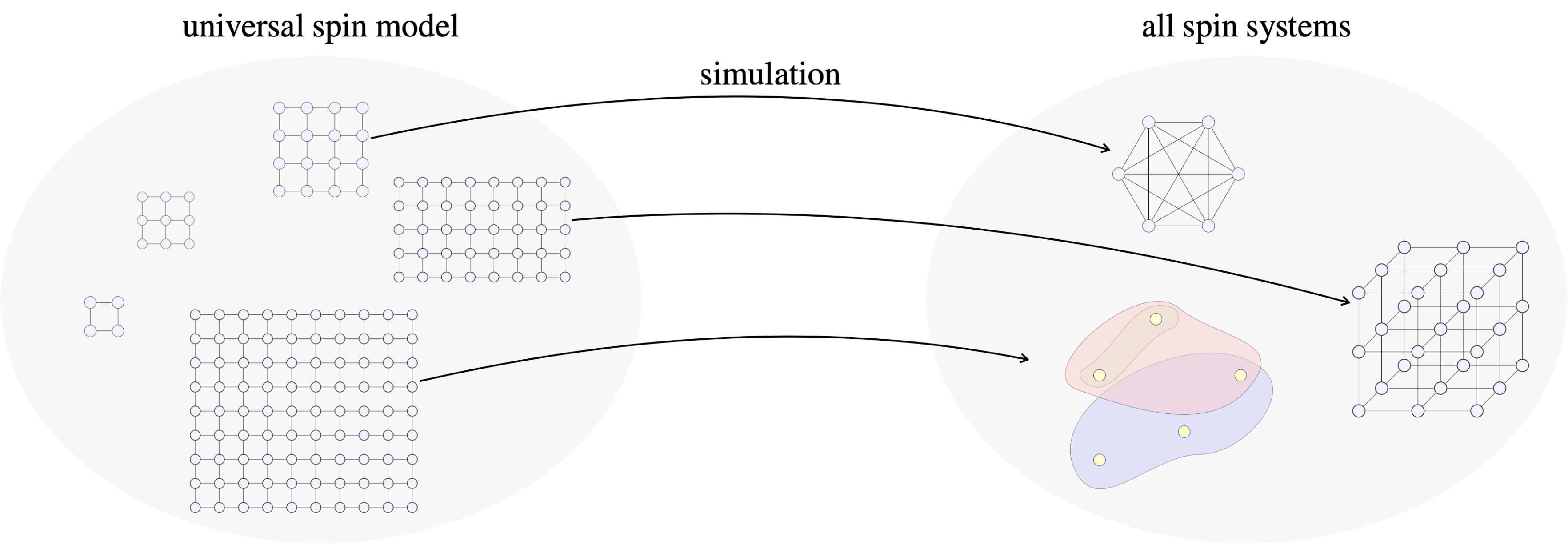}
    \caption{A spin model (i.e.\ a set of spin systems) is universal if it emulates the set of all spin systems. Emulations are efficiently computable simulations with arbitrary cut-off, where simulation is defined between spin systems.  
    }
    \label{fig:intro-universal}
\end{figure*}

While spin models were introduced as toy models for magnetism and thus investigated within condensed matter physics, their study rapidly gained significance for other disciplines. 
Via lattice gauge theories, spin models intersect with high-energy physics \cite{Ko79}. 
More recently, spin models have percolated complex systems (see e.g.\ \cite{So19}), 
neural networks via Hopfield networks and Boltzmann machines, 
and to a lesser extent, formal language theory (see \cite{St21,Re21c} and references therein, as well as \cite{De19f}). 
The study of spin models also intersects with graph theory and computational complexity theory, originating from Barahona's proof of the NP hardness of the 2d Ising with fields \cite{Ba82b}. The latter implies that NP-complete problems can be transformed to the ground-state energy problem of the Ising model \cite{Lu14b}, 
and in particular that the ground-state energy problem of any spin model can be transformed to that of the Ising model. This gives rise to combinatorial optimization solvers, e.g.\ via quantum annealing (see \cite{Le15,Ng23} and \cite{Mo22} for an overview). This intersection thus results in an important application, at the core of which is a notion of transformation of spin models.

Depending on the context, it is meaningful to define transformations of spin models one way or another. 
In condensed matter physics, such transformations may preserve the relevant properties of phase transitions, such as high to low temperature dualities (such as the Kramers--Wannier duality, see e.g.\ \cite{We71}). 
In contrast, with a computational complexity mindset, transformations ought to be reductions that map between solutions of the ground-state energy problems, allowing to compute ground states of the target model from ground states of the source. These transformations, however, need not behave well on high energy levels, so it may not be possible to recover the partition function of the target from that of the source. In other words, such notion of transformation would be too general (or: weak) for the latter purpose. Instead, with a graph theory mindset, we could consider transformations of spin models that relabel the spins, inspired by graph isomorphisms. A less strong notion of transformation could be motivated by graph minor operations (such as vertex deletion, and edge deletion and contraction). As it turns out (and we shall flesh out), such transformations would preserve several graph properties and, as a consequence, limit the reach of spin models unnecessarily, where the reach of a spin model is the set of all spin models to which it can be transformed. For example, in combinatorial optimization approaches, spin models on planar graphs could only solve optimization problems with planar connectivity. Such a notion of transformation would thus be too narrow (or: strong). Yet another mindset could be provided by neural networks, where transformations of spin models should approximately preserve Boltzmann distributions. 

What is common to all these considerations is that transformations be \emph{property-preserving}, i.e.\ such that in a given context, if a source spin model $\mc{S}$ can be transformed to a target spin model $\mc{T}$, then in every situation $\mc{T}$ can be replaced by $\mc{S}$ approximately with arbitrarily small error. Clearly, the minimal set of properties that must be preserved depends on the context. We are after a definition of transformation that applies to all contexts. 

A feature that appears in some of the above notions of transformations is \emph{universality}. A spin model is called universal if it can be transformed to any other spin model. Universal spin models thus have maximal reach. Some notions of transformations allow for universality; for example, reductions allow for certain spin models to be universal in the sense that their ground-state energy problems are NP-complete. This universality is the crucial fact in combinatorial optimization via quantum annealing. Another important example of universality is enabled by the transformations of Ref.\ \cite{De16b} (see also \cite{Ko18}). 

An implicit requirement in many of the above contexts is a certain sense of \emph{modularity}, meaning that transformations have a well-behaved enough structure to enable the (efficient) construction of new transformations, as well as their study and characterization. Modularity allows to construct (complicated) transformations by composing, adding and scaling (simple) transformations. Its theoretical and practical importance cannot be overstated. As we shall see, modularity allows to characterize universality. 

In short, we want a notion of transformation between spin models that 
\begin{enumerate}[label=(\roman*)]
\item is property-preserving for all contexts, 
\item allows for universality, and 
\item is modular. 
\end{enumerate}

\begin{figure*}[th]\centering
        \includegraphics[width=0.9\textwidth]{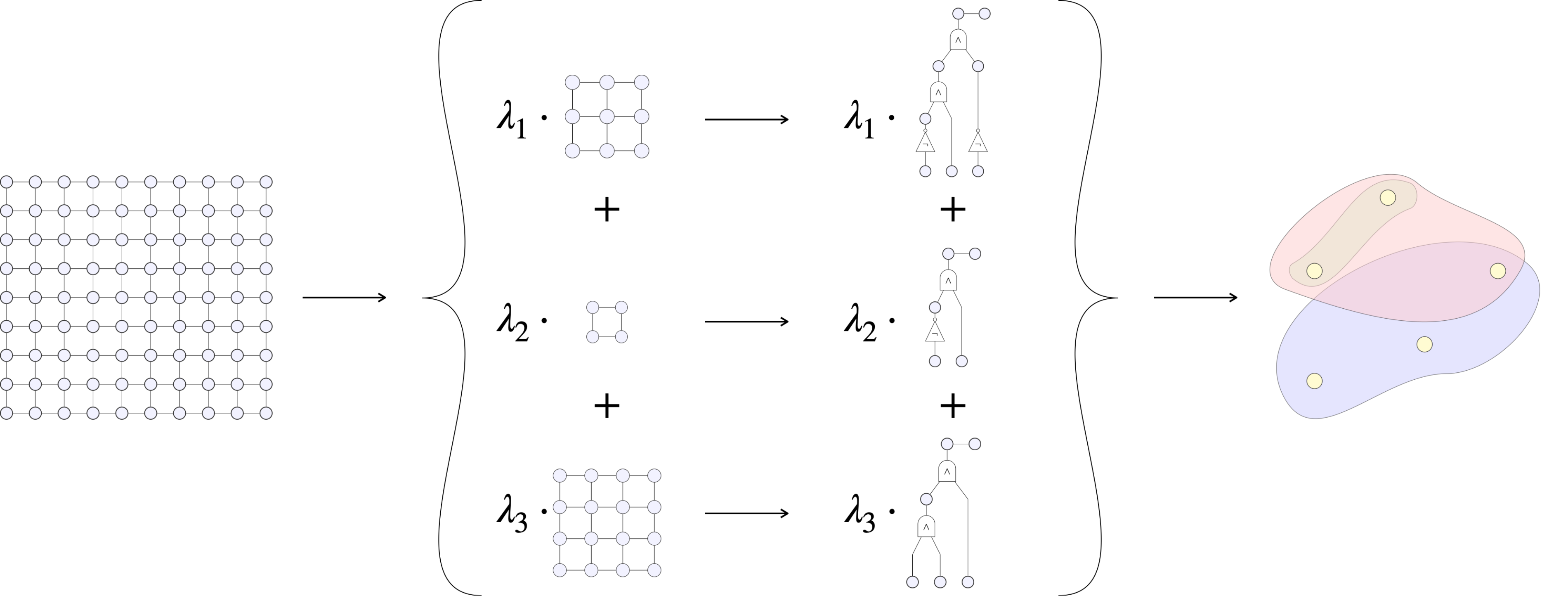}
    \caption{The simulation of a target spin system with a universal spin model can be constructed following the characterization of universality in terms of scalability, closure and functional completeness (\cref{thm:main}). 
    The target (right hand side) is decomposed into a linear combination of flag spin systems, 
    each of which is simulated by a source spin system from the universal spin model (by functional completeness), whose linear combination can in turn be simulated by the universal spin model (by closure and scalability). Because of modularity, this results in the desired simulation.}
    \label{fig:intro-modular}
\end{figure*}

In this paper, we provide a framework for \emph{emulations} between spin models. 
We show that emulations are a meaningful notion of transformation between spin models, as they satisfy the above conditions. Emulations are efficiently computable families of simulations with arbitrary cut-off, where a spin model is a set of spin systems, and a simulation between spin systems is defined as follows. If a source spin system $S$ simulates a target spin system $T$, then below an energy cut-off $\Del$ the Hamiltonians of $S$ and $T$ agree up to a global energy shift $\Shift$ and a local encoding of configurations of $T$ into configurations of $S$. 

We show that emulations are property-preserving, because, among others, they induce reductions between computational problems such as computing ground states, approximating partition functions and approximate sampling from Boltzmann distributions. They also allow for universality; in fact, we characterize universality (\cref{fig:intro-universal}) in terms of three properties of spin models: closure, scalability and functional completeness (\cref{fig:intro-modular}). The characterization is constructive and thus provides the means to explicitly emulate arbitrary targets. Finally, emulations are modular, as they can be composed, scaled and added. In summary, emulations preserve the relevant properties, are general enough to allow for universality and strong enough for modularity. 

We apply our machinery to the 2d Ising model with fields, where we do not only prove that it is universal, but provide a step-by-step guide to emulate arbitrary target models.
We also provide a new crossing gadget for the Ising model. Finally, we show that simulations can be computed by linear programs, and leverage this to provide yet another (even more economic) crossing gadget.

In summary, the framework presented in this paper allows to study and construct emulations, and provides a practical toolbox to this end. Any context (implicitly) involving emulations, such as quantum annealing approaches or neural network architectures, could benefit from this toolbox. More generally, translating the results of this paper to a more abstract language (e.g.\ the categorical framework of Ref.\ \cite{Go23}) may extend the reach of the present results and identify situations featuring universality with a similar characterization. This work could also shed light on the relation of this form of universality to universality classes, or more generally forms of emergence in complex systems (see e.g.\ \cite{So00}). We shall come back to these points in \nameref{sec:outlook}.

This paper is structured as follows. We first study spin systems and simulations (\cref{sec:spin sys}), and then turn to spin models and emulations (\cref{sec:spin models}). We derive consequences of emulations, including universality (\cref{sec:consequences}), and illustrate the framework by proving universality of the 2d Ising model with fields (\cref{sec:2d Ising}). Finally we show that simulations can be computed by linear programs (\cref{sec:lin prog}). We conclude and provide an outlook in \cref{sec:outlook}. The modularity proofs of simulations are presented in \cref{sec:modular}. 

In comparison to Ref.\ \cite{De16b}, this work focuses on the mathematical structure of simulations and emulations. It goes beyond the findings of \cite{De16b} generally in the thoroughness of definitions and theorems, and specifically in the following facts: 
\begin{enumerate}[label=(\roman*)]
\item 
We provide a new definition of simulation (\cref{sssec:def}) and carefully study its properties. 
We prove that simulations preserve spectra and ground states of spin systems (\cref{sssec:spectra}) and approximately preserve the partition function and Boltzmann distribution (\cref{sssec:thermo}).  
We also prove that isomorphisms, weak symmetries and spin type modifications induce simulations (\cref{sssec:symmetries}), and so do graph minors (\cref{sssec:minors}). Finally, we prove that simulations are modular, i.e.\ can be composed, scaled and added (\cref{{ssec:spin sys sim prop}}). 

\item We define emulations (\cref{sssec:def emu}) and prove that they can be composed and lifted to so-called hulls of spin models, implying that they are modular (\cref{sssec:comp lift emu}).  

\item We prove that a spin model is universal if and only if it is closed, scalable and functional complete (\cref{thm:main}). This characterization is new in two respects: 
the new definition of functional completeness (\cref{def:f.c.}), 
and its use of modularity. 
We also define local closure and show that it implies closure (\cref{thm:locally closed}). 

\item We prove that emulations induce polytime reductions between ground-state energy problems (\cref{ssec:gse}), 
partition function approximations (\cref{ssec:partfun}) and approximate sampling (\cref{ssec:sample}), where the latter two are in fact $\text{\sc PTAS}$ reductions. 
We also prove that these three problems are maximally hard for universal spin models. 

\item We apply the new characterization of universality to provide a new proof of universality of the 2d Ising with fields that is new in three respects: 
we prove functional completeness (\cref{ssec:f.c.Ising}), provide a new crossing gadget (\cref{ssec:gadget sim}) and prove local closure (\cref{ssec:l.c.Ising}). 
In addition, this new careful step-by-step proof can serve in practice to construct simulations.

\item The results of \cref{sec:lin prog} are new, namely that simulations can be computed by linear programs, resulting in a more economic crossing gadget (\cref{ssec:alternative crossing}). 

\end{enumerate}

\section{Spin Systems and Simulations}\label{sec:spin sys}

Here we define spin systems (\cref{ssec:spin sys}), simulations  (\cref{ssec:spin sys sim}), and prove that the latter are modular (\cref{ssec:spin sys sim prop}).

\subsection{Spin Systems}\label{ssec:spin sys}

We start by defining spin systems.
\begin{definition}[spin system]\label{def:spin sys}
A \emph{spin system} $S$ is a tuple $(q_S,V_S,E_S,J_S)$ where $q_S\in \mathbb{N}_{\geq 2}$, $V_S$ is a finite set, $E_S\subseteq \mc{P}(V_S)$  and 
\begin{equation}
J_S\colon E_S \rightarrow \{f:[q_S]^W \rightarrow \mathbb{R}_{\geq 0} \mid W \subseteq V_S \}
\end{equation}
such that for all $e\in E_S$ 
\begin{equation}
J_S(e)\colon [q_S]^e \rightarrow \mathbb{R}_{\geq 0},
\end{equation}
where $[q_S]\coloneqq \{1, \ldots, q_S\}$, $[q_S]^W$ denotes the set of functions of type $W \to [q_S]$ and $\mc{P}(V_S)$ the power set of $V_S$.
\end{definition}

Given a spin system $S$, we call $q_S$ its \emph{spin type}, $V_S$ its \emph{spins}, $G_S\coloneqq(V_S,E_S)$ its \emph{interaction hypergraph} and $J_S$ its \emph{local interactions} (thought of as $\{J_S(e) \mid e \in E_S \}$).
We require that there are no isolated spins, that is, for any $v\in V_S$ there exists at least one $e\in E_S$ with $v\in e$. Note that, if this not the case, it can be made so by adding any such edge $e$ with $J_S(e)$ the zero function, thereby yielding a spin system with equivalent local interactions and without isolated spins.

In essence, a spin system is a hyperedge-labelled hypergraph, where hyperedges determine which spins interact, and the hyperedge labels obtained from  $J_S$ are local functions that determine how they interact. 
Once the state of each spin is fixed by means of a \emph{spin configuration} 
\begin{equation}
    \vec{s}\in \mc{C}_S\coloneqq [q_S]^{V_S} , 
\end{equation}
we obtain a positive real number $J_S(e)(\vec{s}\vert_e)$ for each hyperedge $e\in E_S$. This  number can be interpreted as the local energy (or cost) of the configuration $\vec{s}$ according to the interaction of spins in $e$.
Adding these local energies yields a global energy (or cost) function, namely the Hamiltonian of the spin system.
\begin{definition}[Hamiltonian]\label{def:hamiltonian}
A spin system $S$ defines a \emph{Hamiltonian}
\begin{equation}
\begin{split}
    H_S : \mc{C}_S &\rightarrow \mathbb{R}_{\geq 0} \\
    \vec{s} &\mapsto H_S(\vec{s}) \coloneqq \sum_{e \in E_S} J_S(e)(\vec{s}\vert_e)
\end{split}
\end{equation}
\end{definition} 
Note that \cref{def:spin sys} is restricted to $q_S\geq2$, since for $q_S=1$ there exists a single spin configuration, and hence $H_S$ is a constant.  
Further note that the local interactions are required to be non-negative functions. This definition is equivalent to that of arbitrary local interactions, as a spin system with arbitrary local interactions can be transformed to have non-negative local interactions by shifting all local interactions. This procedure is polytime computable and hence defines a polytime reduction for all relevant computational problems that one might consider. As we shall see, non-negative local interactions are more convenient for certain operations that we will define on spin systems.

Subsequently, we might specify configurations by listing their values with respect to (w.r.t.) a fixed ordering of $V_S$, i.e.\ we write $V_S= \{i.j,k, \ldots\}$ and
\begin{equation}
    \vec{s} = (\vec{s}(i), \vec{s}(j), \vec{s}(k), \ldots) = (s_i,s_j,s_k, \ldots).
\end{equation}

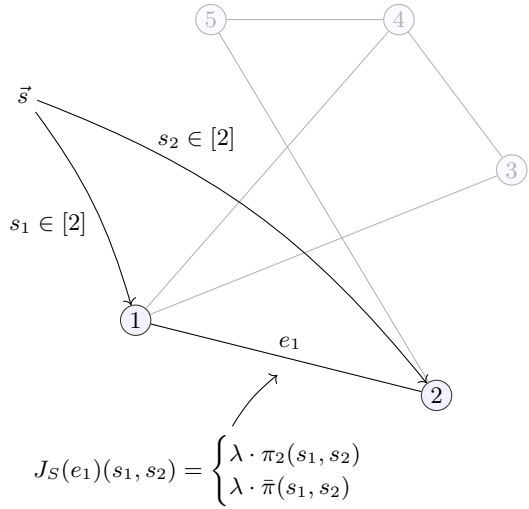
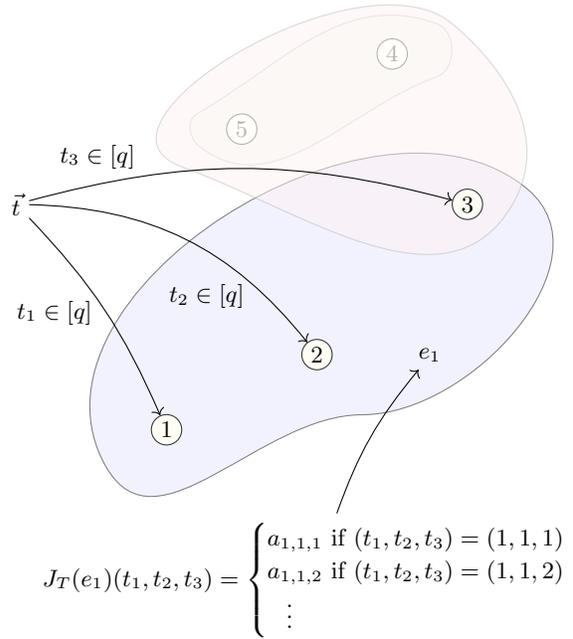
\begin{figure*}[th]
    \centering
    \begin{subfigure}[t]{1.0\columnwidth}
    \centering
    \caption{
}\label{fig:spin sys ex 1}
\vspace{2cm}
        \begin{tikzpicture}
[cgnodeS/.style = {draw=black!80, fill=blue!10, thick, circle, minimum size= {width("$s_{i}$")+4pt}, inner sep = 1pt},
lightNode/.style = {draw=black!30, fill=blue!2, thick, circle, minimum size= {width("$s_{i}$")+4pt}, inner sep = 1pt, text = black!30}
]

\pgfdeclarelayer{bg}    
\pgfsetlayers{bg,main} 

\node[cgnodeS] (1) at (-1,1) {$1$};
\node[cgnodeS] (2) at (3,0) {$2$};
\node[lightNode] (3) at (4,3) {$3$};
\node[lightNode] (4) at (2.5,5) {$4$};
\node[lightNode] (5) at (0,5) {$5$};

\draw[thick] (1) to node[midway, above] { \ $e_1$} node[midway, below] (8) {} (2);
\draw[thick,black!30] (1) to (3);
\draw[thick,black!30] (2) to (5);
\draw[thick,black!30] (3) to (4);
\draw[thick,black!30] (1) to (4);
\draw[thick,black!30] (4) to (5);

\node[] (6) at (-2.5,4) {$\vec{s}$};
\draw[->] (6) to [bend left = 10]  node[midway, below, xshift = -0.6cm] { ${s_1\in [2]}$} (1);
\draw[->] (6) to [bend left = 15] node[pos = 0.35, above, yshift = 0.2cm] { ${s_2 \in [2]}$} (2);
\node[] (7) at (0,-1) { $J_S(e_1)(s_1, s_2) = \begin{cases}
    \lambda \cdot \pi_2(s_1,s_2) \\
    \lambda \cdot \bar{\pi}_2(s_1,s_2)
\end{cases} $};
\draw[->] (7) to [bend left = 10] (8);
\end{tikzpicture}
    \end{subfigure}
    \begin{subfigure}[t]{1.0\columnwidth}
    \centering 
    \caption{}\label{fig:spin sys ex 2}
    \begin{tikzpicture}
[cgnodeS/.style = {draw=black!80, fill=yellow!10, thick, circle, minimum size= {width("$s_{i}$")+4pt}, inner sep = 1pt},
lightNode/.style = {draw=black!30, fill=yellow!2, thick, circle, minimum size= {width("$s_{i}$")+4pt}, inner sep = 1pt, text = black!30}]

\pgfdeclarelayer{bg}    
\pgfsetlayers{bg,main} 

\node[cgnodeS] (1) at (0,0) {$1$};
\node[cgnodeS] (2) at (2,1) {$2$};
\node[cgnodeS] (3) at (4,3) {$3$};
\node[lightNode] (4) at (3,5) {$4$};
\node[lightNode] (5) at (1,4) {$5$};

\begin{pgfonlayer}{bg}
\pgfsetstrokeopacity{0.5}
\pgfsetfillopacity{0.5}
    \filldraw[thick, fill=blue!20] ($(1)+(-0.8,-0.5)$) 
            to[out=300,in=180] ($(2) + (0.6,-0.8)$) 
            to[out=0,in=320] ($(3) + (0.8,0.2)$)
            to[out=140,in=120] ($(1) + (-0.8,-0.5)$);
    \filldraw[thick,fill=green!10, draw = black!30] ($(5)+(-0.6,-0.2)$) 
            to[out=320,in=200] ($(4) + (0.6,-0.3)$)
            to[out=20, in=340] ($(4) + (0.6, +0.4)$)
            to[out=160,in=30] ($(5) + (-0.6,+0.2)$)
            to[out=210, in=130] ($(5) + (-0.6, -0.2)$);
     \filldraw[thick, fill=red!10, draw = black!30] ($(5)+(-0.8,-0.5)$) 
            to[out=340,in=240] ($(3) + (0.6,0)$) 
            to[out=60,in=320] ($(4) + (0.8,0.4)$)
            to[out=140,in=160] ($(5) + (-0.8,-0.5)$);
\end{pgfonlayer}

\node[] (6) at (-2,3) {$\vec{t}$};
\draw[->] (6) to [bend left = 10]  node[pos = 0.4, below, xshift = -0.5cm] { ${t_1 \in [q]}$} (1);
\draw[->] (6) to [bend left = 25]  node[pos = 0.7, above, yshift = -0.6cm, xshift = -0.4cm] { ${t_2 \in [q]}$} (2);
\draw[->] (6) to [bend left = 15]  node[pos = 0.15, above, yshift = 0.1cm] { ${t_3 \in [q]}$} (3);
\node[] (7) at (3.5,1) {$e_1$};
\node[] (8) at (2,-2) { $J_T(e_1)(t_1,t_2,t_3)= \begin{cases}
    a_{1,1,1} \ \text{if } (t_1,t_2,t_3) = (1,1,1) \\
    a_{1,1,2} \ \text{if } (t_1,t_2,t_3) = (1,1,2) \\
    \hphantom{c \ } \vdots 
\end{cases}$};
\draw[->] (8) to [bend left = 10] (7);
\end{tikzpicture}
    \end{subfigure}
    \caption{An Ising system  with pair interactions and spin type $2$ (\ref{fig:spin sys ex 1}) and a generic spin system with $q$-level spins and many-body interactions (\ref{fig:spin sys ex 2}). In contrast to the single parameter $\lambda$ specifying an Ising interaction, a generic $k$-body interaction of $q$-level spins requires $q^k$ parameters.}
    \label{fig:spin sys ex}
\end{figure*}

\begin{example}[Ising system]\label{ex:Ising system}
    A well-studied class of spin systems can be obtained from Ising interactions. In the simplest case, these are pair interactions defined for $2$-level spins that only depend on the parity of the two spins.
    Up to scaling (and relabeling the spins) there are two such interactions
    \begin{equation}
    \begin{split}
        \pi_2(s_i,s_j) &= \begin{cases}
                1 \ \text{if} \ s_i=s_j \\
                0 \ \text{else}
        \end{cases}\\
        \bar{\pi}_2(s_i,s_j) &= \begin{cases}
                1 \ \text{if} \ s_i\neq s_j \\
                0 \ \text{else}.
        \end{cases}
    \end{split}
    \end{equation} 
    An Ising spin system (see \cref{fig:spin sys ex 1}) can be defined by taking an arbitrary graph $(V,E)$ and, for each edge $e\coloneqq\{i,j\}$ from $E$, letting $J(e)$ be either $\lambda \cdot \pi_2$ or $\lambda \cdot \bar{\pi}_2$ for some non-negative real number $\lambda$.
    Additionally, one might include single-spin hyperedges $\{i\}$ in the edge set of the interaction graph.
    Defining 
    \begin{equation}
    \begin{split}
        \pi_1(s_i) &= \begin{cases}
                1 \ \text{if} \ s_i=2 \\
                0 \ \text{else}
        \end{cases}\\
        \bar{\pi}_1(s_i) &= \begin{cases}
                1 \ \text{if} \ s_i=1 \\
                0 \ \text{else}
        \end{cases}
    \end{split}
    \end{equation}
    and taking the local interaction to be either  $\lambda \cdot \pi_1$ or  $\lambda \cdot \bar{\pi}_1$ on these single-spin hyperedges yields a so-called Ising spin system with fields.
\end{example}

Several constructions on spin systems presented in this paper ought to be efficient, that is, polytime computable, where polytime refers to a runtime that is polynomial in the description of the spin system. 
We define the \emph{size} of a spin system $S$ as the length of this description, i.e.\ the number of real valued parameters that are required to specify it, 
\begin{equation}\label{def:size}
\vert S \vert \coloneqq \sum_{e \in E_S}q_S^{\vert e\vert}.
\end{equation} 
Note that, strictly speaking, specifying $S$ amounts to specifying $q_S,V_S,E_S$ and $J_S$ and hence  
\begin{equation}\label{eq:size of input S}
     q_S + \vert V_S \vert + \sum_{e \in E_S} \vert e \vert + \sum_{e \in E_S} (\vert e \vert+q_S^{\vert e \vert})
\end{equation}
real parameters.
Here we used that, for each $e\in E_S$, the space of real valued functions of type $[q_S]^e \to \mathbb{R}$ is a $q_S^{\vert e \vert}$ dimensional real vector space, so specifying $J_S(e)$ requires specifying $e$ together with $q_S^{\vert e \vert}$ real parameters. 
However, since \eqref{eq:size of input S} is polynomial in $\vert S \vert$, 
the two notions of polytime are equivalent. 
Further note that if the arity of local interactions of a set of spin systems with constant spin type admits an upper bound, i.e.\ if the set consists of $q_S$-level spin systems with at most $k$-body interactions, then 
\begin{equation}
\vert S \vert = \mr{poly}(\vert V_S \vert). 
\end{equation} 
That is, polytime w.r.t.\ $\vert S \vert$ is equivalent to polytime w.r.t.\ the number of spins. 
While this applies to many common examples, it need not always be the case. 

One of the key properties of spin systems is their local structure. Each spin system not only gives rise to a global cost function but also encodes how this global function is composed of local terms. 
More conceptually, each spin system can be seen as a sum of elementary spin systems.

It is important to keep this perspective in mind, as constructions might be easy when applied to the local terms but hard when applied globally. To give an example, minimizing all local terms of $H_S$ is polytime computable, while minimizing $H_S$ globally is in general NP-hard.
In this example, minimizing $H_S$ locally does not lead to a global minimum if $S$ suffers from frustration, i.e.\ if the sub-configurations that locally minimize $H_S$ do not agree on the overlaps of local terms, and thus cannot be added to obtain a global configuration.

In contrast, some constructions presented in this work, most importantly that of simulations between spin systems, are compatible with the sum of spin systems. 
This will enables us to obtain simulations of complex spin systems by dividing them into simple building blocks, simulating each of them individually and combining the result for the global simulation.

\begin{definition}[sum of spin systems]\label{def:sum}
    Let $S_1, S_2$ be spin systems with $q \coloneqq q_1 = q_2$. Their \emph{sum} is the spin system 
    \begin{equation}
         S_1 + S_2 \coloneqq (q, V_1 \cup V_2, E_1 \cup E_2, J_{1+2})
    \end{equation}
    where $J_{1+2}$ 
    is defined by 
    \begin{equation}
    J_{1+2}(e)(\vec{s}) \coloneqq \sum_{i \colon e \in E_i}J_i(e)(\Vec{s}|_{e}).
    \end{equation}
\end{definition}
In words, the interaction hypergraph of $S_1+S_2$ is the hypergraph union of $G_1$ and $G_2$. For edges $e \in E_1 \setminus E_2$ or $e \in E_2 \setminus E_1$, its local interactions are $J_1(e)$ or $J_2(e)$, respectively, while for those edges in the overlap of $E_1$ and $E_2$, $e\in E_1 \cap E_2$, its local interactions are given by the sum $J_1(e)+J_2(e)$.
By construction, the sum of spin systems is both associative and commutative.
Moreover, it is straightforward to conclude that 
\begin{equation}
    H_{S_1+S_2}(\vec{s}) = H_{S_1}(\vec{s}\vert_{V_{S_1}}) + H_{S_2}(\vec{s}\vert_{V_{S_2}}).
\end{equation}

Spin systems can also be scaled with non-negative real numbers. 
\begin{definition}[scaling of spin systems]\label{def:scaling}
Let $S$ be a spin system and $\lambda \in \mathbb{R}_{\geq0}$. Then $\lambda \cdot S$ is the \emph{scaled} spin system
\begin{equation}
  \lambda\cdot S \coloneqq(q_S,V_S,E_S,\lambda \cdot J_S)  
\end{equation} 
where $\lambda \cdot J_S$ is understood pointwise, i.e. 
\begin{equation}
    \lambda\cdot J_S(e)(\vec{s}) \coloneqq \lambda\cdot (J_S(e)(\vec{s})).
\end{equation}
\end{definition} 
The scaling of spin systems distributes over sums, i.e.\ $\lambda \cdot (S_1+S_2) = \lambda \cdot S_1 + \lambda \cdot S_2 $.
Moreover, by construction
\begin{equation}\label{eq:scaling:ham}
    H_{\lambda\cdot S}(\vec{s}) = \lambda \cdot H_S(\vec{s}).
\end{equation}

In thermodynamic terms, scaling a system by a positive constant can be understood as changing its temperature (cf.\ \cref{def:part fun}).
In contrast, scaling by a negative number would correspond to swapping high-energy with low-energy configurations, and hence fundamentally change a spin system.

While in general spin systems have multiple hyperedges and are thus sums of multiple local terms, given any function $f \colon [q]^e \rightarrow \mathbb{R}_{\geq 0}$ we can define a \emph{canonical} spin system $S_f$, with a single hyperedge $e$ and corresponding local interaction $J_f(e)=f$. 

\begin{definition}[canonical spin system]\label{def:can sys}
Let 
\begin{equation}
    f:[q]^e \rightarrow \mathbb{R}_{\geq 0}
\end{equation} 
be any function. 
Then the \emph{canonical} spin system corresponding to $f$, denoted $S_f$ (or $T_f$, $S'_f$, etc), is defined by 
\begin{equation}
    \begin{split}
        V_f &\coloneqq e\\
        E_f &\coloneqq \{e\}\\
        J_f(e) &\coloneqq f.
    \end{split}
\end{equation}
\end{definition}
Note that by construction $H_{S_f}=f$.
Moreover, every spin system $S$ is trivially the sum of the canonical systems corresponding to its local interactions
\begin{equation}\label{lem:can sys}
    S=\sum_{e\in E_S} S_{J_S(e)}.
\end{equation}

Note that the sum of spin systems uses the labels of spins (in $V_1$ and $V_2$) as a  bookkeeping tool. More precisely, the labels determine the overlap of the sum. 
In order to achieve a certain overlap in the sum  it might hence be necessary to change the labels of some of the spins.
Similarly to how relabeling vertices in a graph leads to the notion of graph isomorphisms, relabeling spins in a spin system leads to the notion of spin system isomorphisms.
\begin{definition}[isomorphism of spin systems]\label{def:iso spin sys}
    Let $S,T$ be spin systems. An \emph{isomorphism} 
    $\phi \colon S \xrightarrow{\sim} T$
    from $S$ to $T$ is a bijection $\phi \colon V_S \to V_T$ such that for all $e \in \mc{P}(V_S)$ 
    \begin{equation}
        e \in E_S \Leftrightarrow \phi(e) \in E_T ,
    \end{equation}
   and for all $e \in E_S$ and $\vec{s}\in [q_S]^e$
   \begin{equation}\label{eq:iso spin}
       J_S(e)(\vec{s}) = J_T(\phi(e))(\vec{s}\circ \phi^{-1}) .
   \end{equation}
    If there exists an isomorphism $\phi \colon S \xrightarrow{\sim} T$ we write $S \cong T$.
\end{definition}
Note that \cref{eq:iso spin} implies that $q_S = q_T$.
Isomorphic spin systems are equivalent in all aspects, in particular their Hamiltonians agree up to relabelings, i.e.\
\begin{equation}
        H_S(\vec{s}) = H_T(\vec{s}\circ \phi^{-1})
\end{equation}
for all configurations $\vec{s}$ of $S$.

Given a spin system $S$ and a bijection $\phi\colon V_S \to V'$, we denote by $S\langle \phi \rangle$ the spin system isomorphic to $S$ with spins relabeled according to $\phi$, i.e.\
\begin{equation}
\begin{split}
V_{S\langle \phi \rangle}&=V'\\
E_{S\langle \phi \rangle}&=\phi(E_S)\\
    J_{S\langle \phi \rangle}(\phi(e))(\vec{s}\circ \phi^{-1})&= J_S(e)(\vec{s}).
\end{split}
\end{equation}
If $V_S=\{1, \ldots, l \}$, bijections $\phi\colon V_S \to V'$ are orderings of $V'$. We then denote $\phi$ by  $\phi = (v_1,\ldots, v_l)$ with $v_i = \phi(i)$, and the relabeled spin system by $S\langle v_1 \ldots, v_l\rangle$.

\subsection{Simulations}\label{ssec:spin sys sim}

In the previous section we saw that isomorphic spin systems have, up to relabelings, equal Hamiltonians and hence can be seen as equivalent in all aspects. 
There are however situations in which two spin systems $S$ and $T$ can be considered at least approximately equivalent in all relevant aspects, even though they are not isomorphic. For instance, $S$ and $T$ may have equal Hamiltonians given a more general identification of their spins, where one spin of $T$ is identified with multiple spins of $S$. Or the equality of Hamiltonians may only hold in a certain (low) energy regime.
Isomorphisms are thus too strong a notion to describe these weaker kinds of equivalences, which are precisely captured by \emph{spin system simulations}.

We shall first define spin system simulations (\cref{sssec:def}),
show that they preserve spectra and ground states (\cref{sssec:spectra}) as well as thermodynamic properties (\cref{sssec:thermo}), 
and finally illustrate that various common transformations on spin systems (\cref{sssec:symmetries}) as well as graph minors are special cases of simulations (\cref{sssec:minors}).

\subsubsection{Definition} \label{sssec:def}

Whenever $S$ simulates $T$, every occurrence of $T$ can be replaced by $S$ at least in good approximation, even though $S$ and $T$ need not be isomorphic. 
More precisely, a spin system simulation specifies an energy cut-off $\Del$, an energy shift $\Shift$ and functions $\phys, \dec$
that are used to identify configurations $\vec{s}$ of $S$ with configurations $\dec \circ \vec{s} \circ \phys$ of $T$, such that below the cut-off the Hamiltonians of $S$ and $T$ satisfy 
\begin{equation}
    (H_S(\vec{s}) -\Shift ) \big\vert _{<\Del} = H_T(\dec \circ \vec{s} \circ \phys) \big\vert _{< \Del}.
\end{equation}
That is, $\phys$ and $\dec$ describe how low-energy configurations $\vec{s}$ of $S$ are mapped to low-energy configurations $\dec \circ \vec{s} \circ \phys$ of $T$ with equal energy up to $\Shift$.

\begin{definition}[simulation]\label{def:sys simulation}
    Let $S,T$ be a spin systems. A \emph{simulation} $\mc{f}$ of type $S \to T$ consists of the following data:
    \begin{enumerate} 
        \item a positive real number $\Del$ called \emph{cut-off}
        \item a real number $\Shift$ called \emph{energy shift}
        \item a natural number $\degeneracy$ called \emph{degeneracy}
        \item a function $\phys \colon V_T \to V_S^k$ called \emph{physical spin assignment}
        \item a function $\dec \colon [q_S]^k \to [q_T]$ called \emph{decoding} 
        \item a list of functions $\enc = (\enc_i \colon [q_T] \to [q_S]^k)_{i=1,\ldots,m}$ called \emph{encoding}. 
    \end{enumerate}
    Additionally, $\mc{f}$ must satisfy the following conditions: 
    \begin{enumerate}[align=left]
        \item \label{def:sim physical spins}  \emph{Disjoint physical spins}: 
        for all $v,w \in V_T$, $i,j \in \{1, \ldots, k\}$
        \begin{equation}
        \phys^{(i)}(v) = \phys^{(j)}(w) \Rightarrow i=j \ \text{and} \ v=w    
        \end{equation}
        where $\phys^{(i)}$ denotes the $i$-th component of $\phys$

        \item\label{def:sim enc-dec} \emph{Decode-encode compatibility}:  
        for all $i=1, \ldots, m$
        \begin{equation}
            \mr{dec} \circ \mr{enc}_i = \mr{id}
        \end{equation}
        \item\label{def:sim disjoint enc}  \emph{Disjoint encodings}: 
        for all $i, j\in \{1, \ldots, m\}$, $t\in [q_T]$
        \begin{equation}
            \enc_i(t)=\enc_j(t) \Rightarrow i=j
        \end{equation}
        \item \label{def:sim deg} \emph{Constant local degeneracy}:  
        for $i=1, \ldots, m$ and $\vec{t}\in \mc{C}_T$ we define 
        \begin{equation}\label{def:sim simul}
            \simul_i(\vec{t})\coloneqq  \{\vec{s} \in \mc{C}_S \mid \vec{s}\circ \phys = \enc_i \circ \vec{t}, \ H_S(\vec{s})-\Shift < \Del\}
        \end{equation}
        and require that 
        \begin{equation}
            H_T(\vec{t})< \Del \Rightarrow \vert \simul_i(\vec{t})\vert =\degeneracy
        \end{equation} 
        \item \label{def:sim energy} \emph{Matching energies}: 
        for all $\vec{t} \in \mc{C}_T$, $\vec{s}\in \simul(\vec{t})\coloneqq \bigcup_i \simul_i(\vec{t})$
        \begin{equation}\label{eq:sim def}
            H_S(\vec{s})-\Shift = H_T(\vec{t}),
        \end{equation}
        for all $\vec{s} \notin \simSet\coloneqq  \bigcup_{\vec{t}\in \mc{C}_T} \simul(\vec{t}) $
        \begin{equation}
            H_S(\vec{s})-\Shift \geq \Del.
        \end{equation}
    \end{enumerate}
\end{definition}

If $\mc{f}$ is a simulation of type $S \to T$ we also write $\mc{f} \colon S \to T$ or $S \overset{\mc{f}}{\to} T$.
We might also write $S \to T$ as a shorthand for ``there exists a simulation $\mc{f}\colon S \to T$"  and simply say ``$S$ simulates $T$" whenever this is the case. 
Whenever we distinguish between multiple simulations we index their data, i.e.\ we write $\Del_{\mc{f}}, \Shift_{\mc{g}}$, etc. 
If part of the data specifying a simulation is trivial, i.e.\ $\Shift=0$, $\degeneracy=1$, $\vert \enc \vert =1$ and $\enc=\dec= \mr{id}$, or $V_T\subseteq V_S$ and $\phys$ is the corresponding subset injection, we may not give it explicitly and only specify the non-trivial data. 

Given a simulation $\mc{f} \colon S \to T$, 
we call $S$ the \emph{source}  and $T$ the \emph{target} of $\mc{f}$. 
We call $k$, from $V_S^k$ and $[q_S]^k$, the \emph{order} of the encoding, and denote it by $\mr{ord}(\enc)=k$.
Note that $k$ is implicit in the data that specifies $\mc{f}$.
Similarly, the number of individual functions in $\enc$, viz.\ $m$, is implicit. We call it the \emph{size} of the encoding, and denote it by $\vert \enc\vert$.

In condition \ref{def:sim energy} we defined the function $\simul$
mapping configurations $\vec{t}$ of $T$ to sets of configurations $\simul(\vec{t})$ of $S$ with equal energy up to $\Shift$. 
We call $\simul_i$ the $i$-th \emph{simulation assignment} and say that $\vec{s}$ simulates $\vec{t}$ whenever $\vec{s}\in \simul(\vec{t})$.

Given $v\in V_T$ we call $\phys(v)$ \emph{physical spins of $v$}, and 
\begin{equation}
    \physSet \coloneqq \bigcup_{v \in V_T}\phys(v)
\end{equation} 
\emph{physical spins}. The remaining spins of $S$ are called \emph{auxiliary spins}.
If $\vec{s}$ simulates $\vec{t}$, then by construction of $\simul$, the state of $\vec{t}(v)$ can be recovered from $\vec{s}({\phys(v)})$ according to 
\begin{equation}
    \vec{t}(v) = \dec \circ \vec{s} \circ \phys(v) .
\end{equation}
This means that $\vec{t}$ can be recovered from $\vec{s}$ locally. Conversely, if $\vec{t}'$ differs from $\vec{t}$ in the state of a single spin $v$ and $\vec{s'}\in \simul(\vec{t}')$, then $\vec{s}$ and $\vec{s'}$ agree on all physical spins except for those of $v$. 
Note that $\vec{s}$ and $\vec{s'}$ can always differ on auxiliary spins.

Finally let us give an intuition as to why simulations use both decoding and encoding.
The decoding ensures that the individual encodings are compatible. 
That is, if simulations only had encodings, it would not be possible to require condition \ref{def:sim enc-dec}, and hence it might happen that for $s\neq t$, $\enc_i(s)=\enc_j(t)$. In this case, it would not be possible to decode source configurations with energy below $\Del$. 
Secondly, the encoding ensures that whenever a target configuration $\vec{t}$ agrees on spins $v,w$ then any source configuration $\vec{s}$ that simulates $\vec{t}$ agrees on the corresponding physical spins $\phys(v),\phys(w)$.
If this was not the case, target configurations $\vec{t}$ could not be encoded locally.
Specifically, given access to $\vec{t}$ (with $H_T(\vec{t})<\Del$) on a subset $V \subseteq V_T$ it would not be possible to infer $\simul(\vec{t})\vert_{\phys(V_T)}$. 
This would lead to simulation being incompatible with sums (and hence decompositions) of spin systems (cf.\ \cref{thm:sim sum}).

Let us illustrate simulations between spin systems. 

\begin{example}[simulation]\label{ex:simulation}
    Consider the spin system $T$ (see right hand side of \cref{fig:ex sim 1}) defined by 
    \begin{equation}
        \begin{split}
            q_T & = 3\\
            V_T & = \{1,2,3\}\\
            E_T & = \{\{1,2,3\}, \{1,2\}\}, 
        \end{split}
    \end{equation}
    and local interactions  
    \begin{equation}
        \begin{split}
            J_T(e_1)(t_1,t_2,t_3) &= \begin{cases}
                0 \ \text{if } \  (t_1,t_2,t_3) = (1,2,3)\\
                1 \ \text{if } \  t_1=t_2=t_3 \\
                4 \ \text{else},
            \end{cases}\\
            J_T(e_2)(t_1,t_2) &= \begin{cases}
                0 \ \text{if } t_1 = t_2\\
                2 \ \text{else},
            \end{cases}
        \end{split}
    \end{equation}
    where $e_1 = \{1,2,3\}$ and $e_2 = \{1,2\}$.
    
    We construct a simulation with source $S$ being an Ising spin system, target $T$ and cut-off $\Del=3$. There are exactly 4 target configurations with energy below $3$, $\vec{t}_1=(1,1,1)$, $\vec{t}_2=(2,2,2)$, $\vec{t}_3=(3,3,3)$ and $\vec{t}_4=(1,2,3)$,
    \begin{equation}
            H_T(\vec{t}_1) = H_T(\vec{t}_2) = H_T(\vec{t}_3) = 1 \quad \text{and} \quad 
            H_T(\vec{t}_4) = 2. 
    \end{equation}
    All other configurations either have energy $E=4$ or $E=6$ (see right hand side of \cref{fig:ex sim 2}).
    The source system $S$ is an Ising spin system with fields, with $q_S=2$, spins $V_S= \{1, \ldots, 10\}$ and edges illustrated in the left hand side of \cref{fig:ex sim 1}.
    The local interactions are listed in \cref{tab: local int ex}.

\begin{figure*}[th]
    \centering
    \begin{subfigure}{1.0\textwidth}
    \caption{}\label{fig:ex sim 1} 
        \scalebox{1.0}{
\begin{tikzpicture}
[cgnodeS/.style = {draw=black!80, fill=blue!10, thick, circle, minimum size= {width("$s_{i}$")+4pt}, inner sep = 1pt},
cgnodeT/.style = {draw=black!80, fill=yellow!10, thick, circle, minimum size= {width("$s_{i}$")+4pt}, inner sep = 1pt},
cgnodeP/.style = {draw=black!80, thick, ellipse, rotate = -45,  minimum height= 1cm, minimum width = 2.2cm, dashed},
cgnodeP2/.style = {draw=black!80, circle, thick,  minimum size= 0.8cm, dashed}
]

\pgfdeclarelayer{bg}    
\pgfsetlayers{bg,main} 

\node[cgnodeS] (1) at (-1,1) {$1$};
\node[cgnodeS] (2) at (0,0) {$2$};
\node[cgnodeS] (3) at (4,0) {$3$};
\node[cgnodeS] (4) at (5,-1) {$4$};
\node[cgnodeS] (5) at (1,4) {$5$};
\node[cgnodeS] (6) at (2,3) {$6$};
\node[cgnodeS] (7) at (-2,5) {$7$};
\node[cgnodeS] (8) at (0,6.5) {$8$};
\node[cgnodeS] (9) at (5,6.5) {$9$};
\node[cgnodeS] (10) at (7,5) {$10$};

\foreach \x/\y in {1/3, 1/4, 1/5, 1/10, 2/4, 2/5, 2/6, 2/10, 3/4, 3/5, 3/10, 4/6, 4/8, 4/9, 5/6, 5/8, 5/9, 6/10, 8/10, 9/10}
    \draw[thick, black!30] (\x) to (\y);

\node[cgnodeT] (11) at (10,2) {$a$};
\node[cgnodeT] (12) at (12,0) {$b$};
\node[cgnodeT] (13) at (11,4) {$c$};

\begin{pgfonlayer}{bg}
\pgfsetstrokeopacity{0.5}
\pgfsetfillopacity{0.5}
    \filldraw[thick, fill=red!20] ($(11)+(-1,-0.4)$) to[out=280, in=200, looseness=1.1] 
    ($(12)+(0.1,-1)$)  to[out=20, in=290, looseness=1.1] 
    ($(13)+(0.8,0.5)$)   to[out=110, in=100, looseness=1.4] cycle;
    \filldraw[thick, fill=blue!20] ($(11)+(-0.6,-0.4)$) 
            to[out=320,in=200] ($(12) + (0.5,-0.5)$)
            to[out=20, in=340] ($(12) + (0.4, +0.6)$)
            to[out=160,in=20] ($(11) + (-0.6,+0.6)$)
            to[out=200, in=140] ($(11) + (-0.6, -0.4)$);
\end{pgfonlayer}

\node[cgnodeP] (14) at (-0.5,0.5) {};
\node[cgnodeP] (15) at (4.5,-0.5) {};
\node[cgnodeP] (16) at (1.5,3.5) {};

\node[cgnodeP2] (17) at (10,2) {};
\node[cgnodeP2] (18) at (12,0) {};
\node[cgnodeP2] (19) at (11,4) {};

\draw[->] (17) to [bend right = 20] (14);
\draw[->] (18) to [bend right = 10] (15);
\draw[->] (19) to [bend right = 15] (16);

\node[] (20) at (-2.25,0) {$\mathrm{P}(a) = (1,2)$};
\node[] (21) at (3.25,-1.5) {$\mathrm{P}(b) = (3,4)$};
\node[] (22) at (-0.75,3.5) {$\mathrm{P}(c) = (5,6)$};

\end{tikzpicture}       
        } 
    \end{subfigure}
    \begin{subfigure}{1.0\textwidth} 
    \caption{}\label{fig:ex sim 2}
    \scalebox{1.0}{
    
\begin{tikzpicture}
[cgnodeS/.style = {draw=black!80, fill=blue!5, circle, minimum size= {width("$s_{10}$")+6pt}, inner sep = 2pt},
cgnodeS2/.style = {draw=black!80, fill=yellow!5, circle, minimum size= {width("$s_{10}$")+6pt}, inner sep = 2pt}]

\draw[->] (0,-6) to node[pos = 1.05, right] { $H_S- \Shift$} (0,1);
\draw[] (0,-5) to (2,-5);
\draw[] (2,-5) to (2.5,-5.2);
\draw[] (2,-5) to (2.5,-5);
\draw[] (2,-5) to (2.5,-4.8);
\draw[] (2.5,-5.2) to (5,-5.2);
\draw[] (2.5,-5) to node[pos = 1.4] {$\{\vec{s}_1, \vec{s}_2, \vec{s}_3\}$} (5,-5);
\draw[] (2.5,-4.8) to (5,-4.8);

\draw[] (0,-4) to node[pos = 1.2] { $\{ \vec{s}_4\}$} (5,-4);
\draw[] (0,-2) to (5,-2);
\draw[] (0,-2.8) to  (5,-2.8);
\draw[] (0,-2.7) to  (5,-2.7);
\draw[] (0,-1.5) to  (5,-1.5);
\draw[] (0,-1.3) to  (5,-1.3);
\draw[] (0,-0.9) to (5,-0.9);
\draw[] (0,-0.2) to (5,-0.2);
\draw[] (0,0.1) to (5,0.1);

\draw[->] (8,-6) to node[pos = 1.05, right] { $H_T$} (8,1);
\draw[] (8,-5) to (10,-5);
\draw[] (10,-5) to (10.5,-5.2);
\draw[] (10,-5) to (10.5,-5);
\draw[] (10,-5) to (10.5,-4.8);
\draw[] (10.5,-5.2) to (13,-5.2);
\draw[] (10.5,-5) to node[pos = 1.4] {$\{\vec{t}_1, \vec{t}_2, \vec{t}_3\}$} (13,-5);
\draw[] (10.5,-4.8) to (13,-4.8);
\draw[] (8,-4) to node[pos = 1.2] { $\vec{t}_4$} (13,-4);
\draw[] (8,-2) to  (13,-2);
\draw[] (8,0) to  (13,0);

\draw[red,thick] (0,-3) to (15,-3);
\node[] at (16,-3) {$\Delta = 3$};
\draw[] (0,-6) to (15,-6);
\node[] at (16,-6) {$E=0$};
\node[] at (16,-5) {$E=1$};
\node[] at (16,-4) {$E=2$};
\node[] at (16,-2) {$E=4$};
\node[] at (16,0) {$E=6$};

\end{tikzpicture}
    }
    \end{subfigure}
    \caption{
    \cref{fig:ex sim 1} shows the spin systems $S$ (left hand side) and $T$ (right hand side) from \cref{ex:simulation} as well as the physical spin assignment $\phys$ of the simulation $S \to T$. Each spin of $T$ is encoded into a pair of spins of $S$.
    \cref{fig:ex sim 2} shows the spectra of $H_S-\Shift$ (left hand side) and $H_T$ (right hand side). The simulation $S \to T$ has shift $\Shift=3$ and cut-off $\Delta=3$. Below the cut-off, the spectra of $H_S-\Shift$ and $H_T$ are identical, up to identifying each low-energy target configurations $\vec{t}_i$ with the respectively unique configuration $\vec{s}_i \in \simul(\vec{t_i})$.
    Above the cut-off, the two spectra are only drawn schematically.}
\end{figure*}
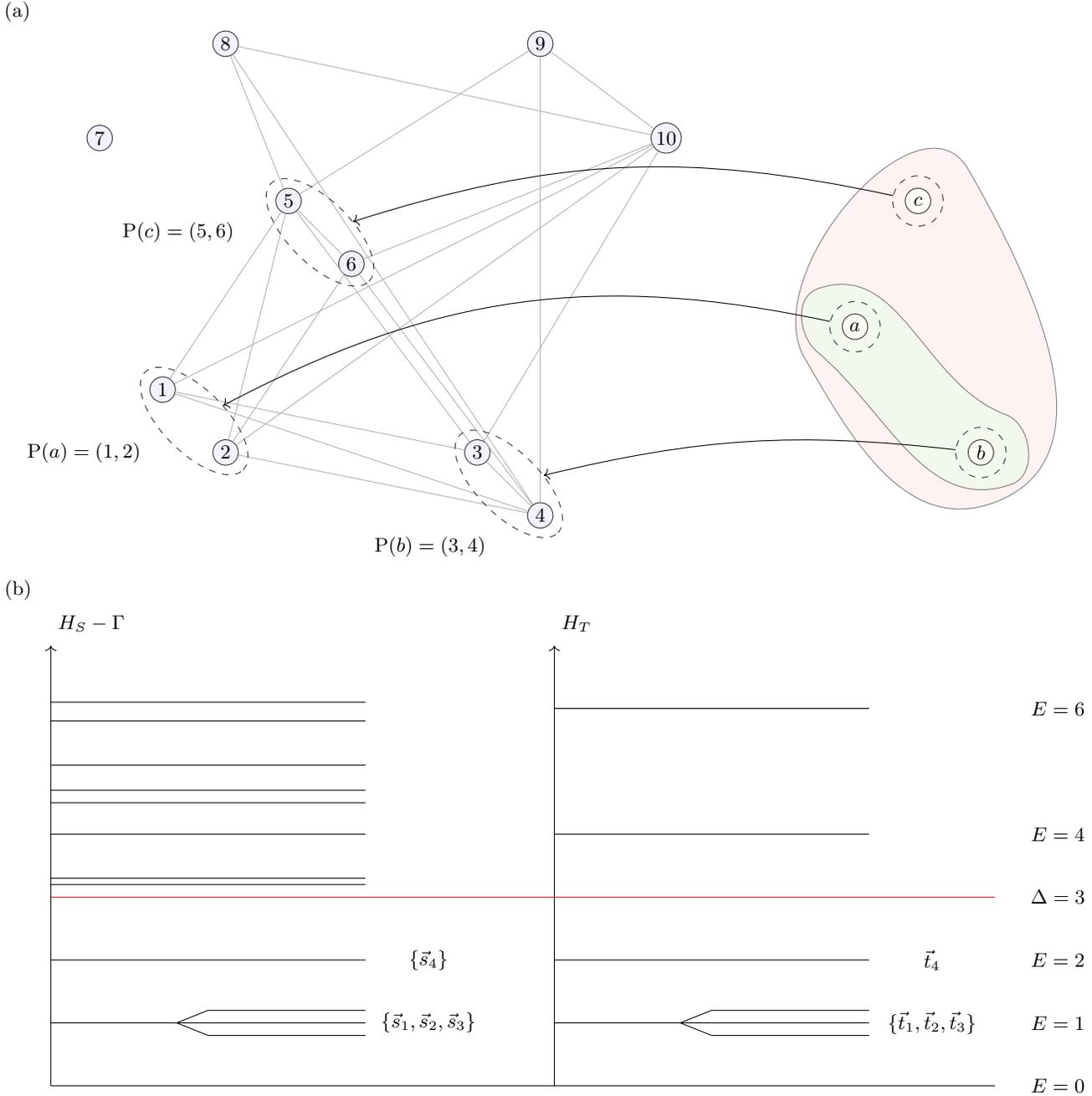

\begin{table}[th]
\centering
\begin{tabular}{ |p{1.8cm}|p{1.8cm}|p{1.8cm}|p{1.8cm}|  }
 \hline
 \multicolumn{2}{|c}{Fields} & \multicolumn{2}{|c|}{Pair interactions} \\
 \hline
 Spin & Field & Edge & Interaction\\
 \hline
$1$   & $\frac{3}{4}\cdot \pi_1$   & $\{1,3\}$   & $1 \cdot \bar{\pi}_2$   \\
$2$   &  $\frac{1}{4}\cdot \pi_1$  &$\{1,4\}$    & $\frac{3}{4}\cdot \pi_2 $  \\
$3$   &  $\frac{1}{4}\cdot \pi_1$   &$\{1,5\}$    & $\frac{1}{4} \cdot \bar{\pi}_2$  \\
$4$   &  $\frac{1}{4}\cdot \pi_1$  &$\{1,10\}$    & $ \frac{3}{4}\cdot \pi_2$  \\
$5$   &  $\frac{1}{4}\cdot \pi_1$  &$\{2,4\}$    &  $\frac{3}{4} \cdot \bar{\pi}_2 $ \\
$6$   &  $\frac{3}{4}\cdot \pi_1$  &$\{2,5\}$    &  $ \frac{1}{4} \cdot \pi_2$ \\
$7$   &   $2\cdot \bar{\pi}_1$ &$\{2,6\}$    &  $ 1 \cdot \bar{\pi}_2$ \\
$8$   &  $1\cdot \bar{\pi}_1$  &$\{2,10\}$    &  $\frac{1}{4}\cdot \pi_2 $  \\
$9$   &  $\frac{5}{4} \cdot \bar{\pi}_1$  &$\{3,4\}$ & $ \frac{1}{4}\cdot \pi_2$     \\
$10$  &   $\frac{1}{4}\cdot \pi_1$  &$\{3,5\}$ & $ \frac{3}{4}\cdot \bar{\pi}_2$   \\
      &    &$\{3,10\}$         & $ \frac{1}{4} \cdot \pi_2$   \\
      &    &$\{4,6\}$          & $ \frac{1}{4} \cdot \bar{\pi}_2$    \\
      &    &$\{4,8\}$             &  $ 1 \cdot \bar{\pi}_2$  \\
      &    &$\{4,9\}$             &  $ \frac{3}{4}\cdot \pi_2$  \\
      &    &$\{5,6\}$             &   $\frac{3}{4}\cdot \pi_2 $ \\
      &    &$\{5,8\}$             &    $1 \cdot \bar{\pi}_2 $ \\
      &    &$\{5,9\}$             &   $\frac{3}{4} \cdot \pi_2$ \\
      &    &$\{6,10\}$             &   $ \frac{3}{4} \cdot \pi_2$ \\
      &    &$\{8,10\}$             &   $ 1 \cdot \pi_2$ \\
      &    &$\{9,10\}$             &   $ \frac{3}{4} \cdot \pi_2$ \\
\hline
\end{tabular}
\caption{Local interactions of the source spin system $S$ from \cref{ex:simulation}. They are all Ising pair interactions $\pi_2, \bar{\pi}_2$ or Ising fields $\pi_1, \bar{\pi}_1$, both defined in \cref{ex:Ising system}.}\label{tab: local int ex}
\end{table}

Since $S$ and $T$ differ w.r.t.\ their spin type, the first step consists of specifying how $3$-level spins form $T$ are encoded into $2$-level spins from $S$. Consider 
\begin{equation}
    \begin{split}
        \enc &\colon [3] \to [2]^2 \\
        \dec &\colon [2]^2 \to [3],
    \end{split}
\end{equation}
defined by $\enc(1) = (1,1)$, $\enc(2) = (1,2)$, $\enc(3)=(2,1)$ and $\dec(1,1)=1$, $\dec(1,2)=2$, $\dec(2,1)=3$, $\dec(2,2)=1$. Clearly, $\dec\circ \enc= \mr{id}$, so $\dec, \enc$ satisfy \cref{def:sys simulation} \ref{def:sim enc-dec}.
Since $\enc$ is the single encoding, it trivially satisfies \cref{def:sys simulation} \ref{def:sim disjoint enc}.

Next we define the physical spin assignment,
\begin{equation}
    \phys \colon V_T \to V_S^2.
\end{equation}
We take $\phys(1)=(1,2)$, $\phys(2)=(3,4)$ and $\phys(3)=(5,6)$. This satisfies \cref{def:sys simulation} \ref{def:sim physical spins}. 
$\phys$ is illustrated in \cref{fig:ex sim 1}.

By direct computation, it can be verified that this defines a simulation of type $S \to T$ with cut-off $\Del =3$, shift $\Shift = 3$ and degeneracy $\degeneracy=1$.
More precisely, the couplings of $S$ are chosen such that there are exactly $4$ source configurations with energy below $6$. 
These are exactly the four configurations $\vec{s}_1, \ldots, \vec{s}_4$, satisfying 
\begin{equation}
    \simul(\vec{t}_i) = \{\vec{s}_i\}.
\end{equation}
These configurations are listed in \cref{tab:ex simul}.
This proves that \cref{def:sys simulation} \ref{def:sim deg} is satisfied with $\degeneracy=1$.

\begin{table}[th]
    \centering
    \begin{tabular}{|c|c|}
    \hline
       Target configuration  $\vec{t}_i$ & Source configuration $\vec{s}_i \in \simul(\vec{t}_i)$  \\
       \hline 
        (1,1,1) & (1,1,1,1,1,1,2,1,2,2) \\
        (2,2,2) & (1,2,1,2,1,2,2,2,2,1) \\
        (3,3,3) & (2,1,2,1,2,1,2,2,2,1)\\
        (1,2,3) & (1,1,1,2,2,1,2,2,1,2)\\
        \hline
    \end{tabular}
    \caption{Simulation assignment of \cref{ex:simulation}: The first column contains the four target configurations $\vec{t}_i$ with energy below $\Del = 3$, while the second column contains the respectively unique source configuration $\vec{s}_i \in \simul(\vec{t}_i)$. 
    }
    \label{tab:ex simul}
\end{table}

Finally, the four source configurations $\vec{s}_i$ satisfy 
\begin{equation}
    H_S(\vec{s}_i) - 3 = H_T(\vec{t}_i).
\end{equation}
Thus, below the cut-off, the spectra of $H_T$ and $H_S-\Shift$ agree (see \cref{fig:ex sim 2}).
Hence, also \cref{def:sys simulation} \ref{def:sim energy} is satisfied, showing that this provides a simulation of type $S \to T$.
\end{example}

\subsubsection{Simulations Preserve Spectra and Ground States} \label{sssec:spectra}

Spin system simulations preserve important properties of spin systems: 
If $S \to T$, their spectra agree below the cut-off (\cref{lem:sim spectrum}) and the ground states of $T$ can be obtained from those of $S$ (\cref{lem:ground state}). These results follow from the crucial observation that simulation increases the degeneracy of each energy level by the same factor (\cref{lem:sim deg}). 
We first need some preparatory results. 

\begin{lemma}\label{lem:sim high energy} Let $S\to T$ be a simulation.
    If $H_T(\vec{t})\geq \Del$ then $\simul(\vec{t})= \emptyset$.
\end{lemma}
\begin{proof}
     Assume that $\vec{s}\in \simul(\vec{t})$. Then by definition of $\simul(\vec{t})$ 
\begin{equation}
    H_S(\vec{s})-\Shift < \Del 
\end{equation}
but by condition \ref{def:sim energy} 
\begin{equation}
    H_S(\vec{s})-\Shift = H_T(\vec{t}) \geq \Del .
\end{equation}
\end{proof}

This implies that $\simul(\vec{t})\neq \emptyset$ if and only if $H_T(\vec{t})< \Del$. 
Similarly, $\vec{s}\in \simSet$ if and only if $H_S(\vec{s}) - \Shift < \Del$ (by condition \ref{def:sim energy}).
In the following we call such configurations \emph{low-energy configurations} (both for target and source), where the term is understood relative to $\Del$ and $\Shift$. 
Low-energy configurations are those that either simulate a target configuration, or are simulated by a source configuration.

\begin{lemma}[lowering the cut-off]\label{thm:sim new delta}
    Let $\mc{f} \colon S \to T$,
    then setting $\Del_\mc{g} \leq \Del_\mc{f}$ and leaving all remaining parameters unchanged defines a simulation $\mc{g}\colon S \to T$.
\end{lemma}

\begin{proof}
    To prove the claim we need to prove condition \ref{def:sim energy} of \cref{def:sys simulation}. The other conditions hold trivially. 
    If $\vec{s}\in \simul_\mc{g}(\vec{t})$ then by definition of $\simul_\mc{g}$, $H_T(\vec{t})<\Del_\mc{g}\leq \Del_\mc{f}$. Since $\mc{f}$ defines a simulation with cut-off $\Del_\mc{f}$ and $\Shift \coloneqq \Shift_\mc{f} = \Shift_\mc{g}$, 
    \begin{equation}
        H_S(\vec{s})-\Shift  =H_T(\vec{t}).
    \end{equation}
    If $\vec{s}\notin \simSet_\mc{g}$ then either $\vec{s}\notin \simSet_\mc{f}$ and therefore 
    \begin{equation}
        H_S(\vec{s})-\Shift > \Delta_\mc{f} \geq \Del_\mc{g} , 
    \end{equation}
    or $\vec{s} \in \simul_\mc{f}(\vec{t})$ for some $\vec{t}$ with $\Del_\mc{f} > H_T(\vec{t}) \geq \Del_\mc{g}$ and therefore 
    \begin{equation}
         H_S(\vec{s})-\Shift = H_T(\vec{t}) \geq \Del_\mc{g}.
    \end{equation}
\end{proof}

\begin{proposition}[constant degeneracy]\label{lem:sim deg}
    Let $S \to T$ and let $\vec{t}\in \mc{C}_T$ be such that $\xi\coloneqq H_T(\vec{t})<\Del$. Then 
    \begin{equation}
        \vert H_S^{-1}(\{\xi+\Shift\}) \vert = \vert \enc \vert \cdot \degeneracy \cdot \vert H_T^{-1}(\{\xi\}) \vert .
    \end{equation}
\end{proposition}

\begin{proof}
    By \cref{def:sys simulation} condition \ref{def:sim energy} 
    only configurations $\vec{s}\in \simul(\vec{t}')$ for some $\vec{t}' \in H_T^{-1}(\xi)$ satisfy $\vec{s} \in H_S^{-1}(\xi+\Shift)$. 
    Hence, 
    \begin{equation}
        H_S^{-1}(\{\xi+\Shift\}) =  \bigcup_{\vec{t}'\in H_T^{-1}(\xi)} \bigcup_{i=1}^{\vert \enc \vert} \simul_i(\vec{t}')  .
    \end{equation}
    By condition \ref{def:sim deg}, for each such $\vec{t}'$ and each $i=1, \ldots, \vert \enc \vert$ there are $\degeneracy$ source configurations in $\simul_i(\vec{t}')$.
    By condition \ref{def:sim disjoint enc}, if $\vec{s}_i \in \simul_i(\vec{t}')$ and $\vec{s}_j \in \simul_j(\vec{t}')$ then $\vec{s}_i \neq \vec{s}_j$, so  $\vert \simul(\vec{t}') \vert = \vert \enc \vert \cdot \degeneracy$.
    In total this proves
    \begin{equation}
       \Big \vert \bigcup_{i=1}^{\vert \enc \vert} \simul_i(\vec{t}') \Big \vert = \degeneracy \cdot \vert \enc \vert.
    \end{equation}
    By condition \ref{def:sim enc-dec}, $\vec{s}\in \simul(\vec{t}')$ implies 
    \begin{equation}
        \dec\circ \vec{s} \circ \phys = \vec{t}',
    \end{equation}
    so source configurations that simulate different target configurations are different, i.e.\ the union over $\vec{t}' \in H_T^{-1}(\{ \xi \})$ is disjoint, 
    which finishes the proof.
\end{proof}

\begin{proposition}[preservation of spectrum]\label{lem:sim spectrum}
    Let $S \to T$ with cut-off $\Del$ and shift $\Shift$. Then 
    \begin{equation}
        \mr{Im}(H_T)_{<\Del} = \mr{Im}(H_S-\Shift)_{<\Del},
    \end{equation}
    where 
    \begin{equation}
        \mr{Im}(H)_{<\Del}\coloneqq \{\xi \in \mr{Im}(H) \mid \xi < \Del\}.
    \end{equation}
\end{proposition}
\begin{proof}
Follows immediately from \cref{lem:sim deg}. \end{proof}

\begin{proposition}[preservation of ground state]\label{lem:ground state}
Let $S\to T$ with cut-off 
\begin{equation}
    \Del > \mr{min}(H_T)
\end{equation}
and let $\vec{s}$ be a \emph{ground state} of $S$, that is, 
\begin{equation}
    \vec{s} \in \mr{GS}(S)\coloneqq H_S^{-1}(\{\mr{min}(H_S)\}).
\end{equation}
Then $\vec{t}=\dec\circ \vec{s} \circ \phys$ is a ground state of $T$.
\end{proposition}

\begin{proof}
By \cref{lem:sim spectrum}, $\mr{min}(H_S)-\Shift = \mr{min}(H_T) < \Del$, so $\vec{s}\in \simSet$.
Since $\vec{t}=\dec\circ \vec{s} \circ \phys$, it must be the case that $\vec{s}\in \simul(\vec{t})$, and thus by condition \ref{def:sim energy} of \cref{def:sys simulation}
\begin{equation}
    H_T(\vec{t}) = H_S(\vec{s})-\Shift = \mr{min}(H_S)-\Shift = \mr{min}(H_T).
\end{equation}\end{proof}

\subsubsection{Simulations Preserve Thermodynamic Quantities}\label{sssec:thermo}

Spin system simulations also preserve thermodynamic quantities:  
If $S \to T$, the partition function (\cref{lem:sim part fun sys}) 
and Boltzmann distribution (\cref{thm:boltzmann sys}) 
of $T$ can be approximated by that of $S$.

\begin{definition}[partition function]\label{def:part fun}
Let $S$ be a spin system. The \emph{partition function} of $S$ is the function
    \begin{equation}
        \begin{split}
            Z_S &\colon \mathbb{R}_{>0} \to \mathbb{R}_{>0} \\
            Z_S(\beta) &\coloneqq \sum_{\vec{s} \in \mc{C}_S} e^{-\beta H_S(\vec{s})}
        \end{split}
    \end{equation}
\end{definition}
Various thermodynamic quantities, such as the free energy or the entropy, can be derived from the partition function. It also appears as the normalization in the Boltzmann distribution $p_{S,\beta}$, the probability distribution of configurations of $S$ in thermal equilibrium at (inverse) temperature $\beta$.

\begin{definition}[Boltzmann distribution]\label{def:Boltzmann}
    Let $S$ be a spin system and $\beta \in \mathbb{R}_{>0}$. The \emph{Boltzmann distribution} of $S$ at $\beta$, $p_{S,\beta}$ is the probability distribution on $\mc{C}_S$, defined by
    \begin{equation}
            p_{S,\beta}(\vec{s})  = \frac{1}{Z_S(\beta)} e^{-\beta H_S(\vec{s})}
    \end{equation}
\end{definition}

If we think of a spin system $S$ as a thermodynamic system, then when coupled to a heat bath at temperature $\beta$, the probability of $S$ being is configuration $\vec{s}$ is precisely the Boltzmann distribution $p_{S,\beta}(\vec{s})$.

We will now see that whenever $S \to T$, both the partition function and Boltzmann distribution of $T$ can be approximated by that of $S$. This means that simulations preserve most (if not all) thermodynamic properties of spin systems.

\begin{proposition}[preservation of partition function]\label{lem:sim part fun sys}
If $S \to T$ with cut-off $\Del$, then, for all $\beta>0$,
\begin{equation}\label{eq:part fun error}
    \left\vert \frac{Z_S(\beta)}{e^{-\Shift \beta} m d}   - Z_T(\beta) \right\vert
    \leq \mr{max}\left( \frac{q_S^{\vert V_S \vert}}{md}  , n_T(\Del)\right)e^{-\beta \Del}
\end{equation}
where $m=\vert \enc \vert $ and $n_T(\Del)$ denotes the number of target configurations $\vec{t}$ with $H_T(\vec{t})\geq \Del$.
\end{proposition}

Note that increasing $\Del$ while keeping the rest of the simulation unchanged, the right hand side of \cref{eq:part fun error} becomes arbitrarily small.
In particular increasing $\Del$ decreases $n_T(\Del)$ and
if $\Del > \mr{max}(H_T)$ then $n_T(\Delta)=0$ and the RHS equals $\frac{1}{md} q_S^{\vert V_S \vert} e^{-\beta \Del}$.

Further note that for constant $\Delta$ \cref{eq:part fun error} still yields a reasonable approximation when $\beta$ is sufficiently large. Intuitively, this means that at low temperature (i.e.\ large $\beta$) even simulations with low cut-off preserve the partition function in good approximation. In this sense, even simulations with low cut-off preserve (at least) the low-temperature behavior of spin systems.

\begin{proof}
    Assume $S \to T$ with cut-off $\Del$ and shift $\Shift$, and let $\beta>0$ be arbitrary. 
    We rewrite the partition function $Z_S(\beta)$ by splitting the sum into those source configurations in $\simSet$, which by 
    \cref{def:sys simulation} \ref{def:sim energy} have energy $H_T(\vec{t})+\Shift$, for some target configuration $\vec{t}$ with energy $H_T(\vec{t})<\Del$ and those
    not in $\simSet$: 
    \begin{equation}\label{eq:part fun first split}
    \begin{split}
        Z_S(\beta) 
        &= \sum_{\vec{s}\in \simSet}e^{-\beta H_S(\vec{s})} + \sum_{\vec{s}\notin \simSet}e^{-\beta H_S(\vec{s})} \\
        &=  e^{-\beta \Shift}  \sum_{\mathclap{\substack{\vec{t}\in \mc{C}_T,\\
        H_T(\vec{t})<\Del}}}  e^{-\beta H_T(\vec{t})} \sum_{\vec{s}\in \simul(\vec{t})}1 
        + \sum_{\vec{s}\notin \simSet}e^{-\beta H_S(\vec{s})}.
        \end{split}
    \end{equation}
    Next, by \cref{lem:sim deg}, for all target configurations $\vec{t}$ with energy below $\Del$,  $\vert \simul(\vec{t}) \vert = m d$. Using this together with the definition of the partition function of $T$, the first summand in the last line of \cref{eq:part fun first split} yields 
    \begin{equation}\label{eq:part fun approx}
    \begin{split}
        e^{-\beta \Shift} \sum_{\mathclap{\substack{\vec{t}\in \mc{C}_T,\\
        H_T(\vec{t})<\Del}}} & e^{-\beta H_T(\vec{t})} \sum_{\mathclap{\vec{s}\in \simul(\vec{t})}}1  \\
        & \hspace{0.5cm} = e^{-\beta \Shift}    m d  \Bigl( Z_T(\beta) -\sum_{\mathclap{\substack{\vec{t}\in \mc{C}_T,\\
        H_T(\vec{t})\geq \Del}}}  e^{-\beta H_T(\vec{t})} \Bigr).
    \end{split}
    \end{equation}
    
    Now by definition of $n_T(\Delta)$ the sum over high-energy target configurations in \cref{eq:part fun approx}
    can be upper bounded as 
    \begin{equation}\label{eq:upper bound part fun 1}
        \sum_{\mathclap{\substack{\vec{t}\in \mc{C}_T,\\
        H_T(\vec{t})\geq \Del}}}  e^{-\beta H_T(\vec{t})} \leq n_T(\Delta)e^{-\beta \Del}.
    \end{equation}

    Further, since all source configurations $\vec{s}\notin \simSet$ have energy at least $H_S(\vec{s})\geq \Del+\Shift$ and
    since there are at most  $q_S^{\vert V_S \vert}$ such, the sum over configurations $\vec{s}\notin \simSet$ in \cref{eq:part fun first split} can be upper bounded as
    \begin{equation}\label{eq:upper bound part fun 2}
        \sum_{\mathclap{\vec{s}\notin \simSet}}e^{-\beta H_S(\vec{s})} \leq q_S^{\vert V_S \vert}e^{-\beta(\Del + \Shift)}.
    \end{equation}
    Inserting \cref{eq:part fun approx} and the two respective upper bounds (\cref{eq:upper bound part fun 1} and \cref{eq:upper bound part fun 2}) into \eqref{eq:part fun first split} results in  \eqref{eq:part fun error}.
\end{proof}

The Boltzmann distribution can also be approximated. To derive this result, we first define the \emph{simulation distribution}.

\begin{definition}[simulation distribution]\label{def:sim prob}
    Let $\mc{f}\colon S \to T$. The \emph{simulation distribution} of $\mc{f}$ at $\beta$, $p_{\mc{f}, \beta}$, is the probability distribution on $\mc{C}_T$, defined as
    \begin{equation}\label{eq:sim distrib}
            p_{\mc{f},\beta}(\vec{t}) = \ \sum_{\mathclap{\substack{\vec{s} \in \mc{C}_S,\\ \dec \circ \vec{s} \circ \phys = \vec{t}}}} \  p_{S,\beta}(\vec{s}).
    \end{equation}
\end{definition}

If $S$ is at thermal equilibrium, then at temperature $\beta$ it is in configuration $\vec{s}$ with probability $p_{S,\beta}$. Now, for each such source configuration we can decode the target configuration that it corresponds to by restricting $\vec{s}$ to the physical spins and applying $\dec$, i.e.\ computing $ \dec \circ \vec{s} \circ \phys = \vec{t}$. 
The probability of obtaining such $\vec{t}$ is given by summing the probabilities of all $\vec{s}$ that decode to $\vec{t}$. This is how the simulation distribution $p_{\mc{f},\beta}$ is constructed.

In more abstract terms, up to decoding, $p_{\mc{f},\beta}$ is constructed from $p_{S,\beta}$ by marginalizing over auxiliary spins. If $\dec$ and $\phys$ are trivial, then $p_{\mc{f},\beta}$ is the marginal distribution of the physical spins.
Note that $p_{\mc{f},\beta}$ is normalized, so it defines a probability distribution on $\mc{C}_T$:
\begin{equation}
\begin{split}
    \sum_{\mathclap{\vec{t} \in \mc{C}_T}} p_{\mc{f}, \beta}(\vec{t}) 
    &= \sum_{\vec{t} \in \mc{C}_T} \sum_{\substack{\vec{s} \in \mc{C}_S,\\ \dec \circ \vec{s} \circ \phys = \vec{t}}} p_{S,\beta}(\vec{s})    \\
    &= \sum_{\vec{s} \in \mc{C}_S} p_{S,\beta}(\vec{s})\\
    &= 1.
    \end{split}
\end{equation}

\begin{proposition}[preservation of Boltzmann distribution]\label{thm:boltzmann sys}
    Let $S,T$ be spin systems, 
    $ \beta >0$ and 
    $\mc{f} \colon S \to T$ with cut-off $\Del>\mr{min}(H_T)$.  
    Then  
    \begin{equation}\label{eq:total variance}
    \begin{split}
        \Vert p_{\mc{f},\beta}- p_{T,\beta}\Vert 
        \leq&   \frac{1}{2} \biggl (1 + n_T(\Del) +  
        \mr{max}\Bigl( 1, \frac{md}{q_S^{\vert V_S \vert}}n_T(\Del)\Bigr)
        \biggr ) \\
        &\cdot \frac{q_S^{\vert V_S \vert}}{md} e^{-\beta (\Del-\mr{min}(H_T))} 
    \end{split}
    \end{equation}
    where $m = \vert \enc \vert$ and $n_T(\Del)$ is the number of high-energy target configurations.
\end{proposition}
Note that the RHS of \cref{eq:total variance} decreases exponentially with the cut-off $\Del$.
Whenever $\Del>\mr{max}(H_T)$, $n_T(\Del)=0$ and \cref{eq:total variance} can be simplified to 
\begin{equation}
    \Vert p_{\mc{f},\beta}- p_{T,\beta}\Vert 
        \leq \frac{q_S^{\vert V_S \vert}}{md} e^{-\beta (\Del-\mr{min}(H_T))}.
\end{equation}

\begin{proof}
    The total variation distance $\Vert \: \Vert$ is defined as
    \begin{equation}\label{eq:def tot var}
         \Vert p_{\mc{f},\beta}- p_{T,\beta}\Vert \coloneqq 
         \frac{1}{2}\sum_{\vec{t} \in \mc{C}_T} \left\vert p_{\mc{f}, \beta}(\vec{t})- p_{T,\beta}(\vec{t})\right\vert .
    \end{equation}
    First, we split the sum over target configurations $\vec{t}$ into a sum over low-energy target configurations $H_T(\vec{t})<\Del$ and a sum over high-energy target configurations $H_T(\vec{t})\geq \Del$.

    For low-energy target configurations, 
    we split the sum over source configurations in the definition of $p_{\mc{f},\beta}$ into low-energy and high-energy configurations.
    Additionally, we use that first, low-energy configurations $\vec{s}$ that satisfy $\dec \circ \vec{s} \circ \phys$ are precisely configurations from $\simul(\vec{t})$ and hence also satisfy  $H_S(\vec{s})-\Shift = H_T(\vec{t})$. Second, by \cref{lem:sim deg}, for each target configuration $\vec{t}$ there are exactly $m\cdot \degeneracy $ source configurations in $\simul(\vec{t})$. In total, we get 
    \begin{equation}
            p_{\mc{f},\beta}(\vec{t}) = \frac{1}{Z_S(\beta)/\alpha} 
            \left( e^{-\beta H_T(\vec{t})} + \xi(\vec{t})  \right)
    \end{equation}
    where $\alpha \coloneqq e^{-\beta \Shift} m d$ and 
    \begin{equation}\label{eq:def xi}
        \xi(\vec{t}) \coloneqq \frac{1}{e^{-\beta \Shift} m d}\sum_{\substack{\vec{s}\notin \simul(\vec{t}),\\ 
             \dec \circ \vec{s} \circ P = \vec{t}}} e^{-\beta H_S(\vec{s})}.
    \end{equation}

    Using this, we get 
    \begin{equation}\label{eq:low energy upper bound}
    \begin{split}
            &\left\vert p_{\mc{f}, \beta}(\vec{t})- p_{T,\beta}(\vec{t})\right\vert \\
            & \leq
            \frac{\xi(\vec{t})}{Z_S(\beta)/\alpha}  
            +\frac{e^{-\beta H_T(\vec{t})}}{Z_T(\beta)Z_S(\beta)/\alpha} \left \vert \frac{Z_S(\beta)}{\alpha} - Z_T(\beta) \right \vert. 
    \end{split}
    \end{equation}
    We now upper bound the contribution to \cref{eq:def tot var} from low-energy target configuration by upper bounding the two respective contributions from the two summands of the RHS of \cref{eq:low energy upper bound} separately.
    
    We start by upper bounding the contribution of the first summand.
    Since by assumption $\Del >\mr{min}(H_T)$ so the minimum energy of $H_S$ is $\mr{min}(H_T) + \Shift$ and by \cref{lem:sim deg} there are at least $md$ source configuration with this energy. Thus  $Z_S(\beta)/\alpha \geq e^{-\beta \mr{min}(H_T)}$. 
    
    Next, for each low-energy target configuration $\vec{t}$, the sum in $\xi(\vec{t})$ (see \cref{eq:def xi}) runs over those source configurations $\vec{s}$ outside of $\simSet$ that decode to $\vec{t}$, i.e.\ that satisfy $\dec\circ\vec{s}\circ \phys= \vec{t}$. For different target configurations these sets of source configurations are disjoint, so summing $\xi(\vec{t})$ for all low-energy target configuration in total yields a sum over at most $q_S^{\vert V_S \vert }$ source configurations.
    Since all these are outside of $\simSet$ they have high energy, i.e.\  satisfy $H_S(\vec{s})-\Shift\geq \Del$. 
    Thus, 
    \begin{equation}
        \sum_{\mathclap{\substack{\vec{t}\in \mc{C}_T,\\
        H_T(\vec{t})< \Del}}} \ \xi(\vec{t}) \leq \frac{1}{md} q_S^{\vert V_S \vert} e^{-\beta \Del}.
    \end{equation}

    Thus, the contribution of the first summand of \cref{eq:low energy upper bound} is at most 
    \begin{equation}
        \sum_{\mathclap{\substack{\vec{t}\in \mc{C}_T,\\
        H_T(\vec{t})< \Del}}} \ \frac{\xi(\vec{t})} {Z_S(\beta) / \alpha} \leq \frac{1}{md} q_S^{\vert V_S \vert} e^{-\beta (\Del-\mr{min}(H_T))}.
    \end{equation}

    Next, we upper bound the contribution of the second summand of \cref{eq:low energy upper bound}. 
    We again use that $Z_S(\beta)/\alpha \geq e^{-\beta \mr{min}(H_T)}$. The sum of $e^{-\beta H_T(\vec{t})}$ over low-energy target configurations can trivially be upper bounded by $Z_T$ and the absolute value of $Z_S(\beta)/
\alpha -Z_T(\beta)$ can be upper bounded by \cref{lem:sim part fun sys}. In total we thus get
    \begin{equation}
    \begin{split}
        &\sum_{\mathclap{\substack{\vec{t}\in \mc{C}_T,\\
        H_T(\vec{t})< \Del}}} \ \frac{e^{-\beta H_T(\vec{t})}}{Z_T(\beta)Z_S(\beta)/\alpha} \left \vert \frac{Z_S(\beta)}{\alpha} - Z_T(\beta) \right \vert \\
         &\hspace{1.25cm}\leq \mr{max}\left( \frac{1}{md}  q_S^{\vert V_S \vert}, n_T(\Del)\right)e^{-\beta( \Del-\mr{min}(H_T))},
    \end{split}
    \end{equation}
    where $n_T(\Del)$ is the number of high-energy target configurations.

    Thus, in total 
    \begin{equation}\label{eq:upper bound low energy}
        \begin{split}
            \sum_{\mathclap{\substack{\vec{t}\in \mc{C}_T,\\
        H_T(\vec{t})< \Del}}} \left\vert p_{\mc{f}, \beta}(\vec{t})- p_{T,\beta}(\vec{t})\right\vert  
        \leq &\left (1 +  
        \mr{max}\left( 1, \frac{md}{q_S^{\vert V_S \vert}}n_T(\Del)\right)
        \right ) \\
        & \cdot \frac{q_S^{\vert V_S \vert}}{md} e^{-\beta (\Del-\mr{min}(H_T))} 
        \end{split}
    \end{equation}

    Finally, we upper bound the contribution to \cref{eq:def tot var} from high-energy target configurations. Since for high energy target configurations $\simul(\vec{t})=\emptyset$, 
    we have 
    \begin{equation}
         p_{\mc{f},\beta}(\vec{t}) = \frac{\xi(\vec{t})}{Z_S(\beta)/\alpha}.
    \end{equation}
    Using the trivial upper bound $\xi(\vec{t})\leq \frac{1}{md}q_S^{\vert V_S \vert} e^{-\beta \Del}$, we thus get 
    \begin{equation}\label{eq:upper bound high 1}
        p_{\mc{f},\beta}(\vec{t}) \leq \frac{1}{md}q_S^{\vert V_S \vert} e^{-\beta (\Del-\mr{min}(H_T))}.
    \end{equation}
    Similarly, we get 
    \begin{equation}\label{eq:upper bound high 2}
        p_{T,\beta}(\vec{t}) \leq e^{-\beta (\Del-\mr{min}(H_T))}.
    \end{equation}
    We now upper bound the absolute value of $p_{\mc{f},\beta}(\vec{t})-p_{T,\beta}(\vec{t})$ by the maximum of the RHS of \cref{eq:upper bound high 1} and the RHS \cref{eq:upper bound high 2}. Since by assumption $\Del > \mr{min}(H_T)$ there must be at least $md$ source configurations, so $\frac{1}{md}q_S^{\vert V_S \vert}\geq 1$. Thus, 
    \begin{equation}
        \left\vert p_{\mc{f}, \beta}(\vec{t})- p_{T,\beta}(\vec{t})\right\vert \leq \frac{1}{md}q_S^{\vert V_S \vert} e^{-\beta (\Del-\mr{min}(H_T))}
    \end{equation}
    and 
    \begin{equation}\label{eq:upper bound high energy}
        \begin{split}
            &\sum_{\mathclap{\substack{\vec{t}\in \mc{C}_T,\\
        H_T(\vec{t})\geq \Del}}} \  \left\vert p_{\mc{f}, \beta}(\vec{t})- p_{T,\beta}(\vec{t})\right\vert \\
        & \hspace{2.25cm}\leq n_T(\Del) \frac{1}{md}q_S^{\vert V_S \vert} e^{-\beta (\Del-\mr{min}(H_T))}.
        \end{split}
    \end{equation} 

    Combining \cref{eq:upper bound low energy} and \cref{eq:upper bound high energy} we thus get \cref{eq:total variance}.
\end{proof}

\subsubsection{Common Transformations as Simulations}\label{sssec:symmetries}

Several `natural' transformations for spin systems are special cases of simulations, including graph isomorphisms (\cref{lem:iso sim}), weak symmetries (\cref{lem:ground state sym sim}) and binary simulations (\cref{lem:bin sim}). 
To see this, note that if $S=T$, for any cut-off $\Del$, we obtain a trivial simulation $S \to T$. 
This can be relaxed to $T$ merely being isomorphic to $S$, showing that simulations can be interpreted as generalized isomorphisms, as mentioned at the beginning of \cref{ssec:spin sys sim}.

\begin{proposition}[isomorphisms induce simulations]\label{lem:iso sim}
    Let $\phi \colon S \to T$ be an isomorphism. Then for arbitrary $\Del >0$, $S \to T$ with $\phys = \phi^{-1}$. 
\end{proposition}

\begin{proof}
    Since only $\Del$ and $\phys$ are non-trivial, 
    we only have to prove  conditions \ref{def:sim deg} and \ref{def:sim energy} from \cref{def:sys simulation}; the remaining conditions hold trivially. 
    
    Condition \ref{def:sim deg} is satisfied, since
    first, using that the decoding is trivial while the physical spin assignment is given by $\phi^{-1}$,
    $ \dec \circ \vec{s} \circ \phys = \vec{t}$ 
    is equivalent to $\vec{s} = \vec{t}\circ \phi$,
    and second, since $\phi$ is an isomorphism of spin systems it preserves the Hamiltonian, i.e.\ $ H_S(\vec{t}\circ \phi) = H_T(\vec{t})$ and thus, 
    for any low-energy target configuration, 
    \begin{equation}
     \simul(\vec{t}) = \{\vec{t} \circ \phi \}   
    \end{equation}
    and hence $\deg = 1$.

    This also proves the 
    case $\vec{s} \in \simul(\vec{t})$ of condition \ref{def:sim energy}.
    The case $\vec{s} \notin \simSet$ is satisfied, since
    any source configuration $\vec{s}$ trivially satisfies
    $\vec{s} = (\vec{s} \circ \phi^{-1})\circ \phi$, 
    so $\vec{s}\notin \simSet$ implies that $\vec{s}$ is a high-energy configuration. \end{proof}

Another class of simulations can be obtained from \emph{weak symmetries}. Given a bijection $\phi\colon V_S \to V_S$, we say that a spin system $S$ has symmetry  $\phi$ if for all configurations $\vec{s}$, $H_S(\vec{s})= H_S(\vec{s}\circ \phi)$.
While isomorphisms define symmetries, symmetries need not preserve the hyperedges and local interactions of $S$, as they are only required to preserve the Hamiltonian.

If a bijection $\phi$ on $V_S$ only preserves the Hamiltonian for configurations with energy below a cut-off $\Del$, we say that $S$ has \emph{weak symmetry} $(\phi,\Del)$.
If $S$ has weak symmetry $(\phi, \Del)$ and $H_S(\vec{s}\circ \phi)< \Del$ then also $H_S(\vec{s})<\Del$. This is the case because $V_S$ is finite, so $\phi$ has finite order, i.e.\ there exists a natural number $n$ with $\phi^n= \mr{id}$ and applying the weak symmetry $n-1$-times to $\vec{s}\circ \phi$ yields
\begin{equation}\label{eq:weak sym ord}
    H_S(\vec{s}\circ \phi) = H_S(\vec{s}\circ \phi^2) = \ldots = H_S(\vec{s}\circ \phi^{n-1}) = H_S(\vec{s}).
\end{equation}
This implies that $\phi$ not only maps low-energy configurations to low-energy configurations, but also high-energy configurations to high-energy configurations. In other words, both the high-energy and the low-energy subspace of the configuration space $\mc{C}_S$ are invariant under the action of $\phi$.

\begin{proposition}[weak symmetries induce simulations]\label{lem:ground state sym sim}
If $S$ has weak symmetry $(\phi,\Del)$ then $\phi$ defines a simulation $S \to S$
with cut-off $\Del$ and $\phys = \phi^{-1}$.
\end{proposition}

\begin{proof}
The proof is analogous to that of \cref{lem:iso sim}.
Condition \ref{def:sim deg} and the first case of condition \ref{def:sim energy} hold because $\phi$ preserves the Hamiltonian for low-energy configurations, and the second case of 
\ref{def:sim energy} holds because $\phi$ maps high-energy configurations to high-energy configurations.\end{proof}

Simulations can also be used to change the spin type of a spin system. We explicitly show how an arbitrary spin system with spin type $q$ can be simulated by a spin system with spin type $2$ by encoding each $q$-level spin into $\lceil \mr{log}_2(q) \rceil$ $2$-level spins, where $\lceil x \rceil$ denotes the smallest integer greater than or equal to $x$. A similar construction works for arbitrary spin type, i.e.\ from $q$- to $p$-level spins.

\begin{proposition}[binary simulation]\label{lem:bin sim}
Let $S$ be any spin system. For any $\Del \geq 0$, there exists 
a spin system $B(S,\Del)$ with spin type $2$ and a simulation $\mc{b}(S,\Del) \colon B(S,\Del) \to S$.
\end{proposition}

\begin{proof} 
Given any spin system $S$ and $\Del\geq 0$ we construct the required spin system $B(S,\Del)$ and simulation $\mc{b}(S,\Del)$.
To lighten the notation, we fix $S,\Del$ and omit the $(S,\Del)$ arguments, i.e.\ we write $B, \mc{b}$ instead of $B(S,\Del), \mc{b}(S,\Del)$. Note however that the entire construction depends on $(S,\Del)$.

We take $\Shift=0$.
Next, we take $k \coloneqq \lceil \mr{log}_2(q) \rceil$ and fix an injection
    \begin{equation}
        \enc:[q_S] \rightarrow [2]^k,
    \end{equation}
    and a left inverse of $\enc$, $\dec$, i.e.\ we have $\dec \circ \enc = \mr{id}$. 
    $\dec$ and $\enc$ clearly satisfy condition \ref{def:sim enc-dec} of \cref{def:sys simulation} and since $\vert \enc \vert =1$, $\enc$ also satisfies condition \ref{def:sim disjoint enc}.
        
    We define  $V_B\coloneqq V_S \times [k]$ and $\phys: V_S \rightarrow V_B^k$ by
    \begin{equation}
        \phys(v) \coloneqq ((v,1), \ldots, (v,k)).
    \end{equation}
    This clearly satisfies condition \ref{def:sim physical spins}. 
    Intuitively, $P(v)$ contains the $k$ spins that are use to encode $v$.

    Next, we set $E_B \coloneqq \{e \times [k] \mid e \in E_S \}$ and define   
    $J_B$ by the following case distinction: 
    \begin{equation}
        J_B(e \times [k])(\vec{s}) = J_S(e)(\dec \circ \vec{s} \circ \phys)
    \end{equation}
    if for all $v\in e$, $\vec{s} \circ \phys (v)$ is contained in the image of $\enc$, 
    and 
    \begin{equation}
        J_B(e \times [k])(\vec{s})= \Del 
    \end{equation}
    otherwise.
    
    By construction there are no auxiliary spins. This implies that 
    \begin{equation}\label{eq:binary config unique}
        \vec{s}_1\circ \phys = \vec{s}_2\circ \phys \Rightarrow \vec{s}_1 = \vec{s}_2.
    \end{equation}
    Let $\vec{t}\in \mc{C}_S$ and define $\vec{s}_{\vec{t}} \in \mc{C}_B$ by
    \begin{equation}
        \vec{s}_{\vec{t}}(v,i) = \enc^{(i)}(\vec{t}(v)).
    \end{equation}
    Then, by construction, we have 
    \begin{equation}\label{eq: bin phys assign}
        \vec{s}_{\vec{t}}\circ \phys = \enc \circ \vec{t}
    \end{equation}
    and by \eqref{eq:binary config unique} $\vec{s}_{\vec{t}}$ is the unique configuration satisfying this condition.
    It follows that $\simul(\vec{t}) \subseteq \{ \vec{s}_{\vec{t}}\}$.
    Note that $\simul(\vec{t})$ could be empty, since we have not yet taken into account the energy condition, $H_S(\vec{s})<\Del$.
    Finally, using the definition of $J_B$
    and the fact that, for all target spins $v$, $\vec{s}_{\vec{t}}\circ \phys(v)$ is contained in the image of $\enc$, we find 
    \begin{equation}
    \begin{split}
            H_{B}(\vec{s}_{\vec{t}}) &= \sum_{e \in E_S} J_B(e\times [k])(\vec{s}_{\vec{t}}) \\
            & = \sum_{e \in E_S} J(e)(\dec \circ \vec{s}_{\vec{t}} \circ \phys)  \\
        &= \sum_{e \in E_S} J(e)(\vec{t}) = H_S(\vec{t}),
        \end{split}
    \end{equation}
    where the second to last equality holds by \eqref{eq:binary config unique}. 
    Since $\Shift=0$, we have that if $H_S(\vec{t})<\Del$ then $\simul(\vec{t}) = \{\vec{s}_{\vec{t}} \}$, so $\mc{b}$ satisfies condition \ref{def:sim deg} with $d=1$.
    
    Moreover, 
    if $\vec{s}\notin \simSet$
    then either
    $\vec{s}= \vec{s}_{\vec{t}}$ for some $\vec{t}$ with $H_S(\vec{t})\geq \Del$ and hence also $H_{B}(\vec{s})= H_S(\vec{t})\geq \Del,$
    or there exists no $\vec{t}\in \mc{C}_S$ with $\vec{s}= \vec{s}_{\vec{t}}$. 
    In the latter case, there must exist at least one $v \in V_S$ with 
    \begin{equation}
        \vec{s}\circ \phys(v) \notin \mr{Im}(\enc).
    \end{equation}
    Then, by definition of $J_B$, whenever $v\in e$,
    \begin{equation}
        J_B(e\times [k])(\vec{s}\vert_{e\times[k]})= \Del ,
    \end{equation}
    and thus also $ H_{B}(\vec{s})\geq \Del $,
    so $\mc{b}$ satisfies condition \ref{def:sim energy}.
\end{proof}

\subsubsection{Graph Minors as Simulations} \label{sssec:minors}

Finally, sequences of graph-minor operations such as deletion of edges and vertices, or contractions of edges can be realized as spin system simulations, too (\cref{thm: minor sim}). 
To this end we consider spin systems which are defined on graphs (instead of hypergraphs) but we expect that similar results also hold for the general case.
Given a graph $G_S$, a minor $G_T$ of $G_S$ and a spin system $T$ on $G_T$, we show how local interactions can be assigned to the edges of $G_S$ that are contracted and deleted such that the resulting spin system $S$ on $G_S$ simulates $T$.
Intuitively, 
\begin{itemize}
\item[$\vartriangleright$] deleting an edge is achieved by setting its coupling to zero, and
\item[$\vartriangleright$] contracting an edge is achieved by assigning it a coupling that leads to zero energy if the two adjacent spins are in equal states, and energy $\Del$ if not.
\end{itemize}
We start by considering the case where $G_T$ is obtained from $G_S$ by a single edge contraction.

\begin{lemma}[edge contraction as a simulation]\label{lem:edge contr}
Consider graphs $G_S, G_T$ such that $G_T$ is obtained from $G_S$ by contracting a single edge. Let $T$ be any spin system on $G_T$.
Then, for any $\Del>0$, we can construct a spin system $S$ on $G_S$ such that $S \to T$.
\end{lemma}

\begin{proof}
Let us first spell out the assumption that $G_T$ is obtained form $G_S$ by contracting a single edge. Denote this edge $\{a,b\}$, so that vertices of $G_T$ are precisely all vertices from $V_S$ except for $b$.
Edges of $T$ are obtained by taking all edges of $S$ except for $\{a,b\}$, and renaming all occurrences of $b$ in these edges into $a$. We define $f_{b\to a} \colon V_S \to V_T$ that renames $b$ to $a$, 
\begin{equation}
f_{b\to a}(c) = \begin{cases}
a \ \text{if} \ c=b \\
c \ \text{else}
\end{cases}
\end{equation}
and further define the contraction map $f \colon E_S\setminus \{\{a,b\}\} \to E_T$ by 
\begin{equation}
    f(\{c,d\}) = \{f_{b\to a}(c), f_{ b \to a}(d)\}.
\end{equation} 
$f$ describes how edges of $G_S$ are mapped to edges of $G_T$ under the contraction of $\{a,b\}$.  
Note that $f$ is surjective but not necessarily injective, i.e.\ different edges from $G_S$ might be mapped to same edge in $G_T$ (consider the contraction of one edge of a triangle).

We construct the spin system $S$ on $G_S$ as follows. For each edge $e'$ of $G_T$, we choose one edge $e$ in $G_S$ that under the contraction is mapped to $e$, and attach the local interaction of $e'$ to $e$.
This amounts to first choosing a right inverse of $f$, $g$. 
We now assign local interactions to $E_S$.
We define
\begin{equation}\label{eq:local int contr}
    J_S(\{a,b\})(s_a,s_b) = \begin{cases}
0 \ \text{if} \ s_a= s_b \\
\Del \ \text{else}.
\end{cases}
\end{equation}
For each edge $e' \in E_T$
we define
\begin{equation}\label{eq:contr int non zero}
    J_S(g(e'))(\vec{t}\circ f_{b \to a}) =  J_T(e')(\vec{t}) ,
\end{equation}
while for the remaining edges $e \in E_S$ not contained in the image of $g$ and different from $\{a,b\}$, we set $   J_S(e) = 0$. 
It is immediate to conclude that $S \to T$ with cut-off $\Del$. 

We prove conditions \ref{def:sim deg} and \ref{def:sim energy}; the rest is trivial.
First, given any target configuration $\vec{t}$, 
we get a source configuration $\vec{s}_{\vec{t}} = \vec{t}\circ f_{b \to a}$.
Source configurations of this are precisely those that assign the same state to $a$ and $b$. 
Using that encoding, decoding and physical spin assignment are trivial, $\vec{s}_{\vec{t}}$ further satisfies
$\vec{s}_{\vec{t}}\circ P = \enc \circ \vec{t}$.
Next, by definition of $J_S$, we have 
\begin{equation}\label{eq:minor energy}
\begin{split}
    H_S(\vec{s}_{\vec{t}}) 
    &= \sum_{e\in E_S\setminus \{ \{ a,b \} \}}J_S(e)(\vec{s}_{\vec{t}}\vert_e)\\
    &=\sum_{e'\in E_T}J_S(g(e'))(\vec{s}_{\vec{t}}\vert _{g(e')}))\\
    &=\sum_{e'\in E_T} J_T(e')(\vec{t}\vert_{e'}) = H_T(\vec{t})
\end{split}
\end{equation}
where the first equality holds because $\vec{s}_{\vec{t}}$ assigns the same state to $a$ and $b$, so the local interaction \eqref{eq:local int contr} is zero.
The second equality follows from splitting the sum over remaining edges from $E_S$ 
into edges $e'$ from $E_T$, and for each such $e'$, summing over its preimage under $f$. 
By definition of $J_S$, the latter sum collapses to a single term, namely that corresponding to $g(e')$. 
The third equality follows by inserting the definition of $\vec{s}_{\vec{t}}$ and using \eqref{eq:contr int non zero}.
This proves that $\simul(\vec{t}) = \{\vec{s}_{\vec{t}}\}$, i.e.\ \ref{def:sim deg} with $\degeneracy=1$, as well as the first case of condition \ref{def:sim energy}.

We still need to show that if $\vec{s}\notin \simSet$ then 
$H_S(\vec{s})\geq \Del$, i.e.\ the second case of condition \ref{def:sim energy}.
This holds because if $\vec{s}\notin\simSet$ then either $\vec{s}= \vec{s}_{\vec{t}}$ for a target configuration $\vec{t}$ with $H_T(\vec{t})\geq \Del$ and the claim follows from \eqref{eq:minor energy}, 
or 
$\vec{s}(a)\neq \vec{s}(b)$ and the claim follows from \eqref{eq:local int contr}.
\end{proof}

Let us now consider the case where $G_T$ is obtained from $G_S$ via a sequence of contraction and deletion operations. 

\begin{proposition}[graph minors as simulations]\label{thm: minor sim}
    Let $G_T$ be a graph minor of $G_S$ and $T$ be a spin system defined on $G_T$. 
    For every $\Del> 0$, we can construct a spin system $S$ on $G_S$ such that $S \to T$.
\end{proposition}

\begin{proof}
    We prove that each step of the transformation of $G_S$ into $G_T$ via edge contraction or deletion can be achieved by a simulation. For single edge contractions this is shown in \cref{lem:edge contr}. Deletions of edges can be achieved by setting the corresponding local interaction to zero, and deletions of isolated vertices can be achieved by an appropriate choice of $\phys$.
    The overall simulation can be obtained by composing all these single-step simulations. This crucially relies on the fact that simulations can be composed, shown in \cref{thm:sim trans}.

    Now, by assumption, $G_T$ is a minor of $G_S$, so there exists $E_{\mr{del}}\in E_S$, $E_{\mr{contr}}\in E_S \setminus E_{\mr{del}}$, $V_{\mr{del}}\subseteq V_S$ such that deleting edges in $E_{\mr{del}}$ from $E_S$, deleting isolated vertices in $V_{\mr{del}}$ and then contracting the  edges in $E_{\mr{contr}}$ yields $G_T$.
    
    We first consider the graph $G_R$ obtained from $G_S$ by delete operations only, i.e.\ $V_R=V_S \setminus V_{\mr{del}}$,  and $ E_R = E_S \setminus E_{\mr{del}}$. 
    As $G_T$ is obtained from $G_R$ by contractions only, iteratively applying \cref{lem:edge contr} we obtain a spin system $R$ on $G_R$ and a simulation $\mc{f}\colon R \to T$. Explicitly, this simulation is constructed as follows.
    First, we fix an ordering of $E_{\mr{contr}} = \{e_1, \ldots, e_n\}$. Note that since the result of contracting edges does not depend on the order of contractions this ordering is arbitrary.
    Next, construct $R$ by iteratively applying \cref{lem:edge contr}.
    For $i\in \{0, \ldots, n\}$ the graph $G_{R_i}$ is obtained from $G_R$ by contracting the edges $\{e_1, \ldots, e_i\}$, so that 
    $G_{R_0}=G_R$, $G_{R_n}=G_T$ and  contracting $e_i$ in $G_{i-1}$ yields $G_{i}$.
    By \cref{lem:edge contr}, we can construct spin systems $R_i$ on $G_{R_i}$ together with simulations 
    \begin{equation}
        \mc{f}_i \colon R_i \to R_{i+1}
    \end{equation}
    with cut-off $\Del$ and identity encoding.
   By \cref{thm:sim trans} with $R\coloneqq R_0$, their composite $\mc{f}\coloneqq \mc{f}_{0} \circ \ldots \circ \mc{f}_{n-1}$ defines a simulation
    \begin{equation}
        \mc{f} \colon R \to T,
    \end{equation}
    with cut-off $\Del$, with $V_R = V_S \setminus V_{\mr{del}}$,  and $ E_R = E_S \setminus E_{\mr{del}}$.
    
    Next we construct a spin system $S$ on $(V_S,E_S)$ together with a simulation 
    \begin{equation}
       \mc{g}\colon  S \to R
    \end{equation}
    with cut-off $\Del$ and identity encoding.
    We define $J_S(e) \coloneqq J_R(e)$ for $e \in E_R$, and $J_S(e)\coloneqq 0$
     for $e \in  E_{\mr{del}}$.
    We let $\phys$ be the inclusion of $V_R$ into $V_S$. This satisfies \cref{def:sys simulation}.
    
    In total we have $\mc{g}\colon S \to R$ and $\mc{f}\colon R \to T$. Using \cref{thm:sim trans} we obtain
    \begin{equation}
        \mc{g}\circ \mc{f} \colon S \to T
    \end{equation}
    with cut-off $\Del$.
\end{proof}

Note that the construction of $S$ in \cref{thm: minor sim} is implicit, as its explicit form depends on the choice of right inverses for the contraction maps in each application of \cref{lem:edge contr}. 
In general, though, $S$ has the following form: 
\begin{enumerate}
\item On edges $e \in E_{\mr{del}}$, $J_S(e)\coloneqq 0$; 
\item  
A sequence of contractions gives rise to a contraction map $F$ describing  how edges from $E_S$ (that are not deleted) are mapped to edges from $E_T$. 
Given a right inverse $F'$ of this contraction map, $J_S$ is constructed analogously to the single edge contraction case. For each edge from $E_T$ we attach its local interaction to $F'(e)$, i.e.\
\begin{equation}\label{eq:edge int}
        J_S(F'(e'))(\vec{s}\circ F) = J_T(e')(\vec{s}),
    \end{equation}
\item 
For all edges not in the image of $F'$, $J_S(e)=0$, 
and for all edges that are contracted, $J_S$ is given by \eqref{eq:local int contr}.
\end{enumerate}

In \cref{thm: minor sim}, $F$ is the composite of single edge contraction maps and $F'$ is the composite of their right inverses. 
However, given $F$ one might also specify its right inverse $F'$ directly,  without specifying right inverses for each single edge contraction involved  in the sequence.
We remark that whenever we use \cref{thm: minor sim} throughout the paper, especially in \cref{sec:2d Ising}, we will explicitly illustrate the choice of $F'$.

\subsection{Modularity of Simulations}\label{ssec:spin sys sim prop}

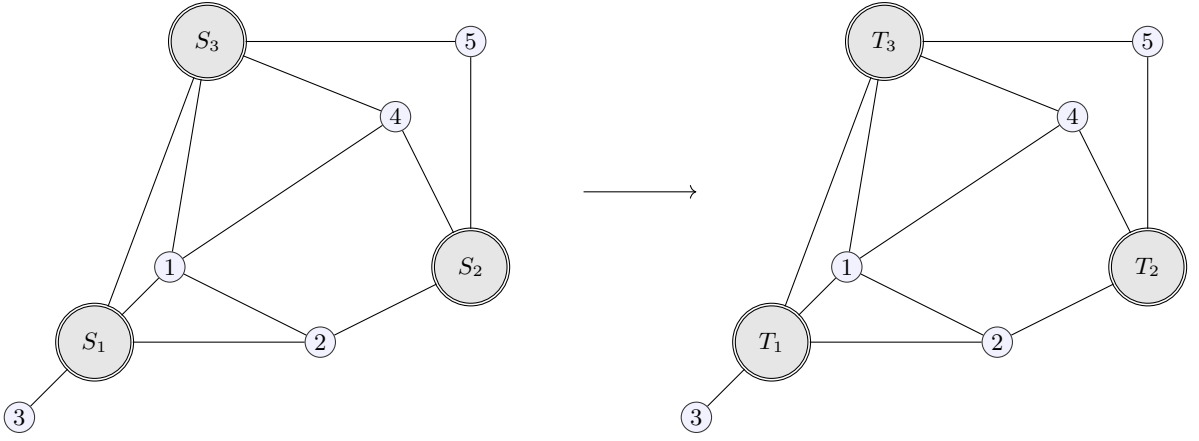
\begin{figure*}[th]
    \centering
   \begin{tikzpicture}
[sysNode/.style = {draw=black, double, fill=gray!20, circle, minimum size= 1cm, inner sep = 2pt},
cgnodeS/.style = {draw=black!80, fill=blue!10, thick, circle, minimum size= {width("$s_{i}$")+4pt}, inner sep = 1pt}]

\node[sysNode] (T1) at (0,0) {$S_1$};
\node[cgnodeS] (s1) at (1,1) {$1$};
\node[cgnodeS] (s2) at (3,0) {$2$};
\node[cgnodeS] (s3) at (-1,-1) {$3$};

\draw[thick] (s1) to (T1);
\draw[thick] (s2) to (T1);
\draw[thick] (s3) to (T1);
\draw[thick] (s1) to (s2);

\node[sysNode] (T2) at (5,1) {$S_2$};
\node[cgnodeS] (s4) at (4,3) {$4$};
\node[cgnodeS] (s5) at (5,4) {$5$};

\draw[thick] (s2) to (T2);
\draw[thick] (s4) to (T2);
\draw[thick] (s5) to (T2);
\draw[thick] (s1) to (s4);

\node[sysNode] (T3) at (1.5,4) {$S_3$};

\draw[thick] (T1) to (T3);
\draw[thick] (s1) to (T3);
\draw[thick] (s4) to (T3);
\draw[thick] (s5) to (T3);

\draw[thick,->] (6.5,2) to (8,2);

\node[sysNode] (2T1) at (9,0) {$T_1$};
\node[cgnodeS] (2s1) at (10,1) {$1$};
\node[cgnodeS] (2s2) at (12,0) {$2$};
\node[cgnodeS] (2s3) at (8,-1) {$3$};

\draw[thick] (2s1) to (2T1);
\draw[thick] (2s2) to (2T1);
\draw[thick] (2s3) to (2T1);
\draw[thick] (2s1) to (2s2);

\node[sysNode] (2T2) at (14,1) {$T_2$};
\node[cgnodeS] (2s4) at (13,3) {$4$};
\node[cgnodeS] (2s5) at (14,4) {$5$};

\draw[thick] (2s2) to (2T2);
\draw[thick] (2s4) to (2T2);
\draw[thick] (2s5) to (2T2);
\draw[thick] (2s1) to (2s4);

\node[sysNode] (2T3) at (10.5,4) {$T_3$};

\draw[thick] (2T1) to (2T3);
\draw[thick] (2s1) to (2T3);
\draw[thick] (2s4) to (2T3);
\draw[thick] (2s5) to (2T3);

\end{tikzpicture}
    \caption{Given simulations $S_i \to T_i$ for $i=1,2,3$, and a spin system $T$ that contains $T_1, T_2, T_3$ as subsystems  (right hand side), we locally replace each $T_i$ with the corresponding $S_i$ and thereby construct a spin system $S$ (left hand side) and a simulation $S \to T$ (by \cref{thm:sim sum}). 
    Note that $S \to T$ is a sum of the simulations $S_i \to T_i$ and the identity simulation for remaining subsystem of $T$. }
    \label{fig:sim sum}
\end{figure*}

We now define several operations that can be used to modify or combine simulations. 
Specifically, we prove that simulations can be 
composed (\cref{thm:sim trans}), 
scaled (by non-negative real numbers, \cref{thm:sim scale}) and 
added (\cref{thm:sim sum}).  

Composition of simulations is defined in the obvious way,
\begin{equation}
    \bigl( S \overset{\mc{g}}{\to}  R\bigr) \circ \bigl(R \overset{\mc{f}}{\to} T \bigr) = S \overset{\mc{g}\circ \mc{f}}{\longrightarrow} T.
\end{equation}
We denote the composition of simulations in diagrammatic order.
Scaling and addition of simulations are such that a positive linear combination of simulations yields a simulation between the corresponding positive linear combinations of the source and target systems, 
i.e.\ for $\mc{f}_i \colon S_i \to T_i$,
\begin{equation}
    \sum_i \lambda_i \cdot \mc{f_i} \colon \bigl( \sum_i \lambda_i \cdot S_i \bigr) \to \bigl( \sum_i \lambda_i \cdot T_i \bigr).
\end{equation}
This implies that we can not only obtain complex spin systems as linear combinations of simple ones, but also complex simulations as linear combinations and compositions of simple ones,  leading to a modular framework for the construction of spin system simulations. 
We term this concept \emph{modularity of simulation} (see \cref{fig:sim sum}). 

\begin{theorem}[composition of simulations]\label{thm:sim trans}
   Let $\mc{f}\colon R\to T$ and $\mc{g}\colon S \to R$ be simulations. Then, the following data defines a simulation $\mc{g}\circ\mc{f}\colon S \to T$:  
   \begin{equation}
   \begin{split}
            \Del_{\mc{g}\circ\mc{f}} &= \mr{min}(\Del_{\mc{f}},\Del_{\mc{g}}-\Shift_{\mc{f}})\\
            \Shift_{\mc{g}\circ\mc{f}} &=  \Shift_{\mc{f}}+\Shift_{\mc{g}}\\
            \degeneracy_{\mc{g}\circ\mc{f}} &= \degeneracy_{\mc{g}}\cdot \degeneracy_{\mc{f}}\\
            \phys_{\mc{g}\circ\mc{f}}^{(m,n)} &= \phys_{\mc{g}}^{(m)} \circ \phys_{\mc{f}}^{(n)}\\
            \dec_{\mc{g}\circ\mc{f}}(\vec{x}_1, \ldots, \vec{x}_{k_1})& = \dec_{\mc{f}}\bigl (\dec_{\mc{g}}(\vec{x}_1), \ldots, \dec_{\mc{g}}(\vec{x}_{k_1})\bigr)\\
            (\enc_{\mc{g}\circ\mc{f}})_{i,j}^{(m,n)} &= (\enc_{\mc{g}})_i^{(m)} \circ (\enc_{\mc{f}})_j^{(n)},
    \end{split}
    \end{equation}
where $k_1,k_2$ are the orders of the encodings $\enc_{\mc{f}}, \enc_{\mc{g}}$, respectively.
\end{theorem}

\begin{proof} See \cref{ssec:composition}.\end{proof}

In the above definition, we identified $V_S^{k_2\cdot k_1}$ with $V_S^{k_2\times k_1}$ and similarly for $[q_S]^{k_2\cdot  k_1}$, 
 where $\vec{x}_r$ is the $r$-th column vector of a matrix in $[q_S]^{k_2\times k_1}$.
The key idea of the construction of $\mc{g}\circ \mc{f}$ is to first use $\mc{f}$ to encode each spin of $T$ into $k_1$ physical spins of $R$, and then use $\mc{g}$ to encode each of these into $k_2$ physical spins of $S$, yielding, in total, for each spin of $T$, an $k_2\times k_1$ matrix of physical spins of $S$.
Also $\dec,\enc$ are defined w.r.t.\ this matrix of physical spins, in particular $\vert \enc_{\mc{g}\circ \mc{f}} \vert = \vert\enc_{\mc{g}} \vert \cdot \vert \enc_{\mc{f}} \vert $. 
It follows that  
\begin{equation}\label{eq:characterization sim3 main}
                (\simul_{\mc{g}\circ\mc{f}})_{i,j}(\Vec{t}) = \{ \Vec{s} \in (\simul_{\mc{g}})_i(\Vec{r}) \mid  \Vec{r} \in (\simul_{\mc{f}})_j(\Vec{t}) \}.
\end{equation}
That is, given any low-energy target configuration $\vec{t}$, we construct a source configuration $\vec{s}$ that simulates $\vec{t}$ w.r.t.\ to $\mc{g}\circ \mc{f}$ by first taking any $\vec{r}$ that simulates $\vec{t}$ w.r.t.\ $\mc{f}$, and then taking $\vec{s}$ that simulates $\vec{r}$ w.r.t.\ $\mc{g}$.
Equation \eqref{eq:characterization sim3 main} is proven in \cref{ssec:composition}.

Note that the composition of simulations is denoted in diagrammatic order, i.e.\ thinking of simulations as transformations between spin systems, $\mc{g}\circ \mc{f}$ means first applying $\mc{g}$ and then $\mc{f}$. 
The reason for this is that, in practice, we often construct simulations starting from the target, that is, given a target system $T$ we first construct $R$ and $\mc{f}\colon R \to T$, and then given $R$ (and possibly information over $\mc{f}$) we construct $S$ and $\mc{g}$. 
In this sense, the diagrammatic order corresponds to the operational order of constructing simulations.
This is also reflected in the fact that the encoding of configurations, specified by $\phys, \enc$ and ultimately $\simul$, is composed in this way. That is, configurations of $T$ are first encoded into configurations of $R$ (as specified by $\mc{f}$), which are then encoded into configurations of $S$ (as specified by $\mc{g}$). 

Finally, since in the low-energy sector (up to encoding configurations) we have $H_S-\Shift_{\mc{g}} = H_R$ and $H_R - \Shift_{\mc{f}} = H_T$, 
the shift of $\mc{g}\circ \mc{f}$ is necessarily of the form $\Shift_{\mc{f}}+\Shift_{\mc{g}}$.
Since low-energy target configurations $\vec{t}$ are precisely those which can be simulated by a configuration $\vec{r}$ such that $\vec{r}$ itself can be simulated by a source configuration $\vec{s}$, we have
\begin{equation}
\begin{split}
    H_R(\vec{r})-\Shift_{\mc{f}} &= H_T(\vec{t}) < \Del_{\mc{f}}\\
    H_S(\vec{s}) - \Shift_{\mc{g}} &= H_R(\vec{r}) < \Del_{\mc{g}} .
\end{split}
\end{equation}
In total, the cut-off of the composite simulation is of the form $\mr{min}(\Del_{\mc{f}}, \Del_{\mc{g}}-\Shift_{\mc{f}})$.

\begin{theorem}[scaling of simulations]\label{thm:sim scale}
Let $\mc{f} \colon S \to T $ and $\lambda \in \mathbb{R}_{\geq 0}$.  Then 
\begin{equation}
    \lambda \cdot \mc{f} \colon \lambda\cdot S \to \lambda \cdot T
\end{equation}
where 
\begin{equation}
    \Del_{\lambda\cdot \mc{f}}= \lambda \cdot \Del_{\mc{f}} \quad \text{and} \quad \Shift_{\lambda\cdot \mc{f}}= \lambda \cdot \Shift_\mc{f}
\end{equation}
and $\lambda \cdot \mc{f}$ agrees with $\mc{f}$ on the remaining data. 
\end{theorem}

\begin{proof}
    Only condition \ref{def:sim deg} and condition \ref{def:sim energy} of \cref{def:sys simulation} are affected by changing $\Del$ and $\Shift$. 
    These hold since 
    \begin{equation}
        H_S(\vec{s})-\Shift_\mc{f}< \Del_\mc{f} \Leftrightarrow \lambda \cdot H_S(\vec{s})- \lambda\cdot \Shift_\mc{f} < \lambda\cdot\Del_\mc{f}
    \end{equation}
    and hence also $\simul_\mc{f}(\vec{t})= \simul_{\lambda\cdot \mc{f}}(\vec{t})$.
\end{proof}

\begin{figure*}[th]
    \centering
    \begin{tikzpicture}
\pgfdeclarelayer{bg}    
\pgfsetlayers{bg,main} 

\coordinate[] (1) at (0,0) ;
\coordinate[] (2) at (2,0.66) ;
\coordinate[] (3) at (3.33,2) ;
\coordinate[] (4) at (1.33,3.33) ;
\coordinate[] (5) at (0,2) ;

\begin{pgfonlayer}{bg}
\pgfsetstrokeopacity{0.5}
\pgfsetfillopacity{0.5}
    \filldraw[very thick, fill=blue!30] ($(1)+(-0.8,-0.5)$) 
            to[out=300,in=180] ($(2) + (0.6,-0.8)$) 
            to[out=0,in=320] ($(3) + (0.8,0.2)$)
            to[out=140,in=120] ($(1) + (-0.8,-0.5)$);
     \filldraw[very thick, fill=red!30] ($(5)+(-0.8,-0.5)$) 
            to[out=340,in=240] ($(3) + (0.6,0)$) 
            to[out=60,in=320] ($(4) + (0.8,0.4)$)
            to[out=140,in=160] ($(5) + (-0.8,-0.5)$);
\end{pgfonlayer}

\coordinate[] (21) at (-10,0.5) ;
\coordinate[] (22) at (-7,1) ;
\coordinate[] (23) at (-4,3) ;
\coordinate[] (24) at (-8,5) ;
\coordinate[] (25) at (-10,3.5) ;

\coordinate (d) at (-7,-1);
\coordinate (u) at (-8,6);

\begin{pgfonlayer}{bg}
\pgfsetstrokeopacity{0.5}
\pgfsetfillopacity{0.5}
    \filldraw[very thick, fill=blue!10] ($(21)+(-0.8,-0.5)$) 
            to[out=300,in=180] ($(d) + (0.6,-0.8)$) 
            to[out=0,in=320] ($(23) + (0.8,0.2)$)
            to[out=140,in=120] ($(21) + (-0.8,-0.5)$);
     \filldraw[very thick, fill=red!10] ($(25)+(-0.8,-0.5)$) 
            to[out=340,in=240] ($(23) + (0.6,0)$) 
            to[out=60,in=320] ($(u) + (0.8,0.4)$)
            to[out=140,in=160] ($(25) + (-0.8,-0.5)$);
\end{pgfonlayer}

\begin{pgfonlayer}{bg}
\pgfsetstrokeopacity{0.5}
\pgfsetfillopacity{0.5}
    \filldraw[dashed, very thick, fill=blue!30] ($(21)+(-0.8,-0.5)$) 
            to[out=300,in=180] ($(22) + (0.6,-0.8)$) 
            to[out=0,in=320] ($(23) + (0.8,0.2)$)
            to[out=140,in=120] ($(21) + (-0.8,-0.5)$);
     \filldraw[dashed, very thick, fill=red!30] ($(25)+(-0.8,-0.5)$) 
            to[out=340,in=240] ($(23) + (0.6,0)$) 
            to[out=60,in=320] ($(24) + (0.8,0.4)$)
            to[out=140,in=160] ($(25) + (-0.8,-0.5)$);
\end{pgfonlayer}

\node[] (t1) at (4,0) {$V_{T_1}$};
\node[] (t2) at (4,3) {$V_{T_2}$};

\node[] (s1) at (-10.5,-1) {$V_{S_1}$};
\node[] (s2) at (-10.5,6) {$V_{S_2}$};

\node[] (p1) at (-8,-0.5) {$\physSet_{\mc{f}}$};
\node[] (p2) at (-9,5.75) {$\physSet_{\mc{g}}$};

\node[] (c1) at (1.5,0.5) {$\vec{t}\vert_{V_{T_1}}$};
\node[] (c2) at (1.5,3) {$\vec{t}\vert_{V_{T_2}}$};

\node[] (c21) at (-8.5,1) {$\vec{s}_1 \in (\simul_{\mc{f}})_i(\vec{t}\vert_{V_{T_1}})$};
\node[] (c22) at (-8.5,4.5) {$\vec{s}_2 \in (\simul_{\mc{g}})_i(\vec{t}\vert_{V_{T_2}})$};

\draw[thick, ->] (c1) to [bend left = 15] (c21);
\draw[thick, ->] (c2) to [bend right = 15] (c22);

\end{tikzpicture}
    \caption{
    Simulations are modular (\cref{thm:sim sum}): If $\mc{f}: S_1\to T_1$ and $\mc{g}: S_2\to T_2$, then $S_1+S_2\to T_1+T_2$ provided certain compatibility conditions are fulfilled. 
    The left hand side shows $V_{S_1+S_2}$ and the physical spin sets $\physSet_{\mc{f}}$ and $\physSet_{\mc{g}}$ of the simulations $\mc{f},\mc{g}$, 
    whereas the right hand side shows $V_{T_1+T_2}$ and a low-energy configuration $\vec{t}$ of $T_1+T_2$.  
  The constructed configurations $\vec{s}_1,\vec{s}_2$ agree on the overlap $V_{S_1}\cap V_{S_2}$ and, once combined, simulate $\vec{t}$.
    }
    \label{fig:thm sum}
\end{figure*}

\begin{theorem}[addition of simulations]\label{thm:sim sum}
    Let $\mc{f}\colon S_1 \to T_1$ and $\mc{g} \colon S_2 \to T_2$ such that 
    \begin{enumerate} 
        \item\label{eq: sim sum dec} $\dec_{\mc{f}} = \dec_{\mc{g}}$
        \item\label{eq:sum phys spins} $\phys_{\mc{f}} \vert_{V_{T_1}\cap V_{T_2}} = \phys_{\mc{g}} \vert_{V_{T_1}\cap V_{T_2}}$
        \item \label{eq:sum phys image} $\mr{Im}(\phys_{\mc{f}}\vert_{V_{T_1}\cap V_{T_2}})= V_{S_1} \cap V_{S_2}$ 
        \item\label{eq:sum common enc} $(\enc_{\mc{f}}) \cap (\enc_{\mc{g}}) \neq \emptyset$ 
        \item\label{eq:sum disjoint enc} 
        if $(\enc_{\mc{f}})_i(s) = (\enc_{\mc{g}})_j(s)$ for some $s\in [q_S]$ then $(\enc_{\mc{f}})_i = (\enc_{\mc{g}})_j$  
    \end{enumerate}
    Then $\mc{f}+\mc{g}\colon S_1+S_2 \to T_1 + T_2$ with 
    \begin{equation}
    \begin{split}
        \Del_{\mc{f}+\mc{g}} &= \mr{min}(\Del_{\mc{f}}, \Del_{\mc{g}})\\
        \Shift_{\mc{f}+\mc{g}} &= \Shift_{\mc{f}} + \Shift_{\mc{g}} \\
        \degeneracy_{\mc{f}+\mc{g}} &= \degeneracy_{\mc{f}}\cdot \degeneracy_{\mc{g}}\\
        \phys_{\mc{f}+\mc{g}}(t) & = \begin{cases}
            \phys_{\mc{f}}(t) \ \text{if} \ t \in V_{T_1}\\
            \phys_{\mc{g}}(t) \ \text{else} 
        \end{cases}\\
        \dec_{\mc{f}+\mc{g}} &= \dec_{\mc{f}}  \\
        \enc_{\mc{f}+\mc{g}} &= \enc_{\mc{f}} \cap \enc_{\mc{g}} 
    \end{split}
    \end{equation}
where in the above
\begin{equation}
\begin{split}
   &\mr{Im}(\phys_{\mc{f}} \vert_{V_{T_1}\cap V_{T_2}}) \\
   & \hspace{0.5cm} =\{\phys_{\mc{f}}^{(i)}(t) \mid 1 \leq i \leq \mr{ord}(\enc_{\mc{f}}), t \in V_{T_1}\cap V_{T_2} \} 
   \end{split}
\end{equation}
and $\enc_{\mc{f}} \cap \enc_{\mc{g}}$ denotes the list of encodings that are contained in both $\enc_{\mc{f}}$ and $\enc_{\mc{g}}$, ordered according to $\enc_{\mc{f}}$. 
\end{theorem}

\begin{proof} See \cref{ssec:addition sim}. \end{proof}

Note that simulations can be added only if they satisfy the compatibility conditions stated in \cref{thm:sim sum}. 

In essence, conditions \ref{eq: sim sum dec} and \ref{eq:sum common enc} are necessary to construct a decoding and encoding of the sum simulation $\mc{f}+\mc{g}$.
Condition \ref{eq:sum disjoint enc}  requires that the encodings $\enc_{\mc{f}}$ and $\enc_{\mc{g}}$ satisfy a condition similar to  \ref{def:sim disjoint enc} in \cref{def:sys simulation}, which  prevents the existence of low-energy source configurations $\vec{s}_i \in \mc{C}_{S_i}$ that agree on the overlap $V_{S_1} \cap V_{S_2}$ but stem from different encodings. Such configurations would have low energy but would not be included in $\simSet$, thus preventing $\mc{f}+\mc{g}$ from defining a simulation.

The remaining conditions, \ref{eq:sum phys spins} and \ref{eq:sum phys image}, ensure that the physical spin assignments agree on the overlap of target systems while auxiliary spins do not overlap. They allow us to construct configurations in $\vec{s} \in (\simul_{\mc{f}+\mc{g}})_i(\vec{t})$ with encoding $\enc_i \in \enc_{\mc{f}}\cap\enc_{\mc{g}}$ by combining configurations  $\vec{s}_1 \in (\simul_{\mc{f}})_i(\vec{t}\vert_{V_{T_1}})$, and $\vec{s}_2 \in (\simul_{\mc{g}})_i(\vec{t}\vert_{V_{T_2}})$. 
Since $\vec{s}_1, \vec{s}_2$ are simulating configurations which are  constructed with the same encoding, $\enc_i$, we have 
\begin{equation}
    \vec{s}_1 \circ \phys_{f}\vert_{V_{T_1}\cap V_{T_2}} =\enc_i \circ \vec{t}\vert_{V_{T_1}\cap V_{T_2}} = \vec{s}_2 \circ \phys_{g}\vert_{V_{T_1}\cap V_{T_2}} .
\end{equation}
Inserting conditions \ref{eq:sum phys spins}  and \ref{eq:sum phys image}
this implies 
\begin{equation}
    \vec{s}_1 \vert _{V_{S_1}\cap V_{S_2}} = \vec{s}_2 \vert _{V_{S_1}\cap V_{S_2}} .
\end{equation}
Thus, $\vec{s}_1, \vec{s}_2$ can be combined to a configuration $\vec{s}$ on $S_1+S_2$.
The crucial part of the proof of \cref{thm:sim sum}, provided in \cref{ssec:addition sim}, shows that all configurations $\vec{s}$ obtained this way in fact simulate the target configuration $\vec{t}$, i.e.\ they satisfy 
\begin{equation}
    H_{S_1+S_2}(\vec{s}) - \Shift_{\mc{f}+\mc{g}} = H_{T_1+T_2}(\vec{t}),
\end{equation}
while all other configurations of $S_1+S_2$ not obtained this way have high energy (see \cref{fig:thm sum}). 
Note that conditions \ref{eq:sum phys spins} and \ref{eq:sum phys image} can always be achieved by relabeling spins in $S_1$ and $S_2$.

Finally, observe that when adding simulations, their cut-offs are not added, but $\Del_{\mc{f}+\mc{g}} = \mr{min}(\Del_{\mc{f}}, \Del_{\mc{g}})$. This guarantees that 
whenever $H_{T_1+T_2}(\vec{t})<\Del_{\mc{f}+\mc{g}}$, both $H_{T_1}(\vec{t}\vert_{V_{T_1}})<\Del_{\mc{f}}$ and $H_{T_2}(\vec{t}\vert_{V_{T_2}})<\Del_{\mc{g}}$ hold.

Next we show that \cref{thm:sim sum} can be applied iteratively to construct sums of more than two simulations.

\begin{corollary}[finite sums of simulations]\label{cor:sim sum finite}
    If $\mc{f_i}\colon S_i \to T_i$ for $i=1, \ldots, n$ such that
\begin{equation}\label{eq:assumption common enc}
        \bigcap_{i=1}^n \enc_{\mc{f}_i} \neq \emptyset
    \end{equation}
    and the remaining conditions of \cref{thm:sim sum} are satisfied pairwise, then there is a simulation
    \begin{equation}
        \mc{f}\coloneqq \sum_{i=1}^n \mc{f}_i \colon \sum_{i=1}^n S_i \to \sum_{i=1}^n T_i 
    \end{equation}
    with 
    \begin{equation}\label{eq:def sim k sum}
    \begin{split}
        \Del_{\mc{f}} &= \mr{min}(\{\Del_{\mc{f}_i} \mid i=1, \ldots, n\})\\
        \Shift_{\mc{f}} &= \sum_{i=1}^n \Shift_{\mc{f}_i}\\
        \degeneracy_{\mc{f}}&= \prod_{i=1}^n \degeneracy_{\mc{f}_i}\\
        \phys_{\mc{f}} &= \phys_{\mc{f}_{i_t}}(t) \ \text{with} \ i_t= \mr{min}(\{i \mid t \in V_{T_{i_t}} \})\\
        \dec_{\mc{f}}&= \dec_{\mc{f_1}}\\
        \enc_{\mc{f}}&= \bigcap_{i=1}^n \enc_{\mc{f}_i}.
    \end{split}
    \end{equation}
\end{corollary}

\begin{proof}
    We prove the claim by induction over $n$. The case $n=2$ follows from \cref{thm:sim sum}.
    Assume the claim holds for $n-1$. 
    We start by writing 
    \begin{equation}
        \sum_{i=1}^n S_i =  \sum_{i=1}^{n-1} S_i + S_n .
    \end{equation}
    By assumption, we have simulations
    \begin{equation}\label{eq:sim k sum 1}
       \mc{f}_{<n} \coloneqq \sum_{i=1}^{n-1} \mc{f}_i \colon  \sum_{i=1}^{n-1}S_i \to \sum_{i=1}^{n-1}T_i
    \end{equation}
    and 
    \begin{equation}\label{eq:sim k sum 2}
    \mc{f_n} \colon S_n \to T_n.
    \end{equation}
    We now prove that these two simulations satisfy the conditions of \cref{thm:sim sum} and hence can be summed to obtain $\mc{f}$.
    
    Condition \ref{eq: sim sum dec} and \ref{eq:sum disjoint enc} follow from the assumption that $\{\mc{f_i} \mid i\leq n \}$ satisfy these conditions pairwise.
    Condition \ref{eq:sum common enc} holds by \eqref{eq:assumption common enc}.
    To see that condition \ref{eq:sum phys spins} holds too, we write
    \begin{equation}\label{eq:union intersection}
        \Bigl( \bigcup_{i=1}^{n-1} V_{T_i} \Bigr) \cap V_{T_n} =  \bigcup_{i=1}^{n-1} \Bigl( V_{T_i} \cap V_{T_n} \Bigr).
    \end{equation}
    Given any $v \in \Bigl( \bigcup_{i=1}^{n-1} V_{T_i} \Bigr) \cap V_{T_n}$, then $v\in V_{T_{i_v}}\cap V_{T_n}$ and by construction
    \begin{equation}
        \phys_{\mc{f}_{<n}} (v) = \phys_{\mc{f}_{i_v}}(v).
    \end{equation}
    Thus, condition \ref{eq:sum phys spins} holds since the original simulations $\{\mc{f}_i \mid i \leq n \}$ satisfy it pairwise. 
    Finally, for condition \ref{eq:sum phys image}, by \eqref{eq:union intersection} we have that
    \begin{equation}\label{eq:phys spin image union}
        \mr{Im}\Bigl(\phys_{\mc{f}_{<n}} \Big \vert_{\bigl( \bigcup_{i=1}^{n-1} V_{T_i}  \bigr) \cap V_{T_n}}\Bigr )= \bigcup_{i=1}^{n-1} \mr{Im}\Bigl( \phys_{\mc{f}_{<n}}  \Big \vert_{V_{T_i} \cap V_{T_n}}
        \Bigr).
    \end{equation}
    Since $\{\mc{f_i} \mid i\leq n \}$ pairwise satisfy condition \ref{eq:sum phys spins},
    \begin{equation}
        \phys_{\mc{f}_{<n}}  \Big \vert_{V_{T_i} \cap V_{T_n}}   = \phys_{\mc{f}_i}\Big \vert_{V_{T_i} \cap V_{T_n}}.
    \end{equation}
    Using this together with the assumption that $\{\mc{f_i} \mid i\leq n \}$ also satisfy condition \ref{eq:sum phys image} pairwise, the right hand side of \eqref{eq:phys spin image union} equals
    \begin{equation}
        \bigcup_{i=1}^{n-1} \bigl (V_{S_i} \cap V_{S_n} \bigr ) = \bigl( \bigcup_{i=1}^{n-1} V_{S_i} \bigr) \cap V_{S_n}.
    \end{equation}

   In total, the simulations $\mc{f}_{<n}$ and $\mc{f}_n$ satisfy the conditions of \cref{thm:sim sum}. They can thus be added to yield  
    \begin{equation}
                \Bigl( \sum_{i=1}^{n-1} \mc{f}_i\Bigr) + \mc{f}_n \colon  \sum_{i=1}^{n}S_i \to \sum_{i=1}^{n}T_i,
    \end{equation}
    which satisfies the requirements of \cref{thm:sim sum}. 
\end{proof}

\section{Spin Models and Emulations}\label{sec:spin models}

In this section we define and study spin models (\cref{ssec:spin model}), emulations (\cref{ssec:spin model simulation}), and finally characterize universal spin models (\cref{ssec:universality}).

\subsection{Spin Models}\label{ssec:spin model}

A spin model is a set of spin systems that is closed under isomorphisms and has constant spin type. 

\begin{definition}[spin model]\label{def:spin model}
A \emph{spin model} $\mc{M}$ is a set of spin systems such that for all $S,T\in \mc{M}$,
\begin{enumerate}
    \item $q_S=q_T$,
    \item\label{cond: spin model closed under iso} if $S'\cong S$ then $S'\in \mc{M}$.
\end{enumerate}
\end{definition}

Given a spin model $\mc{M}$, we denote its \emph{spin type} by $q_{\mc{M}}$.
It is worth noting that many of the following considerations work for general sets of spin systems. Only for some constructions we need the assumptions of \cref{def:spin model}. To distinguish an unrestricted set of spin systems from a spin model, in the following we denote sets of spin systems by $\mc{S}, \mc{T}, \ldots$ and reserve $\mc{M}, \mc{N}, \mc{M}', \ldots$ for spin models.

In practice, we often define spin models by specifying a set of spin systems $\mc{S}$ with constant spin type and consider its closure under isomorphisms. The latter is denoted by $[.]$, i.e.\ $[\mc{S}]$ defines a spin model with $S \in [\mc{S}]$ if and only if $S\cong S'$ for some $S' \in \mc{S}$.We can think of $[\mc{S}]$ as the spin model generated by $\mc{S}$.

We might also define spin models by restricting the available local interactions and  interaction hypergraphs. For example, allowing only Ising pair-interactions on square-grid graphs gives rise to the $2$d-Ising model (cf.\ \cref{ex:2d Ising}). 
\begin{definition}\label{def:spin model alt}
    Let $q\in \mathbb{N}_{\geq 2}$, $\mc{G}$ be a set of hypergraphs and $\mc{J}$ be a set of functions of type 
    \begin{equation}
        \mc{J} \subseteq \{[q]^{\{1, \ldots, l\}} \to \mathbb{R}_{\geq 0} \mid l \in \mathbb{N} \} .
    \end{equation}
    We define the spin model $[\mc{G}, \mc{J}]$ as the set containing all spin systems $S$ which, up to isomorphism, have interaction graph $G_S$ from $\mc{G}$ and local interactions $J_S(e)$ from $\mc{J}$.
\end{definition} 
In words, if $S \in [\mc{G}, \mc{J}]$ then $G_S$ is isomorphic to some $G \in \mc{G}$. In addition, for all $e\in E_S$, there exists a $j\in \mc{J}$ such that, up to isomorphism, i.e.\ up to relabeling spins from $e$, it holds that $J_S(e)=j$.

Conversely, given a spin model $\mc{M}$ we can recover its set of local interactions, resulting in the \emph{local spin model} of $\mc{M}$.
\begin{definition}[local spin model]\label{def:local spin model}
    Let $\mc{M}$ be a spin model. The \emph{local spin model} of $\mc{M}$, denoted $J(\mc{M})$, is defined by 
    \begin{equation}
        J(\mc{M}) = \{ S_{J_S(e)} \mid S \in \mc{M}, e \in E_S \} .
    \end{equation}
\end{definition}
This defines a spin model, as both closure under isomorphisms and constant spin type of $J(\mc{M})$ follow from the respective properties of $\mc{M}$.

Let us now see examples of spin models.
\begin{example}[flags]\label{ex:flag basis}
    Let $l \in \mathbb{N}$ and $\vec{x}\in [2]^{\{1, \ldots, l \}}$. We define 
    \begin{equation}
        \begin{split}
               f_{\vec{x}}& \colon [2]^{\{1, \ldots, l\}} \to \mathbb{R}\\
               f_{\vec{x}}(\vec{s}) &= \begin{cases}
               1 \ \text{if} \ \vec{s}=\vec{x}\\
               0 \ \text{else} . 
               \end{cases}
        \end{split}
    \end{equation}
    The space of functions of type $[2]^{\{1, \ldots, l\}} \to \mathbb{R}$ is a $2^l$ dimensional real vector space and $(f_{\vec{x}})_{\vec{x}\in [2]^{\{1, \ldots, l \}}}$ is its canonical basis.
    According to \cref{def:can sys}, $f_{\vec{x}}$ defines a spin system $S_{f_{\vec{x}}}$.
    We call $f_{\vec{x}}$ the \emph{flag function} on $\vec{x}$ and $S_{f_{\vec{x}}}$ the \emph{flag system} on $\vec{x}$.  
    
    The spin model containing all flag systems of order $l$ is given by  
    \begin{equation}
        \mc{B}_l \coloneqq \bigl [\{S_{f_{\vec{x}}} \mid \vec{x} \in [2]^{\{1, \ldots, l\}}\} \bigr ],
    \end{equation}
    and the spin model containing all flag systems by 
    \begin{equation}
        \mc{B} \coloneqq \bigcup_{l\in \mathbb{N}}\mc{B}_l.
    \end{equation}
\end{example}

    Flag systems and models will play a crucial role in the characterization of universality (\cref{thm:main}), as generic spin systems will be decomposed into non-negative linear combinations of flag systems. 

\begin{example}[generalized Ising model with fields]\label{ex:general Ising}
The Ising pair interactions from \cref{ex:Ising system} can be generalized to $j$-body interactions, for  arbitrary $j\in \mathbb{N}$, as follows: 
\begin{equation}
\begin{split}
\pi_j&\colon [2]^{\{1, \ldots, j\}} \to \mathbb{R}_{\geq 0}\\
\pi_{j}(\vec{s})&\coloneqq \begin{cases} 1 \ &\text{if } 
\vert \vec{s}^{-1}(\{1\}) \vert  \in 2\mathbb{N}\\ 
		0 \ &\text{else}\end{cases}
\end{split}
\end{equation}
and 
\begin{equation}
\begin{split}
\bar{\pi}_j&\colon [2]^{\{1, \ldots, j\}} \to \mathbb{R}_{\geq 0}\\
\bar{\pi}_{j}(\vec{s})&\coloneqq \begin{cases} 1 \ &\text{if } 
\vert \vec{s}^{-1}(\{1\}) \vert  \in 2\mathbb{N}+1\\
		0 \ &\text{else}\end{cases}.
\end{split}
\end{equation}
Note that these are parity interactions, as their energy only depends on the number of spins in state $1$. 

Finally, letting 
\begin{equation}
		\Pi_j \coloneqq \left\lbrace \lambda \cdot \pi_{j} \ |\ \lambda \in \mathbb{R}_{\geq 0} \right\rbrace \cup \left\lbrace \lambda \cdot \bar{\pi}_{j} \ |\ \lambda \in \mathbb{R}_{\geq 0} \right\rbrace 
\end{equation}
and leveraging \cref{def:spin model alt}, we define the \emph{j-body Ising model} as
\begin{equation}
    \mc{I}_j\coloneqq [\mc{G}_j, \Pi_ j ]
\end{equation}
where $\mc{G}_j$ is the set of $j$-uniform hypergraphs, that is, for each $G \in \mc{G}_j$, all of its hyperedges have cardinality $j$. 
Similarly, the \emph{$j$-body Ising model with fields} is defined as 
\begin{equation}
    \mc{I}_{j,f}\coloneqq [\mc{G}_{j,f}, \Pi_ j \cup \Pi_1 ]
\end{equation}
where $\mc{G}_{j,f}$ is obtained by adding all single vertex hyperedges to all hypergraphs in $\mc{G}_j$.
\end{example}

\begin{example}[2d Ising model with fields]\label{ex:2d Ising}
To define the 2d Ising model with fields we restrict the interaction graphs of $\mc{I}_{2,f}$ to $2$ dimensional grid graphs, that is, graphs with vertices $[m]\times [n]$ and edges between those vertices that are separated by distance $1$, either in the vertical or horizontal direction. 
Let $m,n\in \mathbb{N}_{\geq 2}$,
\begin{equation}
    V_{m,n} \coloneqq [m] \times [n], 
\end{equation}
and
\begin{multline}
    E_{m,n} \coloneqq \{ \{v,w\}  \mid v,w \in V_{m,n}, \\
    w = v + (0,1) \text{ or } w = v + (1,0) \}.
\end{multline} 
We define the \emph{$m \times n$ grid graph} 
\begin{equation}\label{eq:grid graph}
G_{m,n} \coloneqq (V_{m,n},  E_{m,n}),
\end{equation}
the set of all such grid graphs 
\begin{equation}
    \mc{G}_{\mr{grid}} \coloneqq \{G_{m,n} \mid m,n \in \mathbb{N}{\geq 2} \},
\end{equation}
and, similar to \cref{ex:general Ising},
the set of all grid graphs that additionally contain all hyperedges of cardinality $1$, as $\mc{G}_{\mr{grid},f}$. 
We now define the \emph{2d Ising model} as 
\begin{equation}
    \mc{I}_{2\mr{d}} \coloneqq [\mc{G}_{\mr{grid}}, \Pi_{2}] ,
\end{equation}
and the \emph{2d Ising model with fields }
\begin{equation}
    \mc{I}_{2\mr{d}, f} \coloneqq [\mc{G}_{\mr{grid},f}, \Pi_{2}\cup \Pi_1].
\end{equation}
\end{example}

\subsection{Emulations}\label{ssec:spin model simulation}

While previously we were only interested in the existence of spin system simulations, often they do not only exist but can be efficiently computed.
The latter may be a crucial property, for example when the simulations ought to serve to construct reductions between computational problems over spin models, such as computing ground states. 
We capture the idea of efficiently computable spin system simulations with \emph{spin model emulations}.

We shall first define them (\cref{sssec:def emu}), illustrate how binary simulations can be promoted to emulations (\cref{sssec:bin emu}) and then turn to their transformations: after defining hulls of spin models  (\cref{sssec:comp}), we will show that emulations can be composed and lifted (\cref{sssec:comp lift emu}). 

\subsubsection{Definition} \label{sssec:def emu}

Emulations are efficient algorithms that compute spin system simulations. More precisely, a spin model emulation of type $\mc{S} \to \mc{T}$ is a polytime computable function which, for each spin system $T\in \mc{T}$ and non-negative real $\delta$ as input, computes a spin system $S \in \mc{S}$ together with a spin system simulation of type $S \to T$ with cut-off $\delta$.

\begin{definition}[emulation]\label{def:spin model sim}
    Let $\mc{S}$ and $\mc{T}$ be sets of spin systems. A \emph{emulation} of type $\mc{S} \to \mc{T}$, denoted $ \text{\sc Emu}\colon \mc{S} \to \mc{T}$, 
    is an algorithm that, on input $(T, \delta)$ with $T \in \mc{T}$ and $\delta \in \mathbb{R}_{\geq 0}$, 
    outputs a pair
    \begin{equation}
        \text{\sc Emu}(T, \delta) = (\text{\sc Emu}^{(1)}(T,\delta), \text{\sc Emu}^{(2)}(T,\delta)),
    \end{equation}
    where $\text{\sc Emu}^{(1)}(T,\delta)$ is a spin system from $\mc{S}$, and  $\text{\sc Emu}^{(2)}(T,\delta)$ is a spin system simulation of type $\text{\sc Emu}^{(1)}(T, \delta) \to T$, such that the following conditions are satisfied: 
    \begin{enumerate}
    \item\label{def:model sim right cutoff} 
    $\text{\sc Emu}^{(2)}(T, \delta)$ has cut-off $\delta$; 
        \item\label{def:efficient sim} $\text{\sc Emu}$ is efficient, i.e.\ on input $(T,\delta)$ $\text{\sc Emu}$ has runtime $\mr{poly}(\vert T \vert)$; 
        \item\label{def:cut-off indep}  $\text{\sc Emu}$  is uniform,  i.e.\ 
        for any $\delta_1,\delta_2 \in \mathbb{R}_{\geq 0}$, 
        the interaction graphs of $\text{\sc Emu}^{(1)}(T,\delta_1)$ and $\text{\sc Emu}^{(1)}(T,\delta_2)$ are equal and the simulations $\text{\sc Emu}^{(2)}(T,\delta_1)$ and $\text{\sc Emu}^{(2)}(T,\delta_2)$ can only differ in their energy shift; and 
        \item\label{def:model sim enc} the decoding and encoding of $\text{\sc Emu}^{(2)}(T,\delta)$ only depend on $q_T$.
    \end{enumerate}
\end{definition}

Condition \ref{def:cut-off indep} states that part of the output of $\text{\sc Emu}$ can be computed independently of the second input, $\delta$. 
More precisely, for each $T \in \mc{T}$, changing $\delta$ may only affect the local interactions of the output spin system and the energy shift of the output simulation.

Condition \ref{def:model sim enc} ensures that systems with the same spin type are simulated by systems with the same spin type. This condition is needed for emulations to be compatible with addition of spin systems (see \cref{thm:model sim sum}), similarly to simulations. 
If $\mc{T}$ is a spin model  condition \ref{def:model sim enc} means that both decoding and encoding are constant.

Finally, similarly to \cref{def:sys simulation}, we might write $\mc{S} \to \mc{T}$ as shorthand for ``there exists an emulation $\text{\sc Emu} \colon \mc{S} \to \mc{T}$'' and call $\mc{T}$ the \emph{target (model)} and $\mc{S}$ the \emph{source (model)} of the emulation.
The first input of an emulation is called the \emph{input spin system} and the second input, the \emph{input cut-off}. 
The first output is called the \emph{output spin system} and the second, the \emph{output simulation}.

While simulations capture the idea of approximately reproducing the properties of a single target system within a source system, emulations capture the idea of efficiently doing so for an entire range of target systems, namely those contained in the target model. However, there is another interesting aspect in which emulations differ from simulation, which is best illustrated by considering a target spin model generated by a single spin system $T$. Then, a simulation of $T$ has a single, fixed cut-off and thus approximates properties of $T$ with constant error (see for instance \cref{lem:sim spectrum}, \cref{lem:sim part fun sys}, etc).
In contrast, an emulation of $[T]$ gives rise to a \emph{family of simulations} 
\begin{equation}
\mc{f}^{\delta}\colon S^{\delta} \to T
\end{equation}
indexed by their cut-off $\delta$. By uniformity (condition \ref{def:cut-off indep}), only the local interactions of $S$ and the shift of $\mc{f}^{\delta}$ are affected by the cut-off. There is thus a sense in which $\mc{f}^{\delta}$ approximates $T$ with arbitrary precision. If we interpret a simulation of $T$ as capturing some aspects of the system, an emulation of $[T]$ captures essentially all aspects of $T$. This is explored in more detail in \cref{sec:consequences}. 

The source of an emulation, $\mc{S}$, can be interpreted as providing the resources for the construction of simulations of spin systems form the target $\mc{T}$.
It is straightforward to see that increasing the source or decreasing the target also yields an emulation. That is, if $\mc{S} \to \mc{T}$ and $\mc{T}' \subseteq \mc{T}$ and $\mc{S} \subseteq \mc{S}'$ then $\mc{S}' \to \mc{T}'$. 
Further note that by \cref{lem:iso sim},  for any set of spin systems $\mc{S}$ we obtain a trivial emulation $\mc{S} \to \mc{S}$ that on input $(T, \delta)$ simply outputs $T$ together with the identity simulation with cut-off $\delta$. Combined with the previous observation, we conclude that whenever $\mc{T} \subseteq \mc{S}$, there is a trivial emulation $\mc{S} \to \mc{T}$.

\begin{remark}[polytime computability]\label{rem:poly}
Polytime computability means polytime w.r.t.\ the number of real parameters that specify the input, assuming that basic, real arithmetic is of zero cost, similar to the real RAM model \cite{Sh78}.
More precisely, in \cref{def:spin model sim}, we assume that spin systems in the input and output of $\text{\sc Emu}$ are presented such that their data can be directly read off. That is, the input $(T,\delta)$ of $\text{\sc Emu}$ consists of a list containing $V_T, E_T$ and, for each $e \in E_T$, the function graph of  $J_T(e)$, as well as the cut-off $\delta$, and similar for the output system. 
Furthermore, we assume that the output simulation is given by a list containing its entire data (according to \cref{def:sys simulation}), where again all functions are presented explicitly in terms of their function graphs.
Polytime computability is always understood w.r.t.\ this presentation of input and output. 
Recall from \eqref{eq:size of input S} that this is equivalent to requiring a runtime that is polynomial in the size of the input system, $\vert T \vert$.

Moreover, we assume that arithmetic of real numbers, addition and scaling of spin systems (see \cref{def:sum}, \cref{def:scaling}), elementary functions of real numbers including those in \cref{thm:model sim scale}, \cref{thm:model sim trans}, \cref{thm:model sim sum}, as well as exponentials and logarithms, are polytime computable. Most of these assumptions could be circumvented by working with rational numbers and/or approximations in the theorems of \cref{sec:consequences}.
\end{remark}

\subsubsection{Common Transformations as Emulations} \label{sssec:bin emu}

Many of the examples of spin system simulations of \cref{ssec:spin sys sim} can be promoted to spin model emulations, essentially because they are polytime computable. We illustrate this for spin system simulations that transform $q$-level to $2$-level spins, i.e.\ we extend \cref{lem:bin sim}.

\begin{proposition}[binary spin model emulation]\label{lem:bin sim model}
    Let $\mc{S}$ be any set of spin systems. Define
    \begin{equation}
        \mc{S}_{\mr{bin}} \coloneqq \{B(S,\delta) \mid S \in \mc{S}, \delta\geq 0 \},
    \end{equation}
    where $B(S,\Del)$ is defined in \cref{lem:bin sim}. Then 
    \begin{equation}
        \text{\sc Emu}_{\mr{bin}} \colon \mc{S}_{\mr{bin}} \to \mc{S},
    \end{equation}
    with 
    \begin{equation}
        \text{\sc Emu}_{\mr{bin}}(S,\delta) = (B(S,\delta), \mc{b}(S,\delta)),
    \end{equation}
   where $\mc{b}(S,\delta)$ is defined in \cref{lem:bin sim}.
\end{proposition}

\begin{proof}
We have to prove that $\text{\sc Emu}_{\mr{bin}}$ satisfies the conditions of \cref{def:spin model sim}. 
Conditions \ref{def:model sim right cutoff},  \ref{def:cut-off indep} and  \ref{def:model sim enc} hold by construction (cf.\ \cref{lem:bin sim}). 
We thus only have to prove condition \ref{def:efficient sim}, namely polytime computability of $ \text{\sc Emu}_{\mr{bin}}$.

By \cref{lem:bin sim}, we have that $V_{B(S, \delta)}$, $E_{B(S, \delta)}$ as well as $\enc$, $\dec$, $\Shift$ and $\phys$ of $\mc{b}(S, \delta)$ can be constructed in polytime. 
We still need to see that $J_{B(S,\delta)}$ can be constructed in polytime.
First, recall that for any $e\in E_S$,  $J_{B(S, \delta)}(e\times [k])$ is defined in terms of the following case distinction:
\begin{equation}
    J_{B(S,\delta)}(e\times [k])(\vec{s}) = J_S(e)(\dec \circ \vec{s} \circ \phys ) 
\end{equation}
if for all $v\in e$, $\vec{s} \circ \phys(v)$ is contained in the image of $\enc$,  and 
\begin{equation}
     J_{B(S,\delta)}(e\times [k])(\vec{s})= \delta 
\end{equation}
otherwise. 

The first case consists of those configurations $\vec{s}$ that satisfy
$\vec{s} = \vec{s}_{\vec{t}}$ for some $\vec{t}\in [q]^e$, where  $\vec{s}_{\vec{t}}$ is defined by $\vec{s}_{\vec{t}}(v,i) = \enc^{(i)}(\vec{t}(v))$. 
In particular, given $\vec{t}$, we have that $\vec{s}_{\vec{t}}$ can be constructed in polytime.
We iterate over all $\vec{t}\in [q]^e$:  
for each we compute $\vec{s}_{\vec{t}}$ and $J_S(e)(\vec{t})$ and append the resulting configuration-energy pair, $(\vec{s}_{\vec{t}},J_S(e)(\vec{t}))$, to a list. 
For each $e\in E_T$ this iteration consists of $q^{\vert e \vert}$, i.e.\ polynomially many steps, each of which is polytime computable, and thus it is itself polytime computable. 

For the second case, we iterate over the remaining configurations $\vec{s}$, i.e.\ those not 
contained in $C_1\coloneqq \{\vec{s}_{\vec{t}} \mid \vec{t}\in [q]^e\}$. 
For each such, we add $(\vec{s}, \delta)$ to the previously constructed list. 
The set of configurations not contained in $C_1$ is denoted $C_2$. 
The complexity of constructing the second part of $J_S(e\times [k])$ thus amounts to the complexity of
constructing $C_2$. 
Since for the first case we have constructed $C_1$, for $C_2$ we construct all $2^{\vert e \vert \cdot k}$ configurations and test membership in $C_1$ for each of them.
This amounts to $2^{\vert e \vert \cdot k}$ membership tests in a set of $2^{\vert e \vert}$ elements, which is polytime computable. Hence, the entire construction of $J_S(e\times [k])$ is polytime computable.

Finally, since there are polynomially many edges $e\in E_T$, the construction of $J_B$ and hence that of $\mc{b}(S,\delta)$ is polytime computable.\end{proof}

\subsubsection{Hulls of Spin Models} \label{sssec:comp}

In \cref{ssec:spin sys sim prop} we proved that simulations are compatible with addition and scaling of spin systems. 
We will now prove analogous results for emulations. 
To this end, we first need to extend addition and scaling to spin models.

For any sets $A,B$, define the following operations:
\begin{enumerate}
        \item $A+B \coloneqq \{ (a,b) \mid a \in A, b \in B\}$,
        \item $\Sigma(A)\coloneqq  \{(a_1, \ldots, a_k) \mid a_i \in A, k < \infty \}$,
        \item $\Scale( A) \coloneqq \{ (\lambda,  a) \mid \lambda \in \mathbb{R}_{\geq 0}, a \in A \}$,
        \item $\Cone(A)\coloneqq \Sigma\circ \Scale(A) $.
\end{enumerate}
We are interested in the case where $A,B$ are sets of spin systems $\mc{S}, \mc{T}$, so that we can use these operations to lift addition and scaling to spin models and emulations. 
We interpret $\mc{S}+\mc{T}$ as containing formal, i.e.\ non-evaluated sums of spin systems $S$ and $T$ from $\mc{S}$ and $\mc{T}$. Similarly, $\Sigma(\mc{S})$ is interpreted as containing all finite, formal sums of spin systems  $(S_1, \ldots, S_k)$ from $\mc{S}$; $\Scale(\mc{S})$ as containing all finite, formal, non-negative scalings of spin systems $S$ from $\mc{S}$; and $\Cone(\mc{S})$ as containing all finite, formal, non-negative linear combinations of spin systems $((\lambda_1,S_1), \ldots, (\lambda_k,S_k))$ from $\mc{S}$.

We term $\mr{id}, \Sigma, \Scale, \Cone$ and arbitrary compositions of these \emph{hull operators}, and write $\mr{hull}$ as shorthand for any such composition. 
This terminology is motivated by the fact that, up to emulations,  $\mr{id}, \Sigma, \Scale$ and $\Cone$ operators are idempotent, i.e.\ define hull operators in the mathematical sense. 
Further, we term the elements of $\mr{hull}(\mc{S})$ \emph{formal expressions}, and write $\mr{expr}$ as shorthand for such. 
Formal expressions are formal, non-negative linear combinations of spin systems, possibly with all coefficients equal to one (in the case of $\Sigma$), or only one summand (in the case of $\Scale$).
We define an \emph{evaluation} of such expressions, $\mr{eval}(\mr{expr})$, to be the result of computing the linear combination defined by $\mr{expr}$, i.e.\ the actual spin system obtained by adding and scaling spin systems according to $\mr{expr}$. 

We now extend the definition of  emulation, \cref{def:spin model sim}, to the case $\mr{hull}_1(\mc{S}) \to \mr{hull}_2(\mc{T})$. 
\begin{definition}[emulation and hull operations]\label{def:emulation mod}
    Let $S,T$ be sets of spin systems and $\mr{hull}_1, \mr{hull}_2$ be hull operators. We extend the definition of spin model emulation to 
    \begin{equation}
        \text{\sc Emu} \colon \mr{hull}_1(\mc{S}) \to \mr{hull}_2(\mc{T}),
    \end{equation}
    by modifying \cref{def:spin model sim} as follows: 
    \begin{enumerate}
        \item inputs of $\text{\sc Emu}$ are pairs containing a formal expressions $\mr{expr}_2 \in \mr{hull}_2(\mc{T})$ and a cut-off $\delta$; 
        \item\label{it:emu mod 2} outputs of  $\text{\sc Emu}$ are pairs containing a formal expressions $\mr{expr}_1 \in \mr{hull}_1(\mc{S})$ and a simulation of type $\mr{eval}(\mr{expr}_1) \to \mr{eval}(\mr{expr}_2)$; and 
        \item polytime is understood w.r.t.\ the total size of the input expression. 
    \end{enumerate}
\end{definition}
Similarly, we modify emulation to cover the cases where source and/or target are of the form $\mc{M}+\mc{N}$.

Consider for example the case $\mr{hull}_1 = \mr{hull}_2 = \Sigma$. An emulation with target $\Sigma(\mc{T})$ is a polytime algorithm that, given finitely many spin systems $(T_1, \ldots, T_k)$ from $\mc{T}$, computes a simulation with target $\sum_{i=1}^kT_i$. While the input to such an emulation is a formal expression, e.g.\ $(T_1, \ldots, T_k)$), its output simulation is a simulation whose target is the evaluation of this expression, i.e.\ the spin system $\sum_{i=1}^k T_i$. 
This distinction is crucial for the compatibility of emulation and hull operator (cf.\ \cref{thm:model sim scale}, \cref{thm:model sim sum}, \cref{lem:sim model cone}).
Similarly, an emulation with source $\Sigma(\mc{S})$ is a polytime algorithm that as first output returns a formal expressions $(S_1, \ldots, S_k)$ from  $\Sigma(\mc{S})$ while its second output is a simulation with source being the evaluation of this formal expression, i.e.\ the spin system $\sum_{i=1}^k S_i$.

It is straightforward to see that for any $\mc{S}$,
\begin{equation}\label{eq:hull identities}
    \begin{split}
        \Scale(\mc{S}) &\to \Scale^2(\mc{S}) \\
        \Sigma(\mc{S}) &\to \Sigma^2(\mc{S}) \\
        \Cone(\mc{S}) &\to \mr{hull}(\mc{S}),
    \end{split}
\end{equation}
where the last emulation exists for arbitrary hull operators on the right hand side. 

These emulations can be constructed by bringing the input expression to normal form,
i.e.\ simplifying it by distributivity and associativity until it is a formal sum of formal scalings of spin systems.
Normal forms can be computed in polytime. Given any expression from any $\mr{hull}(\mc{S})$, its normal form is from $\Cone(\mc{S})$.
If the expression is from $\Scale^2(\mc{S})$ or $\Sigma^2(\mc{S})$, its normal form is from $\Scale(\mc{S})$ or $\Sigma(\mc{S})$, respectively. 
By construction, any expression evaluates to the same spin system as its normal form, as scaling and addition of spin systems satisfy the usual distributive and associative laws. 
In total, the above emulations can simply be defined by returning the normal form of their input expression together with the identity simulation.

\subsubsection{Emulations Can Be Composed and Lifted to Hulls} \label{sssec:comp lift emu}

We shall now see that emulations can be composed (\cref{thm:model sim trans}) and are compatible with hull operators, 
in the sense that we can \emph{lift} any emulation $\mc{M} \to \mc{N}$ to an emulation $\mr{hull}(\mc{M}) \to \mr{hull}(\mc{N})$ via
applying $\Scale$ (\cref{thm:model sim scale}),
$\Sigma$ (\cref{thm:model sim sum}), 
and $\Cone$ (\cref{lem:sim model cone}). 

These results are the spin model counterparts of the \emph{modularity} of simulations, \cref{thm:sim trans}, \cref{thm:sim scale} and \cref{thm:sim sum}, and heavily rely on their spin system cousins. 
They imply that complicated spin model emulations can be constructed by composing and lifting simpler ones, a fact that will be leveraged to prove \cref{thm:main}.

\begin{theorem}[composition of emulations]\label{thm:model sim trans}
Let $\mc{S}, \mc{R}, \mc{T}$ be sets of spin systems, $\mr{hull}_1, \mr{hull}_2, \mr{hull}_3$ be hull operators, and $\text{\sc Emu}_{\mc{f}} \colon \mr{hull}_2(\mc{R}) \to \mr{hull}_3(\mc{T})$ and $\text{\sc Emu}_{\mc{g}} \colon \mr{hull}_1(\mc{S}) \to \mr{hull}_2(\mc{R})$ be emulations. Then we can construct the composite emulation 
\begin{equation}
\text{\sc Emu}_{\mc{g}\circ \mc{f}} \colon \mr{hull}_1(\mc{S})  \to \mr{hull}_3(\mc{T}).
\end{equation}
\end{theorem}

\begin{proof}
We only prove the case where all hull operators are the identity, i.e.\ we have emulations $\text{\sc Emu}_{\mc{f}} \colon \mc{R} \to \mc{T}$ and $\text{\sc Emu}_{\mc{g}} \colon \mc{S} \to \mc{R}$.
The case where $\mr{hull}_1, \ldots, \mr{hull}_3$ are generic hull operators works similarly. 

We construct the composite emulation $\text{\sc Emu}_{\mc{g} \circ \mc{f}}$ from 
$\text{\sc Emu}_{\mc{f}}$ and $\text{\sc Emu}_{\mc{g}}$, together with \cref{thm:sim trans}. 
Say 
\begin{equation}\label{eq:notation f}
    \text{\sc Emu}_{\mc{f}}(T,\delta) = (R(T,\delta), \mc{f}(T,\delta))
\end{equation}
and 
\begin{equation}\label{eq:notation g}
    \text{\sc Emu}_{\mc{g}}(R, \delta) = (S(R,\delta),\mc{g}(R,\delta)) . 
\end{equation}
Then we define $\text{\sc Emu}_{\mc{g} \circ \mc{f}}$ as follows:
\begin{equation}\label{eq:def sim 3}
\begin{split}
\text{\sc Emu}_{\mc{g} \circ \mc{f}}^{(1)}(R, \delta) &= S \bigl( R(T, \delta), \delta + \Shift_{\mc{f}(T,\delta)} \bigr),\\
\text{\sc Emu}_{\mc{g} \circ \mc{f}}^{(2)}(R, \delta) &= \mc{g}\bigl (R(T, \delta), \delta + \Shift_{\mc{f}(T,\delta)} \bigr)  \circ  \mc{f}(T, \delta) .
\end{split}
\end{equation} 
In words and omitting all arguments of $S,T,\mc{f},$, given $T\in \mc{T}$ and $\delta \geq 0$ we first run $\text{\sc Emu}_{\mc{f}}$ on input $(T,\delta)$.
This yields a spin system $R$ and a simulation $\mc{f}\colon R \to T$.
We then run $\text{\sc Emu}_{\mc{g}}$ on input $(R, \delta + \Shift_{\mc{f}})$.
This yields a spin system $S$ and a simulation $\mc{g}\colon S \to R$. 
Finally, 
$\text{\sc Emu}_{\mc{g} \circ \mc{f}}$  returns $S$ and the composite simulation $\mc{g} \circ \mc{f} \colon S \to T$. 
Since $\mc{g}$ has cut-off $\delta + \Shift_{\mc{f}}$, by \cref{thm:sim trans}, $\mc{g} \circ \mc{f}$ has cut-off $\delta$.

Finally, we need to see that $\text{\sc Emu}_{\mc{g}\circ \mc{f}}$ satisfies the conditions of \cref{def:spin model sim}.
We have already argued that condition 
\ref{def:model sim right cutoff} holds by \cref{thm:sim trans}.
Conditions \ref{def:cut-off indep} and \ref{def:model sim enc} hold since they are satisfied by both $\text{\sc Emu}_{\mc{f}}$ and $\text{\sc Emu}_{\mc{g}}$.
Thus, we only have to show that $\text{\sc Emu}_{\mc{g} \circ \mc{f}}$ is polytime computable, i.e.\ satisfies condition \ref{def:efficient sim}.

Computing $\text{\sc Emu}_{\mc{g} \circ \mc{f}}(R, \delta)$ requires running both $\text{\sc Emu}_{\mc{f}}$ and $\text{\sc Emu}_{\mc{g}}$ once. Since both are polytime computable by assumption, this is polytime computable. 
Therefore, to see that $\text{\sc Emu}_{\mc{g} \circ \mc{f}}$ is polytime computable, we have to argue that given $\mc{g}$ and $\mc{f}$, their composite $\mc{g}\circ \mc{f}$ is polytime computable. By \cref{thm:sim trans} computing the data of $\mc{g}\circ \mc{f}$ amounts to applying basic manipulations, such as composition of functions or arithmetic of real numbers, to the data of $\mc{g}$ and $\mc{f}$.

Finally, in the generic case, $\text{\sc Emu}_{\mc{g}\circ \mc{f}}$ is defined in the same way, i.e.\ by \eqref{eq:def sim 3}. By the same reasoning, this defines an emulation, i.e.\ satisfies  \cref{def:emulation mod}. \end{proof}

Next, we prove that emulations can be lifted to hull operators. 
We only prove the case where the original emulation is of the form $\mc{M} \to \mc{N}$, as the results easily extend to the generic case $\mr{hull}_1(\mc{M}) \to \mr{hull}_2(\mc{N})$ (see \cref{rem:mod  lift}).

\begin{theorem}[lifting emulations to $\Scale$]\label{thm:model sim scale}
    Let $\mc{S}, \mc{T}$ be sets of spin systems and $\text{\sc Emu} \colon \mc{S} \to \mc{T}$ be an emulation. 
    Then we can construct
    \begin{equation}
        \text{\sc Emu}_{\Scale} \colon \Scale(\mc{S}) \to \Scale( \mc{T}). 
    \end{equation}
\end{theorem}

\begin{proof}
We prove the claim by leveraging \cref{thm:sim scale}.
Given $T \in \mc{T}$, $\lambda, \delta \geq 0$  
we first compute $\text{\sc Emu}(T, \delta/\lambda)$ to obtain $S(T,\delta/\lambda) \in \mc{S}$ and 
\begin{equation}
    \mc{f}(T, \delta/\lambda) \colon S(T, \delta/\lambda) \to T.
\end{equation}
By construction this simulation has cut-off $\delta/\lambda$.
Then, we set 
\begin{equation}
    \text{\sc Emu}_{\Scale} (\lambda,T,\delta) = ( \lambda, S(T, \delta/\lambda), \lambda \cdot \mc{f}(T, \delta/ \lambda)).
\end{equation}
By \cref{thm:sim scale}, $\lambda \cdot \mc{f}$ defines a simulation of type $\lambda \cdot S \to \lambda \cdot T$ with cut-off $\delta$. 
Hence, $\text{\sc Emu}_{\Scale}$ satisfies \cref{def:emulation mod}.
Finally, since $\text{\sc Emu}$ defines an emulation, i.e.\ satisfies the conditions of \cref{def:spin model sim}, so does $\text{\sc Emu}_{\Scale}$.
In particular, polytime computability of $\text{\sc Emu}_{\Scale}$ follows from poyltime computability of $\text{\sc Emu}$ and the fact that $\lambda \cdot \mc{f}$ can be efficiently computed from $\lambda$ and  $\mc{f}$ by \cref{thm:sim scale}. \end{proof}

\begin{theorem}[lifting emulations to $\Sigma$]\label{thm:model sim sum}
    Let $\mc{M},\mc{N}$ be spin models and let
    $\text{\sc Emu} \colon \mc{M} \to \mc{N}$ be an emulation.   
    Then we can construct an emulation 
    \begin{equation}
        \text{\sc Emu}_{\Sigma} \colon \Sigma(\mc{M}) \to \Sigma(\mc{N}).
    \end{equation}
\end{theorem}

\begin{proof}
We explicitly construct $\text{\sc Emu}_{\Sigma}$ and prove that it satisfies \cref{def:spin model sim}.
The construction relies on \cref{cor:sim sum finite}, where we achieve the compatibility conditions of \cref{cor:sim sum finite} by using that spin models are closed under isomorphism (cf.\ \cref{def:spin model}) and that isomorphisms define spin system simulations (\cref{lem:iso sim}).

On input $(T_1, \ldots, T_k, \delta)$, 
for each $i\leq k$, we first compute $\text{\sc Emu}(T_i, \delta)$. This yields $k$ spin systems $S(T_i, \delta)$ and $k$ simulations 
\begin{equation}
\mc{f}(T_i, \delta) \colon S(T_i, \delta) \to T_i.
\end{equation}
For simplicity, we omit the dependencies of $S(T_i,\delta), \mc{f}(T_i,\delta)$ and write $S_i,\mc{f}_i$. 
To apply \cref{cor:sim sum finite}, we modify the simulations $\mc{f}_i$ such that they satisfy the conditions of \cref{cor:sim sum finite}.
By \cref{def:spin model} and \cref{def:spin model sim}, for any $i,j\leq k$ we have
\begin{equation}\label{eq:sum theorem conditions 1}
    \begin{split}
        q_{S_i} &= q_{S_j}\\
        \dec_{\mc{f}_i} &= \dec_{\mc{f}_j}\\
        \enc_{\mc{f}_i} &= \enc_{\mc{f}_j}.
    \end{split}
\end{equation}

We need to show the simulations satisfy, pairwise, condition \ref{eq:sum phys spins} and condition \ref{eq:sum phys image} of \cref{cor:sim sum finite}, that is, that
\begin{equation}\label{eq:sum cond 3}
    \phys_{\mc{f}_i}\vert_{V_{T_i}\cap V_{T_j}} = \phys_{\mc{f}_j}\vert_{V_{T_i}\cap V_{T_j}}
\end{equation}
and 
\begin{equation}\label{eq:sum cond 4}
    \mr{Im}(\phys_{\mc{f}_i}\vert_{V_{T_i}\cap V_{T_j}}) = V_{S_i}\cap V_{S_j}.
\end{equation}
Both conditions can be achieved by relabeling the spins of the source systems $S_1, \ldots,S_k $. 
Namely, for $i =1, \ldots, k$, we first relabel all spins such that $V_{S_i}\cap V_{S_j}=\emptyset$.
Then we iterate over spins $v \in \bigcup_{i=1}^k V_{T_i}$. For each $v$ we determine 
\begin{equation}
    i_v = \mr{min}\{ i \mid v \in V_{T_i}\}
\end{equation}
and iterate over $j\in \{i_v+1,\ldots,k\}$. Whenever $v \in V_{T_j}$ we relabel the corresponding physical spins $\phys_{\mc{f}_j}(v)$ from $V_{S_j}$ to $\phys_{\mc{f}_{i_v}}(v)$.
Thereby, after relabeling, all source systems intersect precisely on the physical spins corresponding to target spins from intersections of target systems. Moreover, on the intersections of target systems, the physical spin assignments agree. 
That is, after relabeling, \eqref{eq:sum cond 3} and \eqref{eq:sum cond 4} are satisfied.
Note that these relabelings can be performed in polytime.

In total the relabelings define isomorphisms (see \cref{def:iso spin sys})
\begin{equation}
    \phi_i \colon R_i \to S_i,
\end{equation}
where $R_i$ is the relabeled version of $S_i$. 
According to \cref{lem:iso sim}, the $\phi_i$ induce simulations
\begin{equation}
    \mc{g}_i \colon R_i \to S_i
\end{equation}
with arbitrary cut-off. We pick the cut-off
\begin{equation}
    \Del_{\mc{g}_i} = \delta + \Shift_{\mc{f}_i}.
\end{equation}
Then, by \cref{thm:sim trans}, we obtain simulations 
\begin{equation}
   \mc{g}_i \circ  \mc{f}_i \colon R_i \to T_i
\end{equation}
with cut-off $\delta$. 
By construction, these satisfy the compatibility conditions of \cref{cor:sim sum finite}, and by this result 
we obtain 
\begin{equation}
    \sum_{i=1}^k \bigl( \mc{g}_i \circ  \mc{f}_i \bigr) \colon
    \sum_{i=1}^k R_i\to \sum_{i=1}^k T_i.
\end{equation}
Define 
\begin{equation}
  \text{\sc Emu}_{\Sigma} (T_1, \ldots, T_k, \delta)  = \Bigl( (R_1, \ldots, R_k),  \sum_{i=1}^k \bigl( \mc{g}_i \circ  \mc{f}_i \bigr) \Bigr ).
\end{equation}
Note that all spin systems $R_i$, as well as simulations $\mc{f}_i,\mc{g}_i$ depend on $\delta$.
However, for the simulations $\mc{g}_i$ this dependency is only due to the choice of cut-off, as by uniformity of emulation (condition \ref{def:cut-off indep} of \cref{def:sys simulation}) the isomorphisms $\phi_i$ do not depend on $\delta$.

To finish the proof, we need to show that this satisfies \cref{def:emulation mod}. 

First, by construction, all simulations $\mc{g}_i\circ \mc{f}_i$ have cut-off $\delta$, so their sum has cut-off $\delta$, as well. Thus, $\text{\sc Emu}_{\Sigma}$ satisfies condition \ref{def:model sim right cutoff}.

Second, we prove that $\text{\sc Emu}_{\Sigma}$ is efficient, i.e. satisfies condition \ref{def:efficient sim}. 
By assumption, the spins systems $S_i$ and simulations $\mc{f}_i$ are polytime constructable, namely by using the emulation $\text{\sc Emu}$.
We have previously argued that the relabelings $\phi$ and hence both  the spin systems $R_i$ as well as the simulations $\mc{g_i}$ are poyltime constructable. 
As argued in \cref{thm:model sim trans}, composing simulations (according to \cref{thm:sim trans}) is polytime computable and thus also the simulations $\mc{g}_i \circ \mc{f}_i$ can be computed in polytime.
Similarly, the addition of spin system simulations (according to \cref{thm:sim sum}) is polytime computable since it amounts to applying basic manipulations to the individual simulations, which hence implies that also adding the simulations $\mc{g}_i \circ \mc{f}_i$ is polytime computable.
In total, this proves that 
$\text{\sc Emu}_{\Sigma}$ is polytime computable and hence satisfies condition \ref{def:efficient sim}.

To see that $\text{\sc Emu}_{\Sigma}$ is uniform, i.e.\ satisfies condition \ref{def:cut-off indep}, first, note that by assumption $\text{\sc Emu}$ is uniform, so the spin systems $S_i$ satisfy condition \ref{def:cut-off indep}. Second, as argued before, 
 the isomorphisms $\phi_i$ do not depend on the cut-off. 
 Thus, also the spin systems $R_i$ satisfy condition \ref{def:cut-off indep}, and by \cref{lem:iso sim}, also the corresponding simulations $\mc{g}_i$ satisfy condition \ref{def:cut-off indep}. 
 Finally, by \cref{thm:sim trans} and \cref{thm:sim sum},  
 the also $\sum_{i=1}^k \mc{g}_i \circ \mc{f}_i$ and therefore also
 $\text{Emu}_{\Sigma}$ satisfy condition \ref{def:cut-off indep}.

Condition \ref{def:model sim enc} holds because all simulations $\mc{f}_i$ have the same encoding $\enc$ and by \cref{lem:iso sim} all simulations $\mc{g}_i$ have encoding $\mr{id}$ Therefore, by \cref{thm:sim trans} and \cref{thm:sim sum} all output simulations of $\text{\sc Emu}_{\Sigma}$ have encoding $\enc$ and thus $\text{\sc Emu}_{\Sigma}$ satisfies condition \ref{def:model sim enc}. 
\end{proof}

\begin{corollary}[lifting emulations to $\Cone$]\label{lem:sim model cone}
    Let $\mc{M}, \mc{N}$ be spin models and let
    $\text{\sc Emu} \colon \mc{M} \to \mc{N}$ be an emulation, 
    then we can construct 
    \begin{equation}
        \text{\sc Emu}_{\Cone} \colon \Cone(\mc{M}) \to \Cone(\mc{N}).
    \end{equation}
\end{corollary}

\begin{proof}
    The claim follows by first applying \cref{thm:model sim scale} to $\text{\sc Emu}$ and then  \cref{thm:model sim sum} to the resulting emulation.
    Using that $\Cone(\mc{M}) = \Sigma\circ\Scale(\mc{M})$, this yields an emulation of type $\Cone(\mc{M}) \to \Cone(\mc{N})$.
    This construction amounts to applying \cref{thm:model sim sum} to an emulation $\Scale(\mc{M}) \to \Scale(\mc{N})$, which, strictly speaking, is not covered by \cref{thm:model sim sum}.
    We illustrate in \cref{rem:mod  lift} why this nevertheless works.
\end{proof}

    \begin{remark}[lifting emulations to arbitrary hulls]\label{rem:mod  lift}
        Theorems \ref{thm:model sim scale} and \ref{thm:model sim sum} and \cref{lem:sim model cone} can be generalized as follows: Given $\text{\sc Emu} \colon \mr{hull}_1(\mc{M}) \to \mr{hull}_2(\mc{N})$, for arbitrary hull operators $\mr{hull}_1, \mr{hull}_2$ 
        we can construct emulations
            \begin{equation}\label{eq:new emus}
            \begin{split}
               &\text{\sc Emu}_{\Scale} \colon \Scale\circ \mr{hull}_1(\mc{M}) \to \Scale\circ \mr{hull}_2(\mc{N})   \\
               &\text{\sc Emu}_{\Sigma} \colon \Sigma\circ \mr{hull}_1(\mc{M})  \to \Sigma\circ \mr{hull}_2(\mc{N}) \\
               &\text{\sc Emu}_{\Cone} \colon \Cone\circ \mr{hull}_1(\mc{M}) \to \Cone\circ \mr{hull}_2(\mc{N}). 
        \end{split}
            \end{equation}
        $\text{\sc Emu}_{\Scale}$ can be constructed as in \cref{thm:model sim scale}. 
        In particular, $\text{\sc Emu}_{\Scale}$ can be constructed for sets of spin systems $\mc{S},\mc{T}$ instead of spin models $\mc{M},\mc{N}$.
    
        $\text{\sc Emu}_{\Sigma}$ can be constructed by following the proof of \cref{thm:model sim sum}.
        However, now $\text{\sc Emu}$ returns formal expressions $\mr{expr}_1, \ldots, \mr{expr}_k$ instead of spin systems $S_1, \ldots, S_k$. 
        In order to construct the relabelings $\phi_i$ we first have to evaluate these formal expressions. 
        Note, however, that since $\text{\sc Emu}_{\Sigma}$ must return formal expressions from $\Sigma\circ \mr{hull}_1(\mc{M})$, we must apply the relabelings to the formal expressions $\mr{expr}_i$, or more precisely the spin systems contained therein, instead of their evaluations.
        The rest of the construction is analogous to that of \cref{thm:model sim sum}.
        
        Proceeding in this way, the formal expression returned by $\text{\sc Emu}_{\Sigma}$ ultimately evaluates to the source of the output simulation of $\text{\sc Emu}_{\Sigma}$. That is,       $\text{\sc Emu}_{\Sigma}$ satisfies \cref{def:emulation mod} \ref{it:emu mod 2}
        due to the following compatibility between relabelings and evaluations.
        Assume that $\mr{expr}_i$ contains spin systems $S_{i,1}, \ldots, S_{i,r}$ (we write $\mr{expr}_i(S_{i,1} \ldots, S_{i,r})$). Further assume that $\mr{expr}_i$ evaluates to $S_i$, which after relabeling will be denoted by $R_i$, while by relabeling $S_{i,l}$ accordingly, the resulting spin system will be denoted by $R_{i,l}$. Then the following diagram commutes
        \begin{center}
	     \begin{tikzcd}[column sep = huge, row sep = huge]
	          S \arrow[r,"{\text{relabel}}"] & R \\
            \mr{expr}_i(S_{i,1}, \ldots, R_{i,r}) \arrow[u,"{\text{evaluate}}"]
            \arrow[r,"{\text{relabel  }}"] &
            \mr{expr}_i(R_{i,1}, \ldots, R_{i,r}) \arrow[u,"{\text{evaluate}}"].
	     \end{tikzcd}
	 \end{center}
        Since evaluating formal expressions amounts to computing sums and scalings of spin systems, it is polytime computable. Thus, poyltime computability of $\text{\sc Emu}_{\Sigma}$ follows from 
        \ref{thm:model sim sum}.

        Finally, $\text{\sc Emu}_{\Cone}$ can be constructed as described in \cref{lem:sim model cone}.  
    \end{remark}

\subsection{Characterization of Universality}\label{ssec:universality}

Given a spin model $\mc{M}$, it is natural to consider its reach w.r.t.\ spin model emulations, i.e.\ the maximal set of spin systems that $\mc{M}$ emulates.
Surprisingly, there exist fairly simple spin models with maximal reach, i.e.\ that emulate (the set of) all spin systems. Such spin models are called \emph{universal spin models}. 

We shall first define universality as well as other properties of spin models \cref{sssec:prop}, 
and then show how to simulate high order flags with flags of order two (\cref{sssec:basis}), 
to finally characterize universality (\cref{sssec:charac univ}). 

\subsubsection{Universality and Other Properties of Spin Models} \label{sssec:prop}

\begin{definition}[universal]\label{def:universality}
A spin model $\mc{M}$ is \emph{universal} if it emulates the set of all spin systems. 
\end{definition}

Universal spin models, in some sense, can approximate arbitrary spin systems, or more precisely their spectrum (\cref{lem:sim spectrum}), their partition function (\cref{lem:sim part fun sys}) and their Boltzmann distribution (\cref{thm:boltzmann sys}) to arbitrary precision.
Moreover, by \cref{def:spin model sim}, these approximations can be computed in polytime.
We explore consequences of universality in \cref{sec:consequences}.

We now introduce several properties of spin models which shall be used to fully characterize spin model universality.

\begin{definition}[functional complete]\label{def:f.c.}
A spin model $\mc{M}$ is \emph{functional complete} (f.c.) if 
\begin{equation}
    \Cone(\mc{M}) \to \mc{B}_2.
\end{equation}
\end{definition}

Recall from \cref{ex:flag basis} that $\mc{B}_2$ contains all flag systems corresponding to canonical basis functions on the space of functions of type $[2]^{\{v_1,v_2\}}\to \mathbb{R}$, for arbitrary labels $v_1,v_2$.
Functional completeness states that positive linear combinations of spin systems from
$\mc{M}$ suffice to emulate all such basis spin systems.

Since we merely require that positive linear combinations of spin systems from $\mc{M}$ emulate all such basis spin systems, 
we interpret functional completeness as a condition on the local interactions contained in $\mc{M}$.  
In other words, \cref{def:f.c.} takes into account restrictions that might arise from the interaction hypergraphs allowed by $\mc{M}$, and thereby prevent linear combinations of spin systems from $\mc{M}$ from being included in $\mc{M}$.

\begin{definition}[closed]\label{def:closed}
A spin model $\mc{M}$ is \emph{closed} if
\begin{equation}
    \mathcal{M}\to \Sigma(\mc{M}).
\end{equation}
\end{definition}

Closed spin models contain (at least up to emulation) the sums of their spin systems.
We shall refer to the property of being closed as closure. 

In contrast to \cref{def:f.c.}, closure is a condition on the interaction hypergraphs of $\mc{M}$. Clearly, any sum of spin systems from $\mc{M}$ will only contain local interactions from $\mc{M}$, but since the interaction hypergraph of the sum is the union of the individual interaction hypergraphs, it might violate the conditions that $\mc{M}$ poses on its interaction hypergraphs.
Consider for instance the case where $\mc{M}$ only contains 2d grids as interaction graphs (see \cref{ex:2d Ising}); depending on the chosen overlap, the union of two 2d grids might not constitute a 2d grid itself.
This does not imply that the sum of two 2d Ising systems cannot be simulated by a 2d Ising system, but only  that this simulation cannot be trivial.

\begin{definition}[scalable]\label{def:scalable}
    A spin model $\mc{M}$ is \emph{scalable} if
    \begin{equation}
        \mathcal{M}\to \Scale( \mc{M}).
    \end{equation}
\end{definition}
That $\mc{M}$ is scalable means that, up to emulation, $\mc{M}$ contains its spin systems scaled by non-negative real numbers.   
Many spin models are trivially scalable because they have scalable local interactions (see e.g.\ \cref{ex:2d Ising}).

\begin{definition}[locally closed]\label{def:loc closed}
   A spin model $\mc{M}$ is \emph{locally closed} if
   \begin{equation}
       \mc{M} \to \mc{M} + J(\mc{M})
   \end{equation}
   with identity encoding.
\end{definition}

Recall from \cref{def:local spin model} that  $J(\mc{M})$ contains all local interactions of $\mc{M}$. 
Local closure weakens closure as follows: instead of requiring that $\mc{M}$ emulates arbitrary sums of its spin systems, it merely requires that $\mc{M}$ emulates sums of the form $S + R_{J_T(e)}$, where $S, T$ are spin systems from $\mc{M}$ and $R_{J_T(e)}$ is the canonical spin system corresponding to a single local interaction $J_T(e)$ of $T$. 
Note, however, that local closure requires the emulation to use the identity encoding while closure does not restrict the encoding.

\begin{proposition}\label{thm:locally closed}
    If $\mc{M}$ is locally closed 
    then $\mc{M}$ is closed.
\end{proposition}

\begin{proof}
We construct the required emulation 
\begin{equation}
    \text{\sc Emu} \colon \mc{M} \to \Sigma(\mc{M})
\end{equation}
by, given $(T_1, \ldots, T_k)$ from $\Sigma(\mc{M})$,
first decomposing $T_2, \ldots, T_k$ into sums of canonical spin systems and then 
iteratively applying local closure of $\mc{M}$, together with additivity and compositionality of spin model emulations (\cref{thm:model sim sum}, \cref{thm:model sim trans}).

First, on input $(T_1, \ldots, T_k,\delta)$, 
for $i\in \{2, \ldots,  k\} $ we decompose $T_i$ into canonical spin systems, that is, we fix an enumeration $E_{T_i}= \{e_{1}, \ldots, e_{n_i}\}$  and write
\begin{equation}
    T_i = \sum_{j=1}^{n_i} T_{J_{T_i}(e_j)} ,
\end{equation} 
where $T_{J_{T_i}(e_j)}$ is the canonical spin system with $e_{j}$ as the only hyperedge and $J_{T_i}(e_{j})$ as the only local interaction.
Let us write 
$T_{i,j}\coloneqq T_{J_{T_i}(e_j)}$. 
We have 
\begin{equation}
        \sum_{i=1}^k T_i =  T_1 +\sum_{i=2}^k \sum_{j=1}^{n_i}T_{i,j}.
\end{equation}
To simplify the notation for the rest of the proof we rename the spin systems $T_{i,j}$ such that with $n\coloneqq \sum_{i=2}^k n_i$, we have
\begin{equation}
    T = T_1 + \sum_{i=1}^n R_i.
\end{equation}

Now, since $\mc{M}$ is locally closed, there exists a spin model emulation
\begin{equation}
    \text{\sc LC} \colon \mc{M} \to \mc{M} + J(\mc{M}),
\end{equation}
that uses identity encoding. 
We first apply $\text{\sc LC}$ to $(T_1,R_1, \delta)$ to obtain a spin system $S_1 \in \mc{M}$ 
and a spin system simulation
\begin{equation}
    \mc{f}_1 \colon S_1 \to T_1 + R_1.
\end{equation}
Note that both $S_1$ and $\mc{f}_1$ depend on $T$ and $\delta$ but in order to lighten the notation we omit  these dependencies. 

Next, by \cref{lem:iso sim}, the identity isomorphism induces a trivial simulation $\mc{g}_2 \colon R_2 \to R_2$ with cut-off $\delta$.
Similar to the proof of \cref{thm:model sim sum}, we now 
relabel spins in $S_1$ and redefine $\mc{f}_1$ accordingly, such that conditions 
\ref{eq:sum phys spins} and \ref{eq:sum phys image} of \cref{thm:sim sum} are satisfied. 
Note that by definition of local closure, $\mc{f}_1$ has identity encoding and by construction so does $\mc{g}_2$. Therefore, conditions \ref{eq: sim sum dec},  \ref{eq:sum common enc} and \ref{eq:sum disjoint enc} of \cref{thm:sim sum} are satisfied.
By \cref{thm:sim sum} we thus get a simulation
\begin{equation}
    \mc{f}_1+\mc{g}_2 \colon S_1+R_2 \to T_1 + R_1 + R_2.
\end{equation}
Note that $ \mc{f}_1+\mc{g}_2 $ has cut-off $\delta$ and shift $\Shift_{\mc{f}_1}$.

We now apply $\text{\sc LC}$ to $(S_1,R_2, \delta+\Shift_{\mc{f}_1})$ to obtain a spin system $S_2\in \mc{M}$ 
and a spin system simulation
\begin{equation}
    \mc{f}_2 \colon S_2 \to S_1 + R_2.
\end{equation}
Using compositionality of spin system simulations (\cref{thm:sim trans}) we thus have 
\begin{equation}
    \mc{f}_2 \circ (\mc{f}_1 + \mc{g}_2) \colon S_2 \to T_1 + R_1+R_2.
\end{equation}
This simulation, by construction, has cut-off
\begin{equation}
\mr{min}(\delta, \delta + \Shift_{\mc{f}_1}-\Shift_{\mc{f}_1})=\delta.    
\end{equation}
Proceeding iteratively, we obtain the required simulation with target $T$.

To finish the proof we show that this construction defines a spin model emulation, i.e.\ satisfies \cref{def:spin model sim}.
Condition \ref{def:model sim right cutoff} is satisfied by construction.

Condition \ref{def:efficient sim} holds because, 
first, decomposing each $T_i$ into $T_i = \sum_j T_{i,j}$ can simply be read off from the description of $T_i$. Second, $\text{\sc LC}$ is polytime by assumption. We need a total of $n$ applications of $\text{\sc LC}$, where $n$ is the number of hyperedges in $T_2, \ldots, T_k$. In particular,  $n \leq \sum_{i=2}^k \vert T_i \vert$, i.e.\ the above construction needs polynomially many applications of $\text{\sc LC}$ which thus can also be computed in polytime. 
Third, by \cref{thm:model sim trans} and \cref{thm:model sim sum}, composition and addition of spin system simulation is polytime computable and the above construction requires $n$ compositions and $n$ compositions, which thus are also polytime computable.

Condition \ref{def:cut-off indep} is satisfied because it is satisfied by $\text{\sc LC}$, and 
condition \ref{def:model sim enc} is satisfied because the final simulation has identity encoding, independently of the target spin system $T$.                   
\end{proof}

We shall characterize universal spin models as those spin models which are closed, scalable and functional complete (\cref{thm:main}). 
The first step of this characterization consists of proving that closed, scalable and functional complete spin models emulate arbitrary, non-negative linear combinations of flag functions systems from $\mc{B}_2$.
\begin{lemma}\label{lem:closed scalable f.c.}
    If $\mc{M}$ is closed, scalable and functional complete then 
    \begin{equation}
         \mc{M} \to \Cone(\mc{B}_2).
    \end{equation}
\end{lemma}

\begin{proof}
    By assumption of $\mc{M}$ being functional complete, we have an emulation
    \begin{equation}
        \mc{M} \to \mc{B}_2.
    \end{equation}
    Using \cref{thm:model sim scale} we lift it to 
    \begin{equation}
        \Scale(\mc{M}) \to \Scale(\mc{B}_2).
    \end{equation}
    We then use scalability of $\mc{M}$ and compositionality of emulations (\cref{thm:model sim trans}) to obtain 
    \begin{equation}
        \mc{M} \to \Scale(\mc{B}_2).
    \end{equation}
    By \cref{thm:model sim sum} (see \cref{rem:mod lift}) we lift this emulation to 
    \begin{equation}
        \Sigma(\mc{M}) \to \Sigma \circ \Scale(\mc{B}_2) = \Cone(\mc{B}_2),
    \end{equation}
    and finally we use closure of $\mc{M}$ together with \cref{thm:model sim trans} to construct 
    \begin{equation}
        \mc{M} \to \Cone(\mc{B}_2).
    \end{equation}
    \end{proof}
    
    The proof is illustrated in the following diagram, with $\to$ representing emulations that by \cref{thm:model sim trans} can be composed and $\Rightarrow$ representing lifts of emulations: 
    \begin{center}
	     \begin{tikzcd}
	          \mc{M} \arrow[r,"{\text{closed}}"] & \Sigma(\mc{M}) \arrow[rr, " \ref{thm:model sim sum}"] & &  \Cone(\mc{B}_2)\\
	          & \mc{M} \arrow[u, Rightarrow, "\Sigma"] \arrow[r, "{\text{scalable}}"] & \Scale(\mc{M}) \arrow[r, " \ref{thm:model sim scale}"] & \Scale(\mc{B}_2) \arrow[u, Rightarrow, "\Sigma"]  \\
           & & \mc{M} \arrow[u, Rightarrow, "\Scale"]  \arrow[r, "{\text{f.c.}}"] & \mc{B}_2 \arrow[u, Rightarrow, "\Scale"] .
	     \end{tikzcd}
	 \end{center}

The second step to characterize universality (\cref{thm:main}) consists of proving that non-negative linear combinations of flag systems, each acting on two spins, suffice to simulate arbitrary flag systems acting on arbitrarily many spins.
More precisely, we construct a spin model emulation 
\begin{equation}
    \Cone(\mc{B}_2) \to  \mc{B}.
\end{equation}
This situation is reminiscent to Boolean logic, where certain sets of Boolean functions such as $\{\textsc{not}, \textsc{and} \}$, involving only functions of arity $1$ and $2$ suffice to generate all  Boolean functions, including those of higher arity. 
Indeed, we construct $\Cone(\mc{B}_2) \to \mc{B}$ by constructing spin systems from $\Cone(\mc{B}_2)$  that,  in their ground state,  compute the Boolean functions  $\textsc{not}, \textsc{nor}, \textsc{and}$ over some of their spins.  
Generic flag systems from $\mc{B}$ can be constructed by combining these spin systems in an appropriate way.

Since the previous construction involves the explicit definition of several spin systems, 
 while it often suffices to define them up to isomorphism,  
we introduce the following graphical notation. 
\begin{remark}[Graphical notation of spin systems]\label{rem:graphical notation}
    In \cref{lem:f.c.} as well as in \cref{sec:2d Ising} we define spin systems by (graphically) defining their interaction graphs. Their local interactions are represented as colors of the edges, subject to the colour coding provided in \cref{fig:2D close legend}.
    Coloured edges without an attached number represent local interactions scaled by $\delta$, while those with an attached number represent local interactions scaled by that number.
    
    Note that all spin systems defined in this way have 2-level spins and only pair interactions and fields.
    In addition, in the interaction graphs, we only label those spins which are important to understand the (ground state) behavior of the corresponding spin system. In most cases these, spin systems simulate certain target systems with physical spin assignment the identity embedding and identity encoding, in which case we only label the physical spins explicitly.
    Thus, the labels implicitly define the physical spin assignment of the corresponding simulation.
    
    We repeatedly use some of the spin systems defined in \cref{lem:f.c.} and \cref{sec:2d Ising} to construct new spin systems by adding them. 
    To simplify this process we introduce symbols for some of the spin systems, and addition is represented by building graphs out of these symbols (see \cref{fig:AndSim}).
    In the symbols, only physical spins are drawn explicitly; auxiliary spins, in some sense, are internal to the symbols.
    The graphical representation of addition is understood such that all involved spin systems  have pairwise disjoint auxiliary spins, ensuring that the conditions of \cref{thm:sim sum} are always satisfied.
\end{remark}

\begin{figure}[th]
\begin{tikzpicture}[scale = 1, node distance={30mm},
cgnodeS/.style = {draw=black!80, fill=blue!10, thick, circle, minimum size= {width("$v_{i}$")+4pt}, inner sep=1pt},
] 

\node[cgnodeS] (1) at (0.75,2) {$v_i$};

\node[cgnodeS] (2) at (0.75,1) {};

\node[cgnodeS] (3) at (0.75,0) {};
\draw[very thick, red](0.75,0) circle (1em);

\node[cgnodeS] (2-21) at (0.75,-1) {};
\draw[very thick, blue](0.75,-1) circle (1em);

\node[cgnodeS] (2-11) at (0,-2) {};
\node[cgnodeS] (2-12) at (1.5,-2) {};

\node[cgnodeS] (2-31) at (0,-3) {};
\node[cgnodeS] (2-32) at (1.5,-3) {};

\node[cgnodeS] (2-41) at (0,-4) {};
\node[cgnodeS] (2-42) at (1.5,-4) {};

\node[cgnodeS] (2-51) at (0,-5) {};
\node[cgnodeS] (2-52) at (1.5,-5) {};

\node[cgnodeS] (2-61) at (0,-6) {};
\node[cgnodeS] (2-62) at (1.5,-6) {};

\node[cgnodeS] (2-71) at (0,-7) {};
\node[cgnodeS] (2-72) at (1.5,-7) {};

\node[cgnodeS] (2-81) at (0,-8) {};
\node[cgnodeS] (2-82) at (1.5,-8) {};
\node[cgnodeS] (2-91) at (0,-8.5) {};
\node[cgnodeS] (2-92) at (1.5,-8.5) {};

\node[cgnodeS] (2-101) at (0.75,-9.5) {};
\draw[very thick, black](0.75,-9.5) circle (0.9em);
\draw[very thick, brown](0.75,-9.5) circle (1.1em);

\node[cgnodeS] (2-111) at (0,-10.5) {};
\node[cgnodeS] (2-112) at (1.5,-10.5) {};

\node[cgnodeS] (2-121) at (0.75,-11.5) {};
\draw[very thick, black](0.75,-11.5) circle (1em);
\node (lambda) at (1.3,-11.5) {$\lambda$};

\node[cgnodeS] (2-131) at (0,-12.5) {};
\node[cgnodeS] (2-132) at (1.5,-12.5) {};

\draw[red, very thick] (2-11)--(2-12);

\draw[blue, very thick] (2-31)--(2-32);

\draw[green, very thick] (2-41)--(2-42);

\draw[yellow, very thick] (2-51)--(2-52);

\draw[violet, very thick] (2-61)--(2-62);

\draw[orange, very thick] (2-71)--(2-72);

\draw[black, very thick] (2-81)--(2-82);
\draw[brown, very thick] (2-91)--(2-92);

\draw[brown, very thick,shifted path=from 2-111 to 2-112 by 1pt];
\draw[black, very thick, shifted path=from 2-111 to 2-112 by -1pt];

\draw[brown, very thick] (2-131) to node[midway, above] {$\lambda$} (2-132);

\node (d1-1) at (3.2,0) {$\delta\cdot \pi_1$};

\node (d1-3) at (3.2,1) {auxiliary spin};

\node (d1-3) at (3.2,2) {physical spin};

\node (d2-2) at (3.2,-1) {$\delta \cdot \bar{\pi}_1$};

\node (d2-1) at (3.2,-2) {$\delta \cdot \pi_2$};

\node (d2-3) at (3.2,-3) {$\delta\cdot \bar{\pi}_2$};

\node (d2-4) at (3.2,-4) {$\delta\cdot f_{(1,1)}$};

\node (d2-5) at (3.2,-5) {$\delta \cdot f_{(1,2)}$};

\node (d2-6) at (3.2,-6) {$\delta \cdot f_{(2,1)}$};

\node (d2-7) at (3.2,-7) {$\delta \cdot f_{(2,2)}$};

\node[align = center] (d2-8) at (3.2,-8.25) {arbitrary \\
interactions};

\node[align = center] (d2-9) at (3.2,-10) {sum of \\
interactions:\\
$J_1+J_2$};

\node[align = center] (d2-10) at (3.2,-12) {scaling of \\
interactions:\\
$\lambda \cdot J$};

\node[] (t1) at (0.75, 3) {Symbol};
\node[] (t2) at (3.2,3) {Meaning};

\coordinate (top) at (2,3.5) {};
\coordinate (btm) at (2,-13) {};

\coordinate (lt) at (-0.5,3.5) {};
\coordinate (rt) at (4.5,3.5) {};

\draw[-] (lt) to (rt);

\coordinate (sep1-1) at (-0.5,-0.5) {};
\coordinate (sep1-2) at (4.5,-0.5) {};

\coordinate (sep2-1) at (-0.5,-1.5) {};
\coordinate (sep2-2) at (4.5,-1.5) {};

\coordinate (sep3-1) at (-0.5,-2.5) {};
\coordinate (sep3-2) at (4.5,-2.5) {};

\coordinate (sep4-1) at (-0.5,-3.5) {};
\coordinate (sep4-2) at (4.5,-3.5) {};

\coordinate (sep5-1) at (-0.5,-4.5) {};
\coordinate (sep5-2) at (4.5,-4.5) {};

\coordinate (sep6-1) at (-0.5,0.5) {};
\coordinate (sep6-2) at (4.5,0.5) {};

\coordinate (sep7-1) at (-0.5,1.5) {};
\coordinate (sep7-2) at (4.5,1.5) {};

\coordinate (sep8-1) at (-0.5,2.5) {};
\coordinate (sep8-2) at (4.5,2.5) {};

\coordinate (sep9-1) at (-0.5,-5.5) {};
\coordinate (sep9-2) at (4.5,-5.5) {};

\coordinate (sep10-1) at (-0.5,-6.5) {};
\coordinate (sep10-2) at (4.5,-6.5) {};

\coordinate (sep11-1) at (-0.5,-7.5) {};
\coordinate (sep11-2) at (4.5,-7.5) {};

\coordinate (sep12-1) at (-0.5,-9) {};
\coordinate (sep12-2) at (4.5,-9) {};

\coordinate (sep13-1) at (-0.5,-11) {};
\coordinate (sep13-2) at (4.5,-11) {};

\coordinate (lb) at (-0.5,-13) {};
\coordinate (rb) at (4.5,-13) {};

\draw (lb) to (rb);
\draw (lt) to (lb);
\draw (rt) to (rb);

\draw (sep1-1)--(sep1-2);

\draw (sep2-1)--(sep2-2);

\draw (sep3-1)--(sep3-2);

\draw (sep4-1)--(sep4-2);

\draw (sep5-1)--(sep5-2);

\draw (sep6-1)--(sep6-2);

\draw (sep7-1)--(sep7-2);

\draw[] (sep8-1) to (sep8-2);

\draw (sep9-1)--(sep9-2);

\draw (sep10-1)--(sep10-2);

\draw (sep11-1)--(sep11-2);

\draw (sep12-1)--(sep12-2);

\draw (sep13-1)--(sep13-2);

\draw (top)--(btm);

\end{tikzpicture}
\caption{Color coding of local interactions for the graphical definition of spin systems. Note that for $i\neq j$, $f_{(i,j)}$ is not symmetric w.r.t.\ swapping the spins. In the graphical notation, the order of spins for these interactions is understood left to right (top to bottom).
}
\label{fig:2D close legend}
\end{figure}

\subsubsection{Simulating High Order Flags} \label{sssec:basis}

In analogy to Boolean functions of arbitrary arity being simulated by Boolean functions of arity two, we now show that high order flags can be simulated by flags of order two. 

\begin{lemma}[simulation of flags]\label{lem:f.c.}
There exists a spin model emulation 
\begin{equation}
   \text{\sc Basis} \colon \Cone(\mc{B}_2) \to  \mc{B}.
\end{equation}
\end{lemma}
\begin{proof}

We divide the target spin model $\mc{B}$ into three disjoint subsets, 
\begin{equation}
    \mc{B} = \mc{B}_1 \cup \mc{B}_2 \cup \mc{B}_{\geq 3},
\end{equation}
where $\mc{B}_1$ consists of flag systems of order 1 (i.e.\ that act on a single spin), $\mc{B}_1$ consists of flag systems of order $2$ and $\mc{B}_{\geq 3}$ consists of
flag systems of order at least 3.
We then construct separate emulations for these three cases,  
\begin{equation}
\begin{split}
     &\text{\sc Basis}_1 \colon \Cone(\mc{B}_2) \to  \mc{B}_1  \\
    &\text{\sc Basis}_2 \colon \Cone(\mc{B}_2) \to  \mc{B}_2  \\
   &\text{\sc Basis}_{3} \colon \Cone(\mc{B}_2) \to  \mc{B}_{\geq 3}.
\end{split}
\end{equation}
They can be easily combined to obtain the emulation with target spin model $\mc{B}$. Given $(T,\delta)$ as input, we simply first check the cardinality of $V_T$ to decide the subset to which $T$ belongs. Depending on the result apply either $\text{\sc Basis}_1$, $\text{\sc Basis}_2$ or $\text{\sc Basis}_3$.

While the construction of $\text{\sc Basis}_1$ and $\text{\sc Basis}_2$ is straightforward, that of $\text{\sc Basis}_3$ is more involved. 
To construct the latter, we leverage that non-negative linear combinations of spin systems from $\mc{B}_2$ suffice to encode Boolean logic in their ground state.
More precisely we construct spin systems that subject to the identification $1 = \textsc{true}$, $2 = \textsc{false}$ compute the Boolean functions $\{\textsc{not} ,\textsc{and}\}$ in their ground state. Since  $\{\textsc{not} ,\textsc{and}\}$ is functionally complete for Boolean logic, this suffices to compute arbitrary Boolean functions in the ground state. 
We construct spin systems $S^{\delta}_{\vec{x}}$ that, in their ground state, flag configurations $\vec{x}\in [2]^{\{1, \ldots, l\}}$ in the state of a single spin $u_{l-1}$, that is, in the ground state of $S^{\delta}_{\vec{x}}$, $\vec{s}(u_{l-1})=1$ if and only if for all $i$, $\vec{s}(v_i)=\vec{x}(i)$.
We then add interactions to $S^{\delta}_{\vec{x}}$ that yield an energy contribution of $+1$ if and only if $\vec{s}(u_{l-1})=1$. This 
finally results in the required simulation $S^{\delta}_{\vec{x}} \to T_{f_{\vec{x}}}$. 

Let us explain in more detail the constructions of the simulations and prove that they indeed define simulations.
All sources of these simulations will be explicitly constructed as non-negative linear combinations of spin systems from $\mc{B}_2$. For simplicity, we write 
$S \in \Cone(\mc{B}_2)$ or state that $\text{\sc Basis}_i$ returns $S$; according to \cref{def:emulation mod} we actually mean the formal expression obtained from the construction of $S$.

First, we construct $\text{\sc Basis}_1$.
A generic spin system in $\mc{B}_1$ is of the form 
$T_{f_x}\langle v_1\rangle$ (see \cref{def:can sys}), for some $x\in [ 2] \cong [2]^{\{1 \}}$.
Deciding if $x=1$ or $x=2$ can simply be done by reading off the local interaction of $T_{f_x}\langle v_1 \rangle$ from the input.

Consider first the case $x=1$.
We define the spin system
\begin{equation}\label{eq:s1 sum}
    S_1 \langle v_1,v_2\rangle \coloneqq  S_{f_{(1,2)}} \langle v_1,v_2 \rangle  +  S_{f_{(1,1)}} \langle v_1,v_2\rangle.
\end{equation}
Clearly, $S_1 \in \Cone(\mc{B}_2)$.
Moreover, by construction,
\begin{equation}
    H_{S_1\langle v_1,v_2 \rangle}(\vec{s}) = \begin{cases}
        1 \ \text{if} \ \vec{s}(v_1)=1.\\
        0 \ \text{else}
    \end{cases}
\end{equation}
Consequently, for an arbitrary cut-off $\delta$, the physical spin assignment given by the identity injection of $v_1$ into $\{v_1,v_2\}$ yields a simulation $S_1\langle v_1,v_2\rangle \to T_{f_1}\langle v_1 \rangle$ with degeneracy $2$.

The second case, $x=2$, can be treated analogously, by defining 
\begin{equation}
    S_2\langle v_1,v_2\rangle \coloneqq S_{f_{(2,1)}}\langle v_1,v_2 \rangle +S_{f_{(2,2)}} \langle v_1,v_2\rangle.
\end{equation}
Combining these two cases defines the emulation $\text{\sc Basis}_1$.
Clearly, $\text{\sc Basis}_1$ satisfies the four conditions of \cref{def:spin model sim}.

The emulation $\text{\sc Basis}_2$ can be trivially constructed by noting that, according  to \cref{lem:iso sim}, identity isomorphisms define simulations with arbitrary cut-off. We simply take $\text{\sc Basis}_2$ to be the identity emulation, returning the corresponding identity simulation for each target system and cut-off.

To finish the proof, we construct $\text{\sc Basis}_3$ by constructing spin systems that implement Boolean logic in the ground state.

\begin{figure}[th]
    \centering
    \begin{tikzpicture}[circuit logic US, cgnodeS/.style = {draw=black!80, fill=blue!10, thick, circle, minimum size= {width("$v_{i}$")+4pt}, inner sep=1pt}]
        \node[cgnodeS] (v1) at (0,2) {$v_1$};
        \node[cgnodeS] (v2) at (0,0) {$v_2$};

        \draw[green, very thick,shifted path=from v1 to v2 by 1pt];
        \draw[orange, very thick,shifted path=from v1 to v2 by -1pt];

        \node[] (eq) at (1,1) {$\eqqcolon$};

        \node[cgnodeS] (v3) at (2,2) {$v_1$};
        \node[cgnodeS] (v4) at (2,0) {$v_2$};
        \node[not gate,point down, info=center:$\neg$] (g) at (2,1) {};
        \draw[-] (v3) to (g);
        \draw[-] (v4) to (g);

    \end{tikzpicture}
    \caption{Definition of $S^{\delta}_{\textsc{not}}$.}
    \label{fig:notGagdget}
\end{figure}

Let $\delta>0$. 
We first construct the spin system $S^{\delta}_{\textsc{not}}$ (see \cref{fig:notGagdget}) implementing the logical $\textsc{not}$ in its ground state: 
\begin{equation}
    S^{\delta}_{\textsc{not}}\langle v_1,v_2\rangle  = \delta\cdot S_{f_{(1,1)}}\langle v_1,v_2\rangle + \delta\cdot S_{f_{(2,2)}} \langle v_1,v_2\rangle.
\end{equation}
This yields a Hamiltonian
\begin{equation}
    H_{S^{\delta}_{\textsc{not}}\langle v_1,v_2\rangle} (\vec{s})= \begin{cases}
        0 \ \text{if} \ \vec{s}(v_1)\neq \vec{s}(v_2)\\
        \delta \ \text{else}.  
    \end{cases}
\end{equation}
That is, in its ground state, $\vec{s}(v_1)= \textsc{not}(\vec{s}(v_2))$.

Next, we construct $S^{\delta}_{\textsc{and}}$, implementing the logical $\textsc{and}$ (see \cref{fig:AND}) in two steps. We first build $S^{\delta}_{\textsc{nor}}$, which implements the logical $\textsc{nor}$ in its ground state (see \cref{fig:NOR}).
We then use 
\begin{equation}\label{eq:norandid}
    \textsc{and}(x,y) = \textsc{nor}(\textsc{not}(x), \textsc{not}(y))
\end{equation}
to build $S^{\delta}_{\textsc{and}}$ out of $S^{\delta}_{\textsc{not}}$ and $S^{\delta}_{\textsc{nor}}$.

\begin{figure*}[th]
    \centering
    \begin{subfigure}[t]{1.0\columnwidth}
    \caption{}\label{fig:NOR}
    \centering
    \begin{tikzpicture}[circuit logic US, cgnodeS/.style = {draw=black!80, fill=blue!10, thick, circle, minimum size= {width("$v_{i}$")+4pt}, inner sep=1pt}] 

    \node[cgnodeS] (v2) at (2,0) {$v_2$}; 
    \node[cgnodeS] (a4) at (3,0) {}; 

    \node[cgnodeS] (v1) at (-2,0) {$v_1$};
    \node[cgnodeS] (a3) at (-1,0) {};

    \node[cgnodeS] (a2) at (0,-1) {};
    \node[cgnodeS] (a7) at (1,-1) {};

    \node[cgnodeS] (a1) at (0,-3) {};
    \node[cgnodeS] (a6) at (1,-3) {};

    \node[cgnodeS] (v3) at (0,-5) {$v_3$};
    \node[cgnodeS] (a5) at (1,-5) {};

    \
    \draw[very thick, green] (v1) to [out=270,in=135]  node[midway,right] {$2\delta$}  (a1);

    \draw[very thick, green] (v1) to [out=270,in=135] node[midway, left] {$2\delta$} (v3);

    \draw[very thick, green] (a1) to node[midway,left] {$2\delta$} (a2);

    \draw[very thick, green] (v3) to node[midway, left] {$2\delta$} (a1);

    \draw[very thick, green] (v2) to [out=270,in=45] node[midway, right] {$2\delta$} (v3);

    \draw[very thick, green] (v2) to [out=270,in=-45] node[midway, below] {$2\delta$} (a2);

    \draw[violet, very thick,shifted path=from v1 to a3 by 1pt];
    \draw[orange, very thick,shifted path=from v1 to a3 by -1pt];

    \draw[violet, very thick,shifted path=from v2 to a4 by 1pt];
    \draw[orange, very thick,shifted path=from v2 to a4 by -1pt];

    \draw[violet, very thick,shifted path=from v3 to a5 by 1pt];
    \draw[orange, very thick,shifted path=from v3 to a5 by -1pt];

    \draw[violet, very thick,shifted path=from a1 to a6 by 1pt];
    \draw[orange, very thick,shifted path=from a1 to a6 by -1pt];

    \draw[violet, very thick,shifted path=from a2 to a7 by 1pt];
    \draw[orange, very thick,shifted path=from a2 to a7 by -1pt];

     \node (=) at (3,-3) {$\eqqcolon$};

    \node[nor gate, point down, info=center:$\downarrow$] (g) at (5,-3) {};

    \node[cgnodeS] (e21) at (4.5,-2) {$v_1$};
    \node[cgnodeS] (e22) at (5.5,-2) {$v_2$};
    \node[cgnodeS] (e23) at (5,-4) {$v_3$};

    \draw[-] (e21) to (g);
    \draw[-] (e22) to (g);
    \draw[-] (e23) to (g);

    \end{tikzpicture}
\end{subfigure}
    \begin{subfigure}[t]{1.0\columnwidth}
    \caption{}\label{fig:AND}
    \centering 
    \begin{tikzpicture}[circuit logic US, cgnodeS/.style = {draw=black!80, fill=blue!10, thick, circle, minimum size= {width("$v_{i}$")+4pt}, inner sep=1pt},
    ] 

    \node[cgnodeS] (v1) at (-0.5,0) {$v_1$};
    \node[cgnodeS] (v2) at (0.5,0) {$v_2$};

    \node[not gate, point down, info=center:$\neg$] (n1) at (-0.5,-1) {};
    \node[not gate, point down, info=center:$\neg$] (n2) at (0.5,-1) {};

    \draw[-] (v1) to (n1);
    \draw[-] (v2) to (n2);

    \node[nor gate, point down, info=center:$\downarrow$] (g1) at (0,-3) {};

    \node[cgnodeS] (e1) at (-0.5,-2) {};
    \node[cgnodeS] (e2) at (0.5,-2) {};
    \node[cgnodeS] (e3) at (0,-4) {$v_3$};

    \draw[-] (n1) to (e1);
    \draw[-] (n2) to (e2);

    \draw[-] (e1) to (g1);
    \draw[-] (e2) to (g1);
    \draw[-] (e3) to (g1);

    \node (=) at (2,-3) {$\eqqcolon$};

    \node[and gate, point down, info=center:$\land$] (g2) at (4,-3) {};

    \node[cgnodeS] (e21) at (3.5,-2) {$v_1$};
    \node[cgnodeS] (e22) at (4.5,-2) {$v_2$};
    \node[cgnodeS] (e23) at (4,-4) {$v_3$};

    \draw[-] (e21) to (g2);
    \draw[-] (e22) to (g2);
    \draw[-] (e23) to (g2);
    \end{tikzpicture}
\end{subfigure}
\caption{Definition and symbolic representation of spin systems $S^{\delta}_{\textsc{nor}}$ (\ref{fig:NOR}) and $S^{\delta}_{\textsc{and}}$ (\ref{fig:AND}). 
}\label{fig:NORAND}
\end{figure*}
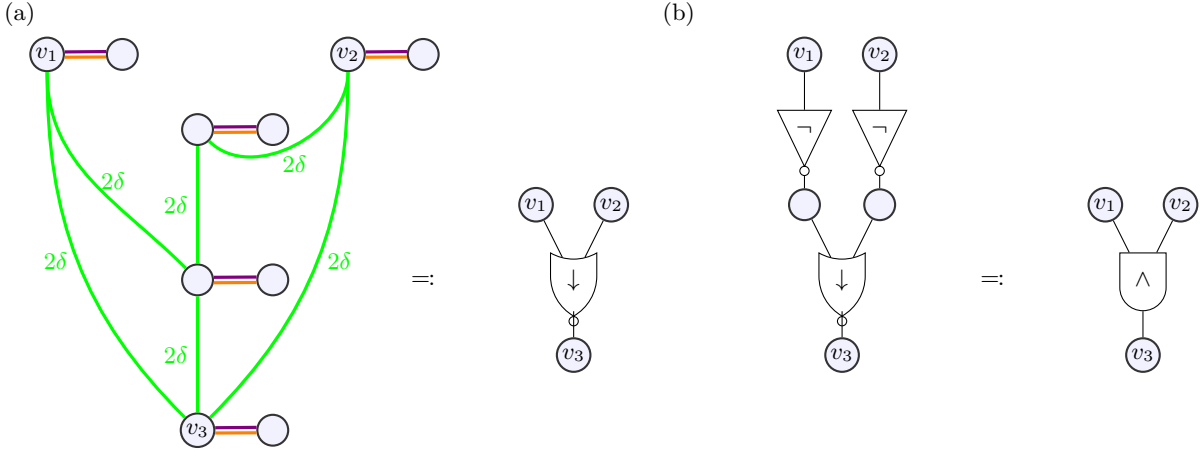

By definition (see \cref{fig:NOR}), it is clear that $S^{\delta}_{\textsc{nor}} \in \Cone(\mc{B}_2) $.
Let $\vec{s}\in [2]^{\{v_1,v_2,v_3, a_1, \ldots, a_7 \}}$, where $a_1,a_2$ correspond to the two internal spins of $S^{\delta}_{\textsc{nor}}$, i.e.\ those connected with a $2\delta\cdot f_{(1,1)}$ interaction and $a_3, \ldots, a_7$ correspond to the $5$ additional spins drawn on the RHS of the $\delta \cdot (f_{(2,1)}+f_{(2,2)})$ interactions. 
It can easily be seen that the ground state of $S^{\delta}_{\textsc{nor}}$ consists of the following configurations
\begin{equation}
\begin{split}
   \mr{GS}_{S^{\delta}_{\textsc{nor}}} \coloneqq \Bigl \{&(2,2,1,2,1,\ldots),(2,1,2,1,2,\ldots),\\
   &(1,2,2,2,1, \ldots), (1,1,2,2,2,\ldots) \Bigr \},
   \end{split}
\end{equation}
where the first $5$ numbers in each tuple correspond to the states of $(v_1,v_2,v_3,a_1,a_2)$ and $,\ldots$ indicates that in the ground state the states of $a_3, \ldots, a_7$ can be chosen arbitrarily, yielding a total of $4\cdot 2^5$ ground states.
Note that, in the ground state, 
\begin{equation}\label{eq:GSNOR}
    \vec{s}(v_3) = \textsc{nor}\bigl(\vec{s}(v_1), \vec{s}(v_2)\bigr).
\end{equation}
Furthermore, the ground-state energy of $S^{\delta}_{\textsc{nor}}$ equals $3\delta$, while all other configurations have energy at least $4\delta$.

Next, we construct $S^{\delta}_{\textsc{and}}$ from $S^{\delta}_{\textsc{not}}$ and $S^{\delta}_{\textsc{nor}}$ (see \cref{fig:AND}).
Denote the top two physical spins of $S^{\delta}_{\textsc{nor}}$ in \cref{fig:AND} by $e_1$ and $e_2$. Then 
\begin{equation}
\begin{split}
    S^{\delta}_{\textsc{and}} \langle v_1,v_2,v_3, \ldots \rangle =&  
    S^\delta_{\textsc{nor}}\langle e_1,e_2,v_3, \ldots\rangle + S^\delta_{\textsc{not}}\langle v_1,e_1 \rangle \\ &+  S^\delta_{\textsc{not}}\langle v_2,e_2 \rangle . 
    \end{split}
\end{equation} 
Thus, in the ground state of $S^{\delta}_{\textsc{and}}$, the states of $v_1$ and $v_2$ satisfy $\vec{s}(v_i)= \textsc{not}(\vec{s}(e_i))$. By \eqref{eq:GSNOR} (with $v_{1,2}$ replaced by $e_{1,2}$) and \eqref{eq:norandid}, 
\begin{equation}
\begin{split}
    \vec{s}(v_3)  
    &= \textsc{nor}\bigl(\textsc{not}(\vec{s}(v_1)), \textsc{not}(\vec{s}(v_2))\bigr)  \\
    &=  \textsc{and}\bigl(\vec{s}(v_1), \vec{s}(v_2)\bigr).
\end{split}
\end{equation}
Note that $S^{\delta}_{\textsc{and}} \in \Cone(\mc{B}_2)$.
The ground-state energy of $S^{\delta}_{\textsc{and}}$ is $3\delta$ and all other configurations have energy at least $4\delta$.

We now use $S^{\delta}_{\textsc{not}}$ and $S^{\delta}_{\textsc{and}}$ to construct the emulation $\text{\sc Basis}_3$.
A generic spin system from $\mc{B}_{\geq 3}$ is of the form $T_{f_{\vec{x}}}\langle v_1, \ldots, v_l\rangle$ for some $\vec{x}\in [2]^{\{1, \ldots, l\}}$ and some $l\geq 3$.
We shall now construct a spin system $S^{\delta}_{\vec{x}}$, which will be used as the source of a simulation of $T_{f_{\vec{x}}}\langle v_1, \ldots, v_l\rangle$.
Without loss of generality (w.l.o.g.) we restrict to the case  $\vec{x}=\vec{x}_1$,
where $\vec{x}_1$ denotes the all one configuration, 
i.e.\ $\vec{x}_1(i)=1$ for all $i \in \{1,\ldots, l\}$.
The generic case can be treated by modifying the following construction
such that all spins in state $2$, i.e.\ all $i$ with $\vec{x}(i)=2$, are first negated. More precisely,
for all such $i$, we add $S^{\delta}_{\textsc{not}}\langle v_i, f_i\rangle$ and replace all  occurrences of $v_i$ by $f_i$ as shown in \cref{fig:AndSim}.
The resulting ground state satisfies that $\vec{s}(f_i)=1$.

$S^{\delta}_{\vec{x}_1}$ is obtained by, first, adding spin systems $S^{\delta}_{\textsc{and}}$ as shown in \cref{fig:AndSim1}.   
Denoting the top physical spins of the various copies of $S^{\delta}_{\textsc{and}}$ in \cref{fig:AndSim1} by $u_1, \ldots, u_{l-1}$, from left to right, 
in the ground state $S^{\delta}_{\vec{x}_1}$ satisfies
\begin{equation}
    \begin{split}
        \vec{s}(u_1) &= \textsc{and}\bigl(\vec{s}(v_1),\vec{s}(v_2)\bigr)\\
        \vec{s}(u_i) &= \textsc{and}\bigl(\vec{s}(u_{i-1}),\vec{s}(v_{i+1})\bigr).
    \end{split}
\end{equation}
In total,
\begin{equation}
    \vec{s}(u_{l-1}) = \textsc{and} \bigl(\vec{s}(v_1), \ldots, \vec{s}(v_l) \bigr),
\end{equation}
i.e.\ in the ground state $\vec{s}(u_{l-1})=1$ if and only if $\vec{s}(v_1, \ldots, v_l)=(1, \ldots, 1)$.

So far we have constructed a sum of $l-1$ copies of $S^{\delta}_{\textsc{and}}$, each having a ground-state energy of $3\delta$ and an energy gap (between the ground-state energy and the remaining energy levels) of $\delta$. For this reason, all ground-state configurations of  $S^{\delta}_{\vec{x}_1}$ have energy $(l-1) \cdot 3\delta$, while all other configurations have energy at least $(l-1)\cdot 3\delta+\delta$.

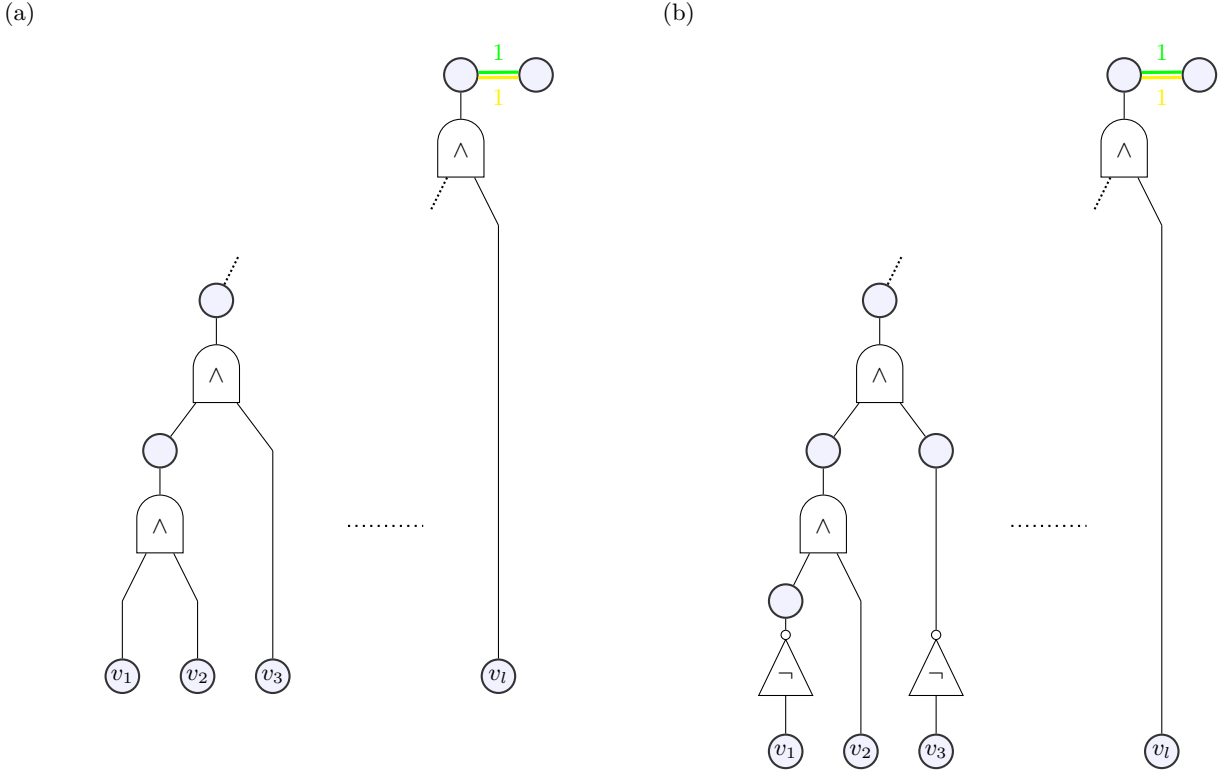
\begin{figure*}[th]
    \centering
    \begin{subfigure}[t]{1.0\columnwidth}
    \caption{}\label{fig:AndSim1}
    \centering
 \begin{tikzpicture}[circuit logic US, cgnodeS/.style = {draw=black!80, fill=blue!10, thick, circle, minimum size= {width("$v_{4}$")+4pt}, inner sep=1pt},
]
	
\node[cgnodeS] (e1) at (0,1) {$v_1$};
\node[cgnodeS] (e2) at (1,1) {$v_2$};
\node[cgnodeS] (e3) at (2,1) {$v_3$};

\draw[-] (e1) to (0,2);
\draw[-] (e2) to (1,2);
\draw[-] (e3) to (2,4);

\node[and gate,  point up, info=center:$\land$] (a1) at (0.5,3) {};
\draw[-] (0,2) to (a1);
\draw[-] (1,2) to (a1);

\node[cgnodeS] (u1) at (0.5,4) {};
\draw[-] (u1) to (a1);

\node[and gate,  point up, info=center:$\land$] (a2) at (1.25,5) {};
\draw[-] (a2) to (u1);
\draw[-] (a2) to (2,4);

\draw[dotted, thick] (3,3) to (4,3);
\node[cgnodeS] (el) at (5,1) {$v_l$};
\draw[-] (el) to (5,7);

\node[cgnodeS] (u2) at (1.25,6) {};
\draw[-] (u2) to (a2);

\node[cgnodeS] (ul) at (4.5,9) {};

\node[and gate,  point up, info=center:$\land$] (al) at (4.5,8) {};

\draw[-] (al) to (ul);

\draw[-] (al) to (5,7);
\draw[densely dotted, thick] (al) to (4.1,7.2);
\draw[densely dotted, thick] (u2) to (1.55,6.6);

\node[cgnodeS] (b) at (5.5,9) {};

\draw[green, very thick,shifted path=from ul to b by 1pt];
\draw[yellow, very thick, shifted path=from ul to b by -1pt];

\node[text = green] (s1) at (5,9.3) {$1$};
\node[text = yellow] (s1) at (5,8.7) {$1$};
	
\end{tikzpicture}
\end{subfigure}
    \begin{subfigure}[t]{1.0\columnwidth}
    \caption{}\label{fig:AndSim2}
    \centering 
   \begin{tikzpicture}[circuit logic US, cgnodeS/.style = {draw=black!80, fill=blue!10, thick, circle, minimum size= {width("$v_{4}$")+4pt}, inner sep=1pt},
]
	
\node[cgnodeS] (e1) at (0,0) {$v_1$};
\node[cgnodeS] (e2) at (1,0) {$v_2$};
\node[cgnodeS] (e3) at (2,0) {$v_3$};

\node[cgnodeS] (f1) at (0,2) {};
\node[cgnodeS] (f3) at (2,4) {};

\node[not gate,  point up, info=center:$\neg$] (n1) at (0,1) {};
\node[not gate,  point up, info=center:$\neg$] (n2) at (2,1) {};

\draw[-] (e1) to (n1);
\draw[-] (n1) to (f1);

\draw[-] (e3) to (n2);
\draw[-] (n2) to (f3);

\draw[-] (e2) to (1,2);

\node[and gate,  point up, info=center:$\land$] (a1) at (0.5,3) {};
\draw[-] (f1) to (a1);
\draw[-] (1,2) to (a1);

\node[cgnodeS] (u1) at (0.5,4) {};
\draw[-] (u1) to (a1);

\node[and gate,  point up, info=center:$\land$] (a2) at (1.25,5) {};
\draw[-] (a2) to (u1);
\draw[-] (a2) to (f3);

\draw[dotted, thick] (3,3) to (4,3);
\node[cgnodeS] (el) at (5,0) {$v_l$};
\draw[-] (el) to (5,7);

\node[cgnodeS] (u2) at (1.25,6) {};
\draw[-] (u2) to (a2);

\node[cgnodeS] (ul) at (4.5,9) {};

\node[and gate,  point up, info=center:$\land$] (al) at (4.5,8) {};

\draw[-] (al) to (ul);

\draw[-] (al) to (5,7);
\draw[densely dotted, thick] (al) to (4.1,7.2);
\draw[densely dotted, thick] (u2) to (1.55,6.6);

\node[cgnodeS] (b) at (5.5,9) {};

\draw[green, very thick,shifted path=from ul to b by 1pt];
\draw[yellow, very thick, shifted path=from ul to b by -1pt];

\node[text = green] (s1) at (5,9.3) {$1$};
\node[text = yellow] (s1) at (5,8.7) {$1$};

\end{tikzpicture}
\end{subfigure}
\caption{
Definition of spin systems $S^{\delta}_{\vec{x}_1}$ (\ref{fig:AndSim1}) and $S^{\delta}_{\vec{x}}$ for a generic $\vec{x} \in [2]^{\{1, \ldots, l\}}$ (\ref{fig:AndSim2}). 
The generic case differs from the $\vec{x}_1$ case in terms of the additional spin systems $S^{\delta}_{\textsc{not}}$ which are added for all $i$ with $\vec{x}(i)=2$, i.e.\ $i\in \{1,3\}$ in the specific case of \cref{fig:AndSim2}.}\label{fig:AndSim}
\end{figure*}

To finish the construction of $S^{\delta}_{\vec{x}_1}$, we add $S_1\langle u_{l-1},b\rangle$. This leads to an energy contribution of $+1$ whenever $\vec{s}(u_{l-1})=1$. Overall, in the ground state
\begin{equation}
    H_{S^{\delta}_{\vec{x}_1}}(\vec{s}) = \begin{cases}
        (l-1)\cdot 3\delta +1 \ \text{ if } \ \vec{s}(v_1, \ldots, v_l)=(1, \ldots, 1) \\
        (l-1)\cdot 3\delta  \ \text{ if } \ \vec{s}(v_1, \ldots, v_l)\neq (1, \ldots, 1). \\
    \end{cases}
\end{equation}
It is straightforward to conclude that letting the physical spin assignment be the identity injection of $\{v_1, \ldots, v_l\}$ and the energy shift be $(l-1)\cdot 3\delta$ results in a simulation of type $S^{\delta}_{\vec{x}_1} \to T_{f_{\vec{x}_1}}$ with cut-off $ \delta$. 
The degeneracy of this simulation can be obtained from the individual degeneracies of $S^{\delta}_{\textsc{and}}$ and $S_1$ as $2 \cdot (4\cdot2^5)^{l-1}$.

We now define the emulation $\text{\sc Basis}_3$ as follows:
On input $(T_{f_{\vec{x}}}\langle v_1, \ldots, v_l \rangle, \delta)$ we return $S^{\delta}_{\vec{x}}$  and the  simulation constructed above, both subject to the modifications necessary to cover the generic case instead of the special case $\vec{x}= \vec{x}_1$, outlined above and shown in \cref{fig:AndSim2}. 

It can easily be checked that this defines an emulation.
By construction, $S^{\delta}_{\vec{x}} \in \Cone(\mc{B}_2)$.
The simulation has cut-off $\delta $ and hence satisfies \cref{def:spin model sim} \ref{def:model sim right cutoff}.
Moreover, this construction is polytime computable and hence satisfies \cref{def:spin model sim} \ref{def:efficient sim}.
Changing the target cut-off of the simulation amounts to modifying the parameter $\delta$; hence, $\text{\sc Basis}_3$ satisfies \cref{def:spin model sim} \ref{def:cut-off indep}.
Finally, the constructed simulation has identity encoding, independent of the target, showing that $\text{\sc Basis}_3$ satisfies \cref{def:spin model sim} \ref{def:model sim enc}.

Finally, all three emulations, $\text{\sc Basis}_1$, $\text{\sc Basis}_2$ and $\text{\sc Basis}_3$ use identity encoding. Hence, they indeed can be combined to obtain the required emulation 
\begin{equation}
        \text{\sc Basis} \colon \Cone(\mc{B}_2) \to \mc{B}.
\end{equation}
\end{proof}

Now consider the spin model containing all spin systems of spin type $2$, $\mc{M}_2$.
We now prove that $\Cone(\mc{B}) \to \mc{M}_2$ by decomposing a generic spin system with spin type $2$ into a non-negative linear combination of spin systems from $\mc{B}$.

\begin{lemma}[decomposing into flag basis]\label{lem:basis decomp}
    \begin{equation}
        \Cone(\mc{B}) \to \mc{M}_2.
    \end{equation}
\end{lemma}
\begin{proof}
Given a spin system $T$ from $\mc{M}_2$, we first decompose it into canonical spin systems
\begin{equation}
    T = \sum_{e \in E_T} T_{J_T(e)}.
\end{equation}
This decomposition holds by definition.
Now recall that $J_T(e)\colon [2]^e \to \mathbb{R}_{\geq 0}$.
Thus, given $e\in E_T, \vec{x}\in [2]^e$, defining $\lambda_{e,\vec{x}}\coloneqq J_T(e)(\vec{x})$ we obtain 
\begin{equation}\label{eq: lin comb}
    T = \sum_{e\in E_T} \sum_{\vec{x} \in [2]^e} \lambda_{e, \vec{x}} \cdot T_{f_{\vec{x}}}.
\end{equation}

Since $T_{f_{\vec{x}}} \in \mc{B}$ and $\lambda_{e,\vec{x}}\geq 0$ the required emulation can be constructed as follows: given $(T, \delta)$, return the formal linear combination corresponding to \cref{eq: lin comb} 
and the identity simulation $T \to T$.
This trivially satisfies \cref{def:spin model sim}. \end{proof}

\subsubsection{Characterization of Universality} \label{sssec:charac univ}

We are finally ready to characterize universal spin models. 

\begin{theorem}[characterization of universality]\label{thm:main}
    A spin model $\mc{M}$ is universal if and only if it is closed, scalable and functional complete.
\end{theorem}

\begin{proof}
First, by \cref{lem:bin sim model}, 
\begin{equation}\label{eq:emu -1}
    (\mc{S}_{\mr{all}})_{\mr{bin}} \to \mc{S}_{\mr{all}},
\end{equation}
where $\mc{S}_{\mr{all}}$ is the set of all spin systems.
Since $(\mc{S}_{\mr{all}})_{\mr{bin}} \subseteq \mc{M}_2$, and by \cref{lem:basis decomp} $\Cone(\mc{B}) \to \mc{M}_2 $, we have 
\begin{equation}\label{eq:emu 0}
    \Cone(\mc{B}) \to (\mc{S}_{\mr{all}})_{\mr{bin}}.
\end{equation}

Next, recall from \cref{lem:f.c.} that 
\begin{equation}\label{eq:emu 1}
    \Cone(\mc{B}_2) \to \mc{B}.
\end{equation}
Applying \cref{lem:sim model cone} (see also \cref{rem:mod  lift}) to this emulation
yields
\begin{equation}
    \Cone^2(\mc{B}_2) \to \Cone(\mc{B}).
\end{equation}
By \eqref{eq:hull identities} for any $\mc{S}$, $\Cone(\mc{S}) \to \Cone^2(\mc{S})$. Together with \cref{thm:model sim trans} we obtain 
\begin{equation}\label{eq:emu 2}
    \Cone(\mc{B}_2) \to \Cone(\mc{B}).
\end{equation}

Finally, recall from \cref{lem:closed scalable f.c.}, that since $\mc{M}$ is closed, scalable and functional complete, 
\begin{equation}\label{eq:emu 3}
    \mc{M} \to \Cone(\mc{B}_2).
\end{equation}
Composing the four emulations \eqref{eq:emu -1} \eqref{eq:emu 0}, \eqref{eq:emu 2} and \eqref{eq:emu 3} (by applying \cref{thm:model sim trans}) yields 
\begin{equation}
    \mc{M} \to \mc{S}_{\mr{all}}.
\end{equation}
The situation is illustrated in the following diagram, with $\to$ representing emulations (that can be composed by \cref{thm:model sim trans}) and the arrows $\Rightarrow$ representing the $\Cone$-lift (according to \cref{lem:sim model cone}).
\begin{center}
	     \begin{tikzcd}[column sep = small]
	          \mc{M} \arrow[r, "\ref{lem:closed scalable f.c.}"] & \Cone(\mc{B}_2) \arrow[r, " \ref{lem:sim model cone}"] &
           \Cone(\mc{B}) \arrow[r, "\ref{lem:basis decomp}"] & (S_{\mr{all}})_{\mr{bin}} \arrow[r, "\ref{lem:bin sim model}"] &
           \mc{S}_{\mr{all}}\\
            & \Cone(\mc{B}_2) \arrow[u, Rightarrow, "\Cone"] \arrow[r, " \ref{lem:f.c.}"] & \mc{B} \arrow[u, Rightarrow, "\Cone"] &  & 
	     \end{tikzcd}
\end{center}

Conversely, if $\mc{M}$ is universal then by definition
\begin{equation}
    \mc{M} \to \mc{S}_{\mr{all}}.
\end{equation}
Since $\mc{B}_2 \subseteq \mc{S}_{\mr{all}}$ we trivially have $\mc{M} \to \mc{B}_2$ and hence functional completeness of $\mc{M}$.

In order to conclude that $\mc{M}$ is closed we construct $\mc{M} \to \Sigma(\mc{M})$. Given spin systems  $(T_1, \ldots, T_k)$ from $\mc{M}$  and a cut-off $\delta$, we first evaluate their sum, $T\coloneqq  \sum_{i=1}^kT_i$. Note that  this is polytime computable and that $T \in \mc{S}_{\mr{all}}$ trivially. We then run the emulation $\mc{M} \to \mc{S}_{\mr{all}}$ on input $(T,\delta)$ to obtain a spin system from $\mc{M}$ and a simulation with target $T$.
Similarly,  scalability of $\mc{M}$ follows by given $(\lambda, T, \delta)$, running  $\mc{M} \to \mc{S}_{\mr{all}}$ on $(\lambda\cdot T, \delta)$, using that also $\lambda \cdot T$ is polytime computable.
\end{proof}

\begin{corollary}\label{cor: universality with locally closed}
    If $\mc{M}$ is locally closed, scalable, and functional complete then $\mc{M}$ is universal.
\end{corollary}
\begin{proof}
    By \cref{thm:locally closed} $\mc{M}$ is closed, and by \cref{thm:main} universal.
\end{proof}

\section{Consequences of Universality}\label{sec:consequences}

In this section we derive several consequences of spin model emulations and universality. 
Building on the results of \cref{ssec:spin model simulation}, we show that if $\mc{M} \to \mc{M}'$ then, from the emulation we can construct the following polytime computable reductions: 
\begin{itemize}
    \item[$\vartriangleright$] From the ground-state energy problem of $\mc{M}'$ to that of $\mc{M}$ (\cref{ssec:gse});
    \item[$\vartriangleright$]  From the partition function problem of $\mc{M}'$ to that of $\mc{M}$ (\cref{ssec:partfun}); and 
    \item[$\vartriangleright$]  From the approximate sampling problem of $\mc{M}'$ to that of $\mc{M}$ (\cref{ssec:sample}).
\end{itemize} 
That is, if $\mc{M}$ emulates $\mc{M}'$ and we have an algorithm that solves one of the above problems for $\mc{M}$, then we can construct an algorithm that solves that problem for $\mc{M}'$. 
The benefit is twofold. First, emulations serve as tools to solve these computational problems. 
Secondly, hardness results for the above problems immediately extend to all universal spin models: If there exists a spin model for which one of the above problems is hard then this problem must be hard for all universal spin models. 
It follows that for universal spin models, the complexity of the above problems is maximal, i.e.\ there do not exist spin models of greater complexity.

\subsection{Ground-State Energy}\label{ssec:gse}

We start by considering the ground-state energy problem in its decision form.

\begin{definition}[ground-state energy problem] \label{def:GSE}
    Let $\mc{M}$ be a spin model. The \emph{ground-state energy problem} of $\mc{M}$, $\text{\sc Gse}_{\mc{M}}$, is the decision problem 
    \begin{equation}
        \text{\sc Gse}_{\mc{M}} \coloneqq \{(S,k) \mid\\
          S \in \mc{M}, k \in \mathbb{R}, 
          \mr{min}(H_S)<k\}. 
    \end{equation}  
\end{definition} 

In words, $\text{\sc Gse}_{\mc{M}}$ is the following decision problem:
\begin{quote}
        Given a spin system $S \in \mc{M}$ and a real number $k$, is the ground-state energy of $S$ smaller than $k$?
\end{quote}
We follow the convention of defining decision problems in terms of their yes instances.

\begin{theorem}\label{thm:sim GSE}
    Let $\mc{M}, \mc{M}'$ be spin models. 
    An emulation $\text{\sc Emu} \colon \mc{M} \to \mc{M}'$ induces a polytime reduction
    \begin{equation}
        r \colon \text{\sc Gse}_{\mc{M}'} \to \text{\sc Gse}_{\mc{M}}.
    \end{equation}
\end{theorem}

\begin{proof}
    Given an instance $(T,k)$ of $\text{\sc Gse}_{\mc{M}'}$, we compute $\text{\sc Emu}(T,k)$.
    This yields a spin system $S(T,k)\in \mc{M}$ and a simulation 
    \begin{equation}
        \mc{f}(T,k) \colon S(T,k) \to T
    \end{equation}
    with cut-off $k$.
    We construct the reduction $r$ as
    \begin{equation}
        r(T,k) \coloneqq (S(T,k),k +\Shift_{\mc{f}(T,k)}).
    \end{equation}
    By \cref{lem:sim spectrum}, below the cut-off $k$ the spectra of $H_T$ and $H_{S(T,k)}-\Shift_{\mc{f}(T,k)}$ agree.
   In particular, $T$ has a configuration with energy less than $k$ if and only if $S(T,k)$ has a configuration with energy less than $k+\Shift_{\mc{f}(T,k)}$.
   Finally, $r$ is polytime since so is $\text{\sc Emu}$.
\end{proof}

Note that we do not only prove the existence of a reduction, but 
given an emulation $\mc{M}\to \mc{M}'$, we construct the reduction $r$. 
In fact, the construction of $r$ in the proof of \cref{thm:sim GSE} works for any emulations. 
This will also be the case for further reductions that build on spin model emulations (see \cref{thm:sim FGSE}, \cref{thm:sim part fun}, \cref{thm:approx sampl}).

\begin{corollary}\label{cor: univ GSE}
    If $\mc{M}$ is universal then $ \text{\sc Gse}_{\mc{M}}$ is $\text{\sc NP}$-hard.
\end{corollary}
\begin{proof}
     By \cite{Ba82b}‚ $\text{\sc Gse}_{\mc{I}_{2d}}$ is $\text{\sc NP}$-hard. Since $\mc{M}$ is universal, $\mc{M} \to \mc{I}_{2d}$, which by \cref{thm:sim GSE} yields a reduction from $\text{\sc Gse}_{\mc{I}_{2d}}$ to $\text{\sc Gse}_{\mc{M}}$.
\end{proof}

We now derive a similar result for the function problem of computing ground-state configurations.
\begin{definition}[ground state function problem] \label{def:FGSE}
    Let $\mc{M}$ be a spin model. The \emph{ground state function problem} of $\mc{M}$, $\text{\sc Fgse}_{\mc{M}}$, is the following function problem: Given $S \in \mc{M}$, compute a configuration from  
    \begin{equation}
       \mr{GS}(S)  \coloneqq H_S^{-1}(\mr{min}(H_S)).
    \end{equation}
    \end{definition}

    A solution to $\text{\sc Fgse}_{\mc{M}}$ is an algorithm that takes $S \in \mc{M}$ as input and outputs a ground-state configuration of $S$.

\begin{theorem}\label{thm:sim FGSE}
    Let $\mc{M}, \mc{M}'$ be spin models. 
    An emulation $\text{\sc Emu} \colon \mc{M} \to \mc{M}'$ induces a polytime reduction
    \begin{equation}
        (f,g) \colon \text{\sc Fgse}_{\mc{M}'} \to \text{\sc Fgse}_{\mc{M}}.
    \end{equation}
\end{theorem}

Note that a reduction between function problems consists of two functions $(f,g)$. In a reduction between decision problems, a single function $r$ sufficed, as the output of a solution consists of a single bit (yes or no). 
For function problems, the output of a solution is non-trivial, and there is freedom in modifying the output, which is exactly the role of $g$.
Other than this, reductions between function problems work similarly as between decision problems. If $(f,g)$ is a reduction from function problem $A$ to function problem $B$, then a solution for $B$ can be turned into a solution for $A$: On input $a$ we apply the solution for $B$ to $f(a)$ to obtain an output $s$ that solves $f(a)$, then by construction $g(a,s)$ solves $a$. 

\begin{proof}
    The proof is similar to that of \cref{thm:sim GSE}.
    We construct two polytime computable functions $f$ and $g$,
    \begin{equation}
        f\colon \mc{M}' \to \mc{M}
    \end{equation}
    and 
    \begin{equation}
        g\colon \mc{M}' \times \mc{C}_{\mc{M}} \to \mc{C}_{\mc{M}'}
    \end{equation}
    that satisfy
    \begin{equation}\label{eq:red fun cond}
        \vec{s} \in \mr{GS}(f(T)) \Rightarrow g(T,\vec{s}) \in \mr{GS}(T),
    \end{equation}
    where $\mc{C}_{\mc{M}} \coloneqq \bigcup_{S\in \mc{M}}\mc{C}_S$.
    
    Given $T\in \mc{M}'$ we first compute an upper bound $k_T$ on $H_T$ by defining
    \begin{equation}
        k_T \coloneqq \sum_{e \in E_T}\mr{max}(J_T(e)).
    \end{equation}
    While computing the exact maximum (or minimum) of $H_T$ might be NP-hard, computing an upper (or lower) bound for $H_T$ by maximization (or minimization) of local terms is polytime computable, as the latter  amounts to an iteration over $e \in E_T$, where each step consists of a search over $\vert [q]^e \vert $ parameters and by definition of $\vert T \vert$,
    \begin{equation}
      \sum_{e\in E_T} \vert  [q_T]^e \vert  < \vert T \vert . 
    \end{equation}
    
    Now we compute $\text{\sc Emu}(T,k_T)$.
    This yields a spin system $S(T, k_T)\in \mc{M}$ and a simulation 
    \begin{equation}
        \mc{f}(T,k_T) \colon S(T, k_T) \to T
    \end{equation}
     with cut-off $k_t\geq \mr{max}(H_T)$. By \cref{lem:ground state}, we have for $\vec{s}\in GS_{\mc{M}}(S(T,k_T))$, 
    \begin{equation}
        \dec_{\mc{f}(T, k_T)} \circ \vec{s} \circ \phys_{\mc{f}(T, k_T)}\in GS_{\mc{M}'}(T).
    \end{equation}
    We define
    \begin{equation}
    \begin{split}
                f(T) &\coloneqq S(T, k_T),\\
                g(T,\vec{s}) &\coloneqq \dec_{\mc{f}(T, k_T)} \circ \vec{s} \circ \phys_{\mc{f}(T, k_T)}.
    \end{split}
    \end{equation}
    This is polytime computable, since  $k_T$, $\text{\sc Emu}$ and composing $\vec{s}$ with $\dec_{\mc{f}(T, k_T)}$ and $\phys_{\mc{f}(T, k_T)}$ are too.
    Moreover, by \cref{lem:ground state} this satisfies \cref{eq:red fun cond}.
    Thus, $(f,g)$ yields a reduction.

    Finally note that, strictly speaking, the right hand side of the definition of $g$ is only defined for $\vec{s}$ being a configuration on $V_{S(T,k_T)}$.
    If this is not the case, the left hand side of the implication of \eqref{eq:red fun cond} evaluates to wrong, and thus we can define $g$ arbitrarily on such inputs $(T,\vec{s})$. 
    Since $S(T,k_T)$ can be computed in polytime, checking whether $\vec{s}$ is a configuration on $V_{S(T,k_T)}$ and hence the modified definition of $g$ are polytime computable, too. 
\end{proof}

\begin{corollary}\label{cor:univ FGSE}
    If $\mc{M}$ is universal then $\text{\sc Fgse}_{\mc{M}}$ is $\text{\sc FNP}$-hard.
\end{corollary}
\begin{proof}
    By \cite{Ba82b},  $ \text{\sc Fgse}_{\mc{I}_{2d}}$ is $\text{\sc FNP}$-hard: A polytime solution for $ \text{\sc Fgse}_{\mc{I}_{2d}}$ would induce a polytime solution for $ \text{\sc Gse}_{\mc{I}_{2d}}$. (But the antecedent does not exist unless $\text{\sc P} = \text{\sc NP}$, a caveat that applies to the conclusions throughout). 
    Since $\mc{M}$ is universal, $\mc{M} \to \mc{I}_{2d}$. By \cref{thm:sim FGSE} we obtain a reduction from  $\text{\sc Fgse}_{\mc{I}_{2d}} $ to $ \text{\sc Fgse}_{\mc{M}}$.
\end{proof} 

\cref{thm:sim GSE} and \cref{thm:sim FGSE} state that  ground states of arbitrary spin systems can be efficiently encoded into ground states of a universal spin model.
Since many $\text{\sc NP}$ problems can be phrased as ground-state energy problems (see \cite{Lu14b}), universal spin models, together with \cref{thm:main} could serve as tools to look for solutions of such problems. 
Not surprisingly, by \cref{cor: univ GSE} and \cref{cor:univ FGSE} their ground-state energy problems are themselves $ \text{\sc NP}$-hard.

One way of circumvent this problem may be to use quantum computing protocols. For example, \cite{Le15} and \cite{Ng23} consider specific spin models and describe, first, how generic $ \text{\sc NP}$ problems can be encoded into its ground state, and, second, how its ground states can be computed via quantum annealing. 
The framework developed in this work could serve to construct novel quantum annealing protocols (by constructing appropriate emulations), as well as  analyze the (dis)advantages of existing protocols. It would be interesting to study what properties of the emulation determine the performance of adiabatic quantum annealing, i.e.\ the adiabatic time scale (see \cite{Ng23,Eb22}). 
We shall come back to this point in the \nameref{sec:outlook}. 

\subsection{Partition Function Approximation}\label{ssec:partfun}
In this section we study the complexity of approximating the partition function of a spin system from  a spin model $\mc{M}$. 

\begin{definition}[partition function problem]\label{def:part fun prob}
    Let $\mc{M}$ be a spin model. The \emph{partition function problem} of $\mc{M}$, $\text{\sc Part}_{\mc{M}}$, is the following function problem: Given $S\in \mc{M}$, $\beta>0$, compute 
    \begin{equation}
            \mr{Part}_{\mc{M}}(S,\beta) \coloneqq Z_S(\beta).
    \end{equation}
\end{definition}

In certain cases it suffices to compute $\text{\sc Part}_{\mc{M}}$ approximately. 
An \emph{approximate solution} of $\text{\sc Part}_{\mc{M}}$ consists of an algorithm that, 
given $S\in \mc{M}$, $\beta>0$ and an error parameter $\epsilon>0$,  returns a real number $z$ that satisfies 
\begin{equation}
    \vert z - \mr{Part}_{\mc{M}}(S,\beta) \vert < \epsilon \cdot \mr{Part}_{\mc{M}}(S,\beta).
\end{equation} 
If the runtime is polynomial in the size of the input $(S,\beta)$ it is called a \emph{polynomial time approximation scheme} ($\text{\sc PTAS}$). If it is polynomial in both the size of the input and the inverse approximation error, $\epsilon^{-1}$, it is called a \emph{fully polynomial time approximation scheme} ($\text{\sc FPTAS}$). If it is polynomial in both the size of input and the inverse approximation error, but the algorithm uses randomness and returns a valid approximation with probability at least $3/4$, it is called a \emph{fully polynomial randomized approximation scheme} ($\text{\sc FPRAS}$) \cite{Mo11}.

For approximation problems, a common notion of reduction is that of $\text{\sc PTAS}$ reductions \cite{We05}. They can be used to transform approximate solutions of any of the above introduced types, that is, $\text{\sc PTAS}$, $\text{\sc FPTAS}$ and $\text{\sc FPRAS}$. 

\begin{theorem}\label{thm:sim part fun}
    Let $\mc{M}$ and $\mc{M}'$ be spin models.  An emulation $\text{\sc Emu} \colon \mc{M}\to \mc{M}'$ induces a PTAS reduction 
    \begin{equation}
        (f,g,\alpha) \colon \text{\sc Part}_{\mc{M}'} \to \text{\sc Part}_{\mc{M}}.
    \end{equation}
\end{theorem}

A $\text{\sc PTAS}$ reduction $(f,g,\alpha): A \to B$ between function problems $A$ and $B$ consist of polytime computable functions $f,g, \alpha$ that transform an approximate solution for $B$ into an approximate solution for $A$. Given an instance $a$ of $A$ and an error parameter $\epsilon$, the approximate solution for $B$ is applied to $(f(a,\epsilon),\alpha(\epsilon))$. This yields an output $s$ that is an $\alpha(\epsilon)$ approximation of the solution for $f(a,\epsilon)$. By construction, $g(a,s,\epsilon)$ is an $\epsilon$ approximation of the solution for $a$.

\begin{proof} 
    We need to construct polytime computable functions
    \begin{equation}
        f \colon \mc{M}' \times \mathbb{R}_{>0} \times \mathbb{R}_{>0} \to \mc{M} \times \mathbb{R}_{>0},
    \end{equation}
    \begin{equation}
        g \colon \mc{M}' \times \mathbb{R}_{>0} \times  \mathbb{R}_{>0} \times \mathbb{R} \to \mathbb{R}_{>0}
    \end{equation}
    and 
    \begin{equation}
        \alpha \colon \mathbb{R} \to \mathbb{R},
    \end{equation}
    such that 
    if $z$ is an $\alpha(\epsilon)$ approximation of $\mr{Part}_{\mc{M}}(f(T,\beta,\epsilon))$ then $g(T,\beta,z,\epsilon)$ is an $\epsilon$ approximation of $\mr{Part}_{\mc{M}'}(T,\beta)$. That is, 
    \begin{equation}\label{eq:PTAS red}
    \begin{split}
        &\left\vert z - \mr{Part}_{\mc{M}}(f(T,\beta,\epsilon)) \right\vert 
        < \alpha(\epsilon) \cdot \mr{Part}_{\mc{M}}(f(T,\beta,\epsilon)) \\
        &\Rightarrow 
        \left\vert g(T,\beta,z,\epsilon) - \mr{Part}_{\mc{M}'}(T,\beta) \right\vert < \epsilon \cdot \mr{Part}_{\mc{M}'}(T,\beta).
        \end{split}
    \end{equation}
    
    By \cref{lem:sim part fun sys}, whenever $\mc{f}\colon S \to T$ with cut-off $\Del_{\mc{f}} > \mr{max}(H_T)$, we have 
    \begin{equation}\label{eq:part fun 2}
        \left\vert \frac{1}{e^{-\Shift_{\mc{f}} \beta} m \degeneracy_{\mc{f}}}\cdot Z_S(\beta) - Z_T(\beta) \right\vert
        \leq \frac{1}{m \degeneracy_{\mc{f}}} \cdot q_S^{\vert V_S \vert} \cdot e^{-\beta \Del_{\mc{f}}},
    \end{equation}
    where $m = \vert \enc_{\mc{f}} \vert$.

    In particular, if $S=S(T,\delta)$ and $\mc{f}= \mc{f}(T, \delta)$ are the output  of $\text{\sc Emu}(T, \delta)$ then $q_S=q_{\mc{M}}$, 
    and by \cref{def:spin model sim} condition \ref{def:cut-off indep} the parameters $\vert \enc_{\mc{f}} \vert ,d_{\mc{f}},\vert V_{S}\vert$ only depend on $T$, not on $\delta$. Thus, for fixed $T$, increasing $\delta$ makes the right hand side of \eqref{eq:part fun 2} arbitrary small. 
    This is the key idea to construct the reduction $(f,g,\alpha)$.

    We start with $f$.
    Given $(T,\beta, \epsilon)$ we first compute $\text{\sc Emu}(T, \delta')$, where $\delta'$ is arbitrary. This yields
    $S(T, \delta')\in \mc{M}$ and a simulation 
    \begin{equation}
    \mc{f}(T, \delta') \colon S(T, \delta')      \to T.
    \end{equation}
    From this, we read off $\vert \enc_{\mc{f}(T, \delta')}\vert ,d_{\mc{f}(T, \delta')}$ and $\vert V_{S(T, \delta')} \vert$. 
    By \cref{def:spin model sim} condition \ref{def:cut-off indep},   these parameters are independent of the cut-off $\delta'$. Henceforth we write 
    \begin{equation}
        \begin{split}
            m_T &\coloneqq \vert \enc_{\mc{f}(T, \delta')}\vert \\
            d_T &\coloneqq d_{\mc{f}(T, \delta')}\\
            v_T &\coloneqq \vert V_{S(T, \delta')} \vert .
        \end{split}
    \end{equation}

    We now evaluate $H_T$ on an arbitrary configuration $\vec{t}$ to obtain $\xi_T \coloneqq H_T(\vec{t})$.
    By construction it satisfies
    \begin{equation}
        e^{-\beta \xi_T} \leq Z_T(\beta).
    \end{equation}
    We compute
    \begin{equation}
        \xi_1\coloneqq \sum_{e \in E_T}\mr{max}(J_T(e)) 
    \end{equation}
    and 
    \begin{equation}
        \xi_2(\epsilon, \beta) = \frac{1}{\beta} \cdot \mr{ln} \Bigl (\frac{2}{\epsilon}\cdot  \frac{(\frac{\epsilon}{2}+1)q_{\mc{M}}^{v_T}}{ m_T d_T} \Bigr) + \xi_T
    \end{equation}
    and define 
    \begin{equation}
        \delta_{\epsilon, \beta} \coloneqq \mr{max}(\xi_1,\xi_2(\epsilon, \beta)).
    \end{equation}
    We use $\text{\sc Emu}$ to obtain a spin system
    $S(T, \delta_{\epsilon, \beta})$ and a simulation
    \begin{equation}
        \mc{f}(T, \delta_{\epsilon, \beta}) \colon S(T, \delta_{\epsilon, \beta}) \to T,
    \end{equation}
    with cut-off $\delta_{\epsilon, \beta}$. 
    We now define $(f,g,\alpha)$ by 
    \begin{equation}
    \begin{split}
        f(T,\beta, \epsilon) &= (S(T, \delta_{\epsilon, \beta}),\beta)\\
        g(T,\beta, z, \epsilon) &= \frac{1}{e^{-\beta \Shift_{\mc{f}(T, \delta_{\epsilon, \beta})}} m_T d_T} \cdot z\\
        \alpha(\epsilon) &= \frac{\epsilon}{2}.
    \end{split}
    \end{equation}
    Note that these three functions are polytime computable (see \cref{rem:poly}).
    In particular, they require two executions of $\text{\sc Emu}$, a first one with arbitrary cut-off $\delta'$ to obtain $m_T, d_T, v_T$ and thereby compute $\delta_{\epsilon, \beta}$, 
    and a second one with cut-off $\delta_{\epsilon, \beta}$ to compute $S(T, \delta_{\epsilon, \beta})$ and $\mc{f}_{T, \delta_{\epsilon, \beta}}$, which are used to define the reduction.
    
    To finish the proof, we prove that this defines a $\text{\sc PTAS}$ reduction, namely satisfies \cref{eq:PTAS red}.
    We shall omit the $(T, \delta_{\epsilon, \beta})$ arguments and write $S$ instead of $S(T, \delta_{\epsilon, \beta})$, and similarly for $\mc{f}$, as well as the subscripts of the simulation parameters of $\mc{f}$.
    Assume 
    \begin{equation}
        \vert z - Z_{S}(\beta) \vert < \alpha(\epsilon)\cdot Z_{S}(\beta) = \frac{\epsilon}{2} \cdot Z_{S}(\beta).
    \end{equation}
    We insert the definition of $g$ into the right hand side of \cref{eq:PTAS red}: 
    \begin{widetext}
\begin{equation}
\begin{split}
            \left\vert \frac{1}{e^{-\beta \Shift}m_T d_T} z - Z_T(\beta) \right\vert 
            &\leq 
            \left\vert \frac{1}{e^{-\beta \Shift}m_T d_T}z - \frac{1}{e^{-\beta \Shift}m_T d_T}Z_{S}(\beta) \right\vert 
            +\left\vert \frac{1}{e^{-\beta \Shift}m_T d_T }Z_{S}(\beta) - Z_T(\beta) \right\vert 
             \\
            & < 
            \frac{\epsilon}{2e^{-\beta \Shift}m_T d_T} Z_{S}(\beta) 
            + \left\vert \frac{1}{e^{-\beta \Shift}m_T d_T}Z_{S}(\beta) - Z_T(\beta) \right\vert 
             \\
            &= \frac{\epsilon}{2} Z_T(\beta) +  \frac{\epsilon}{2e^{-\beta \Shift}m_T d_T} \sum_{\vec{s} \notin \simSet} e^{-\beta H_{S}(\vec{s})}  
            + \frac{1}{e^{-\beta \Shift}m_T d_T} \sum_{\vec{s} \notin \simSet} e^{-\beta H_{S}(\vec{s})} 
             \\
            &< \frac{\epsilon}{2} Z_T(\beta) 
            + \frac{\frac{\epsilon}{2}+1}{m_T d_T}  \cdot q_{\mc{M}}^{v_T} e^{-\beta \delta_{\epsilon,\beta}},
            \label{eq:part fun red err}
\end{split}
\end{equation}
\end{widetext}
    where the first inequality follows from the triangle inequality, 
    the second from the left hand side of \cref{eq:PTAS red},
    the equality from \cref{lem:sim part fun sys} (using additionally that $\delta_{\epsilon,\beta}> \mr{max}(H_T)$), 
    and the last inequality since $H_S(\vec{s})-\Shift\geq\delta_{\epsilon, \beta}$ for configurations $\vec{s}\notin \simSet$, and the number of such configurations is less than the total number of configurations of $S$, $q_{\mc{M}}^{v_T}$.

    Next, we use 
    \begin{equation}
        e^{-\beta \delta_{\epsilon, \beta}} = e^{-\beta(\delta_{\epsilon,\beta}-\xi_T)}\cdot e^{-\beta \xi_T}\leq e^{-\beta(\delta_{\epsilon,\beta}-\xi_T)} \cdot Z_T(\beta).
    \end{equation}
    Inserting this into \eqref{eq:part fun red err} yields
\begin{equation}\label{eq:part fun red err 2}
\begin{split}
       &\left\vert \frac{1}{e^{-\beta \Shift}m_T d_T} z - Z_T(\beta) \right\vert   \\
        & \hspace{0.75cm}< \left( \frac{\epsilon}{2} + 
       \frac{(\frac{\epsilon}{2}+1)q_{\mc{M}}^{v_T}}{m_T d_T}    \cdot e^{-\beta(\delta_{\epsilon,\beta}-\xi_T)} \right)
         \cdot Z_T(\beta)
\end{split}
\end{equation}
    By construction of $\delta_{\epsilon,\beta}$,  
    \begin{equation}\label{eq:Delta eps}
        e^{-\beta (\delta_{\epsilon, \beta}-\xi_T)} \leq
        e^{-\beta (\xi_2(\epsilon, \beta)-\xi_T)} \leq 
        \frac{\epsilon}{2}\cdot 
        \frac{ m_T d_T}{(\frac{\epsilon}{2} + 1)q_{\mc{M}}^{v_T}}.
    \end{equation}
    Inserting this into \eqref{eq:part fun red err 2} finally yields 
    \begin{equation}
        \left\vert \frac{1}{e^{-\beta \Shift_{\mc{f}}}m_T d_T} z - Z_T(\beta) \right\vert < \epsilon \cdot Z_T(\beta) .
    \end{equation}
\end{proof}

\begin{corollary}\label{cor:part univ}
    If $\mc{M}$ is universal, there does not exist a $\text{\sc FPRAS}$ for $\text{\sc Part}_{\mc{M}}$.
\end{corollary}

\begin{proof}
  In \cite{Sl12} it is proven that there does not exist a $\text{\sc FPRAS}$ for the partition function problem of the homogeneous Ising model on regular graphs with degree at most 3, where the homogeneous couplings are restricted to below a certain threshold.
  We denote this spin model by $\mc{I}_{\mr{hom}, 3}$, details regarding its precise definition can be found in \cite{Sl12}.
  Since $\mc{M}$ is universal it emulates $\mc{I}_{\mr{hom}, 3}$. By \cref{thm:sim part fun}, an $\text{\sc FPRAS}$ that approximates the partition function of $\mc{M}$ would induce an $\text{\sc FPRAS}$ that approximates the partition function $\mc{I}_{\mr{hom}, 3}$. 
  Since such a $\text{\sc FPRAS}$ does not exist, the $\text{\sc FPRAS}$ for the partition function of $\mc{M}$ cannot exist in the first place.
\end{proof}

The natural class of problems that can be rephrased as partition functions are counting problems, $\#$P. With a universal spin model $\mc{M}$ we can approximate arbitrary partition functions by \cref{thm:sim part fun}, and thus approximately solve arbitrary counting problems. 
However, partition functions in $\mc{M}$ cannot be approximated in polynomial time by \cref{cor:part univ}. Similarly to the ground-state energy case, quantum computing protocols may help circumvent the hardness of approximating partition functions from $\mc{M}$.

\subsection{Approximate Sampling}\label{ssec:sample}
In this section we consider the problem of approximately sampling configurations from the Boltzmann distribution of spin systems from a spin model $\mc{M}$.
\begin{definition}[approximate sampling problem]\label{def:approx sampl}
Let $\mc{M}$ be a spin model. The \emph{approximate sampling problem} of $\mc{M}$ $\text{\sc Asp}_{\mc{M}}$, is the following problem: Given $S \in \mc{M}$, $\beta>0$ and an error parameter $\epsilon >0$, compute a configuration $\vec{s}\in \mc{C}_S$ drawn from a probability distribution $q_{S,\beta,\epsilon}$ on $\mc{C}_S$ that satisfies
\begin{equation}
    \Vert p_{S,\beta}-q_{S,\beta,\epsilon} \Vert < \epsilon.
\end{equation}
\end{definition}
A solution to the approximate sampling problem is a probabilistic algorithm that, on input $(S,\beta,\epsilon)$, samples configurations form a probability distribution which is $\epsilon$-close to $p_{S,\beta}$ w.r.t.\ the total variation distance $\Vert \: \Vert $ (defined in \eqref{eq:def tot var}).

\begin{theorem}\label{thm:approx sampl}
    An emulation $\text{\sc Emu} \colon \mc{M}\to \mc{M}'$ induces a $\text{\sc PTAS}$ reduction
    \begin{equation}
        (f,g,\alpha)\colon \text{\sc Asp}_{\mc{M}'} \to \text{\sc Asp}_{\mc{M}}.
    \end{equation}
\end{theorem}

\begin{proof}
The proof is similar to that of \cref{thm:sim part fun}.
First, assume $\mc{f}\colon S \to T$ is a simulation and let $q$ be a probability distribution on $\mc{C}_S$.
Similar to \cref{def:sim prob}, we let
\begin{equation}
    q_{\mc{f}}(\vec{t}) \coloneqq \sum_{\mathclap{\substack{\vec{s} \in \mc{C}_S,\\ \dec \circ \vec{s} \circ P = \vec{t}}}} q(\vec{s}).
\end{equation}
Then, for any $\beta>0$, 
\begin{equation}\label{eq:tot var sim}
\begin{split}
    \Vert q_{\mc{f}} - p_{\mc{f}, \beta} \Vert 
    &= \frac{1}{2} \sum_{\vec{t}\in \mc{C}_T} \Bigl \vert \sum_{\substack{\vec{s} \in \mc{C}_S,\\ \dec \circ \vec{s} \circ P = \vec{t}}} (q(\vec{s})-p_{S,\beta}(\vec{s}))   \Bigr \vert\\
    &\leq \frac{1}{2} \sum_{\vec{t}\in \mc{C}_T}  \sum_{\substack{\vec{s} \in \mc{C}_S,\\ \dec \circ \vec{s} \circ P = \vec{t}}} \vert q(\vec{s})-p_{S,\beta}(\vec{s})  \vert \\
    &=
    \frac{1}{2}
    \sum_{\vec{s}\in \mc{C}_S} \vert q(\vec{s})-p_{S,\beta}(\vec{s}) \vert\\
    &= \Vert q - p_{S,\beta} \Vert,
\end{split}
\end{equation}
where $ p_{\mc{f}, \beta} $ is given in \cref{def:sim prob}.
The first equality follows from the definition of the total variation distance, $q_{\mc{f}}$ and $p_{\mc{f},\beta}$, 
the next line from the triangle inequality, 
and the last two equalities form combining the two sums into one sum over all configurations $\vec{s}$.

Second, recall that by \cref{thm:boltzmann sys}, whenever $\mc{f}\colon S \to T$ with cut-off $\Del > \mr{max}(H_T)$, then 
\begin{equation}\label{eq:tot var err}
    \Vert p_{\mc{f}, \beta} - p_{T,\beta} \Vert < \frac{1}{md} \cdot q_S^{\vert V_S \vert} \cdot e^{-\beta(\Delta + \mr{min}(H_T))},
\end{equation}
where $m = \vert \enc \vert$.
Following the proof of \cref{thm:sim part fun}, given $T \in \mc{M}'$ and $\epsilon, \beta > 0$ we define $\delta_{\epsilon, \beta}$ such that the right hand side of \eqref{eq:tot var err} becomes smaller than $\epsilon/2$. By the same arguments as in \cref{thm:sim part fun} this can be computed in polytime.

Next,  we compute $\text{\sc Emu}(T, \delta_{\epsilon, \beta})$ to obtain a spin system $S(T, \delta_{\epsilon, \beta})$ and a simulation $\mc{f}(T, \delta_{\epsilon, \beta}) \colon S(T, \delta_{\epsilon, \beta}) \to T$.
In the following we omit the arguments of $S$ and $\mc{f}$. 
Note that all simulation parameters refer to the computed simulation $\mc{f}$ and hence depend on $T,\beta$ and $\epsilon$. 
We define the reduction $(f,g,\alpha)$ by 
\begin{equation}
\begin{split}
    f(T,\beta, \epsilon) &= (S,\beta)\\
    g(T,\beta, \vec{z}, \epsilon) &= \dec \circ \vec{z} \circ \phys\\
    \alpha(\epsilon) &= \frac{\epsilon}{2}.
\end{split}
\end{equation}
Finally, let us show that this defines a $\text{\sc PTAS}$ reduction. 
By assumption, $\vec{z}$ is distributed according to $q$, with 
\begin{equation}\label{eq:assumption approx sample}
    \Vert q - p_{S, \beta} \Vert < \epsilon/2.
\end{equation}
By definition of $g$, $g(T,\beta, \vec{z}, \epsilon)$ is distributed according to 
$q_{\mc{f}}$. 
The total variation distance satisfies 
\begin{equation}
    \Vert q_{\mc{f}} - p_{T,\beta} \Vert \leq  
    \Vert q_{\mc{f}} - p_{\mc{f}, \beta} \Vert + \Vert p_{\mc{f}, \beta} - p_{T, \beta} \Vert .
\end{equation}
Now, by \eqref{eq:assumption approx sample} and \eqref{eq:tot var sim} the first summand is smaller than $\epsilon/2$, and by construction of $\delta_{\epsilon, \beta}$, the second summand too.
In total,  
\begin{equation}
    \Vert q_{\mc{f}} - p_{T,\beta} \Vert < \epsilon.
\end{equation}\end{proof}

\begin{corollary}\label{cor:uni boltz}
 If $\mc{M}$ is universal then there does not exist a $\text{\sc FPRAS}$ for $\text{\sc Asp}_{\mc{M}}$.
\end{corollary}

\begin{proof}
    To prove the claim, we use the interreducibility between approximate counting (such as approximating the partition function) and approximate sampling for self-reducible problems (see \cite{Je86, Je89}), and then apply \cref{cor:part univ}.
    
    In \cite{ko18b} (see also \cite{ST07}) it is proven that a $\text{\sc FPRAS}$ for approximate sampling of the homogeneous Ising model on arbitrary interaction graphs induces a $\text{\sc FPRAS}$ for approximating its partition function.
    Now assume there exists a $\text{\sc FPRAS}$ for approximate sampling from $\mc{M}$. Then by
    \cref{thm:approx sampl} we obtain a $\text{\sc FPRAS}$ for approximate sampling from $\mc{I}_{\mr{hom}, 3}$, by \cref{cor:part univ}.
    This, by \cite{ko18b}, induces a $\text{\sc FPRAS}$ for approximating the partition function of $\mc{I}_{\mr{hom}, 3}$, which by \cite{Sl12}, does not exist. Hence, the $\text{\sc FPRAS}$ for approximate sampling from $\mc{M}$ cannot exist in the first place.
\end{proof}

Similarly to the previous section, the hardness result of \cref{cor:uni boltz} prevents us from directly using \cref{thm:main} to construct algorithms that sample from arbitrary spin models.

In practice, sampling from spin systems is done by constructing Markov chains on the configuration space that have the corresponding Boltzmann distribution as stationary distribution. 
A common choice of Markov chain is the so-called \emph{Glauber dynamics} \cite{Gl63}.
For most interesting examples (and, by \cref{cor:uni boltz}, for universal spin models) this approach suffers from exponential \emph{mixing times}, that is, the convergence of the Markov chain to its stationary distribution is exponentially slow \cite{Di09}.
However, it seems that the properties of a spin model that determine its mixing time are not fully understood \cite{Di09b}.
\cref{thm:approx sampl} implies that whenever $\mc{M}'$ has exponential mixing times and $\mc{M} \to \mc{M}'$, then so does $\mc{M}$. This seems to indicate that the properties that lead to exponential mixing times are preserved by emulations.

Thus, it could be interesting to apply \cref{thm:approx sampl} to cases where the mixing time of the target model is either not known at all, or known to be not exponential.
This could yield, depending on the mixing time of the source model of the emulation, a novel, efficient, approximate sampling algorithm for this target model. It could also help us understand what properties of spin models determine their mixing times, and reveal which of them are preserved by emulations.

 \cref{thm:approx sampl} is very reminiscent of the universal approximation theorem of restricted Boltzmann machines \cite{LR08,Re24}. We shall return to this point in \nameref{sec:outlook}.

\bigskip 

\section{Universality of the 2d Ising Model with Fields}\label{sec:2d Ising}

Universality of a spin model is satisfied by a widely studied model, the 2d Ising model with fields $\mathcal{I}\textsubscript{2d}$, defined in \cref{ex:2d Ising}.

\begin{theorem}\label{thm:2dIsing}
The 2d Ising model with fields is universal. 
\end{theorem}

The goal of this section is to prove this theorem. 
Because of \cref{cor: universality with locally closed} and the fact that $\mathcal{I}\textsubscript{2d}$ is trivially scalable, we only need to show that $\mathcal{I}\textsubscript{2d}$ is functional complete and locally closed. 
We prove functional completeness in \cref{ssec:f.c.Ising}. 
In \cref{ssec:gadget sim} we construct a crossing gadget, 
which is the essential piece to prove local closure in \cref{ssec:l.c.Ising}.

\subsection{Functional Completeness}\label{ssec:f.c.Ising}
In this subsection we prove functional completeness of $\mathcal{I}\textsubscript{2d}$.

\begin{lemma}\label{thm:2D f.c.}
The 2d Ising model with fields is functional complete.
\end{lemma}

\begin{proof}
    We construct an emulation $\Cone(\mc{I}_{\textsubscript{2d}}) \to \mc{B}_2$.
    A generic spin system from $\mc{B}_2$ is of the form $T_{f_{(i,j)}}\langle v_1, v_2 \rangle$, for $i,j \in [2]$.
    First, we consider the case $i=j=1$. 
    By \cref{thm:sim new delta}, w.l.o.g.\ we construct the required simulation for $\delta > 3$. $S^{\delta}_{1,1}$, defined in \cref{fig:IsingFlag}, has the following $3$ ground states with energy $+1/2$: 
    \begin{equation}
        \{ (2,2,1,1),(2,1,1,1),(1,2,1,1)\}.   
    \end{equation}
    Configuration $(1,1,1,1)$ has energy $+3/2$, and 
    all other configurations have energy $\geq \delta$.
    Hence, it yields a simulation $S^{\delta}_{1,1} \to T_{f_{(1,1)}}$ with cut-off $\delta-1/2$ and shift $+1/2$.
    Note that the cut-off can be arbitrarily high by modifying $\delta$.

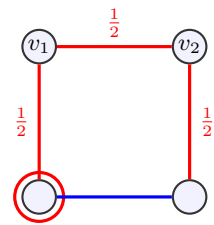
\begin{figure}[th]
    \centering
     \begin{tikzpicture}[circuit logic US, cgnodeS/.style = {draw=black!80, fill=blue!10, thick, circle, minimum size= {width("$v_{4}$")+4pt}, inner sep=1pt}] 

    \node[cgnodeS] (v1) at (-2,0) {$v_1$};
    \node[cgnodeS] (v3) at (0,0) {$v_2$};
    \draw[very thick, red] (v1) to node[midway, above] {$\frac{1}{2}$} (v3);

    \node[cgnodeS] (x1) at (-2,-2) {};  
    \draw[very thick, red] (-2,-2) circle (1em);

    \node[cgnodeS] (x2) at (0,-2) {};
    \draw[very thick, red] (v1) to node[midway, left] {$\frac{1}{2}$} (x1);
    \draw[very thick, blue] (x2) to (x1);
    \draw[very thick, red] (v3) to node[midway, right] {$\frac{1}{2}$} (x2);

    \end{tikzpicture}
    \caption{Definition of $S^{\delta}_{1,1}$.}
    \label{fig:IsingFlag}
\end{figure}

The remaining cases $(i,j)\neq (1,1)$ can be treated similarly, by defining 
\begin{widetext}
\begin{equation}
\begin{split}
S^{\delta}_{1,2}\langle v_1,v_2,w_1,w_2\rangle &=  
\frac{1}{2} \cdot \Bigl( S_{\bar{\pi}_2}\langle v_1,v_2 \rangle  +  S_{\pi_2}\langle v_2,w_2 \rangle 
		+  S_{\bar{\pi}_2}\langle v_1,w_1 \rangle\ \Bigr)
        + \delta \cdot S_{\bar{\pi}_1}\langle  w_1 \rangle
		 + \delta\cdot S_{\bar{\pi}_2} \langle w_1,w_2 \rangle   \\ 
S^{\delta}_{2,1}\langle v_1,v_2,w_1,w_2\rangle &=  
\frac{1}{2} \cdot \Bigl( S_{\bar{\pi}_2}\langle v_1,v_2 \rangle  + 
		 S_{\pi_2}\langle v_2,w_2 \rangle 
		+  S_{\bar{\pi}_2}\langle v_1,w_1 \rangle\ \Bigr)
        + \delta \cdot S_{\pi_1}\langle  w_1 \rangle
		 + \delta\cdot S_{\bar{\pi}_2} \langle w_1,w_2 \rangle   \\
	S^{\delta}_{2,2}\langle v_1,v_2,w_1,w_2\rangle &=  
\frac{1}{2} \cdot \Bigl(S_{\pi_2}\langle v_1,v_2 \rangle  + 		 S_{\pi_2}\langle v_2,w_2 \rangle 		+  S_{\pi_2}\langle v_1,w_1 \rangle\ \Bigr)        + \delta \cdot S_{\bar{\pi}_1}\langle  w_1 \rangle
		 + \delta\cdot S_{\bar{\pi}_2} \langle w_1,w_2 \rangle .
   \label{eq: S_22}
\end{split}
\end{equation}
\end{widetext}
This construction satisfies \cref{def:spin model sim}.    
\end{proof}

\subsection{The Crossing Gadget}\label{ssec:gadget sim}

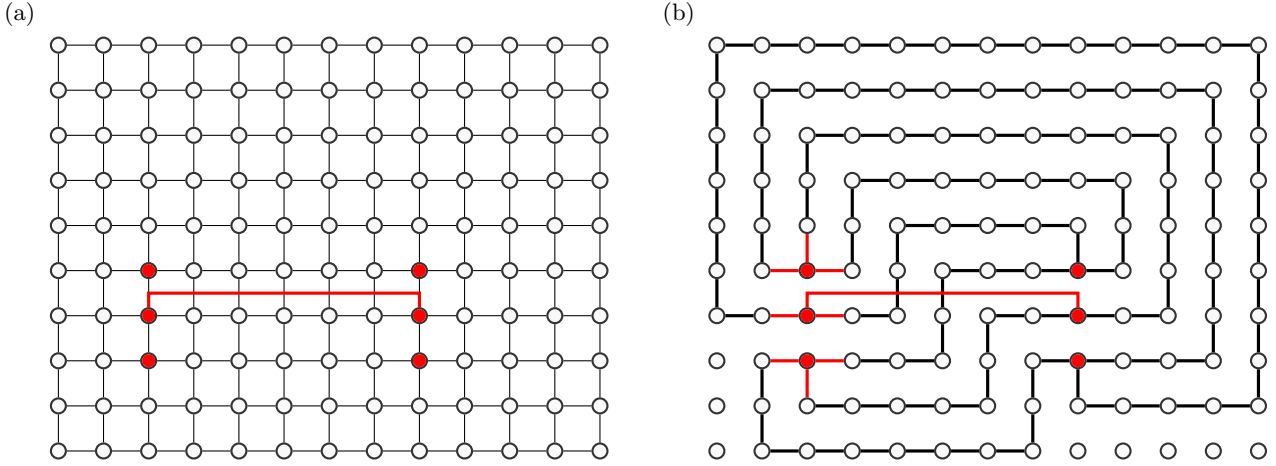
\begin{figure*}[th]
\centering 
\begin{subfigure}{1\columnwidth}
\caption{}\label{fig:lattice additional edge}
\centering 
\begin{tikzpicture}[cgnodeS/.style = {draw=black!80, fill=gray!5, thick, circle, minimum size= 0.5em, inner sep=2pt},scale = 0.6] 
\foreach \x in {1,...,13} {
    \foreach \y in {1,...,10} {
        \node[cgnodeS] (\x a\y) at (\x, \y)  {} ;
        }
}
\foreach[evaluate=\x as \a using int(\x +1)] \x in {1,...,12} {
    \foreach[evaluate=\y as \b using int(\y +1)] \y in {1,...,9} {
        \draw[-] (\x a\y) to (\x a\b);    
        \draw[-] (\x a\y) to (\a a\y);             

    }
}
\foreach[evaluate=\x as \a using int(\x +1)] \x in {1,...,12} {
\draw[-] (\x a10) to (\a a10);
}
\foreach[evaluate=\y as \b using int(\y +1)] \y in {1,...,9} {
\draw[-] (13a\y) to (13a\b);
}
\draw[very thick, red] (3a4) to (3,4.5) to (9,4.5) to (9a4);

\filldraw[red](3,5) circle (0.4em);
\filldraw[red](3,4) circle (0.4em);
\filldraw[red](3,3) circle (0.4em);

\filldraw[red](9,5) circle (0.4em);
\filldraw[red](9,4) circle (0.4em);
\filldraw[red](9,3) circle (0.4em);

\end{tikzpicture}
\end{subfigure}
\begin{subfigure}{1\columnwidth}
    \caption{}
    \label{fig:Lattice3-3Minor}
    \centering 
    \begin{tikzpicture}[cgnodeS/.style = {draw=black!80, fill=gray!5, thick, circle, minimum size= 0.5em, inner sep=2pt},
    scale = 0.6] 

    \foreach \x in {1,...,13} {
        \foreach \y in {1,...,10} {
            \node[cgnodeS] (\x a\y) at (\x, \y)  {} ;
    }
  }

    \filldraw[red](3,5) circle (0.4em);
    \filldraw[red](3,4) circle (0.4em);
    \filldraw[red](3,3) circle (0.4em);

    \filldraw[red](9,5) circle (0.4em);
    \filldraw[red](9,4) circle (0.4em);
    \filldraw[red](9,3) circle (0.4em);

    \draw[very thick, red] (3a5) to (2a5);

    \draw[-, very thick] (2a5) to (2a6) to (2a7) to (2a8) to (2a9) to  (3a9) to (4a9) to (5a9) to (6a9) to (7a9) to (8a9) to (9a9) to (10a9) to (11a9) to (12a9) to (12a8) to (12a7) to (12a6) to (12a5) to (12a4) to (12a3) to (11a3) to (10a3) to (9a3);

    \draw[very thick, red] (3a5) to (3a6);

    \draw[-, very thick] (3a6) to (3a7) to (3a8) to (4a8) to (5a8) to (6a8) to (7a8) to (8a8) to (9a8) to (10a8) to (11a8) to (11a7) to (11a6) to (11a5) to (11a4) to (10a4) to (9a4);

    \draw[very thick, red] (3a5) to (4a5);

    \draw[-, very thick] (4a5) to (4a6) to (4a7) to (5a7) to (6a7) to (7a7) to (8a7) to (9a7) to (10a7) to (10a6) to (10a5) to (9a5);

    \draw[very thick, red] (3a4) to (2a4);

    \draw[-, very thick] (2a4) to (1a4) to (1a5) to (1a6) to (1a7) to (1a8) to (1a9) to (1a10) to (2a10) to (3a10) to (4a10) to (5a10) to (6a10) to (7a10) to (8a10) to (9a10) to (10a10) to (11a10) to (12a10) to (13a10) to (13a9) to (13a8) to (13a7) to (13a6) to (13a5) to (13a4) to (13a3) to (13a2) to (12a2) to (11a2) to (10a2) to (9a2) to (9a3);

    \draw[very thick, red] (3a4) to (4a4);

    \draw[-, very thick] (4a4) to (5a4) to (5a5) to (5a6) to (6a6) to (7a6) to (8a6) to (9a6) to (9a5);

    \draw[very thick, red] (3a3) to (4a3);

    \draw[-, very thick] (4a3) to (5a3) to (6a3) to (6a4) to (6a5) to (7a5) to (8a5) to (9a5);

    \draw[very thick, red] (3a3) to (3a2);

    \draw[-, very thick] (3a2) to (4a2) to (5a2) to (6a2) to (7a2) to (7a3) to (7a4) to (8a4) to (9a4);

    \draw[very thick, red] (3a3) to (2a3);

    \draw[-, very thick] (2a3) to (2a2) to (2a1) to (3a1) to (4a1) to (5a1) to (6a1) to (7a1) to (8a1) to (8a2) to (8a3) to (9a3);  

    \draw[very thick, red] (3a4) to (3,4.5) to (9,4.5) to (9a4);
    
    \end{tikzpicture}
    \end{subfigure}
    \caption{Graph $G$ consisting of a $13\times 10$ lattice with an additional edge between vertex $(3,4)$ and vertex $(9,4)$ drawn in red (\ref{fig:lattice additional edge}). Proof that $G$ is non-planar (\ref{fig:Lattice3-3Minor}): merging along the black edges and deleting all other edges and isolated vertices yields $K_{3,3}$ as minor of $G$ (cf.\ Wagner's theorem), namely each of the three red vertices on the left is connected to each of the three red vertices on the right. 
    Here colors just highlight elements and do not refer to the graphical notation for spin systems of \cref{fig:2D close legend}.}\label{fig:lattice minor}
\end{figure*}

Here we define the crossing gadget $I^{\delta}_{\times}$, which we shall use to prove of local closure of $\mc{I}_{\textsubscript{2d}}$ in \cref{ssec:l.c.Ising}.
To this end, we shall first explain why non-planarity is unavoidable in the construction of the emulation $\mc{I}_{\textsubscript{2d}} \to \mc{I}_{\textsubscript{2d}} + J(\mc{I}_{\textsubscript{2d}})$, which is necessary to prove local closure. 
We shall then construct the crossing gadget and explain how it can be used to simulate arbitrary non-planar spin systems with planar source systems.

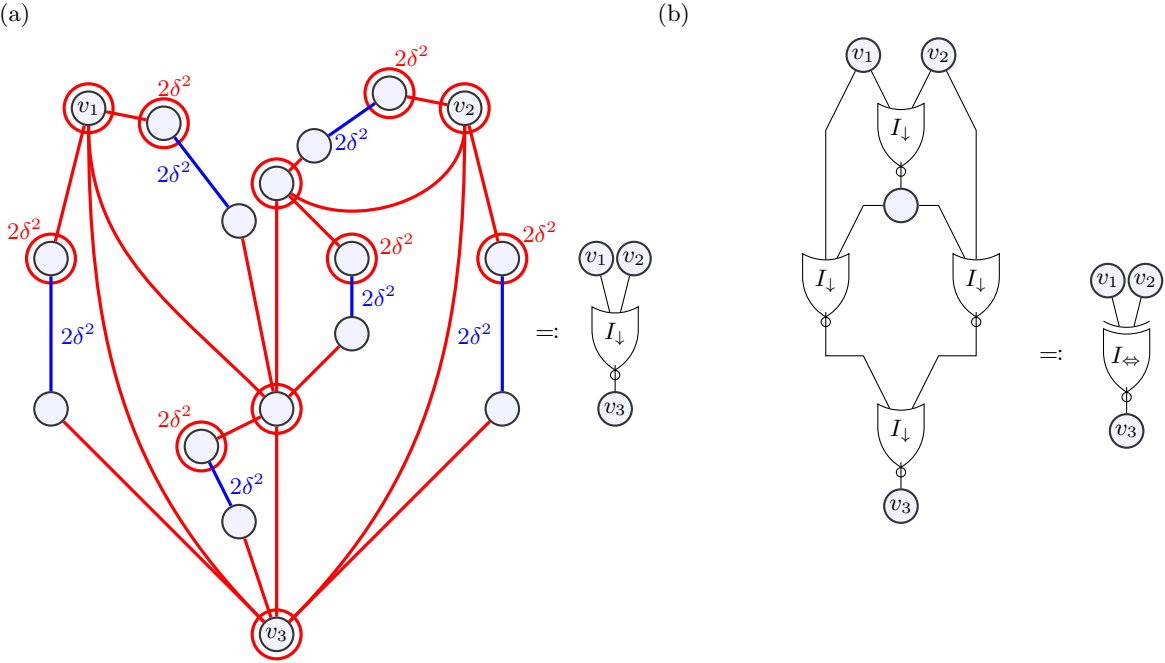
\begin{figure*}[th]
    \centering
    \begin{subfigure}[t]{1\columnwidth}
    \caption{}\label{fig:IsingNor}
    \begin{tikzpicture}[circuit logic US, cgnodeS/.style = {draw=black!80, fill=blue!10, thick, circle, minimum size= {width("$v_{4}$")+4pt}, inner sep=1pt}] 

    \node[cgnodeS] (v2) at (2.5,0) {$v_2$}; 
    \draw[very thick, red](2.5,0) circle (1em);

    \node[cgnodeS] (v1) at (-2.5,0) {$v_1$};
    \draw[very thick, red](-2.5,0) circle (1em);

    \node[cgnodeS] (a2) at (0,-1) {};
    \draw[very thick, red](0,-1) circle (1em);

    \node[cgnodeS] (a1) at (0,-4) {};
    \draw[very thick, red](0,-4) circle (1em);

    \node[cgnodeS] (v3) at (0,-7) {$v_3$};
    \draw[very thick, red](0,-7) circle (1em);

    \draw[very thick, red] (v1) to [out=270,in=135] (a1);

    \draw[very thick, red] (v1) to [out=270,in=135] (v3);

    \draw[very thick, red] (a1) to (a2);

    \draw[very thick, red] (v3) to (a1);

    \draw[very thick, red] (v2) to [out=270,in=45] (v3);

    \draw[very thick, red] (v2) to [out=270,in=-35] (a2);

    \node (=) at (3.6,-3) {$\eqqcolon$};

    \node[nor gate, point down, info=center:$I_\downarrow$] (g) at (4.5,-3) {};

    \node[cgnodeS] (e21) at (4.25,-2) {$v_1$};
    \node[cgnodeS] (e22) at (4.75,-2) {$v_2$};
    \node[cgnodeS] (e23) at (4.5,-4) {$v_3$};

    \draw[-] (e21) to (g);
    \draw[-] (e22) to (g);
    \draw[-] (e23) to (g);

    \node[cgnodeS] (x1) at (-3,-2) {};  
    \draw[very thick, red] (-3,-2) circle (1em);
    \node[text = red] (d1) at (-3.35,-1.6) {$2\delta^2$};

    \node[cgnodeS] (x2) at (-3,-4) {};
    \draw[very thick, red] (v1) to (x1);
    \draw[very thick, blue] (x2) to node[midway, right] {$2\delta^2$} (x1);
    \draw[very thick, red] (v3) to (x2);

    \node[cgnodeS] (x3) at (3,-2) {};    
    \draw[very thick, red] (3,-2) circle (1em);
    \node[text = red] (d2) at (3.5,-1.65) {$2\delta^2$};
    \node[cgnodeS] (x4) at (3,-4) {};
    \draw[very thick, red] (v2) to (x3);
    \draw[very thick, blue] (x4) to node[midway, left] {$2\delta^2$} (x3);
    \draw[very thick, red] (v3) to (x4);

    \node[cgnodeS] (x5) at (-1,-4.5) {}; 
    \draw[very thick, red] (-1,-4.5) circle (1em);
    \node[text = red] (d1) at (-1.35,-4.1) {$2\delta^2$};
    \node[cgnodeS] (x6) at (-0.5,-5.5) {};
    \draw[very thick, red] (a1) to (x5);
    \draw[very thick, blue] (x5) to node[midway, right] {$2\delta^2$} (x6);
    \draw[very thick, red] (v3) to (x6);

    \node[cgnodeS] (x7) at (1,-2) {};  
    \draw[very thick, red] (1,-2) circle (1em);
    \node[text = red] (d3) at (1.6,-1.8) {$2\delta^2$};
    \node[cgnodeS] (x8) at (1,-3) {};
    \draw[very thick, red] (a2) to (x7);
    \draw[very thick, blue] (x7) to node[midway, right] {$2\delta^2$}  (x8);
    \draw[very thick, red] (a1) to (x8);

    \node[cgnodeS] (x9) at (-1.5,-0.2) {};   
    \draw[very thick, red] (-1.5,-0.2) circle (1em);
    \node[text = red] (d4) at (-1.35,0.3) {$2\delta^2$};
    \node[cgnodeS] (x10) at (-0.5,-1.5) {};
    \draw[very thick, red] (v1) to (x9);
    \draw[very thick, blue] (x9) to node[midway, left] {$2\delta^2$}  (x10);
    \draw[very thick, red] (a1) to (x10);

    \node[cgnodeS] (x11) at (1.5,0.2) {}; 
     \draw[very thick, red] (1.5,0.2) circle (1em);
    \node[text = red] (d5) at (1.8,0.7) {$2\delta^2$};
    \node[cgnodeS] (x12) at (0.5,-0.5) {};
    \draw[very thick, red] (v2) to (x11);
    \draw[very thick, blue] (x11) to node[midway, below] {$2\delta^2$}  (x12);
    \draw[very thick, red] (a2) to (x12);
    \end{tikzpicture}
\end{subfigure}
\begin{subfigure}[t]{1\columnwidth}
    \caption{}\label{fig:iffIsing}
    \begin{tikzpicture}[circuit logic US, cgnodeS/.style = {draw=black!80, fill=blue!10, thick, circle, minimum size= {width("$v_{4}$")+4pt}, inner sep=1pt}] 

    \node[nor gate, point down, info=center:$I_\downarrow$] (g1) at (0,0) {};
    \node[nor gate, point down, info=center:$I_\downarrow$] (g2) at (-1,-2) {};
    \node[nor gate, point down, info=center:$I_\downarrow$] (g3) at (1,-2) {};
    \node[nor gate, point down, info=center:$I_\downarrow$] (g4) at (0,-4) {};

    \node[cgnodeS] (v1) at (-0.5,1) {$v_1$};
    \node[cgnodeS] (v2) at (0.5,1) {$v_2$};

    \draw[-] (v1) to (g1); 
    \draw[-] (v2) to (g1);

    \draw[-] (v1) to (-1,0) to (g2); 
    \draw[-] (v2) to (1,0) to (g3);

    \node[cgnodeS] (v3) at (0,-1) {};
    \draw[-] (v3) to (g1); 
    \draw[-] (v3) to (-0.5,-1) to (g2); 
    \draw[-] (v3) to (0.5,-1) to (g3);

    \draw[-] (g2) to (-1,-3) to (-0.5,-3) to (g4);
    \draw[-] (g3) to (1,-3) to (0.5,-3) to (g4);

    \node[cgnodeS] (v4) at (0,-5) {$v_3$};
    \draw[-] (g4) to (v4);

    \node (eq) at (2,-3) {$\eqqcolon$};

    \node[xnor gate, point down, info=center:$I_\Leftrightarrow$] (g5) at (3,-3) {};

    \node[cgnodeS] (v11) at (2.75,-2) {$v_1$};
    \node[cgnodeS] (v21) at (3.25,-2) {$v_2$};
    \node[cgnodeS] (v31) at (3,-4) {$v_3$};

    \draw[-] (v11) to (g5);
    \draw[-] (v21) to (g5);
    \draw[-] (v31) to (g5);
  
    \end{tikzpicture}
\end{subfigure}
\caption{Definition and symbolic notation of $I^{\delta}_{\textsc{nor}}$ (\ref{fig:IsingNor}) and  $I^{\delta}_{\textsc{iff}}$ (\ref{fig:iffIsing}).}\label{fig:Iff}
\end{figure*}

A graph is planar if it can be embedded into the 2d plane such that none of its edges cross. 
Planar graphs can be fully characterized in terms of their graph minors: According to Wagner's theorem, a graph is planar if and only if it does not contain the complete graph on $5$ vertices, $K_5$, or the complete, bipartite graph on $3+3$ vertices, $K_{3,3}$, as minor.

Clearly, all interaction graphs from $\mc{I}_{\textsubscript{2d}}$, i.e.\ all 2d lattices, are planar. 
But adding a single edge to a 2d lattice, as in $\mc{I}_{\textsubscript{2d}} + J(\mc{I}_{\textsubscript{2d}})$, generally leads to non-planarity (see \cref{fig:lattice minor}). 
Spin systems over non-planar graphs can be simulated with spin systems over planar graphs by means of the crossing gadget. 

\begin{figure*}[th]
        \centering 
        \begin{subfigure}[t]{1\columnwidth} 
        \caption{}\label{fig:cross1}
        \begin{tikzpicture}[circuit logic US, cgnodeS/.style = {draw=black!80, fill=blue!10, thick, circle, minimum size= {width("$v_{4}$")+4pt}, inner sep=1pt},
        and/.style = {draw=black, double,  rectangle, minimum size = 2em }] 
         \node[xnor gate, point down, info=center:$I_\Leftrightarrow$] (g1) at (0,0) {};
         \node[xnor gate, point down, info=center:$I_\Leftrightarrow$] (g2) at (-1,-2) {};
        \node[xnor gate, point down, info=center:$I_\Leftrightarrow$] (g3) at (1,-2) {};

        \node[cgnodeS] (v1) at (-1,1.5) {$v_1$};
        \node[cgnodeS] (v2) at (1,1.5) {$v_2$};

        \draw[-] (v1) to (-0.25, 1) to (g1); 
        \draw[-] (v2) to (0.25,1) to (g1);

        \draw[-] (v1) to (-1,0) to (g2); 
        \draw[-] (v2) to (1,0) to (g3);

        \node[cgnodeS] (v3) at (0,-1) {};
        \draw[-] (v3) to (g1); 
        \draw[-] (v3) to (-0.5,-1) to (g2); 
        \draw[-] (v3) to (0.5,-1) to (g3);

        \node[cgnodeS] (v4) at (-1,-3) {$v_3$};
        \node[cgnodeS] (v5) at (1,-3) {$v_4$};

        \draw[-] (g2) to (v4);
        \draw[-] (g3) to (v5);

        \node (=) at (2,-1) {$\eqqcolon$};
    
        \node[cgnodeS] (e1) at (2.5,0) {$v_1$};
        \node[cgnodeS] (e2) at (3.5,0) {$v_2$};
        \node[cgnodeS] (e3) at (2.5,-2) {$v_3$};
        \node[cgnodeS] (e4) at (3.5,-2) {$v_4$};
    
        \node[and] (g4) at (3,-1) {$\times$};
        \draw[-] (e1) to (g4);
        \draw[-] (e2) to (g4);
        \draw[-] (e3) to (g4);
        \draw[-] (e4) to (g4);
        \end{tikzpicture} 
        
        \end{subfigure}
    \begin{subfigure}[t]{1\columnwidth}
    \caption{}\label{fig:cross2}
     \begin{tikzpicture}[circuit logic US, cgnodeS/.style = {draw=black!80, fill=blue!10, thick, circle, minimum size= {width("$v_{4}$")+4pt}, inner sep=1pt},
        and/.style = {draw=black, double,  rectangle, minimum size = 2em }] 
        \node[cgnodeS] (e1) at (2.5,0) {$v_1$};
        \node[cgnodeS] (e2) at (3.5,0) {$v_2$};
        \node[cgnodeS] (e3) at (2.5,-2) {};
        \node[cgnodeS] (e4) at (3.5,-2) {};
    
        \node[and] (g4) at (3,-1) {$\times$};
        \draw[-] (e1) to (g4);
        \draw[-] (e2) to (g4);
        \draw[-] (e3) to (g4);
        \draw[-] (e4) to (g4);

        \node[cgnodeS] (e5) at (4,-3) {$v_4$};
        \node[cgnodeS] (e6) at (2,-3) {$v_3$};

        \draw[very thick, brown] (e4) to (e5);
        \draw[very thick, black] (e3) to (e6);

        \node[cgnodeS] (f1) at (6,-0.5) {$v_1$};
        \node[cgnodeS] (f2) at (7,-0.5) {$v_2$};
        \node[cgnodeS] (f3) at (6,-2.5) {$v_3$};
        \node[cgnodeS] (f4) at (7,-2.5) {$v_4$};

        \draw[very thick, brown] (f1) to (f4);
        \draw[very thick, black] (f2) to (f3);

        \draw[->] (4,-1.5) to (5.5,-1.5);
       
        \end{tikzpicture} 
        
    \end{subfigure}
         \caption{Definition and symbolic notation of $I^{\delta}_{\times}$ (\ref{fig:cross1}), where $I_{\Leftrightarrow}$ is defined in \cref{fig:Iff}. Simulation of two crossing edges with arbitrary local interactions by $I^{\delta}_{\times}$ (\ref{fig:cross2}).}\label{fig:CrossGadget}
\end{figure*}
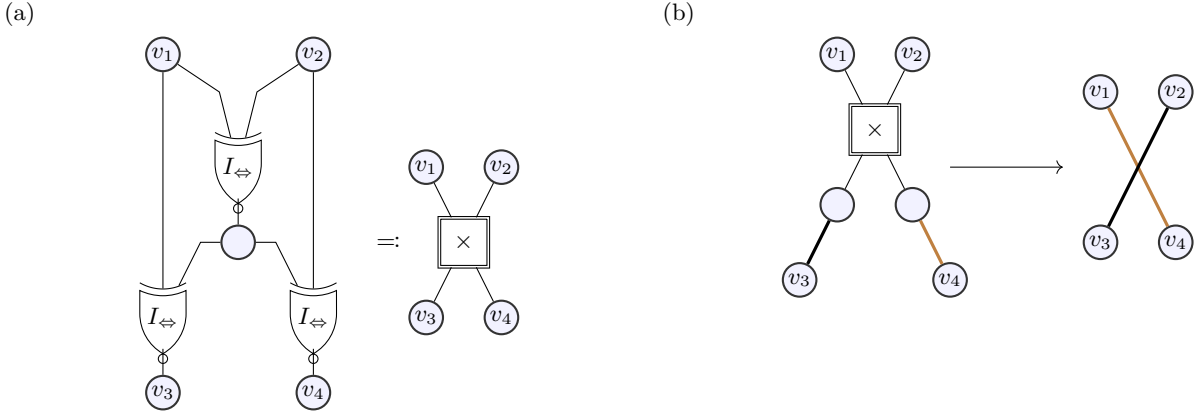

To construct the crossing gadget, we leverage that 2d Ising systems with fields can encode Boolean logic in their ground state (cf.\ \cref{thm:2D f.c.}). 
First, we construct an Ising system $I^{\delta}_{\textsc{nor}}$ (defined in \cref{fig:IsingNor}) which in the ground state computes the logical $\textsc{nor}$ of the states of two of its spins $v_1,v_2$ and stores the result in the state of a third spin $v_3$. It is obtained by replacing the $\mc{B}_{2}$ interactions in $S^{\delta}_{\textsc{nor}}$ (see \cref{fig:NOR}) with the corresponding Ising interactions. More precisely, by \cref{thm:2D f.c.}, $S^{\delta}_{1,1}$ (defined in \cref{fig:IsingFlag}) simulates $f_{(1,1)}$ with cut-off $\delta-1/2$ and shift $1/2$. Thus, by \cref{thm:sim scale}, 
\begin{equation}
    2\delta \cdot S^{\delta}_{1,1} \to 2\delta\cdot f_{1,1}
\end{equation}
with cut-off $2\delta (\delta - 1/2)$ and shift $\delta$.

We now replace 
all $2\delta \cdot f_{(1,1)}$ interactions by $2\delta \cdot S^{\delta}_{1,1}$ spin systems,
and all $\delta \cdot f_{(2,1)}+ \delta \cdot f_{(2,2)}$ by $\delta \cdot \pi_1$ fields.
In the following we assume that 
$2\delta^2>10\delta$, which guarantees that in the ground state of $I^{\delta}_{\textsc{nor}}$ the subsystems $2\delta \cdot S^{\delta}_{1,1}$  correctly simulate $2\delta\cdot f_{1,1}$ interactions.
Then, similarly to $S^{\delta}_{\textsc{nor}}$, the ground state of $I^{\delta}_{\textsc{nor}}$ satisfies
\begin{equation}\label{eq:nor}
    \vec{s}(v_3) = \textsc{nor}(\vec{s}(v_1),\vec{s}(v_2)).
\end{equation}
Its ground-state energy is $9\delta$, and all other configurations have energy greater than $10\delta$.
Note that while $S^{\delta}_{\textsc{nor}}$ has $4\cdot 2^5$ ground states, 
$I^{\delta}_{\textsc{nor}}$ only has $4$ ground states, one for each possibility of satisfying \cref{eq:nor}. This results from replacing the $5$ copies of $\delta \cdot f_{(2,1)}+ \delta \cdot f_{(2,2)}$ each of which has a 2-fold degenerate ground state, with $\delta \cdot \pi_1$ which has a non-degenerate ground state.

Next, we construct an Ising system $I^{\delta}_{\textsc{iff}}$ (defined in \cref{fig:iffIsing}) that in its ground state implements the logical bidirectional $\Leftrightarrow$, also called $\textsc{xnor}$.
It is constructed from $S^{\delta}_{\textsc{nor}}$, using the fact that $x \Leftrightarrow y$ can be expressed in terms of $\textsc{nor}$ as
\begin{equation}
   ( x \Leftrightarrow y ) = \textsc{nor}\bigl(\textsc{nor}(x,\textsc{nor}(x,y)), \textsc{nor}(y, \textsc{nor}(x,y)) \bigr).
\end{equation}
In the ground state of $I^{\delta}_{\textsc{iff}}$, $\vec{s}(v_3) = \bigl(\vec{s}(v_1) \Leftrightarrow \vec{s}(v_2)\bigr)$, which can also be stated as 
\begin{equation}\label{eq:iff}
    \vec{s}(v_3) =  \begin{cases}
        1 \ \text{if } \ \vec{s}(v_1) = \vec{s}(v_2) \\
        2 \ \text{else}.
    \end{cases}
\end{equation}
Its ground-state energy is $36\delta$, and all other configurations have energy at least $37\delta$.
Note that $I^{\delta}_{\textsc{iff}}$ has $4$ ground states only, one for each possibility of satisfying \eqref{eq:iff}.

Finally, the crossing gadget $I^{\delta}_{\times}$ (defined in \cref{fig:CrossGadget}) is constructed from $I^{\delta}_{\textsc{iff}}$. From the identity 
\begin{equation}
  x = \bigl(x \Leftrightarrow (x \Leftrightarrow y) \bigr) 
\end{equation}
we conclude that in the ground state of $I^{\delta}_{\times}$,  
\begin{equation}\label{eq:cross}
    \vec{s}(v_3) = \vec{s}(v_2) \quad \text{and} \quad \vec{s}(v_4) = \vec{s}(v_1).
\end{equation}
Its ground-state energy is $108\delta$, and all other configurations have energy at least $109\delta$.
Note that $I^{\delta}_{\times}$ has precisely $4$ ground states, one for each possibility of satisfying \eqref{eq:cross}.

So, in the ground state, $I^{\delta}_{\times}$ copies the state of $v_2$ to $v_3$ and the state of $v_1$ to $v_4$. Since both $I^{\delta}_{\textsc{nor}}$ and $I^{\delta}_{\textsc{iff}}$ are planar, so is $I^{\delta}_{\times}$.
We will now show how to simulate two crossing edges with a planar source system by means of 
$I^{\delta}_{\times}$ (see \cref{fig:cross2}).

\begin{lemma}[crossing gadget]\label{lem:crossGadget}
    Let $T$ be a spin system consisting of two edges $\{v_1,v_3\}$ and $\{v_2,v_4\}$ with arbitrary local interactions. Then \cref{fig:cross2} defines a simulation with cut-off $\delta$ and shift $108\delta$.
\end{lemma}

\begin{proof}
    It follows from the ground-state behavior of $I^{\delta}_{\times}$. 
\end{proof}
By \cref{thm:sim sum}, the crossing gadget can simulate an arbitrary non-planar spin system $T$ (containing only pair interactions and fields) with a planar source system. One ought to iterate over all crossings  of $T$ (w.r.t.\ some embedding into the plane), and locally replace the two crossing edges according to \cref{fig:cross2} while keeping the rest of $T$ unchanged (see \cref{fig:sim sum} for $S_i\to T_i$ being the simulation from \cref{lem:crossGadget}).

\subsection{Local Closure}\label{ssec:l.c.Ising}

In this subsection we prove local closure of $\mathcal{I}\textsubscript{2d}$ and thereby finish the proof of its universality. To this end we construct an emulation $\mc{I}_{\textsubscript{2d}} \to \mc{I}_{\textsubscript{2d}} + J(\mc{I}_{\textsubscript{2d}})$ by means of the crossing gadget defined in \cref{fig:CrossGadget}.

\begin{lemma}\label{thm:2d local closed}
The 2d Ising model with fields is locally closed.
\end{lemma}

\begin{proof}
    Given as input spin systems $T\in \mc{I}_{\textsubscript{2d}}$ and $K \in J(\mc{I}_{\textsubscript{2d}})$, as well as a positive real number $\delta$, we provide a polytime construction of 
    a spin system $S\in \mathcal{I}\textsubscript{2d}$ and a
    simulation $S \to T+K$ with cut-off $\delta$.  
    
    First, a generic spin system $T$ from $ \mathcal{I}\textsubscript{2d}$ is isomorphic to a spin system $\tilde{T}$ with interaction graph $G_{m,n}$ (cf.\ \eqref{eq:grid graph}) for some $m,n\geq 2$.
    We now argue that w.l.o.g.\ we can assume that $T$ has interaction graph $G_{m,n}$.
    By \cite{Bu85}, we can construct a graph isomorphism between $V_T$ and $G_{m,n}$ in polytime. This induces a spin system isomorphism $\phi$ between $\tilde{T}$ and $T$, which by \cref{lem:iso sim} induces a simulation $\tilde{T} \to T$. By \cref{thm:sim sum} we therefore obtain a simulation 
    \begin{equation}
        \mc{r}\colon \tilde{T}+\tilde{K} \to T + K,
    \end{equation}
    where $\tilde{K}$ is obtained by relabeling spins in $V_K\cap V_T$ according to $\phi$.   
    Thus, if the interaction graph of $T$ is not of the form $G_{m,n}$, we simply work with $\tilde{T}$, use it to construct the simulation of $\tilde{T}+\tilde{K}$ and finally compose the result with the simulation $\mc{r}$.
    We henceforth call spins from $T$ with degree $2$ or $3$ the boundary spins, and those with degree $4$ the interior, where the degree of a spin is the number of edges it is included in.

    A spin system $K$ from $J(\mc{I}_{\textsubscript{2d}})$ consists either of a single spin $v$ with local interaction given by an Ising field $\pi_1, \bar{\pi}_1$, or of two spins $v,w$ sharing an edge with local interaction given by an Ising pair interaction $\pi_2, \bar{\pi}_2$.
    
    \subsubsection{Single Spin $v$}
    We first consider the case of $K$ consisting of a single spin $v$. 
    There are two sub-cases, either $v\in V_T$ or $v\notin V_T$.

    \medskip
    \paragraph*{Sub-case $v \in V_T$.}
    We start with 
    $v \in V_T$.  We further distinguish the two sub-sub-cases of $J_T(\{v\})$ and $J_K(\{v\})$ being of the same type, i.e.\ both scalar multiples of $\pi_1$ or $\bar{\pi}_1$, or $J_T(\{v\})$ and $J_K(\{v\})$
    being of different types. 
    
If they are of the same type, say $J_T(\{v\}) = \lambda_1 \cdot \pi_1$ and $J_K(\{v\}) = \lambda_2 \cdot \pi_1$, then their sum equals $(\lambda_1 + \lambda_2) \cdot \pi_1$ and thus is itself an Ising field. Hence, $T + K \in \mathcal{I}\textsubscript{2d}$, and we take $S = T+K$ together with the identity simulation. 
     
If they are of different types then their sum is not an Ising interaction.
We construct an Ising interaction that equals their sum up to a constant shift. 
    Say $J_T(\{v\}) = \lambda_1 \cdot \pi_1$ and $J_K(\{v\}) = \lambda_2 \cdot \bar{\pi}_1$ with $\lambda_1 \geq \lambda_2$, then 
    \begin{equation}\label{eq:Ising field sum}
        (\lambda_1-\lambda_2)\cdot \pi_1 + \lambda_2 = \lambda_1 \cdot \pi_1 + \lambda_2 \cdot \bar{\pi}_1.
    \end{equation}
    Let $S$ be equal to $T$ except for the local interaction $J_S(\{v\})$ which we take to be $(\lambda_1-\lambda_2)\cdot \pi_1$. This yields a simulation $S \to T+K$ which is trivial except for the non-trivial shift of $-\lambda_2$.
   
    \medskip
    \paragraph*{Sub-case $v \notin V_T$.}
    Consider now the sub-case $v\notin V_T$. We define $S$ by 
    taking $G_{m+1,n+1}$ as interaction graph (where $G_{m,n}$ is the interaction graph of $T$) and defining
    \begin{equation}
        J_S(e) = \begin{cases}
            J_T(e) \ \text{if } e \in E_T \\
            J_K(e) \ \text{if } e=\{(m+1,1) \} \\
            0 \ \text{else},
        \end{cases}
    \end{equation}
    i.e.\ on hyperedges from $G_{m,n}$, $S$ has the same interactions as $T$, 
    the field $J_K(\{v\})$ acts on 
    the additional spin $(m+1,1)$,
    and all other additional interactions are zero.
    This yields a simulation $S \to T+K$ with physical spin assignment $\phys(v) = (m+1,n)$.

    This can also be interpreted as follows. 
    First, we extend the interaction grid graph of $T$ by one row and one column only containing  zero interactions. This defines a spin systems $T_+$. 
    By \cref{thm: minor sim} we trivially get a simulation $T_+ \to T$.
    Then we apply an isomorphism to $K$ that identifies $v$ with the border spin $(m+1,n)$ of $T_+$.
    This defines a simulation $\tilde{K} \to K$. 
    By \cref{thm:sim sum} we have $T_+ +\tilde{K} \to T+K$ and by construction 
    $T_+, \tilde{K}$ belong to the sub-case $v \in V_T$. This effectively reduces the sub-case $v \notin V_T$ to the sub-case $v \in V_T$.

\begin{figure*}[th]
\centering 
\begin{subfigure}{0.5\columnwidth}
\caption{}\label{fig:lc 1}
\centering 
\begin{tikzpicture}[cgnodeS/.style = {draw=black!80, fill=blue!10, thick, circle, minimum size= 0.5em, inner sep=2pt},scale = 0.6] 
\foreach \x in {1,...,6} {
    \foreach \y in {1,...,5} {
        \node[cgnodeS] (\x a\y) at (\x, \y)  {} ;
        }
}
\foreach[evaluate=\x as \a using int(\x +1)] \x in {1,...,5} {
    \foreach[evaluate=\y as \b using int(\y +1)] \y in {1,...,4} {
        \draw[-] (\x a\y) to (\x a\b);    
        \draw[-] (\x a\y) to (\a a\y);             

    }
}
\foreach[evaluate=\x as \a using int(\x +1)] \x in {1,...,5} {
\draw[-] (\x a5) to (\a a5);
}
\foreach[evaluate=\y as \b using int(\y +1)] \y in {1,...,4} {
\draw[-] (6a\y) to (6a\b);
}

\draw[very thick, brown] (2a2) to (3a2);

\end{tikzpicture}
\end{subfigure}
\begin{subfigure}{0.5\columnwidth}
    \caption{}
    \label{fig:lc 2}
    \centering 
    \begin{tikzpicture}[cgnodeS/.style = {draw=black!80, fill=blue!10, thick, circle, minimum size= 0.5em, inner sep=2pt},
    scale = 0.6] 

   \foreach \x in {1,...,6} {
    \foreach \y in {1,...,5} {
        \node[cgnodeS] (\x a\y) at (\x, \y)  {} ;
        }
}
\foreach[evaluate=\x as \a using int(\x +1)] \x in {1,...,5} {
    \foreach[evaluate=\y as \b using int(\y +1)] \y in {1,...,4} {
        \draw[-] (\x a\y) to (\x a\b);    
        \draw[-] (\x a\y) to (\a a\y);             

    }
}
\foreach[evaluate=\x as \a using int(\x +1)] \x in {1,...,5} {
\draw[-] (\x a5) to (\a a5);
}
\foreach[evaluate=\y as \b using int(\y +1)] \y in {1,...,4} {
\draw[-] (6a\y) to (6a\b);
}
\node[cgnodeS] (v1) at (7,5) {};
\draw[very thick, brown] (6a5) to (v1);

\node[cgnodeS] (w1) at (7,1) {};
\node[cgnodeS] (w2) at (7,2) {};
\draw[very thick, brown] (w1) to (w2);

    \end{tikzpicture}
    \end{subfigure}
\begin{subfigure}{0.5\columnwidth}
    \caption{}
    \label{fig:lc 3}
    \centering 
    \begin{tikzpicture}[cgnodeS/.style = {draw=black!80, fill=blue!10, thick, circle, minimum size= 0.5em, inner sep=2pt},
    scale = 0.6] 

   \foreach \x in {1,...,6} {
    \foreach \y in {1,...,5} {
        \node[cgnodeS] (\x a\y) at (\x, \y)  {} ;
        }
}
\foreach[evaluate=\x as \a using int(\x +1)] \x in {1,...,5} {
    \foreach[evaluate=\y as \b using int(\y +1)] \y in {1,...,4} {
        \draw[-] (\x a\y) to (\x a\b);    
        \draw[-] (\x a\y) to (\a a\y);             

    }
}
\foreach[evaluate=\x as \a using int(\x +1)] \x in {1,...,5} {
\draw[-] (\x a5) to (\a a5);
}
\foreach[evaluate=\y as \b using int(\y +1)] \y in {1,...,4} {
\draw[-] (6a\y) to (6a\b);
}

\node[cgnodeS] (w) at (7,4) {};
\draw[very thick, brown] (2a2) to [out = 20, in = 270] (w);

\end{tikzpicture}
\end{subfigure}
\begin{subfigure}{0.5\columnwidth}
    \caption{}
    \label{fig:lc 4}
    \centering 
    \begin{tikzpicture}[cgnodeS/.style = {draw=black!80, fill=blue!10, thick, circle, minimum size= 0.5em, inner sep=2pt},
    scale = 0.6] 

   \foreach \x in {1,...,6} {
    \foreach \y in {1,...,5} {
        \node[cgnodeS] (\x a\y) at (\x, \y)  {} ;
        }
}
\foreach[evaluate=\x as \a using int(\x +1)] \x in {1,...,5} {
    \foreach[evaluate=\y as \b using int(\y +1)] \y in {1,...,4} {
        \draw[-] (\x a\y) to (\x a\b);    
        \draw[-] (\x a\y) to (\a a\y);             

    }
}
\foreach[evaluate=\x as \a using int(\x +1)] \x in {1,...,5} {
\draw[-] (\x a5) to (\a a5);
}
\foreach[evaluate=\y as \b using int(\y +1)] \y in {1,...,4} {
\draw[-] (6a\y) to (6a\b);
}
\draw[very thick, brown] (2a2) to [out = 15, in = 260] (5a4);

\end{tikzpicture}
\end{subfigure}
    \caption{If $K$ consists of two interacting spins, we consider four sub-cases: sub-case \ref{it:lc1} (\ref{fig:lc 1}), 
the two sub-sub-cases of sub-case \ref{it:lc2} (\ref{fig:lc 2}), 
sub-case \ref{it:lc3} (\ref{fig:lc 3}), 
and sub-case \ref{it:lc4} (\ref{fig:lc 4}). 
}\label{fig:lc}
\end{figure*}
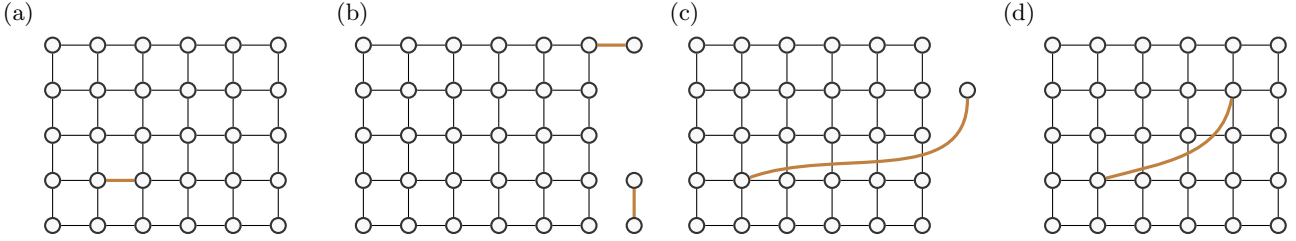

\subsubsection{Two spins $v,w$}
    Next, we consider the case of $K$ consisting of two spins $V_K=\{v,w\}$ that interact via $\lambda \cdot \pi_2$ or $\lambda \cdot \bar{\pi}_2$.
    We distinguish 4 sub-cases (see \cref{fig:lc}):
    \begin{enumerate}
        \item\label{it:lc1} $\{v,w\} \in E_T$
        \item\label{it:lc2} $\{v,w\} \cap V_T = \emptyset$ or  $\{v,w\} \cap V_T = \{v\}$ (or $\{w\}$) and $v$ (or $w$) is part of the boundary of $T$ 
        \item\label{it:lc3} $\{v,w\} \cap V_T  = \{v\}$ (or $\{w\}$) and $v$ (or $w$) is part of the interior of $T$
        \item\label{it:lc4} $\{v,w\} \subseteq V_T $ but $\{v,w\} \notin E_T$.
    \end{enumerate}

\medskip
\paragraph*{Sub-cases \ref{it:lc1} and \ref{it:lc2}.}
Sub-cases \ref{it:lc1} and \ref{it:lc2} can be treated similarly to the sub-cases of $K$ consisting of a single spin. In sub-case \ref{it:lc1} we let $S$ be equal to $T$ and change the local interaction $J_T(\{v,w\})$.
Similarly to $\pi_1, \bar{\pi}_1$, also $\pi_2, \bar{\pi}_2$ satisfy 
\begin{equation}\label{eq:Ising int sum}
    (\lambda_1-\lambda_2) \cdot \pi_2 + \lambda_2 = \lambda_1 \cdot \pi_2 + \lambda_2 \cdot \bar{\pi}_2.
\end{equation}
In sub-case \ref{it:lc2} we construct $S$ by enlarging the interaction graph of $T$ by one row and one column and take all additional local interactions equal to zero, except for the local interaction corresponding to $J_K(\{v,w\})$. 

\medskip
\paragraph*{Sub-case \ref{it:lc3}.}
Sub-case \ref{it:lc3} reduces to sub-case \ref{it:lc4} by enlarging the interaction graph of $T$ by one row and one column, and identifying $w$ with a border spin of the so defined spin system $T_+$. 

\medskip
\paragraph*{Sub-case \ref{it:lc4}.}
Sub-case \ref{it:lc4} is the only non-trivial sub-case.
We assume w.l.o.g.\ that $v=(i,j)$ and $w=(i+r,j+s)$ for $r,s>0$; all other possibilities can be treated similarly.
We construct $S$ and the simulation $S \to T+K$ in four steps, namely spin systems $S,S_{\mr{deg}}, S_{\mr{cross}}$ and $S_{\mr{part}}$, as well as simulations 
\begin{equation}
    S \to S_{\mr{deg}} \to S_{\mr{cross}} \to S_{\mr{part}} \to T + K.
\end{equation}
More precisely, starting from $T+K$ and proceedings `backwards', we construct $S_{\mr{part}}$ by partitioning the additional edge $\{v,w\}$ into a sequence of edges, each of which crosses precisely one edge of $T$. We construct $S_{\mr{cross}}$ by replacing each such crossing with a crossing gadget, resulting in a planar interaction graph. We construct $S_{\mr{deg}}$ by introducing additional vertices such that the interaction graph has degree $\leq 4$. Finally, $S$ is obtained by embedding this interaction graph of degree $\leq 4$ into a large enough grid graph.
In the following, we devote one paragraph to each of these four steps.

\begin{figure}[th]
    \centering
    \begin{tikzpicture}[cgnodeS/.style = {draw=black!80, fill=blue!10, thick, circle, minimum size= 0.5em, inner sep=2pt}, scale = 0.9] 

   \foreach \x in {1,...,4} {
    \foreach \y in {1,...,3} {
        \node[cgnodeS] (\x a\y) at (\x, \y)  {} ;
        }
}
\foreach[evaluate=\x as \a using int(\x +1)] \x in {1,...,3} {
    \foreach[evaluate=\y as \b using int(\y +1)] \y in {1,...,2} {
        \draw[-] (\x a\y) to (\x a\b);    
        \draw[-] (\x a\y) to (\a a\y);             

    }
}
\foreach[evaluate=\x as \a using int(\x +1)] \x in {1,...,3} {
\draw[-] (\x a3) to (\a a3);
}
\foreach[evaluate=\y as \b using int(\y +1)] \y in {1,...,2} {
\draw[-] (4a\y) to (4a\b);
}

\draw[very thick, brown] (1a1) to [out = 30, in = 180] (2.5,1.5) to [out = 0, in = 225] (3.5,1.5) to [out = 45, in = 270] (3.5,2.5) to [out = 90, in = 225] (4a3);

\foreach[evaluate=\x as \a using int(\x -5)] \x in {1,...,4} {
    \foreach \y in {1,...,3} {
        \node[cgnodeS] (\x b\y) at (\a, \y)  {} ;
        }
}
\foreach[evaluate=\x as \a using int(\x +1)] \x in {1,...,3} {
    \foreach[evaluate=\y as \b using int(\y +1)] \y in {1,...,2} {
        \draw[-] (\x b\y) to (\x b\b);    
        \draw[-] (\x b\y) to (\a b\y);             

    }
}
\foreach[evaluate=\x as \a using int(\x +1)] \x in {1,...,3} {
\draw[-] (\x b3) to (\a b3);
}
\foreach[evaluate=\y as \b using int(\y +1)] \y in {1,...,2} {
\draw[-] (4b\y) to (4b\b);
}

\node[cgnodeS] (i1) at (-2.5,1.5) {};
\node[cgnodeS] (i2) at (-1.5,1.5) {};

\draw[very thick, brown] (1b1) to [out = 30, in = 180] (i1);
\draw[very thick, red] (i1) to [out = 0, in = 225] (i2) to [out = 45, in = 270] (-1.5,2.5) to [out = 90, in = 225] (4b3);

\draw[->] (-0.4, 2) to (0.4, 2);

\foreach \x in {1,...,4}{
\draw[-] (\x a3) to ( [yshift = 10.0] \x a3);
\draw[-] (\x a1) to ( [yshift = -10.0] \x a1);

\draw[-] (\x b3) to ( [yshift = 10.0] \x b3);
\draw[-] (\x b1) to ( [yshift = -10.0] \x b1);
}

\foreach \y in {1,...,3}{
\draw[-] (1a\y) to ( [xshift = -10.0] 1a\y);
\draw[-] (4a\y) to ( [xshift = 10.0] 4a\y);

\draw[-] (1b\y) to ( [xshift = -10.0] 1b\y);
\draw[-] (4b\y) to ( [xshift = 10.0] 4b\y);
}

\node[] (v1) at (-4.2,0.8) {$v$};
\node[] (v2) at (0.8,0.8) {$v$};

\node[] (w1) at (-0.8,3.2) {$w$};
\node[] (w2) at (4.2,3.2) {$w$};

\node[] (i3) at (-2.5,1.8) {$v_1$};
\node[] (i4) at (-1.7,1.75) {$v_2$};

\draw[] (-4,0.3) to node[below] {$r$} (-1,0.3);
\draw[] (-4,0.2) to (-4,0.4);
\draw[] (-1,0.2) to (-1,0.4);

\draw[] (-4.7,1) to node[left] {$s$} (-4.7,3);
\draw[] (-4.8,1) to (-4.6,1);
\draw[] (-4.8,3) to (-4.6,3);

\end{tikzpicture}
    \caption{The simulation $S_{\mr{part}} \to T+K$ for $r=3, s=2$. By introducing two spins, $v_1,v_2$, the edge $\{v,w\}$ is partitioned into 3 edges, the first of which carries the original interaction, while the other two carry $\delta\cdot \pi_2$ interactions which, below the cut-off, force the two interacting spins to be in equal states. Note that drawing the additional spins and edges as illustrated leads to precisely one crossing for each of the three additional edges. }
    \label{fig:crossing path}
\end{figure}
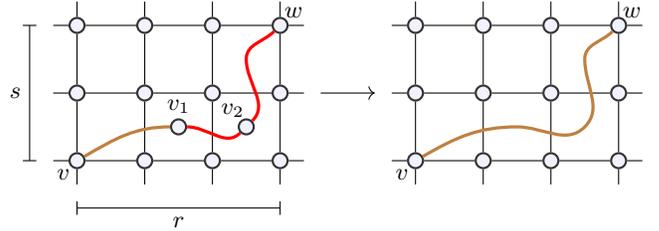

\medskip
\paragraph*{Step 1:  $S_{\mr{part}}$.}
First, we partition the edge $\{v,w\}$ into a path of $r+s-2$ edges (\cref{fig:crossing path}), each of which crosses precisely one edge from $T$. This amounts to modifying the interaction graph of $T$ by removing edge $\{v,w\}$ and introducing $r+s-3$ additional vertices $v_1, \ldots, v_{r+s-3}$ and $r+s-2$ additional edges, $\{v,v_1\}$, $\{v_i,v_{i+1}\}$ for $i = 1, \ldots, r+s-4$ and $\{v_{r+s-3},w\}$. We denote the resulting graph by $G_{\mr{part}}$. 

We now construct a spin system $S_{\mr{part}}$ with interaction graph $G_{\mr{part}}$, by
choosing the local interaction attached to the first additional edges, $\{v,v_1\}$, as $J_K(\{v,w\})$. We choose $\delta \cdot \pi_2$ for all other additional edges, and  $J_T(e)$ for all edges $e$ which are contained in $T$. 
By \cref{thm: minor sim}, we obtain a simulation 
\begin{equation}\label{eq:sim S part}
    S_{\mr{part}} \to T+K 
\end{equation}
with cut-off $\delta$ and shift $0$. Note that the edges with  $\delta\cdot \pi_2$ interactions are precisely those  that must be contracted to recover $T+K$ from $S_{\mr{part}}$. The $\delta\cdot \pi_2$ interactions implement this contraction by forcing the two interaction spins to be equal whenever a configuration has energy below the cut-off.

\begin{figure*}[th]
\centering 

\begin{subfigure}{1\columnwidth}
    \caption{}\label{fig:cross sim 1}
\begin{tikzpicture}[circuit logic US, cgnodeS/.style = {draw=black!80, fill=blue!10, thick, circle, minimum size=1.0em, inner sep=2pt},
        and/.style = {draw=black, double,  rectangle, minimum size = 2em }] 
        \node[cgnodeS] (11) at (0,1) {};
        \node[cgnodeS] (12) at (0,2) {};
        \node[cgnodeS] (21) at (1,1) {};
        \node[cgnodeS] (22) at (1,2) {};
        \node[cgnodeS] (31) at (2,1) {};
        \node[cgnodeS] (32) at (2,2) {};
        \node[cgnodeS] (x1) at (1.5,1.5) {};
        \node[] (v1) at (-0.3,0.8) {$v$};
        \node[] (v2) at (1.6,1.8) {$v_1$};

        \draw[very thick] (11) to (12);
        \draw[very thick] (11) to (21);
        \draw[very thick] (12) to (22);
        \draw[very thick] (21) to (22);
        \draw[very thick] (21) to (31);
        \draw[very thick] (31) to (32);
        \draw[very thick] (22) to (32);

        \draw[very thick, brown] (11) to [out = 15, in = 200] (x1);

        \node[cgnodeS] (11b) at (-6,0) {};
        \node[cgnodeS] (12b) at (-6,3) {};
        \node[cgnodeS] (21b) at (-4,0) {};
        \node[cgnodeS] (22b) at (-4,3) {};
        \node[cgnodeS] (31b) at (-2,0) {};
        \node[cgnodeS] (32b) at (-2,3) {};
        \node[cgnodeS] (x2) at (-2.8,2) {};
        \node[] (v3) at (-6.3,-0.2) {$v$};
        \node[] (v4) at (-2.7,2.3) {$v_1$};

        \node[and] (c) at (-4,1.5) {$\times$};
        \draw[] (22b) to (c);
        \draw[] (x2) to [out = 200, in = 70] (c);
        \node[cgnodeS] (c1) at (-5,0.75) {};
        \node[cgnodeS] (c2) at (-4,0.75) {};

        \draw[very thick] (11b) to (12b);
        \draw[very thick] (11b) to (21b);
        \draw[very thick] (12b) to (22b);
        \draw[very thick] (21b) to (31b);
        \draw[very thick] (31b) to (32b);
        \draw[very thick] (22b) to (32b);

        \draw[very thick, brown] (11b) to (c1);
        \draw[very thick] (21b) to (c2);
        \draw[] (c1) to [out = 45, in = 240] (c);
        \draw[] (c2) to (c);

        \draw[->] (-1.5,1.5) to (-0.5,1.5);

        \draw[very thick, red] (x1) to (2.25,1.75);

        \draw[very thick, red] (x2) to (-1.8,2.2);

        \end{tikzpicture}
        \end{subfigure}
        \begin{subfigure}{1\columnwidth}
    \caption{}\label{fig:cross sim 2}
\begin{tikzpicture}[circuit logic US, cgnodeS/.style = {draw=black!80, fill=blue!10, thick, circle, minimum size= 1.0em, inner sep=2pt},
        and/.style = {draw=black, double,  rectangle, minimum size = 2em }] 
        \node[cgnodeS] (11) at (0,1) {};
        \node[cgnodeS] (12) at (0,2) {};
        \node[cgnodeS] (21) at (1,1) {};
        \node[cgnodeS] (22) at (1,2) {};
        \node[cgnodeS] (31) at (2,1) {};
        \node[cgnodeS] (32) at (2,2) {};
        \node[cgnodeS] (x1) at (1.5,1.6) {};
        \node[cgnodeS] (x2) at (0.5,1.4) {};
        \node[] (v1) at (0.4,1.7) {$v_1$};
        \node[] (v2) at (1.6,1.3) {$v_2$};

        \draw[very thick, brown] (x2) to (-0.2,1.2); 
         \draw[very thick, red] (x1) to (2.2,1.8); 

        \draw[very thick] (11) to (12);
        \draw[very thick] (11) to (21);
        \draw[very thick] (12) to (22);
        \draw[very thick] (21) to (22);
        \draw[very thick] (21) to (31);
        \draw[very thick] (31) to (32);
        \draw[very thick] (22) to (32);

        \draw[very thick, red] (x1) to (x2);

        \node[cgnodeS] (11b) at (-6,0) {};
        \node[cgnodeS] (12b) at (-6,3) {};
        \node[cgnodeS] (21b) at (-4,0) {};
        \node[cgnodeS] (22b) at (-4,3) {};
        \node[cgnodeS] (31b) at (-2,0) {};
        \node[cgnodeS] (32b) at (-2,3) {};
        \node[cgnodeS] (x3) at (-5.5,0.8) {};
        \node[cgnodeS] (x4) at (-2.8,2.5) {};
        \node[] (v3) at (-5.6,1.1) {$v_1$};
        \node[] (v4) at (-2.7,2.2) {$v_2$};

        \node[and] (c) at (-4,2) {$\times$};
        \draw[] (22b) to (c);
        \draw[] (x4) to [out = 200, in = 70] (c);
        \node[cgnodeS] (c1) at (-4.6,1.2) {};
        \node[cgnodeS] (c2) at (-4,1.2) {};

        \draw[very thick] (11b) to (12b);
        \draw[very thick] (11b) to (21b);
        \draw[very thick] (12b) to (22b);
        \draw[very thick] (21b) to (31b);
        \draw[very thick] (31b) to (32b);
        \draw[very thick] (22b) to (32b);

        \draw[very thick, red] (x3) to (c1);
        \draw[very thick] (21b) to (c2);
        \draw[] (c1) to [out = 45, in = 240] (c);
        \draw[] (c2) to (c);

        \draw[->] (-1.5,1.5) to (-0.5,1.5);

        \draw[very thick, brown] (x3) to (-6.2,0.6);
         \draw[very thick, red] (x4) to (-1.8,2.6);

        \end{tikzpicture}
        \end{subfigure}
        \caption{The simulation $S_{\mr{cross}} \to S_{\mr{part}}$ to replace the first (\ref{fig:cross sim 1}) and second (\ref{fig:cross sim 2}) crossing with a crossing gadget (see \cref{fig:CrossGadget}).
        The full simulation is obtained by treating every crossing this way, where for the last $s-1$ crossings, i.e.\ those that cross horizontal edges, the construction must be rotated. 
        Only the spins $v,v_1,v_2$ are labelled explicitly.
        }
    \label{fig:cross sim}  
\end{figure*}
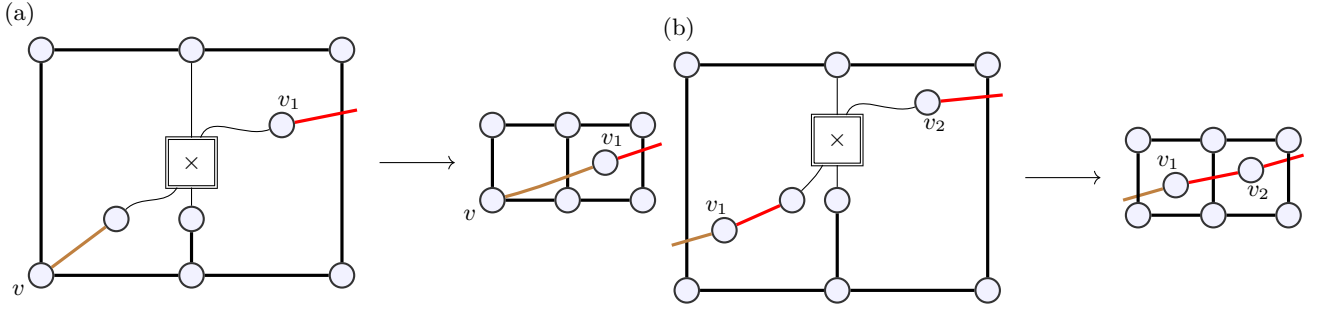

\medskip
\paragraph*{Step 2: $S_{\mr{cross}}$.}
We replace each pair of crossing edges with the crossing gadget (see \cref{fig:CrossGadget}). For  a single crossing this is illustrated  in \cref{fig:cross sim}.
This defines a spin system $S_{\mr{cross}}$, 
and using the simulation of \cref{fig:cross2} (see also \cref{lem:crossGadget}) together with \cref{thm:model sim sum}, a simulation 
\begin{equation}\label{eq:sim S cross}
    S_{\mr{cross}} \to S_{\mr{part}}
\end{equation}
with cut-off $\delta$ and shift $\Shift_{\mr{cross}} \coloneqq 108\delta \cdot (r+s-2)$.  
Note that $S_{\mr{cross}}$ is a planar spin system with interactions from $\mathcal{I}\textsubscript{2d}$.

Observe that we can compute the $r+s-2$ crossings of $S_{\mr{part}}$ in polytime.
By construction, the horizontal distance between $v$ and $w$ is $r$. This means that the first $r-1$ of the additional edges cross vertical edges of $T$, i.e.\ the first additional edge $\{v,v_1\}$ crosses edge $\{(i+1,j),(i+1,j+1)\}$, and for $k=1, \ldots, r-2$, edge $\{v_k,v_{k+1}\}$ crosses $\{(i+k+1,j),(i+k+1,j+1)\}$. The remaining $s-1$ additional edges cross vertical edges of $T$, i.e.\ 
for $l=1, \ldots,s-2 $,  edge $\{v_{r+l-2},v_{r+l-1} \}$ crosses $\{(i+r-1,j+l),(i+r,j+l) \}$ and the last additional edge $\{v_{r+s-3}, w\}$ crosses  edge $\{(i+r-1,j+s-1),(i+r,j+s-1) \}$.
Therefore, also $S_{\mr{cross}}$ and the simulation \eqref{eq:sim S cross} can be computed in polytime.

\medskip
\paragraph*{Step 3: $S_{\mr{deg}}$.}
For the construction of $S_{\mr{deg}}$, we start by noting that there are three mechanisms that might lead to vertices of $S_{\mr{cross}}$ having degree $> 4$: 
\begin{itemize}
\item[(a)] Vertices of the crossing gadget \cref{fig:CrossGadget} of degree $>4$; 
\item[(b)]External vertices of the crossing gadget ($v_1, \ldots, v_4$ in \cref{fig:cross1}) of degree $>1$ (note that in \cref{fig:cross sim 2} such external vertices are connected to vertices which have degree at most $3$); 
and 
\item[(c)] The two vertices $v,w$ belonging to the additional edge from $K$.
\end{itemize}

We construct $S_{\mr{deg}}$  
by distributing the edges of these spins with degree $>4$ over multiple spins which are connected amongst each other with $\delta'\cdot \pi_2$ interactions.
The latter ensure that in the low-energy sector, the spins have equal states and thus mimic a single spin.
Since the previous 
simulation has shift $\Shift_{\mr{cross}}$, the relevant low-energy sector is below $\Shift_{\mr{cross}}+\delta$. We thus pick
\begin{equation}
    \delta' \coloneqq \delta + \Shift_{\mr{cross}}
\end{equation}
to achieve this.

The precise replacement rules in the construction of $S_{\mr{deg}}$ are listed in \cref{fig:deg}.  By \cref{thm: minor sim},  applying any number of such rules results in a simulation, in our case of type $S_{\mr{deg}} \to S_{\mr{cross}}$.

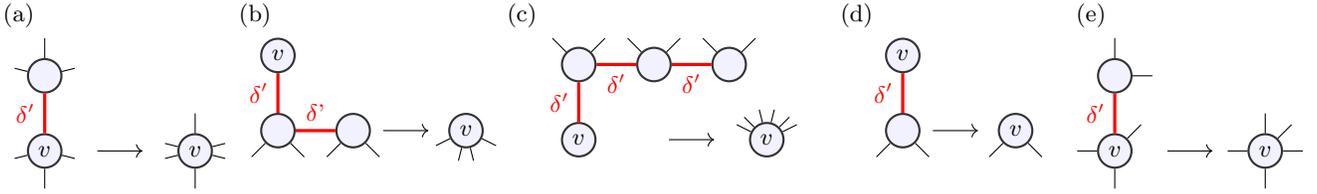
\begin{figure*}[th]
\centering
\begin{subfigure}[t]{0.35\columnwidth}
\caption{}\label{fig:deg1}
\begin{tikzpicture}[circuit logic US, cgnodeS/.style = {draw=black!80, fill=blue!10, thick, circle, minimum size= {width("$v_{4}$")+4pt}, inner sep=1pt},
        and/.style = {draw=black, double,  rectangle, minimum size = 2em }] 
       \node[cgnodeS] (v) at (0,0) {$v$};
       \node[cgnodeS] (w) at (0,1) {};
       \draw[] (v) to (-0.4,-0.1);
        \draw[] (v) to (0.4,-0.1);
       \draw[] (v) to (0,-0.5);
       \draw[very thick, red] (w) to node[midway,left] {$\delta'$} (v);
       \draw[] (w) to (-0.4,1.1);
       \draw[] (w) to (0.4,1.1);
       \draw[] (w) to (0,1.5);

\draw[->] (0.7,0) to (1.3,0);
\node[cgnodeS] (v2) at (2,0) {$v$};
\draw[] (v2) to (2,-0.5); 
\draw[] (v2) to (1.6,-0.1); 
\draw[] (v2) to (2.4,-0.1);
\draw[] (v2) to (2,0.5);
\draw[] (v2) to (1.6,0.1); 
\draw[] (v2) to (2.4,0.1); 

\end{tikzpicture}
\end{subfigure}
\begin{subfigure}[t]{0.4\columnwidth}
\caption{}\label{fig:deg2}
\begin{tikzpicture}[circuit logic US, cgnodeS/.style = {draw=black!80, fill=blue!10, thick, circle, minimum size= {width("$v_{4}$")+4pt}, inner sep=1pt},
        and/.style = {draw=black, double,  rectangle, minimum size = 2em }] 
       \node[cgnodeS] (v) at (0,0) {};
       \node[cgnodeS] (w1) at (0,1) {$v$};
       \node[cgnodeS] (w2) at (1,0) {};

       \draw[] (v) to (-0.35,-0.35);
       \draw[] (v) to (0.35,-0.35);
       \draw[very thick, red] (w1) to node[midway, left] {$\delta'$} (v);
        \draw[very thick, red] (w2) to node[midway, above] {$\delta$'} (v);
       \draw[] (w2) to (0.65,-0.35);
       \draw[] (w2) to (1.35,-0.35);

\draw[->] (1.4,0) to (2,0);
\node[cgnodeS] (v2) at (2.5,0) {$v$};
\draw[] (v2) to (2.1,-0.2); 
\draw[] (v2) to (2.9,-0.2); 
\draw[] (v2) to (2.4,-0.4);
\draw[] (v2) to (2.6,-0.4);
\end{tikzpicture}
\end{subfigure}
\begin{subfigure}[t]{0.5\columnwidth}
\caption{}\label{fig:deg3}
\begin{tikzpicture}[circuit logic US, cgnodeS/.style = {draw=black!80, fill=blue!10, thick, circle, minimum size= {width("$v_{4}$")+4pt}, inner sep=1pt},
        and/.style = {draw=black, double,  rectangle, minimum size = 2em }] 
       \node[cgnodeS] (v) at (0,0) {$v$};
       \node[cgnodeS] (w1) at (0,1) {};
       \node[cgnodeS] (w2) at (1,1) {};
        \node[cgnodeS] (w3) at (2,1) {};
        \draw[very thick, red] (w1) to node[midway, left] {$\delta'$} (v);
        \draw[very thick, red] (w2) to node[midway, below] {$\delta'$} (w1);
        \draw[very thick, red] (w2) to node[midway, below] {$\delta'$} (w3);

       \draw[] (w1) to (-0.35,1.35);
       \draw[] (w1) to (0.35,1.35);

       \draw[] (w2) to (0.65,1.35);
       \draw[] (w2) to (1.35,1.35);
       
       \draw[] (w3) to (1.65,1.35);
       \draw[] (w3) to (2.35,1.35);

\draw[->] (1.2,0) to (1.8,0);
\node[cgnodeS] (v2) at (2.5,0) {$v$};
\draw[] (v2) to (2.1,0.2); 
\draw[] (v2) to (2.9,0.2); 
\draw[] (v2) to (2.2,0.3); 
\draw[] (v2) to (2.8,0.3); 
\draw[] (v2) to (2.4,0.4);
\draw[] (v2) to (2.6,0.4);
\end{tikzpicture}
\end{subfigure}
\begin{subfigure}[t]{0.35\columnwidth}
\caption{}\label{fig:deg4}
\begin{tikzpicture}[circuit logic US, cgnodeS/.style = {draw=black!80, fill=blue!10, thick, circle, minimum size= {width("$v_{4}$")+4pt}, inner sep=1pt},
        and/.style = {draw=black, double,  rectangle, minimum size = 2em }] 
       \node[cgnodeS] (v) at (0,0) {};
       \node[cgnodeS] (w1) at (0,1) {$v$};

       \draw[] (v) to (-0.35,-0.35);
       \draw[] (v) to (0.35,-0.35);
       \draw[very thick, red] (w1) to node[midway, left] {$\delta'$} (v);

\draw[->] (0.4,0) to (1,0);
\node[cgnodeS] (v2) at (1.5,0) {$v$};
\draw[] (v2) to (1.15,-0.35); 
\draw[] (v2) to (1.85,-0.35);
\end{tikzpicture}
\end{subfigure}
\begin{subfigure}[t]{0.35\columnwidth}
\caption{}\label{fig:deg5}
\begin{tikzpicture}[circuit logic US, cgnodeS/.style = {draw=black!80, fill=blue!10, thick, circle, minimum size= {width("$v_{4}$")+4pt}, inner sep=1pt},
        and/.style = {draw=black, double,  rectangle, minimum size = 2em }] 
       \node[cgnodeS] (v) at (0,0) {$v$};
       \node[cgnodeS] (w) at (0,1) {};
       \draw[] (v) to (-0.5,0);
        \draw[] (v) to (0.35,0.35);
       \draw[] (v) to (0,-0.5);
       \draw[very thick, red] (w) to node[midway, left] {$\delta'$} (v);
       \draw[] (w) to (0.5,1);
       \draw[] (w) to (0,1.5);

\draw[->] (0.7,0) to (1.3,0);
\node[cgnodeS] (v2) at (2,0) {$v$};
\draw[] (v2) to (2,-0.5); 
\draw[] (v2) to (1.5,0); 
\draw[] (v2) to (2.5,0);
\draw[] (v2) to (2,0.5);
\draw[] (v2) to (2.35,0.35); 

\end{tikzpicture}
\end{subfigure}
    \caption{Replacement rules used in the construction of $S_{\mr{deg}}$. Each rule defines a simulation with cut-off $\delta'$, by \cref{thm: minor sim}.}
    \label{fig:deg}
\end{figure*}

We first explain how to modify the crossing gadget so that it has no vertices of degree $>4$ and its external vertices have degree $1$, thereby treating cases (a) and (b) of high degree vertices of $S_{\mr{cross}}$.   
Recall that the crossing gadget is constructed from three copies of $I^{\delta}_{\textsc{iff}}$, 
each of which is constructed from four copies of 
$I^{\delta}_{\textsc{nor}}$ (see \cref{fig:Iff} and \cref{fig:CrossGadget}).

$I^{\delta}_{\textsc{nor}}$ contains one internal spin with degree $>4$. We apply rule \ref{fig:deg1} to this spin. On top of that, the upper two external spins $v_1,v_2$ of $I^{\delta}_{\textsc{nor}}$ have degree $4$ and the lower external spin $v_3$ has degree $6$. We apply rule \ref{fig:deg2} to the upper external spins and rule \ref{fig:deg3} to the lower external spin. In total, the modified spin system $\tilde{I}^{\delta}_{\textsc{nor}}$ has no spin with degree $>4$ and all of its external spins have degree $1$.

Next, we construct $\tilde{I}^{\delta}_{\textsc{iff}}$ by first following the construction of \cref{fig:iffIsing} but replacing each $I^{\delta}_{\textsc{nor}}$ with the corresponding $\tilde{I}^{\delta}_{\textsc{nor}}$. 
This leads to the upper external spins having degree $2$. We apply rule \ref{fig:deg4} to those. The resulting spin system, $\tilde{I}^{\delta}_{\textsc{iff}}$, has the same functionality as  $I^{\delta}_{\textsc{iff}}$, but all of its internal spins have degree $<4$ and its three external spins $v_1,v_2,v_3$ have degree $1$.

Finally, we construct $\tilde{I}^{\delta}_{\times}$ by first following \cref{fig:cross sim 1}, but replacing $I^{\delta}_{\textsc{iff}}$ with $\tilde{I}^{\delta}_{\textsc{iff}}$ and then again, applying rule \ref{fig:deg4} to the two upper external spin. This ultimately yields a crossing gadget with no spin of degree $>4$. Additionally, all of its external spins have degree $1$. Hence, connecting this modified crossing gadget as illustrated in \cref{fig:cross sim} does not create additional spin of degree $>4$. 

Since this modification concerns the crossing gadget itself, i.e.\ is independent of the input spin systems $T,K$, it can be computed in constant time. As the number of crossing gadgets in $S_{\mr{cross}}$, $r+s-2$ is clearly polynomial in $\vert T \vert$, applying the resulting modifications to $S_{\mr{cross}}$ is polytime computable.

We finish the construction of $S_{\mr{deg}}$ by 
applying rule \ref{fig:deg5} to $v$ and $w$. Since this modification does not depend on the size of the input either, it is polytime computable too. 
Finally, using \cref{thm: minor sim} we obtain 
\begin{equation}
    S_{\mr{deg}} \to S_{\mr{cross}}
\end{equation}
with cut-off $\delta'$.

\medskip 
\paragraph*{Step 4: $S$.}
Finally, we construct $S\in \mathcal{I}\textsubscript{2d}$, together with $S\to S_{\mr{deg}}$, by embedding $S_{\mr{deg}}$ into a large enough grid graph $G_{m',n'}$. 
Specifically, we use the linear time grid embedding algorithm presented in \cite{Ta89}. 
Given a planar graph $G$ of degree at most $4$ as input, this algorithm identifies vertices of $G$ with grid coordinates, and edges of $G$ with lattices paths between the appropriate grid coordinates. We apply this algorithm to the interaction graph of $S_{\mr{deg}}$ and denote by $G_{m',n'}$ the minimal 2-dimensional grid graph that contains all obtained coordinates of vertices of $G_{S_{\mr{deg}}}$.
By construction $S_{\mr{deg}}$ is planar and has degree $4$, i.e.\ meets the requirements of the grid embedding algorithm.

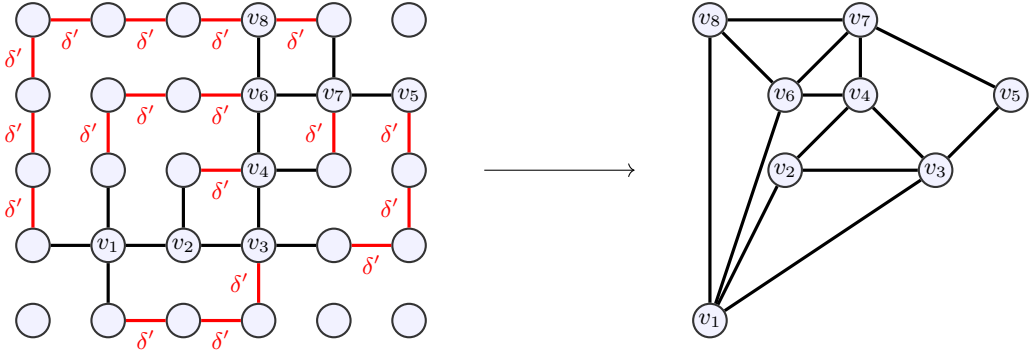
\begin{figure*}[th]
        \centering 
               \begin{tikzpicture}[circuit logic US, cgnodeS/.style = {draw=black!80, fill=blue!10, thick, circle, minimum size= {width("$v_{4}$")+4pt}, inner sep=1pt},
        and/.style = {draw=black, double,  rectangle, minimum size = 2em }] 
        \node[cgnodeS] (1) at (0,0) {$v_1$};
        \node[cgnodeS] (2) at (1,2) {$v_2$};
        \node[cgnodeS] (3) at (3,2) {$v_3$};
        \node[cgnodeS] (4) at (2,3) {$v_4$};
        \node[cgnodeS] (5) at (4,3) {$v_5$};
        \node[cgnodeS] (6) at (1,3) {$v_6$};
        \node[cgnodeS] (7) at (2,4) {$v_7$};
        \node[cgnodeS] (8) at (0,4) {$v_8$};

        \draw[very thick] (1) to (8);
        \draw[very thick] (1) to (3);
        \draw[very thick] (1) to (6);
        \draw[very thick] (1) to (2);
        \draw[very thick] (2) to (3);
        \draw[very thick] (2) to (4);
        \draw[very thick] (3) to (4);
        \draw[very thick] (3) to (5);
        \draw[very thick] (4) to (7);
        \draw[very thick] (5) to (7);
        \draw[very thick] (4) to (6);
        \draw[very thick] (6) to (7);
        \draw[very thick] (6) to (8);
        \draw[very thick] (7) to (8);

\foreach[evaluate=\x as \a using int(\x -10)] \x in {1,...,6} {
    \foreach[evaluate=\y as \b using int(\y -1)] \y in {1,...,5} {
        \node[cgnodeS] (\x a\y) at (\a, \b)  {} ;
        }
}

\node[] (v1) at (-6,4) {$v_8$};
\node[] (v1) at (-6,3) {$v_6$};
\node[] (v1) at (-6,2) {$v_4$};
\node[] (v1) at (-6,1) {$v_3$};
\node[] (v1) at (-5,3) {$v_7$};
\node[] (v1) at (-4,3) {$v_5$};
\node[] (v1) at (-7,1) {$v_2$};
\node[] (v1) at (-8,1) {$v_1$};

\draw[very thick, red] (4a5) to node[midway, below] {$\delta'$}(3a5);
\draw[very thick, red] (3a5) to  node[midway, below] {$\delta'$}(2a5);
\draw[very thick, red] (2a5) to  node[midway, below] {$\delta'$}(1a5);
\draw[very thick, red] (1a5) to  node[midway, left] {$\delta'$}(1a4);
\draw[very thick, red] (1a4) to  node[midway, left] {$\delta'$}(1a3);
\draw[very thick, red] (1a3) to  node[midway, left] {$\delta'$}(1a2);
\draw[very thick] (1a2) to (2a2);

\draw[very thick] (4a5) to (4a4);

\draw[very thick] (4a4) to (4a3);

\draw[very thick] (4a3) to (4a2);

\draw[very thick, red] (4a4) to 
 node[midway, below] {$\delta'$}(3a4);
\draw[very thick, red] (3a4) to  node[midway, below] {$\delta'$}(2a4);
\draw[very thick, red] (2a4) to  node[midway, left] {$\delta'$}(2a3);
\draw[very thick] (2a3) to (2a2);

\draw[very thick] (2a2) to (3a2);

\draw[very thick, red] (4a3) to 
 node[midway, below] {$\delta'$}(3a3);
\draw[very thick] (3a3) to (3a2);

\draw[very thick] (3a2) to (4a2);

\draw[very thick] (2a2) to (2a1);
\draw[very thick, red] (2a1) to  node[midway, below] {$\delta'$}(3a1);
\draw[very thick, red] (3a1) to  node[midway, below] {$\delta'$}(4a1);
\draw[very thick, red] (4a1) to  node[midway, left] {$\delta'$}(4a2);

\draw[very thick, red] (4a5) to  node[midway, below] {$\delta'$}(5a5);
\draw[very thick] (5a5) to (5a4);

\draw[very thick] (4a4) to (5a4);

\draw[very thick] (5a4) to (6a4);

\draw[very thick] (4a3) to (5a3);
\draw[very thick, red] (5a3) to  node[midway, left] {$\delta'$}(5a4);

\draw[very thick] (4a2) to (5a2);
\draw[very thick, red] (5a2) to  node[midway, below] {$\delta'$}(6a2);
\draw[very thick, red] (6a2) to  node[midway, left] {$\delta'$}(6a3);
\draw[very thick, red] (6a3) to  node[midway, left] {$\delta'$} (6a4);

\draw[->] (-3,2) to (-1,2);
   
\end{tikzpicture}
\caption{Example of simulation obtained form the grid embedding algorithm (see \cite{Ta89}). Black edges represent arbitrary interactions. For simplicity, the original spin system (right) has no fields.}\label{fig:gridEmbedd}
\end{figure*}

We now construct the 2d Ising system $S$ by first choosing the fields $J_S(\{v\})$ of $S$ to be $J_{\mr{S}}(\{w\})$ whenever the embedding algorithm maps $w$ to $v$, and zero for all those spins in $S$ that do not correspond to spins from $S_{\mr{deg}}$. Second, we choose pair interactions of $S$ such that for each edge $\{v,w\}$ in $S_{\mr{deg}}$, one edge of the lattice path corresponding to $\{v,w\}$ has interaction $J_{S_{\mr{deg}}}(\{v,w\})$ while all other edges of the lattice path have interaction $\delta'\cdot\pi_2$. Moreover, for all those edges in $G_{m',n'}$ that are not part of the lattice paths representing the edges of $S_{\mr{deg}}$ we pick the local interaction to be zero.
By \cref{thm: minor sim} we obtain a simulation 
\begin{equation}
    S \to S_{\mr{deg}}
\end{equation}
with cut-off $\delta'$. Note that the physical spin assignment of this simulation is precisely the map assigning grid coordinates to the spin from $S_{\mr{deg}}$ obtained from the algorithm.
The situation is illustrated in \cref{fig:gridEmbedd}.

\medskip
\paragraph*{Combining the four steps.}
Finally, combining the four constructed simulations yields the desired simulation $S \to T + K$.
Note that  $S \to S_{\mr{deg}}$ and $S_{\mr{deg}} \to S_{\mr{cross}}$ have cut-off $\delta'=\delta + \Shift_{\mr{cross}}$, while all other simulations have cut-off $\delta$. Moreover, $S_{\mr{cross}} \to S_{\mr{part}}$ has shift $\Shift_{\mr{cross}}$ while all other simulations have shift $0$. By \cref{thm:sim trans}, $S \to T+K$ has cut-off $\delta$.
As argued above, each step of the construction of $S \to T+K$ is polytime computable, and hence by \cref{thm:model sim trans} the construction itself is poyltime computable.
Further note that $\delta$ is merely a parameter in this construction, and changing the cut-off  amounts to changing this parameter, so the constructed emulation satisfies \cref{def:spin model sim} \ref{def:cut-off indep}. 
Finally, by construction $S \to T+K$ has identity encoding.
We conclude that the construction of $S$ and $S \to T+K$ satisfies \cref{def:spin model sim}, finishing the proof.
\end{proof}

\section{Computing Simulations by Linear Programs}\label{sec:lin prog}

Many of the problems studied in this work are of the form: 
Given a target spin system $T$, construct a simulation $S\to T$ with source $S$ with certain properties, often specified by requiring that $S\in \mc{M}$ for a spin model $\mc{M}$.
While our modular framework for simulations can be used to construct the global simulation $S\to T$ from simpler, local simulations of subsystems of $T$ (\cref{fig:sim sum}), we have so far constructed these simpler simulations by hand (see e.g.\ \cref{fig:Iff}, \cref{fig:CrossGadget} and \cref{fig:cross sim}).

Yet, this is not necessary: Simulations can be computed by linear programs, as we show in this section.  

We consider the problem:
\begin{quote}
Given a target system $T$, a cut-off $\delta$, and a finite set of local interactions $\mc{J}\coloneqq\{J_1, \ldots, J_n\}$ of spin type $q$, construct a simulation $S \to T$ where, up to isomorphism, the local interactions of $S$ are non-negative linear combinations of those from $\mc{J}$.
\end{quote}
We show how such spin systems $S$ and simulations $S\to T$ can be constructed via linear programs, 
first, for fixed interaction hypergraph of $S$ (\cref{ssec:lin prog fixed}) and then for an arbitrary one (\cref{ssec:lin prog arb}).  
We then illustrate the power of this approach by constructing an alternative crossing gadget for the 2d Ising model with fields (\cref{ssec:alternative crossing}). 

Note that to set up the respective linear program it is necessary to know the low-energy behavior of the target system $T$, specifically, to know all target configurations with energy below the desired cut-off. For a generic target system, computing these low energy target configurations is $\#$P-hard. This implies that computing simulations of a global target system is unfeasible. 
However, the the modular framework for simulations allows us to isolate local subsystems of the target $T$, compute simulations of these and combine the results to a global simulation (\cref{fig:sim sum}). In many cases these local subsystems are of constant size and hence the respective simulations can be computed efficiently by linear programs (see e.g.\ \cref{fig:Iff} and \cref{fig:cross sim}). 
In \cref{ssec:alternative crossing}, we explicitly show this by constructing a new, more efficient crossing gadget for the 2d Ising model with fields (as alternative to \cref{fig:CrossGadget}). 

For simplicity, in the following we consider simulations with $\degeneracy=\vert \enc\vert =1$ and where all interactions in $\mc{J}$ are invariant under permutation of spins. If this is not the case, each local interaction acting on $k$ spins (via relabellings) gives rise to up to $k!$ ways to be attached to a hyperedge of order $k$, increasing the number of variables associated to each hyperedge in the linear program. 
We write ``$S$ has local interactions from $\Cone(\mc{J})$" for ``up to isomorphisms, the local interactions of $S$ are non-negative linear combinations from $\mc{J}$".

\subsection{Fixed Interaction Hypergraph}\label{ssec:lin prog fixed}

We first consider the case where $S$ has fixed interaction hypergraph $G_S$. 

\begin{theorem}\label{thm:lin prog 1}
   Let $T,\mc{J},\delta$ and $G_S$ be as above. A spin system $S$ on $G_S$ with local interactions from $\Cone(\mc{J})$ and a simulation $S\to T$ with cut-off $\delta$ can be computed by solving a system of linear inequalities. 
\end{theorem}

Systems of linear equations can be solved by a linear program. 

\begin{proof}
We construct the system of linear inequalities.
For each natural number $a\geq 1$, denote by $\mc{J}_a$ the subset of local interactions from $\mc{J}$ which have arity $a$.
Denote by $I_a$ the set of indices of local interactions from $\mc{J}$ which are contained in $\mc{J}_a$, i.e.\ 
$\mc{J}_a = \{J_i \mid i \in I_a \}$.
For each $e\in E_S$, $\mc{J}_{\vert e \vert}$ are those local interactions that could potentially be used for the construction of $J_S(e)$.
Thus, a general $J_S(e)$ is a non-negative linear combination of local interactions from $\mc{J}_{\vert e \vert}$.  
For each $e \in E_T$ and for each $i \in I_{\vert e \vert}$ we introduce a variable $\lambda_{e,i}$ corresponding to the scalar coefficient of $J_i \in \mc{J}_{\vert e \vert}$ in $J_S(e)$. That is, we define
\begin{equation}
    J_S(e) \coloneqq \sum_{i\in I_{\vert e \vert}}\lambda_{e,i}\cdot J_i.
\end{equation}
This way, for each source configuration $\vec{s}$, $H_S(\vec{s})$ is a linear expression containing variables $\{\lambda_{e,i} \mid e \in E_S, i \in I_{\vert e \vert}\}$.
Adjusting the local interactions of $S$ such that $S \to T$ amounts to computing appropriate values for these variables.

Next we choose an arbitrary encoding, compatible decoding and physical spin assignment. 
Each low-energy target  configuration $\vec{t}$ determines by $\vec{s}\circ \phys = \enc \circ \vec{t}$ a corresponding sub-configuration of physical spins, and extending $\vec{s}$ to auxiliary spins yields the configuration of $S$ which simulates $\vec{t}$. This extension can be done arbitrarily. 
This defines a function $\simul$ mapping low-energy target configurations $\vec{t}$ to source configurations $\vec{s}$, which satisfy $\vec{s}\circ \phys = \enc \circ \vec{t}$.
Note that we have not yet imposed the energy condition (\cref{def:sys simulation} \ref{def:sim energy}) necessary for $\simul$ to define a simulation assignment. 
We denote the set of source configurations in the image of $\simul$ by $\simSet$.

Finally, imposing the energy conditions 
amounts to requiring that the variables $\{\lambda_{e,i}\}$ and $\Shift$ satisfy the inequalities: 
\begin{align}
\label{eq:lin prog 0} & \lambda_{e,i} \geq 0   && \forall  e \in E_S, i \in I_{\vert e \vert} \\
   \label{eq:lin prog 1} &H_S(\vec{s})-\Shift = H_T(\dec \circ \vec{s} \circ \phys)  &&\forall \vec{s} \in \simSet \\
   \label{eq:lin prog 2} & H_S(\vec{s})-\Shift \geq \delta  &&  \forall \vec{s} \notin \simSet.
\end{align} 
In total, the simulation $S \to T$ can be computed by solving this system of linear inequalities.\end{proof}

The first condition ensures that the local interactions of $S$ are non-negative functions\,---\,this condition should be omitted if one used a definition of a spin system with arbitrary interactions.

The objective function of this linear program can be chosen such that additional properties of the simulation $S \to T$ are imposed. 
For instance, one can minimize the absolute value of $\Shift$ to obtain a spin system $S$ whose ground-state energy is maximally close to the ground-state energy of the target spin system,  
or  minimize the sum of variables $\lambda_{e,i}$ to force a large number of them to be equal to zero and  obtain a sparser connectivity in $S$.

Other choices in the above construction, such as that of $\enc, \dec, \phys$ and $\simul$, are arbitrary.
They might, however, determine whether the resulting linear program has a solution. 
In practice, 
one may need to 
set up linear programs for several of these choices.

\subsection{Arbitrary Interaction Hypergraph}\label{ssec:lin prog arb}

We now consider the case where the interaction hypergraph of $S$ is not fixed. The task is thus: 
\begin{quote}
Given a set of local interactions $\mc{J}$ and a target system, construct a simulation $S\to T$ such that $S$ has interactions from $\Cone(\mc{J})$ but can have any interaction graph. 
\end{quote}
The construction $I^{\delta}_{\textsc{nor}}$ from \cref{fig:Iff} is of this type.

The lack of restrictions on the interaction hypergraph (and hence on the number of spins of $S$) gives rise to an infinite search space for $S$.
Inspired by \cite[Lemma 3.7]{Tr00} we show how this infinite search space can be restricted to a finite one, and thereby reduce the case of unfixed $G_S$ to that of fixed $G_S$.

We assume that the local interactions $\mc{J}$ are such that, for every $J\in \mc{J}$ acting on spins $e$ 
with $\vert e \vert \geq 2$ and every pair of distinct spins $s_i,s_j\in e$, there exists an interaction $J_{s_i,s_j}\in \mc{J}$ acting on spins $e'\coloneqq e\setminus\{s_j\}$ such that for all configurations $\vec{s}$ on $e$ that satisfy $\vec{s}(s_i)=\vec{s}(s_j)$,
\begin{equation}\label{eq:hered}
    J_{s_i,s_j}(\vec{s}\vert_{e'}) = J(\vec{s}).
\end{equation}

Restricting a local interaction with arity $k$ to configurations that agree on two fixed spins $s_i,s_j$ effectively defines a local interaction of arity $k-1$. 
The above condition then states that 
$\mc{J}$ is closed under such restrictions, i.e.\ if $J \in \mc{J}$ then all of its restrictions are contained in $\mc{J}$, too. 
Following \cite{Tr00}, 
we term sets of local interactions satisfying the above condition \emph{hereditary}.
This is for instance satisfied for Ising interactions of arbitrary arity. 

We will prove that for hereditary $\mc{J}$ the search space is finite, by showing that, given $T$ and $\delta$, we can compute the number of auxiliary spins needed to simulate $T$. This allows us to fix the interaction hypergraph of $S$ and resort to the case of \cref{ssec:lin prog fixed}.

For simplicity, we restrict to $\enc$ and $\phys$ trivial. We expect similar considerations to apply for the general case. 

We start by introducing terminology.
Assume the target system $T$ has $r$ spins $s_1, \ldots, s_r$ and $k$ low-energy configurations $\vec{t}_1, \ldots, \vec{t}_k$.
Since $\enc$ and $\phys$ are trivial, the physical spins of $S$ are given by $s_1, \ldots, s_r$.
Denote the auxiliary spins of $S$ by $s_{r+1}, \ldots, s_{r+a}$. Specifying $\simul$ amounts to giving the states of these $a$ auxiliary spins for each of the $k$ low-energy target configurations. The states of physical spins of $\simul$ are fixed by the requirement $\vec{s}\circ \phys = \enc \circ \vec{t}_i$.  
We denote the state of spin $s_j$ w.r.t.\ configuration $\simul(\vec{t_i})$ by $\mr{state}_i(s_j)$.
We define 
\begin{equation}
    \mr{state}(s_j) \coloneqq (\mr{state}_1(s_j), \ldots, \mr{state}_k(s_j)),
\end{equation}
and call $\mr{state}(s_j)$ the \emph{low-energy vector} of spin $s_j$.
Specifying $\simul$ amounts to giving the low-energy vectors of all auxiliary spins.

The relation between $\simul(\vec{t}_i)$ and $ \mr{state}(s_j)$ can be given as follows: 
$\simul$ is fully determined by a $k\times (r+a)$ matrix with rows being the individual values  $\simul(\vec{t}_i)$,
\begin{equation}\label{eq:simul mat}
\begin{pmatrix}
 - &\simul(\vec{t}_1) & - \\
  & \vdots &  \\
  - &\simul(\vec{t}_k) & -  
\end{pmatrix}
= \begin{pmatrix}
\mid &  & \mid \\
 \mr{state}(s_1) & \hdots & \ \mr{state}(s_{r+a}) \\
 \mid &  & \mid 
\end{pmatrix}.
\end{equation}
The columns of this matrix are the low-energy vectors.

Finally, two source spins $s_i,s_j$ are termed \emph{indistinguishable} if their low-energy vectors are equal, and distinguishable otherwise. The former correspond to equal columns in $\simul$ (cf.\ \eqref{eq:simul mat}). 
In the low-energy sector, indistinguishable spins have equal states.
The number of indistinguishable physical spins is denoted by $i_s$.
Since $\enc,\phys$ are trivial, $i_s$ is determined by the low-energy target configurations $\vec{t}_i$.
In the presence of symmetries of the local interactions $\mc{J}$, it might suffice to call spins indistinguishable if their low-energy vectors agree up to this symmetry.

\begin{lemma}[bounding the search space]\label{thm:lin prog 2}
    Let $T$ be a spin system and $\mc{J}$ a hereditary set of local interactions. If there exists a simulation $S \to T$ where $S$ has local interactions from $\Cone(\mc{J})$, then there exists a simulation $S'\to T$ with equal cut-off, where $S'$ has local interactions from $\Cone(\mc{J})$ and at most $q^k-i_s$ auxiliary spins (and hence $q^k-i_s+r$ total spins).
\end{lemma}

\begin{proof}
We follow \cite[Lemma 3.7]{Tr00}.
Assume there exists a simulation with more than $q^k-i_s$ auxiliary spins.
This has a simulation assignment $\simul$ which determines the
low-energy vectors for all spins of $S$.
Since each low-energy vector consists of $k$ numbers from $[q]$, there exist $q^k$ different low-energy vectors.
By assumption, $i_s$ of these are associated to physical spins.
Since $S$ has more than $q^k-i_s$ auxiliary spins, at least one low-energy vector appears at least twice in $\simul$. That is, there exists an auxiliary spin $s_j$ such that $\mr{state}(s_j)=\mr{state}(s_i)$ for some source spin $s_i$ different from $s_j$.
Note that $s_i$ might be an auxiliary or a physical spin.

We now construct a spin system $S'$ with spins $V_S\setminus \{s_j\}$. 
First, we explain how $E_{S'}$ is constructed.
Given any edge $e\in E_S$, we define an edge $r_j(e)$ by distinguishing the following three cases: 
\begin{enumerate}
    \item \label{it:hered 1} 
    If $e$ does not contain $s_j$ then $r_j(e)$ equals $e$; 
    \item \label{it:hered 2} If $e$ contains $s_j$ but does not contain $s_i$ then $r_j(e)$ is obtained from $e$ by relabeling $s_j$ with $s_i$;
    \item \label{it:hered 3} If $e$ contains both $s_j$ and $s_i$ then $r_j(e)$ equals $e\setminus\{s_j\}$.
\end{enumerate}
We now define 
\begin{equation}
    E_{S'} \coloneqq \{ r_j(e) \mid e \in E_S \}.
\end{equation}

Next, we construct $J_{S'}$.
By assumption, local interactions from $S$ are non-negative linear combinations from $\mc{J}$.
Given any local interaction $J_S(e) = \sum_i \lambda_i \cdot J_i$ from $S$ we first define a local interaction $R_j(J_S(e))$ by again distinguishing the same three cases 
\begin{enumerate}
    \item If $e$ does not contain $s_j$ then $R_j(J_S(e))$ equals $J_S(e)$;
    \item If $e$ contains $s_j$ but does not contain $s_i$ then $R_j(J_S(e))$ is obtained from $J_S(e)$ by relabeling $s_j$ to $s_i$ (as in \cref{def:iso spin sys});
    \item If $e$ contains both $s_i$ and $s_j$ then $R_j(J_S(e))$ equals $\sum_{i}\lambda_i\cdot (J_i)_{s_i,s_j}$.
\end{enumerate}
First, note that by construction $R_j(J_S(e))$ defines a local interaction on $r_j(e)$.
Second, since $\mc{J}$ is hereditary, $R_j(J_S(e))$  is a local interaction from $\Cone(\mc{J})$.

Now, given any $e'\in E_{S'}$, we define 
\begin{equation}
    J_{S'}(e') = \ \sum_{\mathclap{e \in r_j^{-1}(\{ e'\})}} \ R_j(J_S(e)), 
\end{equation}
which clearly defines a local interaction from $\Cone(\mc{J})$, acting on hyperedge $e'$.
Note that since $r_j$ is not necessarily injective, the preimage of $e'$, $r_j^{-1}(\{ e'\})$ might contain multiple edges from $E_S$.
Consider for instance
\begin{equation}
    r_j(\{a,b,s_j\}) = r_j(\{a,b,s_i,s_j \}) = \{a,b,s_i\}.
\end{equation}

Finally, using \cref{eq:hered}, for any $e'\in E_{S'}$ and any configuration $\vec{s}\in \mc{C}_S$ that agrees on $s_i$ and $s_j$ the following holds:
\begin{equation}
    J_{S'}(e')(\vec{s}\vert _{e'}) = \
    \sum_{\mathclap{e \in r_j^{-1}(\{ e'\})}} \  J_S(e)(\vec{s}\vert_e).
\end{equation}
This implies for such configurations
\begin{equation}
    H_{S'}(\vec{s}\vert_{V_{S'}}) = H_S(\vec{s}).
\end{equation}

Since all configurations from $\simSet$ agree on $s_i$ and $s_j$, defining 
\begin{equation}
    \simul'(\vec{t}) = \simul(\vec{t})\vert_{V_{S'}}
\end{equation}
implies that $H_{S'}$ satisfies   \cref{eq:lin prog 1}.
Note that this definition of $\simul'$ precisely corresponds to removing the column $\mr{state}(s_j)$ from $\simul$.

Since for each $\vec{s'}\notin\simSet'$ there exists a configuration $\vec{s}\notin\simSet$ that agrees on $s_i, s_j$ and restricts to $\vec{s'}$, $H_{S'}$ also satisfies \cref{eq:lin prog 2}.
Finally, since
the construction of $J_{S'}$ does not change the coefficients of linear combinations of local interactions and 
$S$ satisfies \cref{eq:lin prog 0}, so does $S'$.

By \cref{thm:lin prog 1} we have constructed a simulation $S'\to T$ such that $S'$ has one spin less than $S$.
The simulation assignment of this simulation in its matrix form is obtained by removing the column containing the low-energy vector $\mr{state}(s_j)$ from $\simul$.
As long as the resulting spin system $S'$ still has at least $q^k-i_s+r$ spins we are guaranteed the existence of another auxiliary spins $s_{j'}$ indistinguishable from some other spin $s_{i'}$.
Thus, the claim follows by applying  this construction iteratively, until the number of auxiliary spins equals $q^k-i_s$.
\end{proof}

We now use \cref{thm:lin prog 2} to prove that also simulations with arbitrary interaction hypergraph $G_S$ can be computed by linear programs.
Denote by $n_a$ the number of local interactions from $\mc{J}$ of arity $a$ and by $a_{\mr{max}}$ the maximum arity of local interactions from $\mc{J}$.

\begin{theorem}[simulation as a linear program]\label{thm:lin prog 3}
   Let $T$ be a target system with $r$ spins, $\delta>0$ and  $\mc{J}$ a hereditary set of local interactions. If there exists a spin system $S$ with local interactions from $\Cone(\mc{J})$ that simulates $T$ with cut-off $\delta$ and with trivial $\enc, \phys$, then such a spin system can be computed by a linear program with
   \begin{equation}\label{eq:num var}
       1 + \sum_{a=1}^{a_{\max}} n_a \cdot \binom{q^k-i_s+r}{a} 
   \end{equation}
   variables.
\end{theorem}

\begin{proof}
By \cref{thm:lin prog 2}, we can restrict to source spin systems $S$ with no more than $q^k-i_s+r$ spins.
By \cref{thm:lin prog 1}, any choice of interaction hypergraph $G_{S}$ and  simulation assignment $\simul$ 
determines a system of linear inequalities whose solutions are simulations $S \to T$.

In the following we construct the most general interaction hypergraph and simulation assignment compatible with the assumptions.
Denoting the corresponding system of linear inequalities by $L_{\mr{gen}}$,
any solution of $L_{\mr{gen}}$ yields a simulation which satisfies the requriements.
We then argue that any other interaction hypergraph and simulation assignment amounts to adding equations to $L_{\mr{gen}}$.
Thus, if there exists a spin system $S$ that simulates $T$ (with the required properties), then by \cref{thm:lin prog 1} the linear system containing $L_{\mr{gen}}$ plus the corresponding additional equations corresponding, has a solution. Hence, $L_{\mr{gen}}$ must have a solution, too. 
The claim then follows by counting the number of variables of $L_{\mr{gen}}$.

Let us construct $L_{\mr{gen}}$.
Given the target system $T$ and cut-off $\delta$, we first compute all $k$ low-energy target configurations $\vec{t}_i$.
Since $\enc$ and $\phys$ are trivial, from the low-energy target configurations we can read off the low-energy vectors of  physical spins, $\mr{state}(s_1), \ldots, \mr{state}(s_r)$.
In particular, we can read off the number of different such vectors, $i_s$.
By \cref{thm:lin prog 2}, we thus can compute the number of spins necessary for $S'$, $q^k-i_s+r$.

We construct the simulation assignment $\simul'$ by assigning 
each of the $q^k-i_s$ low-energy vectors not corresponding to a physical spin to one auxiliary spin. In other words, $\simul'$, in its matrix form, contains each of the $q^k$ possible low-energy vectors as a column-vector. 
We define $G_{S'}$ to be the hypergraph with vertices $s_1, \ldots, s_{q^k-i_s+r}$ and all possible hyperedges of order $\leq a_{\mr{max}}$. Given this, we set up the linear program described in \cref{thm:lin prog 1}.

Next, we count the number of variables of $L_{\mr{gen}}$.
$G_{S'}$ contains 
\begin{equation}
\binom{q^k-i_s+r}{a}
\end{equation}
hyperedges of order $a$, and $\mc{J}$ contains $n_a$ local interactions of arity $a$.
Following \cref{thm:lin prog 1}, each hyperedge $e$ of order $a$ thus yields $n_a$ variables $\lambda_{e,i}$. 
Including the final variable, $\Shift$, we obtain the total number of variables of \eqref{eq:num var}.
Solving this linear program yields a spin system $S'$ and simulation $S'\to T$ with the desired properties.

To finish the proof, we argue that any other choice of interaction hypergraph $G_S$ and simulation assignment $\simul$ amounts to adding equations to the linear program $L_{\mr{gen}}$ and thus, assuming that 
there exists some spin system $S$ such that $S \to T$ with the desired properties then also $L_{\mr{gen}}$ must have a solution.

First, any choice of interaction hypergraph with no more than $q^k-i_s+r$ spins can be obtained from $G_{S'}$ by removing hyperedges. This can be achieved by adding equations to the linear program that set the corresponding variables to zero.
Second, by \cref{thm:lin prog 2}, any choice of simulation assignment $\simul$, w.l.o.g.\ has a minimum number of indistinguishable spins.
Since $\simul'$ contains all possible low-energy vectors, 
$\simul$ can only differ from $\simul'$ by missing some low-energy vectors.
Since low-energy vectors of physical spins are determined by the low-energy target configurations, the missing low-energy vectors must correspond to auxiliary spins. 
Imposing this constraint amounts to adding equations to the linear program that set the variables corresponding to hyperedges containing these auxiliary spins to zero.  
\end{proof}

Let us see an example of computing simulations by linear programs, where we solved the system of linear inequalities with the tools of \cite{ortools}.

\begin{example}[revisiting Ising]
Let us revisit the simulation provided in \cref{ex:simulation}.
Since $\mc{J}$ contains all Ising fields and Ising pair interactions, it is hereditary.
Due to the symmetry of Ising interactions, we consider spins indistinguishable if their low-energy vectors are equal up to flipping all states, i.e.\ exchanging states, $1 \to 2, 2 \to 1$. 
Fixing the encoding, decoding and physical spin assignment as described in \cref{ex:simulation} leads to $4$ different low-energy vectors which are associated to physical spins (see \cref{tab:ex simul}). 
Accounting for the Ising symmetry of local interactions, since $T$ has $4$ low-energy configurations, there exists a total of $8$ different low-energy vectors, which implies that the simulation requires $4$ auxiliary spins. 

The interaction graph of $S$ and simulation assignment are chosen as described in \cref{thm:lin prog 3}, i.e.\
$G_S$ is  the  complete graph on 10 vertices (including all single vertex hyperedges to allow for fields) and 
$\simul$ is provided on the right hand side of \cref{tab:ex simul}. 
The solution of the resulting linear program is provided in \cref{tab: local int ex}, where we used \cref{eq:Ising field sum} and \cref{eq:Ising int sum} to rewrite linear combinations of Ising interactions in terms of a single Ising interaction.
The shift of the constructed simulation is $3$.
Even though $G_S$ is chosen to be the complete graph on $10$ vertices, after solving the linear program several variables are set to zero. We hence removed the corresponding edges from \cref{fig:ex sim 1}.
\end{example}

\subsection{Alternative Crossing Gadget}\label{ssec:alternative crossing}

Let us put these ideas in practice: We shall provide a new crossing gadget for the Ising model with fields, with the same functionality as $I^{\delta}_{\times}$ (cf.\ \cref{fig:CrossGadget}), but notably simpler. 

We construct the alternative crossing gadget for cut-off $\delta=1$; arbitrary cut-offs can be obtained by scaling local interactions and the shift according to \cref{thm:sim scale}.  

Consider the spin system $T$ acting on four two-level spins $1, 2,3, 4$, which in the zero energy ground state imposes
\begin{equation}\label{eq: gs T}
    \vec{t}(1) = \vec{t}(4) \ \ \text{and} \ \ \vec{t}(2) = \vec{t}(3),
\end{equation}
while all other configurations have energy $\geq \delta=1$.
Any spin system that simulates $T$ with identity encoding and cut-off $\delta = 1$ also satisfies \cref{lem:crossGadget}. Given that it also satisfies certain planarity requirements, it can be used as an alternative crossing gadget. 
Thus, computing an alternative crossing gadget amounts to constructing a simulation of $T$.

We apply \cref{thm:lin prog 1} to construct an Ising system with fields $S$ that simulates $T$.
We choose the interaction graph $G_S$ such that applying $S$ according to \cref{fig:cross sim} leads to planar interaction graphs.
We pick trivial encoding and physical spin assignment.
\cref{eq: gs T} gives rise to $4$ low-energy configurations of $T$ which contain $2$ of the (up to Ising symmetries) $8$ possible low-energy vectors. We thus take $S$ to have $10$ spins, $4$ physical and $6$ auxiliary. Moreover, we define $\simul$ as described in \cref{thm:lin prog 3}, i.e.\ such that it contains all of the $8$ indistinguishable low-energy vectors:
\begin{equation}
\simul = 
    \begin{pmatrix}
        2 & 2 & 2 & 2 & 2 & 1 & 2 & 2 & 2 & 2 \\ 
        2 & 1 & 1 & 2 & 2 & 2 & 2 & 1 & 1 & 2 \\
        1 & 2 & 2 & 1 & 2 & 2 & 2 & 2 & 1 & 1\\
        1 & 1 & 1 & 1 & 2 & 2 & 1 & 2 & 2 & 2
    \end{pmatrix}.
\end{equation}

\begin{figure}[t]
    \centering
\begin{tikzpicture}[circuit logic US, cgnodeS/.style = {draw=black!80, fill=blue!10, thick, circle, minimum size= {width("$v_{4}$")+4pt}, inner sep=1pt}] 

    \node[cgnodeS] (1) at (-1, 4) {$1$}; 
    \node[cgnodeS] (2) at (1, 4) {$2$}; 
    \node[cgnodeS] (3) at (-1, 0) {$3$}; 
    \node[cgnodeS] (4) at (1, 0) {$4$}; 
    \node[cgnodeS] (5) at (-2, 2) {$5$}; 
    \node[cgnodeS] (6) at (-1, 2) {$6$}; 
    \node[cgnodeS] (7) at (1, 2) {$7$}; 
    \node[cgnodeS] (8) at (2, 2) {$8$}; 
    \node[cgnodeS] (9) at (0, 3) {$9$}; 
    \node[cgnodeS] (10) at (0, 1) {$10$};

    \draw[thick, black] (1) to (2);
    \draw[thick, black] (1) to [bend right = 85] (3);
    
    \draw[thick, black] (1) to (6);
    \draw[thick, black] (1) to (9);

    \draw[thick, black] (2) to (7);
    \draw[thick, black] (2) to (8);
    \draw[thick, black] (2) to (9);
    \draw[thick, black] (3) to (4);
    
    \draw[thick, black] (3) to (6);
    \draw[thick, black] (3) to (10);
    \draw[thick, black] (4) to (7);
    \draw[thick, black] (4) to (8);
    \draw[thick, black] (4) to (10);
   
    \draw[thick, black] (6) to (7);
    \draw[thick, black] (6) to (9);
    \draw[thick, black] (6) to (10);
    
    \draw[thick, black] (7) to (9);
    \draw[thick, black] (7) to (10);
    \end{tikzpicture}
    \caption{The alternative crossing gadget, namely $S$ in \cref{ssec:alternative crossing}. Its interactions are listed in \cref{tab: local int S}.}
    \label{fig:int lin S}
\end{figure}
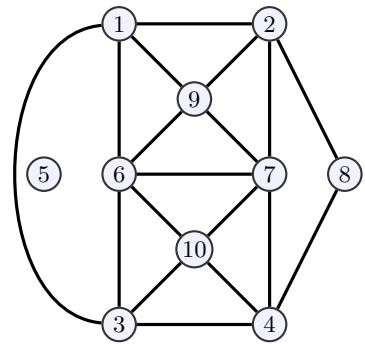

\begin{table}[t]
\centering
\begin{tabular}{ |p{1.8cm}|p{1.8cm}|p{1.8cm}|p{1.8cm}|  }
 \hline
 \multicolumn{2}{|c}{Fields} & \multicolumn{2}{|c|}{Pair interactions} \\
 \hline
 Spin & Field & Edge & Interaction\\
 \hline
$1$   & $2\cdot\bar{\pi}_1$   & $\{1,2\}$   & $2 \cdot \pi_2$   \\
$2$   &  $4\cdot \pi_1$  &$\{1,3\}$    & $2\cdot \pi_2 $  \\
$3$   &  $4\cdot \bar{\pi}_1$   &$\{1,6\}$    & $6\cdot \pi_2$  \\
$4$   &  $2\cdot \pi_1$  &$\{1,9\}$    & $ 2\cdot \pi_2$  \\
$5$   &  $2\cdot \bar{\pi}_1$  &$\{2,7\}$    &  $6\cdot \bar{\pi}_2 $ \\
$6$   &  $8\cdot \bar{\pi}_1$  &$\{2,8\}$    &  $ 2 \cdot \bar{\pi}_2$ \\
$7$   &   $8\cdot \bar{\pi}_1$ &$\{2,9\}$    &  $ 4 \cdot \bar{\pi}_2$ \\
$8$   &  $2\cdot \bar{\pi}_1$  &$\{3,4\}$    &  $2\cdot \bar{\pi}_2 $  \\
$9$   &  $6 \cdot \bar{\pi}_1$  &$\{3,6\}$ & $ 4\cdot \pi_2$     \\
$10$  &   $6\cdot \bar{\pi}_1$  &$\{3,10\}$ & $ 4\cdot \pi_2$   \\
      &    &$\{4,7\}$         & $ 2 \cdot \bar{\pi}_2$   \\
      &    &$\{4,8\}$          & $ 2 \cdot \pi_2$    \\
      &    &$\{4,10\}$             &  $ 4 \cdot \bar{\pi}_2$  \\
      &    &$\{6,7\}$             &  $ 4\cdot \pi_2$  \\
      &    &$\{6,9\}$             &   $6\cdot \pi_2 $ \\
      &    &$\{6,10\}$             &    $2 \cdot \pi_2 $ \\
      &    &$\{7,9\}$             &   $8 \cdot \pi_2$ \\
      &    &$\{7,10\}$             &   $ 2 \cdot \pi_2$ \\
\hline
\end{tabular}
\caption{Local interactions of the alternative crossing gadget of \cref{fig:int lin S} for cut-off $\delta=1$.  
For a simulation with arbitrary cut-off $\delta$, every field and interaction must be multiplied by $\delta$.}\label{tab: local int S}
\end{table}

The interaction graph of $S$ is shown in \cref{fig:int lin S}; its local interactions, obtained from solving the corresponding linear program (by \cite{ortools}) are listed in \cref{tab: local int S}, where we again used \cref{eq:Ising field sum} and \cref{eq:Ising int sum}. The shift of the resulting simulation is $\Shift=36$.
Again, we have removed those edges which after solving the linear program have zero interactions. 
Further note that spin $5$, in the low-energy sector, merely contributes a constant shift.
We could thus further simplify $S$ by removing spin $5$ and removing its low-energy vector $(2,2,2,2)$ from $\simul$.

\section{Conclusions and Outlook}\label{sec:outlook}

Let us now summarize our main findings (\cref{ssec:concl}), and discuss applications and extensions of the work (\cref{ssec:outlook}). 

\subsection{Conclusions} \label{ssec:concl}

Spin system simulations (\cref{def:sys simulation}) capture the idea of encoding the low-energy behavior of a target spin system $T$ into the low-energy behavior of a source spin system $S$. Below the cut-off, simulations preserve spectra (\cref{lem:sim spectrum}) with constant degeneracy (\cref{lem:sim deg}) as well as ground states (\cref{lem:ground state}). Additionally, simulations approximately preserve thermodynamic quantities, such as the partition function (\cref{lem:sim part fun sys}) and the Boltzmann distribution (\cref{thm:boltzmann sys}).
Crucially, the approximation error scales as $\mathcal{O}(e^{-\Del})$, so that increasing the cut-off allows for arbitrary precision. 
In addition, common transformations for spin systems (or graphs, in general) are  simulations, such as generalized symmetries (\cref{lem:iso sim}, \cref{lem:ground state sym sim}), modifications of spin type (\cref{lem:bin sim}) and graph minor relations (\cref{thm: minor sim}).
Finally, simulations are modular, i.e.\ can be composed (\cref{thm:sim trans}), scaled (\cref{thm:sim scale}) and added (\cref{thm:sim sum}), so that complicated simulations can be constructed from simple ones. 

Spin models (\cref{def:spin model}) are sets of spin systems and spin model emulations (\cref{def:spin model sim}) are efficiently computable simulations between spin models.
Most results for simulations extend to emulations; for example, they can also modify the spin type (\cref{lem:bin sim model}). Most importantly, they are modular, i.e.\ can be composed (\cref{thm:model sim trans}), scaled (\cref{thm:model sim scale}) and added (\cref{thm:model sim sum}). We define three properties of spin models: functional completeness (\cref{def:f.c.}), closure (\cref{def:closed} and \cref{def:loc closed}) and scalability (\cref{def:scalable}) and characterize universality (\cref{def:universality}) in terms of them (\cref{thm:main}). Since the characterization is constructive, it provides a step-by-step recipe for the efficient construction of a simulation with arbitrary target spin system $T$ by an arbitrary universal spin model $\mc{M}$. 
The recipe consists of three parts: 
(i) we decompose $T$ into a linear combination of flag systems (\cref{lem:basis decomp}), 
(ii) we use Boolean algebra to simulate arbitrary flag systems with linear combinations of flag systems of order $2$ (\cref{lem:f.c.}), 
and (iii) leverage functional completeness, closure and scalability to simulate linear combinations of order $2$ flag systems with spin systems from $\mc{M}$ (\cref{thm:main}).
It heavily relies on the modularity of emulations.

As for the consequences of universality, we show that emulations induce polytime computable reductions for the problems of computing ground states, approximating partition functions, and approximate sampling from Boltzmann distributions (\cref{thm:sim GSE}, \cref{thm:sim FGSE}, \cref{thm:sim part fun}, \cref{thm:approx sampl}). 
This means that emulations can be used to construct efficient solutions for these problems, but also that universal spin models are maximally hard for these problems (\cref{cor: univ GSE}, \cref{cor:univ FGSE}, \cref{cor:part univ}, \cref{cor:uni boltz}). 

We then show that the 2d Ising model with fields is universal, because it is scalable, functional complete (\cref{thm:2D f.c.}), locally closed and thus closed (\cref{thm:locally closed}). 
The main challenge to prove closure are non-planar target systems (\cref{fig:lattice minor}), for which we construct a crossing gadget (\cref{lem:crossGadget}).
This again relies on the modularity of simulations/emulations. 

Finally, simulations can be computed by linear programs (\cref{thm:lin prog 3}).  
This allows us to construct an alternative, more frugal crossing gadget for the Ising model with fields (\cref{ssec:alternative crossing}). 

\subsection{Outlook}\label{ssec:outlook}

This framework for simulations/emulations opens several questions, which can be divided into those concerning applications and those concerning extensions.

\subsubsection{Applications}

\paragraph*{Quantum annealing.}
As mentioned in \cref{ssec:gse}, \cref{thm:sim GSE} and \cref{thm:sim FGSE} could be applied to quantum optimization. 
Quantum annealing protocols essentially correspond to emulations $\mc{M} \to \mc{N}$, 
where $\mc{N}$ is chosen to correspond to the optimization problem to be solved, $\text{\sc Gse}_{\mc{N}}$ (e.g.\ by following \cite{Lu14b}), and $\mc{M}$ is chosen so that one can compute $\text{\sc Gse}_{\mc{M}}$ via quantum annealing. 

Many quantum annealing protocols suffer from exponential, adiabatic time scales \cite{Ng23,Eb22}. However, it seems to be unclear what precise properties of the emulation lead to such exponential scales. For example, it seems unclear if for fixed target model $\mc{N}$ the exponential time scales can be avoided by a clever choice of other data of the emulation, or, conversely, if certain properties of $\mc{N}$ unavoidably lead to exponential timescales. 
In the latter case, such properties are probably preserved by emulations. 
It would be interesting to use our framework to shed light on these properties. 
For example, if lattices of large connectivity can be efficiently simulated by 2d lattices
and `exotic' many-body interactions can be efficiently simulated by Ising pair interactions and fields the relevant properties 
cannot solely depend on the connectivity or the type of local interactions since these are clearly not preserved by emulations.
Emulations may provide the correct notion of transformation of spin models to extract the relevant properties.

\medskip
\paragraph*{Sampling algorithms.}
As explained in \cref{ssec:sample}, \cref{thm:approx sampl} may allow to 
construct algorithms that sample from Boltzmann distributions. 
Such algorithms are often obtained from Markov chains, which often suffer from exponential mixing times. 
Yet, what properties of spin models lead to exponential mixing times seems to be poorly understood. 
\cref{thm:approx sampl} seems to indicate that such properties are preserved by emulation; the latter could thus be used to isolate these properties.

\medskip
\paragraph*{Weaker types of universality.}
Not surprisingly, universal spin models are computationally maximally hard (\cref{sec:consequences}), so we cannot expect to solve a computational problem of $\mc{N}$ by simulating it with a universal model.
We might however consider the converse situation: Given a spin model $\mc{M}$ and an algorithm to solve a computational problem of $\mc{M}$, 
we can extend this algorithm to all spin models $\mc{N} \in \mr{Reach}(\mc{M})$,  i.e.\ all those spin models that can be emulated by $\mc{M}$.

While $\mr{Reach}(\mc{M})$ may be hard to characterize for a generic $\mc{M}$, we may be able to lower bound it by considering weaker types of universality, as follows. 
For a set of spin systems $\mc{S}_{\alpha}\subseteq \mc{S}_{\mr{all}}$, call $\mc{M}$ $\alpha$-universal if $\mc{M}\to \mc{S}_{\alpha}$, and characterize it similarly to \cref{thm:main}, i.e.\ by properties weaker than functional completeness, closure and scalability. Deciding whether $\mc{S}_{\alpha} \subseteq \mr{Reach}(\mc{M})$ amounts to checking if $\mc{M}$ is $\alpha$-universal which then could be done by this charcaterization. 
If the characterization of $\alpha$-universality is constructive, this yields algorithms for all spin models $\mc{N}$ which are subsets of $\mc{S}_{\alpha}$.

\subsubsection{Extensions}

\paragraph*{Emulation and universality classes.} 
Spin models are often used to describe phase transitions. 
In many cases, the nature of their thermodynamic quantities at the phase transition does not depend on details of their definition, but can be grouped in so-called universality classes \cite{Ni11}.

Several of the thermodynamic quantities which might determine the universality class of a spin model can be obtained from its partition functions. Since emulation approximately preserves partition functions,
it would be interesting to study if it also preserves universality classes. 
Given an emulation $\mc{M} \to \mc{N}$, denote by $\mc{M}_{\mc{N}}\subseteq \mc{M}$ the submodel of $\mc{M}$ consisting of those systems which are used in the emulation of $\mc{N}$.  
If $\mc{N}$ belongs to universality class $\iota$, does $\mc{M}_{\mc{N}}$ also belong to $\iota$? 

If this is the case, it would follow that a universal spin model contains submodels from all universality classes. We would thereby obtain a new theory of universality classes in terms of submodels of $\mc{M}$.
In addition, studying if emulation preserves universality classes may lead to a refinement of the very notion of emulation. 

A related question is how our notion of emulation extends to the thermodynamic limit. 
Particularly, one could study the approximation errors of the partition function and Boltzmann distribution in this limit. 

\medskip
\paragraph*{Alternative gadgets.}
In the proof of closure of the 2d Ising model with fields, the crossing gadget is essential to overcome the constraint of  planar interaction graphs.
What gadgets are required to overcome other constraints on the interaction hypergraphs of a spin model?  

Consider a spin model $\mc{M}$ defined by a set of local interactions $\mc{J}$ and a set of allowed interaction graphs $\mc{G}$. 
Assume that $\mc{G}$ is specified in terms of a graph property $P$, namely $\mc{G}$ contains all graphs that satisfy $P$, e.g.\ all planar graphs or all loop-free graphs.
Then proving closure of $\mc{M}$ amounts to constructing a set of gadgets $S_1, \ldots, S_n$ from $\mc{M}$ that suffice to overcome the constraint $P$.
If $\mc{J}$ stands for Ising interactions and $P$ for planarity, the relevant gadget is the crossing gadget from \cref{sec:2d Ising}. For other graph properties such as loop-free interaction graphs, it is less clear how the corresponding gadgets can be characterized. 

We expect that the relevant gadgets can be characterized as follows: Given the graph property $P$, one can derive a set of graphs $G_1,\ldots, G_n$ that encode how a generic graph fails to satisfy $P$. The gadgets are  precisely simulations of all spin systems (with interactions from $\mc{J}$) which can be defined on these graphs.
For example, for planarity, $G_1$ is the graph containing two (crossing) edges, and the crossing gadget from \cref{sec:2d Ising} suffices for closure since it can be used to simulate all Ising system defined on $G_1$ (\cref{lem:crossGadget}).
We also expect that $G_1, \ldots, G_n$ can be extracted from the forbidden minor characterization (by the Robertson--Seymour theorem) of $P$. 
This would allow to characterize closure in terms of the existence of gadgets for the appropriate graph property, and thereby improve the characterization of universality.

\medskip
\paragraph*{Widening the scope.} 
It should be possible to reproduce the characterization of universality of \cref{thm:main} in contexts other  than spin systems/models, possibly by first translating it to a more abstract language \cite{Go23}. 
These could include continuous spin variables, quantum spin systems (cf.\ \cite{Cu17}) and spin systems with couplings drawn from a probability distribution, as in spin glasses. 
This would allow to study emulations in other contexts as well as their potential modularity and universality. Ultimately, it may shed light on our understanding of the (surprising) nature and scope of universality.

\section{Statements and Declarations}
We thank Sebastian Stengele for discussions on the early stages of this work. 
We also thank William Slofstra for bringing \cite{Sh78} to our attention, 
and Roger Melko for pointing out high to low temperature dualities.   
We acknowledge support of the Austrian Science Fund (FWF) via the START Prize (project Y1261-N). For open access purposes, the authors have applied a CC BY public copyright license to any accepted manuscript version arising from this submission. 
All authors have no conflicts of interest. No new data were created or analysed in this study.

\appendix

\section{Proofs for Modularity of Simulation}\label{sec:modular}

In \cref{ssec:spin sys sim prop} we stated that simulations can be composed (\cref{thm:sim trans}) 
and added (\cref{thm:sim trans}). Let us prove these claims, respectively, in \cref{ssec:composition} and \cref{ssec:addition sim}. 

\subsection{Simulations Can Be Composed}\label{ssec:composition}

To prove \cref{thm:sim trans} we show that the definition of $\mc{g}\circ \mc{f}$ satisfies the five conditions of \cref{def:sys simulation}.

\medskip
\paragraph*{\ref{def:sim physical spins}.\ Disjoint physical spins.}   
        We need to prove that
            \begin{equation}
            \begin{split}
               & \phys_{\mc{g}\circ\mc{f}}^{(m_2,m_1)}(v) = \phys_{\mc{g}\circ\mc{f}}^{(n_2,n_1)}(v') \\
                &\Rightarrow m_i = n_i \ \text{and}\ v=v'.
                \end{split}
            \end{equation}
        Inserting the definition of $P_{\mc{g}\circ\mc{f}}$ the left hand side of this implication becomes
        \begin{equation}
            \phys_{\mc{g}}^{(m_2)} \circ \phys_{\mc{f}}^{(m_1)} (v) = \phys_{\mc{g}}^{(n_2)} \circ \phys_{\mc{f}}^{(n_1)}(v').
        \end{equation}
        Since $\mc{g}$ has disjoint physical spins, this implies that $m_2=n_2$ and
        \begin{equation}
            \phys_{\mc{f}}^{(m_1)} (v)=\phys_{\mc{f}}^{(n_1)}(v').
        \end{equation}
        Since also $\mc{f}$ has disjoint physical spins, we conclude that $m_1=n_1$ and $v=v'$, which proves the claim.

\medskip
\paragraph*{\ref{def:sim enc-dec}.\ Decode-encode compatibility.} 
        We need to prove that 
            \begin{equation}
                \dec_{\mc{g}\circ\mc{f}} \circ (\enc_{\mc{g}\circ\mc{f}})_{i,j} = \mr{id}.
            \end{equation}
        Inserting the definition of $\dec_{\mc{g}\circ\mc{f}}$ and $\enc_{\mc{g}\circ\mc{f}}$ yields 
            \begin{equation}
            \begin{split}
                \dec_{\mc{g}\circ\mc{f}} \circ (\enc_{\mc{g}\circ\mc{f}})_{i,j} 
                &= \dec_{\mc{f}} \bigl(\dec_{\mc{g}}\circ (\enc_{\mc{g}})_i\circ (\enc_{\mc{f}})_j^{(1)},\\
                &\hphantom{=} \ldots , \dec_{\mc{g}}\circ (\enc_{\mc{g}})_i\circ (\enc_{\mc{f}})_j^{(k_1)}\bigr).
            \end{split}
            \end{equation}
        Using decode-encode compatibility of $\mc{g}$ the right hand side becomes
        \begin{equation}
        \dec_{\mc{f}}\bigl((\enc_{\mc{f}})_j^{(1)},\ldots , (\enc_{\mc{f}})_j^{(k_1)}\bigr) = \dec_{\mc{f}}\circ (\enc_{\mc{f}})_j,
            \end{equation}
            which by decode-encode compatibility of $\mc{f}$ equals $\mr{id}$. 

\medskip            
\paragraph*{\ref{def:sim disjoint enc}.\ Disjoint encodings.} 
        We need to prove that 
            \begin{equation}
                (\enc_{\mc{g}\circ\mc{f}})_{i,j}(s) =(\enc_{\mc{g}\circ\mc{f}})_{k,l}(s) \Rightarrow (i,j)=(k,l) .
            \end{equation}
        By definition of $\enc_{\mc{g}\circ\mc{f}}$ the left hand side 
        implies that  for all $n \in [k_1]$
        \begin{equation}\label{eq:disjoint encoding 2}
                (\enc_{\mc{g}})_i \circ (\enc_{\mc{f}})_j^{(n)}(s) = (\enc_{\mc{g}})_k \circ (\enc_{\mc{f}})_l^{(n)}(s).
            \end{equation}
        Applying $\dec_{\mc{g}}$ to both sides implies
        \begin{equation}
               (\enc_{\mc{f}})_j(s) = (\enc_{\mc{f}})_l(s).
        \end{equation} 
        Since $\mc{f}$ has disjoint encodings this implies $j=l$.
        Reinserting this into \cref{eq:disjoint encoding 2} and using that $\mc{g}$ has disjoint encodings finally implies $i=k$.
        
\medskip
\paragraph*{\ref{def:sim deg}.\ Constant degeneracy.} 
        We have to prove if  $H_T(\vec{t})< \Del_{\mc{g}\circ\mc{f}}$, then for all $i,j$, $\vert (\simul_{\mc{g}\circ\mc{f}})_{i,j}(\vec{t}) \vert = \degeneracy_{\mc{g}\circ\mc{f}}$.
        We split the proof in two parts.
        First, we prove that 
         \begin{equation}\label{eq:characterization sim3}
                (\simul_{\mc{g}\circ\mc{f}})_{i,j}(\Vec{t}) = \{ \Vec{s} \in (\simul_{\mc{g}})_i(\Vec{r}) \mid  \Vec{r} \in (\simul_{\mc{f}})_j(\Vec{t}) \}
            \end{equation}
        and then, prove that for low-energy configurations $\vert (\simul_{\mc{g}\circ\mc{f}})_{i,j}(\vec{t}) \vert = \degeneracy_{\mc{g}\circ\mc{f}}$.
        
        Starting with the first part, let $\vec{s}\in (\simul_{\mc{g}})_i(\Vec{r})$ with $\vec{r} \in (\simul_{\mc{f}})_j(\Vec{t})$. Then by definition of $\phys_{\mc{g}\circ\mc{f}}^{(m,n)}$
        we have
        \begin{equation}
        \begin{split}
            \vec{s} \circ \phys_{\mc{g}\circ\mc{f}}^{(m,n)} =
                \Vec{s} \circ \phys_{\mc{g}}^{(m)} \circ \phys_{\mc{f}}^{(n)} = (\enc_{\mc{g}})_i^{(m)} \circ \Vec{r} \circ \phys_{\mc{f}}^{(n)},
                \end{split}
        \end{equation}
        where the last equality holds since 
        $\vec{s} \in (\simul_{\mc{g}})_i(\Vec{r})$. 
        Similarly, since  $\vec{r}\in (\simul_{\mc{f}})_j(\Vec{t})$ we get
        \begin{equation}
                (\enc_{\mc{g}})_i^{(m)} \circ \Vec{r} \circ \phys_1^{(n)} = (\enc_{\mc{g}})_i^{(m)} \circ  (\enc_{\mc{f}})_j^{(n)} \circ \Vec{t},
        \end{equation}
        which in total shows that 
        \begin{equation}
            \Vec{s}\circ \phys_{\mc{g}\circ\mc{f}} = (\enc_{\mc{g}\circ\mc{f}})_{i,j} \circ \vec{t}.    
        \end{equation}
        Moreover, since $\vec{s}\in (\simul_{\mc{g}})_i(\Vec{r})$ we have 
        \begin{equation}
            H_S(\Vec{s}) - \Shift_{\mc{g}} =
            H_R(\Vec{r}) < \Del_{\mc{g}} 
        \end{equation}
        and since $\vec{r}\in (\simul_{\mc{f}})_j(\Vec{t})$  
        \begin{equation}
            H_R(\Vec{r}) - \Shift_{\mc{f}} = H_T(\vec{t})  < \Del_{\mc{f}}
        \end{equation}
        and hence 
        \begin{equation}
        \begin{split}
            H_S(\Vec{s}) - \Shift_{\mc{g}\circ\mc{f}}=
            H_T(\vec{t}) < \min(\Del_{\mc{f}}, \Del_{\mc{g}}-\Shift_{\mc{f}}) = \Del_{\mc{g}\circ\mc{f}}.
            \end{split}
        \end{equation}
        Thus, we have shown that 
        \begin{equation}
                (\simul_{\mc{g}\circ\mc{f}})_{i,j}(\Vec{t}) \supseteq \{ \Vec{s} \in (\simul_{\mc{g}})_i(\Vec{r}) \mid  \Vec{r} \in (\simul_{\mc{f}})_j(\Vec{t}) \}
        \end{equation}
        Conversely, if $\vec{s}\in (\simul_{\mc{g}\circ\mc{f}})_{i,j}(\vec{t})$ then
        \begin{equation}\label{eq:s energy}
            H_S(\vec{s})-\Shift_{\mc{g}\circ\mc{f}}<\Del_{\mc{g}\circ\mc{f}},
        \end{equation}
        which implies that 
        \begin{equation}
            H_S(\vec{s})-\Shift_{\mc{g}} < \Del_{\mc{g}}
        \end{equation}
        and hence that $\vec{s}\in (\simul_{\mc{g}})_i(\vec{r})$ for some $i$ and $\vec{r}$. Note that this $\vec{r}$ can be obtained via 
        \begin{equation}\label{eq: r config}
            \vec{r} = \dec_{\mc{g}} \circ \vec{s} \circ \phys_{\mc{g}}
        \end{equation}
        Since $\vec{s}$ satisfies 
        \begin{equation}
            \Vec{s} \circ \phys_{\mc{g}}^{(m)} \circ \phys_{\mc{f}}^{(n)} = (\enc_{\mc{g}})_i^{(m)} \circ  (\enc_{\mc{f}})_j^{(n)} \circ \Vec{t},
        \end{equation}
        applying $\dec_{\mc{g}}$ to both sides and inserting \eqref{eq: r config} we find 
        \begin{equation}
            \vec{r}\circ \phys_{\mc{f}}^{(n)} = (\enc_{\mc{g}})_i^{(n)} \circ \vec{t}.
        \end{equation}
        Finally, \eqref{eq:s energy} also implies that 
        \begin{equation}
            H_R(\vec{r})-\Shift_{\mc{f}} < \Del_{\mc{f}}
        \end{equation}
        and hence $\vec{r}\in (\simul_{\mc{f}})_j(\Vec{t})$, 
        which shows that 
        \begin{equation}
            (\simul_{\mc{g}\circ\mc{f}})_{i,j}(\Vec{t}) \subseteq \{ \Vec{s} \in (\simul_{\mc{g}})_i(\Vec{r}) \mid  \Vec{r} \in (\simul_{\mc{f}})_j(\Vec{t}) \}. 
        \end{equation}
        This completes the proof of the first part,  \cref{eq:characterization sim3}. 
        
        We continue with the second part, proving that for low-energy configurations $\vert (\simul_{\mc{g}\circ\mc{f}})_{i,j}(\vec{t}) \vert = \degeneracy_{\mc{g}\circ\mc{f}} $.
        First, the right hand side of \eqref{eq:characterization sim3} can be rewritten as
        \begin{equation}
             \{ \Vec{s} \in (\simul_{\mc{g}})_i(\Vec{r}) \mid  \Vec{r} \in (\simul_{\mc{f}})_j(\Vec{t}) \}   = \ \bigcup_{\mathclap{\vec{r}\in (\simul_{\mc{f}})_j(\Vec{t})}}  \ (\simul_{\mc{g}})_i(\Vec{r}), 
        \end{equation}
        where by \cref{eq: r config} the union is disjoint so 
        \begin{equation}
            \vert (\simul_{\mc{g}\circ\mc{f}})_{i,j}(\Vec{t}) \vert = \ \sum_{\mathclap{\vec{r}\in (\simul_{\mc{f}})_j(\Vec{t})}} \ \vert (\simul_{\mc{g}})_i(\Vec{r}) \vert.
        \end{equation}
        As argued before, for $\vec{s}\in  (\simul_{\mc{g}})_i(\Vec{r})$ with $\vec{r} \in (\simul_{\mc{f}})_j(\Vec{t})$ and $H_T(\vec{t})<\Del_{\mc{f}}$ we have 
        $H_R(\vec{r})<\Del_{\mc{g}}$ so since $\mc{f}$ and $\mc{g}$ satisfy condition \ref{def:sim deg}
        \begin{equation}
            \vert (\simul_{\mc{f}})_j(\Vec{t}) \vert = \degeneracy_{\mc{f}}
        \end{equation}
        and for all $\vec{r} \in (\simul_{\mc{f}})_j(\Vec{t}) $
        \begin{equation}
            \vert (\simul_{\mc{g}})_i(\Vec{r}) \vert = \degeneracy_{\mc{g}},
        \end{equation}
        so in total 
        \begin{equation}
            \vert (\simul_{\mc{g}\circ\mc{f}})_{i,j}(\Vec{t}) \vert = \degeneracy_{\mc{f}} \cdot \degeneracy_{\mc{g}} = \degeneracy_{\mc{g}\circ\mc{f}}.
        \end{equation}

\medskip
\paragraph*{\ref{def:sim energy}.\ Matching energies.}
        By \cref{eq:characterization sim3}, if  $\vec{s} \in (\simul_{\mc{g}\circ\mc{f}})_{i,j}(\vec{t})$ then
         $\vec{s} (\simul_{\mc{g}})_i(\Vec{r})$ for some $\vec{r} \in (\simul_{\mc{f}})_j(\Vec{t})$
         and hence 
        \begin{equation} 
        \begin{split}
            H_S(\Vec{s}) &= H_R(\Vec{r}) +\Shift_{\mc{g}} \\
             &= H_T(\Vec{t}) + \Shift_{\mc{f}} + \Shift_{\mc{g}} \\
             &= H_T(\Vec{t}) + \Shift_{\mc{g}\circ\mc{f}}.
            \end{split}
        \end{equation} 
        What remains to be shown is that $\forall \Vec{s} \notin \simSet_{\mc{g}\circ\mc{f}}$, 
        \begin{equation}
            H_S(\Vec{s}) \geq \Del_{\mc{g}\circ\mc{f}} + \Shift_{\mc{g}\circ\mc{f}}.
        \end{equation}
        First, note that  $\Vec{s} \notin \simSet_{\mc{g}\circ\mc{f}}$ if and only if either $\Vec{s} \notin \simSet_{\mc{g}}$ or $\vec{s} \in (\simul_{\mc{g}})(\vec{r})$ with $\vec{r} \notin \simSet_{\mc{f}}$. 
        In the first case, we have
        \begin{equation}
        \begin{split}
            H_S(\vec{s})\geq \Del_{\mc{g}}+\Shift_{\mc{g}} =\Del_{\mc{g}} - \Shift_{\mc{f}} + \Shift_{\mc{g}\circ\mc{f}}  
            \geq \Del_{\mc{g}\circ\mc{f}} + \Shift_{\mc{g}\circ\mc{f}}.
        \end{split}
        \end{equation}
        In the second case, we have 
        \begin{equation}
        \begin{split}
             H_S(\Vec{s}) = H_R(\Vec{r}) + \Shift_{\mc{g}}\geq  \Del_{\mc{f}} + \Shift_{\mc{f}} + \Shift_{\mc{g}}  
             = \Del_{\mc{g}\circ\mc{f}} + \Shift_{\mc{g}\circ\mc{f}},
        \end{split}
        \end{equation}
        which finishes the proof.

\subsection{Simulations Can Be Added}\label{ssec:addition sim}

We prove \cref{thm:sim sum} by showing that the definition of $\mc{f}+\mc{g}$ satisfies the five conditions of \cref{def:sys simulation}.  

\medskip
\paragraph*{\ref{def:sim physical spins}.\ Disjoint physical spins.}
            We have to prove that 
            \begin{equation}\label{eq:sum phys spin cond}
                \phys_{\mc{f}+\mc{g}}^{(i)}(t) = \phys_{\mc{f}+\mc{g}}^{(j)}(t') \Rightarrow i=j \ \text{and} \ t=t'.
            \end{equation}
            By definition of $\phys_{\mc{f}+\mc{g}}$ we have 
            \begin{equation}
              \phys_{\mc{f}+\mc{g}}\vert _{V_{T_1}} = \phys_{\mc{f}}.
            \end{equation}
            Since by assumption \ref{eq:sum phys spins}, $\phys_{\mc{f}}, \phys_{\mc{g}}$ agree on the overlap of $T_1$ and $T_2$ 
            we also have
            \begin{equation}
         \phys_{\mc{f}+\mc{g}}\vert _{V_{T_2}}=\phys_{\mc{g}}.
            \end{equation}
            Hence, if either $t,t'\in V_{T_1}$ or $t,t'\in V_{T_2}$, $\phys_{\mc{f}+\mc{g}}$ satisfies condition \ref{def:sim physical spins}, since so do $\phys_{\mc{f}}$ and $\phys_{\mc{g}}$.
            We finish proof of condition \ref{def:sim physical spins} by proving that there is no third case, i.e.\ by proving that  $t,t'$ satisfy the left hand side of \cref{eq:sum phys spin cond} only if either $t,t'\in V_{T_1}$ or $t,t'\in V_{T_2}$. 
            More precisely, we show that 
            \begin{equation}\label{eq:phys disjoint 1}
                 \mr{Im}(\phys_{\mc{f}+\mc{g}}\vert_{V_{T_1}\setminus V_{T_2}})   \subseteq  V_{S_1}\setminus V_{S_2}
            \end{equation}
            and 
            \begin{equation}\label{eq:phys disjoint 2}
                \mr{Im}(\phys_{\mc{f}+\mc{g}}\vert_{V_{T_2}\setminus V_{T_1}})  \subseteq  V_{S_2}\setminus V_{S_1}.
            \end{equation}

            To prove \cref{eq:phys disjoint 1}, assume there exists $t\in V_{T_1}\setminus V_{T_2}$ with
            \begin{equation}
                \phys_{\mc{f}+\mc{g}}^{(i)}(t) =  \phys_{\mc{f}}^{(i)}(t)\in V_{S_1}\cap V_{S_2}.
            \end{equation}
            By assumption \ref{eq:sum phys image}, 
             there exists $r\in V_{T_1}\cap V_{T_2}$ and $l\in \{1, \ldots, k_1 \}$ with 
            \begin{equation}
                \phys_{\mc{f}}^{(l)}(r)= \phys_{\mc{f}}^{(i)}(t)
            \end{equation}
            Since $t\notin V_{T_2}$ it must hold that $t\neq r$ but since $\phys_{\mc{f}}$ has disjoint physical spins, i.e.\ satisfies \cref{def:sys simulation} \ref{def:sim physical spins} it must be that $t=r$. Thus there exist no such $t$.
            Similarly, \cref{eq:phys disjoint 2} holds since
            $\phys_{\mc{g}}$ satisfies \cref{def:sys simulation} \ref{def:sim physical spins}.
 
        \medskip
        \paragraph*{\ref{def:sim enc-dec}.\ Decode-encode compatibility.}
        We have to prove that 
        \begin{equation}
            \dec_{\mc{f}+\mc{g}} \circ (\enc_{\mc{f}+\mc{g}})_i = \mr{id}.
        \end{equation}
            By definition $\dec_{\mc{f}+\mc{g}} = \dec_{\mc{f}}$ and $\enc_{\mc{f}+\mc{g}}=\enc_{\mc{f}}\cap \enc_{\mc{g}}$.
            Since $\dec_{\mc{f}}, \enc_{\mc{f}}$ satisfy condition \ref{def:sim enc-dec}, so do
            $\dec_{\mc{f}+\mc{g}}, \enc_{\mc{f}+\mc{g}}$. 

    \medskip
        \paragraph*{\ref{def:sim disjoint enc}.\ Disjoint encodings.} 
        We have to prove that 
        \begin{equation}
            (\enc_{\mc{f}+\mc{g}})_i(s) = (\enc_{\mc{f}+\mc{g}})_j(s) \Rightarrow i = j.
        \end{equation}
             Since $\enc_{\mc{f}+\mc{g}} = \enc_{\mc{f}} \cap \enc_{\mc{g}}$ and $\enc_{\mc{f}}$ satisfies condition \ref{def:sim disjoint enc} so does $\enc_{\mc{f}+\mc{g}}$.
             
\medskip
\paragraph*{\ref{def:sim deg}.\ Constant degeneracy.} 
        We have to prove that if $H_{T_1+T_2}(\vec{t})<\Del_{\mc{f}+\mc{g}}$ 
        then $\vert \simul_{\mc{f}+\mc{g}}(\vec{t}) \vert = \degeneracy_{\mc{f}+\mc{g}}$.
        We split the proof into two parts.
        First, we characterize $(\simul_{\mc{f}+\mc{g}})_i(\vec{t})$ for low-energy configurations $\vec{t}$ and then, using this characterization, we prove that $\vert (\simul_{\mc{f}+\mc{g}})_i(\vec{t}) \vert = \degeneracy_{\mc{f}+\mc{g}}$.
            
        Let $q_S\coloneqq q_{S_1} = q_{S_2}$ and $\Vec{s}_1 \in \mc{C}_{S_1}$, $\Vec{s}_2 \in \mc{C}_{S_2}$, we define 
            $ \vec{s}_1+\vec{s}_2 \in \mc{C}_{S_1+S_2}$ by
            \begin{equation}
                 (\Vec{s}_1 +\Vec{s}_2)(v)  \coloneqq \begin{cases}
            \Vec{s}_1(v) \ \text{if} \ v \in V_{S_1}\\
            \Vec{s}_2(v) \ \text{else} .
            \end{cases}
            \end{equation}
        Furthermore, for $(\enc_{\mc{f}+\mc{g}})_i \in \enc_{\mc{f}+\mc{g}}$ 
        we define
        \begin{equation}
        \begin{split}
                &(\simul_{\mc{f}} + \simul_{\mc{g}})_i(\Vec{t})
                \coloneqq \{ \Vec{s}_1 + \Vec{s}_2 \mid \\
                & \hspace{0.75cm} \Vec{s}_1 \in (\simul_{\mc{f}})_{i_1}(\Vec{t}\vert_{V_{T_1}}),\  \Vec{s}_2 \in (\simul_{\mc{g}})_{i_2}(\Vec{t}\vert _{V_{T_2}})  \},
                \end{split}
            \end{equation}
            where $i_1,i_2$ are such that 
            \begin{equation}
          (\enc_{\mc{f}})_{i_1}=(\enc_{\mc{f}+\mc{g}})_i=(\enc_{\mc{g}})_{i_2}.
            \end{equation} 
            We now prove that if $H_{T_1+T_2}(\vec{t})<\Del_{\mc{f}+\mc{g}}$ then 
         \begin{equation}\label{eq:sim equality}
         (\simul_{\mc{f}}+\simul_{\mc{g}})_i(\Vec{t}) = (\simul_{\mc{f}+\mc{g}})_i(\Vec{t}).
            \end{equation}
            We start with the $\subseteq$ inclusion. 
            Let $\vec{s}_1+\vec{s}_2 \in (\simul_{\mc{f}}+\simul_{\mc{g}})_i(\Vec{t})$
            and $v \in V_{T_1}$ then
            \begin{multline}
                (\Vec{s}_1 + \Vec{s}_2) \circ \phys_{\mc{f}+\mc{g}} (v) =  \Vec{s}_1\circ \phys_{\mc{f}}(v) \\
                = (\enc_{\mc{f}})_{i_1} \circ \vec{t}(v)
                 = (\enc_{\mc{f}+\mc{g}})_i \circ \Vec{t}(v),
            \end{multline}
            where the last equality holds since  $\Vec{s}_1\in (\simul_{\mc{f}})_{i_1}(\Vec{t}\vert _{V_{T_1}})$.

            Similarly, for $v\in V_{T_2}\setminus V_{T_1}$,  
            \begin{multline}
                (\Vec{s}_1 + \Vec{s}_2) \circ \phys_{\mc{f}+\mc{g}} (v) = \Vec{s}_2\circ \phys_{\mc{g}}(v) \\
                =(\enc_{\mc{g}})_{i_2} \circ \vec{t}(v)
                 = (\enc_{\mc{f}+\mc{g}})_i \circ \Vec{t}(v),
            \end{multline}
            where for the first equality, we additionally used by \cref{eq:phys disjoint 2} in this case
            $\phys_{\mc{f}+\mc{g}}(v) \in V_{S_2}\setminus V_{S_1}$. 
        Therefore, in total 
            \begin{equation}
                (\Vec{s}_1 + \Vec{s}_2) \circ \phys_{\mc{f}+\mc{g}}  = (\enc_{\mc{f}+\mc{g}})_i \circ \Vec{t}.
            \end{equation}
            Next, if $v\in V_{S_1}\cap V_{S_2}$ then, by assumption,
            \begin{equation}
                v = \phys_{\mc{f}}^{(l)}(t) = \phys_{\mc{g}}^{(l)}(t)
            \end{equation}
            for some $t\in V_{T_1}\cap V_{T_2}$ and $l\in \{1, \ldots, k_1\}$ and hence 
            \begin{equation}
            \begin{split}
                &(\vec{s}_1 + \vec{s}_2)(v) = \vec{s}_1(v)= \vec{s}_1 \circ \phys_{\mc{f}}^{(l)}(t)  \\
                & \hspace{0.5cm} =(\enc_{\mc{f}+\mc{g}})_i^{(l)} \circ \vec{t} (t) 
                = \vec{s}_2 \circ \phys_{\mc{g}}^{(l)}(t) = \vec{s}_2 (v),
                \end{split}
            \end{equation}
             i.e.\ $\vec{s}_1$ and $\vec{s}_2$ agree on the overlap $V_{S_1} \cap V_{S_2} $ and therefore
            \begin{equation}
            \begin{split}
                        H_{S_1+S_2}(\vec{s}_1 + \Vec{s}_2)  
                        &= H_{S_1}(\vec{s_1})+H_{S_2}(\vec{s}_2)\\
                        &= H_{T_1}(\vec{t}\vert_{V_{T_1}})+ \Shift_{\mc{f}}+ H_{T_2}(\vec{t}\vert_{V_{T_2}})+\Shift_{\mc{g}} \\
                        &=H_{T_1+T_2}(\vec{t})+ \Shift_{\mc{f}+\mc{g}}\\
                        &< \Del_{\mc{f}+\mc{g}} + \Shift_{\mc{f}+\mc{g}}.
            \end{split}
            \end{equation}
           In total this proves that
         \begin{equation}\label{eq:sim incl 1}
                (\simul_{\mc{f}}+\simul_{\mc{g}})_i(\Vec{t}) \subseteq ((\simul_{\mc{f}+\mc{g}}))_i(\Vec{t}).
            \end{equation}
            
            Next, we prove the $\supseteq$ inclusion.
            Let $\vec{s}\in (\simul_{\mc{f}+\mc{g}})_i(\Vec{t})$ and $t\in V_{T_1}$. Then 
            \begin{equation}
            \begin{split}
               & \vec{s} \vert _{V_{S_1}} \circ \phys_{\mc{f}}(t) = \vec{s} \circ \phys_{\mc{f}+\mc{g}} (t)  \\
               & \hspace{0.5cm} = (\enc_{\mc{f}+\mc{g}})_i \circ \vec{t} = (\enc_{\mc{f}})_{i_1} \circ \vec{t}\vert_{V_{T_1}}.
            \end{split}
            \end{equation}
        Similarly, if $t\in V_{T_2}$ then 
            \begin{equation}
            \begin{split}
                &\vec{s} \vert _{V_{S_2}} \circ \phys_{\mc{g}}(t) = \vec{s} \circ \phys_{\mc{f}+\mc{g}} (t) \\
                & \hspace{0.5cm} =(\enc_{\mc{f}+\mc{g}})_i \circ \vec{t} = (\enc_{\mc{g}})_{i_2} \circ \vec{t}\vert_{V_{T_2}}.
            \end{split}
            \end{equation}
            Since by assumption $H_{S_1+S_2}(\vec{s}) < \Del_{\mc{f}+\mc{g}} + \Shift_{\mc{f}+\mc{g}}$ we in particular find that 
            \begin{equation}
                H_{S_1}(\vec{s}\vert_{V_{S_1}}) < \Del_{\mc{f}} + \Shift_{\mc{f}},
            \end{equation}
            where we used that by \cref{lem:sim spectrum} 
          \begin{equation}\label{eq:gs gamma}
                H_{S_2}- \Shift_{\mc{g}}>0.
            \end{equation}
            Similarly, we obtain 
            \begin{equation}
                H_{S_2}(\vec{s}\vert_{V_{S_2}}) < \Del_{\mc{g}} + \Shift_{\mc{g}}.
            \end{equation}
            Thus, we have shown that $\vec{s}\vert_{V_{S_1}}\in (\simul_{\mc{f}})_{i_1}(\vec{t}\vert_{V_{T_1}})$ and $\vec{s}\vert_{V_{S_2}}\in (\simul_{\mc{g}})_{i_2}(\vec{t}\vert_{V_{T_2}})$
            Finally, by definition
            \begin{equation}
                \vec{s} = \vec{s}\vert _{V_{S_1}} + \vec{s}\vert _{V_{S_2}},
            \end{equation}
            so we conclude 
            \begin{equation}
                \vec{s} \in (\simul_{\mc{f}}+\simul_{\mc{g}})_i(\vec{t})
            \end{equation} 
            and hence 
           \begin{equation}\label{eq:sim incl 2}
                (\simul_{\mc{f}}+\simul_{\mc{g}})_i(\Vec{t}) \supseteq (\simul_{\mc{f}+\mc{g}})_i(\Vec{t}).
            \end{equation}
            In contrast to \eqref{eq:sim incl 1}, this inclusion  is independent of the assumption $H_{T_1+T_2}(\vec{t})<\Del_{\mc{f}+\mc{g}}$.

            To finish the prove of condition \ref{def:sim deg}, using \cref{eq:sim equality}
            we prove that for low-energy configurations $\vert (\simul_{\mc{f}+\mc{g}})_i(\vec{t})\vert = \degeneracy_{\mc{f}+\mc{g}}$.
            First, by \cref{eq:sim equality} we have 
            \begin{equation}
               \vert (\simul_{\mc{f}+\mc{g}})_i(\vec{t}) \vert = \vert (\simul_{\mc{f}} + \simul_{\mc{g}})_i(\vec{t}) \vert .
            \end{equation}
            Note that two configurations
            \begin{equation}
            \vec{s}_1+\vec{s}_2, \vec{s}_1'+\vec{s}_2' \in (\simul_{\mc{f}} + \simul_{\mc{g}})_i(\vec{t})    
            \end{equation}
            are different if and only if either 
             $\vec{s}_1 \neq \vec{s}_1'$ or $\vec{s}_2 \neq \vec{s}_2'$.
            Thus, 
            \begin{equation}
            \begin{split}
                &\vert (\simul_{\mc{f}} + \simul_{\mc{g}})_i(\vec{t}) \vert  \\ 
                & \hspace{0.5cm} = \vert (\simul_{\mc{f}})_{i_1}(\Vec{t}\vert_{V_{T_1}}) \vert \cdot \vert (\simul_{\mc{g}})_{i_2}(\Vec{t}\vert _{V_{T_2}}) \vert
                = \degeneracy_{\mc{f}+\mc{g}}.
                \end{split}
            \end{equation}

\medskip         
\paragraph*{\ref{def:sim energy}.\ Matching energies.} 
        By \eqref{eq:sim incl 2}, if $\vec{s}\in (\simul_{\mc{f}+\mc{g}})_i(\vec{t})$ for any target configuration $\vec{t}$, then $\vec{s} = \vec{s}_1+\vec{s}_2$ for some $\vec{s}_1+\vec{s}_2 \in (\simul_{\mc{f}}+\simul_{\mc{g}})(\vec{t})$. 
        This implies  
        \begin{equation}
        \begin{split}
            H_{S_1+S_2}(\vec{s}) &= H_{S_1}(\vec{s}_1)+H_{S_2}(\vec{s}_2) \\
            &=H_{T_1}(\vec{t}\vert_{V_{T_1}}) + \Shift_{\mc{f}} + H_{T_2}(\vec{t}\vert_{V_{T_2}}) + \Shift_{\mc{g}}\\
            &=
            H_{T_1+T_2}(\vec{t})+\Shift_{\mc{f}+\mc{g}}
        \end{split}
        \end{equation}
    
        If, on the other hand, $\vec{s}\notin \simSet_{\mc{f}+\mc{g}}$, then either $\vec{s}\in (\simul_{\mc{f}}+\simul_{\mc{g}})_i(\vec{t})$ for $H_{T_1+T_2}(\vec{t})\geq\Del_{\mc{f}+\mc{g}}$ or 
        for all $i,\vec{t}$, $\vec{s}\notin (\simul_{\mc{f}}+\simul_{\mc{g}})_i(\vec{t})$.
        In the first case we have 
        \begin{equation}
            H_{S_1+S_2}(\vec{s}) = H_{T_1+T_2}(\vec{t})+\Shift_{\mc{f}+\mc{g}} 
            \geq \Del_{\mc{f}+\mc{g}} + \Shift_{\mc{f}+\mc{g}}.
        \end{equation}
        
        In the second case, either $\vec{s}\vert_{V_{S_1}}\notin \simSet_{\mc{f}}$ and/or $\vec{s}\vert_{V_{S_2}}\notin \simSet_{\mc{g}}$.
        If $\vec{s}\vert_{V_{S_1}}\notin \simSet_{\mc{f}}$, since $\mc{f}$ satisfies condition \ref{def:sim energy},
        $H_{S_1}(\vec{s}\vert_{V_{S_1}})-\Shift_{\mc{f}}\geq\Del_{\mc{f}}$ and hence
        \begin{equation}
        \begin{split}
            H_{S_1+S_2}(\vec{s}) - \Shift_{\mc{f}+\mc{g}} &\geq   \Del_{\mc{f}} + H_{S_2}(\vec{s}\vert_{V_{S_2}})-\Shift_{\mc{g}}   \\
            &>\Del_{\mc{f}} \geq \Del_{\mc{f}+\mc{g}},
            \end{split}
        \end{equation}
        using again that by \cref{lem:sim spectrum}, $H_{S_2}-\Shift_{\mc{g}}>0$.
        The other cases, $\vec{s}\vert_{V_{S_2}}\notin \simSet_{\mc{g}}$, can be treated similarly.
        In total, if $\vec{s}\notin \simSet_{\mc{f}+\mc{g}}$ then
        $H_{S_1+S_2}(\vec{s})>\Del_{\mc{f}+\mc{g}}+\Shift_{\mc{f}+\mc{g}}$.

\end{document}